%% file: thesis.tex
\documentclass[oneside, 12pt, letterpaper] {memoir}

\input{config.tex}
\input{macros.tex}

\setlrmargins{1in}{*}{*} 

\OnehalfSpacing 

\chapterstyle{crosshead}

\title{Proofs and Programs about Open Terms}

\author{Francisco Ferreira Ruiz\\ School of Computer Science\\ McGill University, Montr\'eal}

\date{December 2017}

\renewcommand{\maketitlehookd}
{
\vfill

\center

\textsc{A thesis submitted to McGill University in partial
  fulfillment of the requirements of the degree of Doctor of
  Philosophy}\\[1em]

Copyright \textcopyright\ 2017 by Francisco Ferreira Ruiz
}

\begin{document}
\pagenumbering{roman}

\maketitle
\thispagestyle{empty} \setcounter{page}{0}
\clearpage


\begin{abstract}
  \input{abstract}
\end{abstract}

\clearpage

\renewcommand{\abstractname}{R\'esum\'e}
\begin{abstract}
  \input{resume}
\end{abstract}

\clearpage

\hfill
\parbox{0.65\linewidth}{%
\centering
A Juan, a la Nanana y a la Chongola.
}




\clearpage

\section*{Acknowledgements}

\footnotesize
Oh, the places I’ve seen! But getting to this point in the adventure
would not have been possible or as much fun without the great people
that supported me. First, I want to thank my doctoral advisor Brigitte
Pientka whose passion and drive will always be examples to me. And
Stefan Monnier for his confidence and always relevant guidance. I want
to thank Prakash Panangaden for showing me the excitement that new
ideas can provide. Laurie Hendren, a person that makes every
interaction an insightful occasion, and whose laughter can cheer you
up even through two doors!

Andrew Cave and I started at the same time, I will be for ever in debt
to all the insights I get from each discussion. Andrew’s discipline
and empathy are a treasure to me. David Thibodeau, we shared so many
discussions, pair programming, so many ideas. I could have not asked
for a better friend to share these adventures with.

Annie Ying who offered support and understanding that you can only get
from shared experiences. Pablo Duboue a dear friend that always has an
interesting project or two and he was kind enough to let me
participate in some.

The many friends that kept me sane: Stefan Knudsen, Shawn Otis, Rohan
Jacob-Rao, Steven Thephsourinthone, Agata Murawska and François Thiré,
Milena Scaccia, Caroline Berger, Yam Chhetri, Gayane Petrosyan, Aïna
Linn Georges, Vincent Foley, Eric Lavoie, Larry Diehl, Alanna De
Bortoli, Mathieu Boespflug, Tao Xue, \dots

I do not have space to enumerate my Argentine chosen family, your love
and support keep me going every day of my life!

Getting to this point was difficult. But when the night was darkest I
knew that Juan Buzzetti would be unconditionally there, no questions
asked and that made all the difference in the world.

Finally, I want to thank my fantastic defense comittee: Prakash
Panangaden, Clark Verbrugge, Stefan Monnier and Joerg Kienzle for their
kind and insightful questions. Finally, special thanks to the external
reviewer James Cheney.

I could not have done it without all of you. Thank you, so very much!
\normalsize

\clearpage

\section*{Contributions of the Author}

\begin{itemize}
\item Chapter~\ref{chp:recon} is an extended version
  \citep{Ferreira:2014} where I am first author and I developed the
  ideas under my co-author supervision.
\item Chapter~\ref{chp:babybel} is an extended version
  \citep{Ferreira:2017} where I am also first author and where I
  developed the idea and prototype under my advisor's (also my
  co-author) supervision.
\item Chapter~\ref{sec:ctandtt} is a heavily extended version of
  \citep{Ferreira:2017b}, work that I did with David Thibodeau and
  Brigitte Pientka. The idea started as a consequence of the work on
  chapter~\ref{chp:babybel} and we then developed together the
  implementation and the ideas that lead to the theory. This chapter
  spells out the theory for the first time.
\end{itemize}

\clearpage
\tableofcontents
\clearpage
\listoffigures
\clearpage

\pagenumbering{arabic}


\input{introduction.tex}

\chapter{Reconstruction of Implicit Parameters}\label{chp:recon}

\input{reconstruction/rmacros.tex}
\input{reconstruction/reconstruction.tex}
\input{reconstruction/description.tex}
\input{reconstruction/target.tex}
\input{reconstruction/elaboration.tex}
\input{reconstruction/elbranch.tex}
\input{reconstruction/soundness.tex}
\input{reconstruction/related.tex}
\input{reconstruction/conclusion.tex}

\chapter[Contextual Types and Programming Languages]{Contextual Types and Programming Languages}\label{chp:babybel}
\chaptermark{In Programming Languages}

\input{babybel/bmacros.tex}
\input{babybel/introduction.tex} 
\input{babybel/example.tex} 
\input{babybel/smallml.tex} 
\input{babybel/sf.tex} 
\input{babybel/mlct.tex} 
\input{babybel/gadt.tex}
\input{babybel/implementation.tex} 
\input{babybel/poc.tex} 
\input{babybel/conclusion.tex} 

\chapter{Contextual Types and Type Theory}\label{sec:ctandtt}
\input {orca/omacros.tex}
\input{orca/intro.tex}
\input{orca/calculus.tex}
\input{orca/prototype.tex}
\input{orca/conclusion.tex}

\chapter{Conclusion}
\input{conclusion.tex}
\appendix

\input{reconstruction/extra.tex} 
\input{babybel/proof.tex}
\input{orca/beltran.tex} 

\input {bibliography} 

\end{document}

%% file: config.tex
\usepackage[hidelinks,draft=false]{hyperref}
\usepackage{natbib}
\usepackage{proof}
\usepackage{amsthm}
\usepackage{xspace}
\usepackage{tabularx}
\usepackage[utf8]{inputenc}
\usepackage{rotating}
\usepackage{textcomp} 

\setsecnumdepth{subsubsection}

\usepackage{amsmath}
\usepackage{amssymb}
\usepackage{stmaryrd}

\usepackage{pxfonts} 
\usepackage{sourcecodepro}

\usepackage[dvipsnames, svgnames]{xcolor} 
\usepackage{color}
\usepackage{colortbl}
   \definecolor{orange}{rgb}{0.75,0.5,0.1}
   \definecolor{myblue}{rgb}{0.8,0.85,1}
   \definecolor{mydarkblue}{rgb}{0,0.08,0.45}
   \definecolor{myotherblue}{rgb}{0,0.1,0.7}
   \definecolor{myotherblue}{rgb}{0.5,0,0.5}
   \definecolor{dimgrey}{rgb}{0.2,0.4,1.0}

\usepackage{cancel}

\usepackage{tikz}
\usetikzlibrary{arrows,positioning,patterns, shapes}

\usetikzlibrary{calc,
  arrows,decorations.pathmorphing,
  backgrounds,fit,positioning,shapes.symbols,chains}
\usetikzlibrary{decorations.pathreplacing}

\usepackage[final]{listings}

\lstdefinelanguage{Babybel}
{
  morekeywords={
    type, @type, @@@signature, def, open, let, rec, fun, function
  },
  basicstyle=\ttfamily,
  keepspaces=true,
  sensitive,
  morecomment=[l]{\%},
  morecomment=[n]{\%\{}{\}\%},
  morestring=[b]",
  showstringspaces=false
}[keywords,comments,strings]

\lstdefinelanguage{Beluga}
{
  morekeywords={and, block, case, of, mlam, fn, impossible, let, in, schema,
    some, rec, type, ctype, prop, stratified, inductive, coinductive, LF, if, then,
    else, total},
  keepspaces=true,
  sensitive,
  morecomment=[l]{\%},
  morecomment=[n]{\%\{}{\}\%},
  morestring=[b]"
}[keywords,comments,strings]

\lstdefinelanguage{Orca}
{
  morekeywords={
    set, spec, where, def, thm, data, ctx
  },
  basicstyle=\scriptsize\ttfamily,
  keywordstyle=\bfseries,
  keepspaces=true,
  sensitive,
  morecomment=[l]{(*)},
  morecomment=[n]{(*} {*)},
  morestring=[b]",
  showstringspaces=false
}[keywords,comments,strings]



%% file: macros.tex
\newcommand{\bnfas}{\mathrel{::=}}
\newcommand{\bnfalt}{\mathrel{\mid}}

\newcommand{\code}[1]{{\normalfont\texttt{#1}}}

\newcommand{\oft}{{\ensuremath{\,\mathrel{:}\,}}}

\newcommand{\checks}{\ensuremath{\mathrel{\Leftarrow}}}
\newcommand{\synths}{\ensuremath{\mathrel{\Rightarrow}}}

\newcommand{\wf}{\,\code{wf}} 

\newcommand{\tran}[1]{\ensuremath{\llbracket {#1} \rrbracket}}

\newcommand{\rl}[1]{\text{\small\code{#1}}}

\newcommand{\const}[1]{\ensuremath{\textbf{#1}}}

\newcommand{\cons}{\ensuremath{\mathrel{::}}}

\newcommand{\app}[2]{\ensuremath{{#1}\,{#2}}}

\newcommand{\lam}[1]{\ensuremath{\lambda\,{#1}.}}

\newcommand{\pit}[2]{\ensuremath{\Pi\,#1\oft#2.}}

\newcommand{\vs}{\\[1em]} 

\newcommand{\ssub}[2]{\ensuremath{[#1/#2]\,}}
\newcommand{\assub}[3]{\ensuremath{\app{\ssub{#1}{#2}}{#3}}}

\newcommand{\asisub}[2]{\ensuremath{\app{#1}{[#2]}}}

\theoremstyle{definition}
\newtheorem{thm}{Theorem}[section]
\newtheorem{lem}{Lemma}

\newcommand{\casee}[2]{\mbox{\code{case}}\,#1\,\mbox{\code{of}}\,#2}
\newcommand{\ibox}[1]{\ensuremath{[#1]}} 
\newcommand{\fne}[2]{\code{fn}\,#1 \Rightarrow #2}

\newcommand{\rece}[3]{\mbox{\code{rec}}\,#1\oft#2=#3}
\newcommand{\mlame}[2]{\mbox{\code{mlam}}\,#1\Rightarrow#2}

\newcommand{\oarr}[1]{\ensuremath{\overrightarrow{#1}}}

\newcommand{\wfctx}{\ \code{ctx}}
\newcommand{\wfsig}{\ \code{sig}}
\newcommand{\wftype}{\ \code{type}}
\newcommand{\wfkind}{\ \code{kind}}

\newcommand{\hsub}[4]{\ensuremath{[#3/#4]^{#1}_{#2}}}

\newcommand{\erased}[1]{\ensuremath{(#1)^-}}


%% file: abstract.tex
Formal deductive systems are very common in computer science. They are
used to represent logics, programming languages, and security systems.
Moreover, writing programs that manipulate them and that reason about
them is important and common. Consider proof assistants, language
interpreters, compilers and other software that process input
described by formal systems. This thesis shows that contextual types
can be used to build tools for convenient implementation and reasoning
about deductive systems with binders. We discuss three aspects of
this: the reconstruction of implicit parameters that makes writing
proofs and programs with dependent types easier, the addition of
contextual objects to an existing programming language that make
implementing formal systems with binders easier, and finally, we
explore the idea of embedding the logical framework LF using
contextual types in fully dependently typed theory. These are three
aspects of the same message: programming using the right abstraction
allows us to solve deeper problems with less effort. In this sense we
want: easier to write programs and proofs (with implicit parameters),
languages that support binders (by embedding a syntactic framework
using contextual types), and the power of the logical framework LF
with the expressivity of dependent types.

%% file: resume.tex
Les systèmes critiques comme les systèmes embarqués dans les avions
requièrent un niveau de sécurité élevé qui peut seulement être obtenu
par des systèmes formels. Les compilateurs certifiés ainsi que les
assistants de preuve sont des programmes qui manipulent et qui
raisonnent sur ces systèmes. Dans cette thèse, nous utilisons les
types contextuels afin de construire de tels outils permettant
notamment de raisonner sur des systèmes formels avec lieurs (binders).
En particulier, nous abordons les trois points suivants : la
reconstruction des paramètres implicites, ce qui simplifie l'écriture
des preuves et programmes avec types dépendants ; l'ajout d'objets
contextuels à un langage de programmation existant qui facilitent la
mise en œuvre des systèmes formels avec des lieurs ; et l'intégration
des types contextuels avec des types dépendants au-dessus du cadre
logique (logical frameworks) LF. Ces trois facettes reflètent le
message suivant : le choix d'une abstraction adéquate nous permet de
construire des outils capables de résoudre des problèmes plus
complexes plus facilement. Ceci se traduit par la construction de
langages de programmation utilisant des paramètres implicites,
supportant des lieurs (en intégrant un cadre syntaxique utilisant des
types contextuels), et qui utilisent la puissance du cadre logique LF
jointe à une théorie des types a la Martin-L\"of.

%% file: introduction.tex
\chapter{Introduction}

Proofs are fundamental in mathematics and computer science.  For the
purpose of this thesis we will consider that a proof is an irrefutable
argument in favour of a statement (i.e. a theorem). A formal proof is
one that is presented as a step by step argument in some foundational
system like ZFC set theory. The validity of a formal proof is reduced
to mechanically checking all of its steps in regards to the rules
established in its foundational system. On the one hand, formal proofs
are, as one might expect, very verbose and long to validate. On the
other hand, because they are a sequence of elementary steps in some
theory, they are easy to verify by computers.

A formal proof is straightforward to check but laborious to
construct. Therefore, proof assistant software was written as soon as
computers became fast enough to validate them. The
AUTOMATH~\citep{debruijn:1983} system was a trailblazer in this
field. The system was meant as a language to formalize mathematics and
it introduced many ideas that we take for granted today. It pioneered,
for example, ideas like using strongly typed $\lambda$-calculi as a
formalism. Today there are many proof assistants that are in use. They
are based on diverse formalisms, for example the
Isabelle~\citep{Paulson88cade} system that can use ZFC or Higher-Order
Logic, or systems based in type theory like Coq~\cite{CoqManual}, or
Agda~\citep{Norell:phd07} and others.

The job of a proof assistant is not only to validate proofs but also
to assist the user in describing the full formal proof. This is
usually achieved by allowing one to omit parts that can be
automatically reconstructed, invoking an automated decision procedure,
or tactics (a mechanism for defining programs that compute proofs).
Mechanizing a theorem is describing a theorem in a way that a proof
assistant is able to generate and validate the complete formal proof.

The history of proof assistants or interactive theorem provers while
not long is rich and eventful. A good reference is offered
by~\citet{Harrison:2014}. As mentioned before, there are proof
assistants based on several formalisms, among them several variations
on type theory. Some (non-exhaustive) examples and their formalisms
are:
\begin{itemize}
\item Church's simple theory of types~\citep{Church:1940}
  \begin{itemize}
  \item HOL4~\citep{Slind:2008},
  \item HOL light~\citep{Harrison09a} and
  \item Isabelle/HOL~\citep{Nipkow-Paulson-Wenzel:2002}),
  \end{itemize}

\item Martin-L\"of's type theory~\citep{Martin-Loef84a} (MLTT)
  \begin{itemize}
  \item the Nuprl~\citep{Constable:1986} proof assistant that is based on an extension of MLTT.
  \item Coq~\citep{CoqManual} that uses an extension of the calculus of
    constructions~\citep{Coquand:1988},
  \item Agda~\citep{Norell:phd07} roughly
    based on UTT~\citep{Luo:1994}
  \end{itemize}

\item the logical framework LF~\citep{Harper93jacm}
  \begin{itemize}
  \item Twelf~\citep{Pfenning99cade} that uses logic programming with
    LF definitions
  \item Beluga~\citep{Pientka:CADE15} that implements a reasoning
    language based on first order logic with induction on LF
    specifications.
  \end{itemize}
\end{itemize}

In this thesis, we will mostly discuss systems that take advantage of
the Curry-Howard correspondence~\citep{Howard80}, where propositions
correspond to types and proofs correspond to terms (some proof
assistants based on this idea are Agda, Coq and Beluga). In these
systems, higher-order statements become dependent functions, that is,
functions where the resulting type depends on the value of the
parameter. Theorems are then represented as function types and proofs
by induction are represented by well-founded recursion and pattern
matching for case analysis. Because type theory can be seen as a
programming language~\citep{NordstroemPetersonSmith90a} these proof
assistants can be used simultaneously as provers and as programming
environments. This is a key advantage of constructive logics and type
theory in particular, and it is instrumental to the subject matter of
this thesis.

Proof assistants based on MLTT or extensions of the Calculus of
Constructions like Coq are highly expressive and significant
mathematical results have been formalized in them. For example the
formalization of the odd order theorem by~\citet{Gonthier:2013} not
only fully mechanizes the existing proof, but in order to do so it
provides a library of algebraic definitions and theorems. With many
logics, expressivity is not a problem; however, having the right
setting and properties is crucial. For example, inductive types can be
internally encoded in the pure calculus of
constructions~\citep{Pfenning:1990}. Nevertheless, the calculus of
inductive constructions was designed to have inductive types as a
first-class construct in order to have nicer computational behaviour
and better usability~\citep{Paulin-Mohring:TLCA93}. A theme of this
work is to take advantage of good representations (namely, using the
logical framework LF and contextual types) to allow for a
straightforward representation of theorems and programs about
structures with binders and hypothetical judgments.

The kind of systems we want to represent are deductive systems. These
are systems specified by axioms and deduction rules and they often
include the notion of variables and binders. This class of systems is
very large, and includes, for example: programming languages, logics,
and their meta-theory together with their implementation. When dealing
with deductive systems we distinguish two related activities: the
first one is implementing formal proofs about deductive systems (e.g.:
a language is type safe, or a logic is normalizing). When talking
about some object language (i.e.: the programming language or logic
under study) we construct proofs about aspects of its meta-theory. We
refer to this as \emph{reasoning about the specification of the
  language} or working on the meta-theory of the specified deductive
system. Second, we want to specify and perform computations over these
systems, for example: compilers, evaluators and normalizers. We will
refer to this as \emph{computing with specification language} or
simply writing programs that manipulate objects in the deductive
systems.

Deductive systems that represent programming languages and logics
typically have the idea of bound variables, that is some terms
introduce new variables. When writing proofs or programs about such
specifications one usually needs to recursively inspect the expression
under a binder. Consider the first order formula: $\forall x . P(x)$
where one finds the predicate $P$ that may contain free occurrences of
$x$ a variable that is bound outside by the universal quantifier. In
this situation the sub-term might have free variables (i.e.: a
variable bound outside of the term). Such terms are called open terms.
Implementing and reasoning about open terms of systems with variables
and binders is a common activity for computer scientists. Some typical
examples of this are implementing new programming languages, reasoning
about the meta-theory of languages and logics, and implementing proof
assistants. Therefore, when working with open objects (terms with free
variables) one finds oneself needing to represent variables, binders
and substitutions.

The meaning of free variables cannot be ignored. One solution commonly
used while reasoning both on paper and formally with a proof
assistant, is to track free variables with a context that gives
meaning to all the free variables in the term. Here we say the context
binds the variables. This thesis explores a particular approach to
this problem. Concretely, we will explore the use of contextual
types~\citep{Nanevski:ICML05}, that represent the type of an
expression together with the context that gives meaning to all its
free variables, to program and reason about higher-order abstract
syntax (HOAS). Higher-order representations like HOAS, are used to
represent binders reusing the function space, and function application
to represent substitution, this approach is exemplified by the logical
framework LF. While contextual types and the logical framework LF
provide support for specifying formal systems, one can reason about
these structures by implementing pattern matching and recursion over
them. This has been done before, for example languages like
Beluga~\citep{Pientka:CADE15} and other languages like
Delphin~\citep{Poswolsky:DelphinDesc08}. Actually, Delphin does not
have an idea of first class contexts but it does manipulate terms
using HOAS.

This thesis shows that contextual types can be used to build tools for
convenient implementation and reasoning about deductive systems with
binders. We discuss the specification, reasoning and programming with
open terms from the following points of view:
\begin{itemize}
\item How to reconstruct types in dependently typed systems. This is
  important to make the system accessible. Otherwise dependently typed
  programs are very verbose. We describe a formal algorithm for the
  sound reconstruction of implicit parameters (i.e.: parameters that
  the user does not write and the system infers) in systems with
  dependent types and a rich index domain like Beluga.

\item How to integrate contextual types in an industrial strength
  functional programming language (in this case
  OCaml~\citep{ocamlManual}) to allow for type safe programming with
  binders.

\item How to extend Martin-L\"of's type theory~\citep{Martin-Loef84a}
  (MLTT) as the reasoning/programming language for a system with
  contextual types.
\end{itemize}

These three points of view are intimately related. The objective of
implicit parameter reconstruction is to simplify writing proofs about
specifications and dependent pattern matching (which can be seen as
case analysis). Integrating contextual object with existing
programming languages allows for ease in writing programs that
manipulate objects with variables and binders. To conclude, combining
contextual objects and specifications with fully dependently typed
language allows for implementing proofs about specifications,
implementing programs over specifications, and particularly for
implementing proofs about said programs.

\section*{Main Contributions}

\subsection*{Reconstruction of Implicit Parameters}

This chapter presents the design of a source language with index types
and dependent pattern matching together with an elaboration phase that
reconstructs omitted arguments. We differentiate between implicit
arguments; those that the user does not write, and explicit arguments,
which must be provided by the user. This language and its elaboration
describe the corresponding reconstruction of the computational
language of the Beluga system.

The elaboration is type directed and it infers omitted arguments to
produce a closed well-typed program. Notably, we describe the
reconstruction of pattern matching expressions. Specifically, using
the design of pattern matching inspired by Beluga that provides nested
dependent pattern matching without type annotations. This is an
important distinction from other systems that either do not provide
nested dependent matching statements (such as
Agda~\citep{Norell:phd07} or Idris~\citep{Brady:JFP13}) or require
annotations for the return type (as the Coq~\citep{CoqManual} proof
assistant).

Finally, we prove that this is sound, that is, that the successful
elaboration of a term implies that it is well-typed in the target
language. Part of this work was published in \citet{Ferreira:2014}.

\subsection*{Contextual Types and Programming Languages}

Implementing contextual types in a current and existing programming
language brings some of the power of Beluga and HOAS or $\lambda$-tree
definitions~\citep{Miller:1999} to existing (simply-typed) functional
programming languages. Using contextual types and a syntactic
framework (SF) based on modal S4~\citep{Nanevski:ICML05,Davies:ACM01},
programmers can manipulate open objects by pattern matching with a
type system that guarantees the binders do not escape their scopes. We
show that this language extension can be translated into a language
that supports Generalized Abstract Data Types (i.e.: GADTs) using a
deep embedding of SF.

The other contribution is Babybel, a prototype implementation of these
ideas. Babybel is implemented as a syntax extension of the OCaml
language. It takes advantage of OCaml's type system to ensure that the
translation is type preserving and allows the users to take advantage
of GADTs to represent some inductive predicates over syntax, like
context relations. Finally, part of this work was published
in~\citet{Ferreira:2017}.

\subsection*{Contextual Types and  Type Theory}

The main contribution of this chapter is a calculus that integrates
the logical framework LF for specifications with Martin-L\"of's type
theory as a reasoning/computation language. This is achieved using
contextual types and it can be seen as an extension of the technique
for adding contextual types to programming languages. Morally it
presents an extended version of the Beluga language with fully
dependent typing. MLTT does not have a phase separation between
type-checking and evaluation allowing for the interleaving of
computation and specification. Additionally, because MLTT allows for
reasoning about functions, the theory permits writing computations and
proving properties about the computations.

A further contribution, is the Orca prototype that implements these
ideas. The design of the Orca language also provides an interesting
type directed syntax reconstruction to be able to disambiguate between
computational terms and specification terms which provides a nicer
user experience. Some of these ideas were presented at \citet{Ferreira:2017b}.

\chapter{Deductive Systems with Binders}\label{chp:dedsys}

\section{Introduction}

Deductive systems presented using axioms and deduction rules are
designed to formally describe logics, programming languages, and proof
assistants. Therefore they are widely used in the study of programming
languages and their properties. Examples abound: from the description
of modal logic systems in~\citet{Pfenning01mscs}, to the presentation
of the core calculus of a reactive programming
language~\citep{Cave:POPL14}, to the specification of the addition of
dependent types to Haskell, a real world general purpose
language~\citep{Weirich:2017}

Most of the formal systems that we discuss in this thesis require
variables and binders (the place where new variables are introduced).
This a delicate aspect of the presentation of a formal system, where
one must be aware of issues like variables not escaping their scopes,
comparing terms up-to the renaming of bound variables (i.e.:
$\alpha$-equivalent terms), and that the substitution operation shall
not capture free variables. Consider for example, the untyped
$\lambda$-calculus, its syntax is:

\begin{displaymath}
  \begin{array}{rrclr}
    \mbox{Terms} & M, N & \bnfas & x & \mbox{variables}\\
    & & \bnfalt & \lam x M & \mbox{function abstraction}\\
    & & \bnfalt & \app M N & \mbox{application}
  \end{array}
\end{displaymath}

Note that $x$ is a name from a set that contains countably many
distinct variable names. An important operation in the
$\lambda$-calculus is that of substitution, that allows the instantiation of
a variable with a term. We write $\assub M x N$ to say replace
every occurrence of variable $x$ for term $M$ in term $N$. The usual
definition is done inductively on the structure of $N$ in the
following way:

\begin{displaymath}
  \begin{array}{rrcll}
    \ssub M x& y & = & M &\mbox{when $x = y$} \\
    \ssub M x& y & = & y &\mbox{when $x \neq y$} \\
    \ssub M x& \lam z N & = & \lam z {\assub M x N} &\mbox{with $z$ not free in $M$}\\
    \ssub M x& \app N {N'} & = & \app {(\assub M x N)} {(\assub M x {N'})} \\
  \end{array}
\end{displaymath}

The definition is straightforward except that at first sight, the
reader might think that it is not a total operation. After all, in the
abstraction case, what is one supposed to do if $z$ does appear free
in $M$? (i.e.: a variable is free if its binder is not part of the
term). The answer is that the substitution needs to be applied to a
term where $z$ has been renamed (i.e.: the substitution continues with
an $\alpha$-equivalent version of the term). This issue is the central
idea in capture avoiding substitution. While this is usually left
implicit in a description, it remains an issue in implementations and
mechanized formal presentations.

\section{The Representation of Binders}

Representing variables as string and adding side conditions when
necessary for substitution works well enough for paper or black-board
presentations. However, names as strings are less common in computer
implementations. For example, N. G. de~Bruijn, when designing
Automath~\citep{deBruijNAMS:1991} (one of the first proof assistants)
proposed a nameless representation~\citep{de1972lambda} where a
variable is the distance, in number of binders, to the place where the
variable is bound. On one hand this eliminates issues with name
capture, but on the other hand, humans usually find these terms very
difficult to understand. This idea was later refined
by~\citet{Altenkirch:TLCA93} as well-scoped de~Bruijn indices where
they use dependent types to enforce scoping invariants.

An alternative approach is using nominal logic~\citep{Gabbay:LICS99}
or categorical approaches such as~\citep{Fiore:2008}, to give a
precise mathematical definition the ideas of $\alpha$-equality and
name capture.

A final approach, and the one that is discussed in this thesis the
most, is using different versions of the $\lambda$-calculus as
representation frameworks and reusing their function space to
introduce binders. This allows for a simple implementation of
substitution as one simply reuses the notion of substitution of the
underlying calculus. This was first used by Church in
\citep{Church:1940}, but the logical framework LF~\citep{Harper93jacm}
takes full advantage of the idea. LF proposes a dependently typed
typed $\lambda$-calculus as a representation logic that is able to
encode syntax together with judgments. This technique is usually known
as Higher-Order Abstract Syntax (HOAS) because of the reutilization of
the the function space to implement binders. A related technique with
similar presentation is the use of
$\lambda$-trees~\citep{Miller:1999}.

\section{Higher Order Abstract Syntax}

The idea of representing binders using the function space of
$\lambda$-calculi to represent binders and dependent types to
represent judgments is a key insight provided by LF and allows for the
mechanization of deductive systems in a direct and high-level way. The
use of the function space of LF frees the user from thinking about the
representation of binders, the implementation of substitution and even
proving some substitution lemmas as all this infrastructure is
inherited from the framework itself.
\filbreak

In regular programming languages (that provide higher-order functions)
it is possible to represent binders using the function space of the
language, like in this OCaml example:

\lstset{language=Caml}
\begin{lstlisting}
type exp =
  | Lam of (exp -> exp) (* abstractions *)
  | App of exp * exp (* applications *)

let omega = Lam (fun x -> App (x, x)) (* little omega *)

let rec eval : exp -> exp = function
  | Lam f -> Lam f (* abstractions are values *)
  | App (e1, e2) ->
     begin match eval e1 with
     | Lam f -> eval (f e2) (* function application is substitution *)
     | stuck -> App (e1, e2)
     end
\end{lstlisting}

This short example is enough to show two of the crucial aspects of
HOAS. The first appears in the declaration of the type
\lstinline!exp!, where the constructor for abstractions uses a
function to represent the body of the $\lambda$-expression that
contains a variable. For that reason it is introduced as a
higher-order function. The second aspect is that in this setting the
function application represents substitution, and we can notice this
when reducing applications. To perform the substitution it suffices to
apply the function because binders are implemented by functions.

This short example is also enough to show two important problems: the
first is that the OCaml function space is too rich and there are many
functions that do not represent terms in the $\lambda$-calculus. For
example:
\begin{lstlisting}
let exotic : exp =
  Lam (fun x -> match x with App (_, e) -> e | e -> e)
\end{lstlisting}
This exotic term pattern matches on the shape of its argument
(dropping all the terms in function position) while in the
$\lambda$-calculus variables can only be substituted for a term and
cannot analyze the shape of any term. Thus, this representation, while
it might be convenient, is not adequate and it is not a replacement
for LF. The right framework should provide an adequate representation
(i.e. a function space that is weak enough so that no exotic terms
exist) and should allow for the intensional inspection of the
represented functions. LF is designed with these features in mind
(together with the ability to represent judgments and completely
represent deductive systems).

The other problem is that because functions operate as black boxes, it
is not easy to operate on open terms as the only operation on
functions is application, and thus the only way to get to the body of
the abstraction is to perform a substitution. The issue is that the
function space in OCaml behaves like an extensional function (i.e.:
one can only observe the result of a function application). When
operating with open terms one would need an intensional function space
that gives access to the structure of the implementation of the
function~\citep{Pfenning:LICS01}.

One possible solution is to separate the language that describes
computation from the language that is used for specifications. This
approach was started by \citet{Schurmann:TCS01} when they proposed a
computation language with primitive recursion and a simply typed
language for specifications. This idea was pushed forward in
Delphin~\citep{Poswolsky:DelphinDesc08} when they added support for
dependent types and the logical framework LF. Finally, the
Beluga~\citep{Pientka:POPL08} system uses contextual types together
with logical framework LF. Contextual types describe potentially open
terms (terms with free variables) together with their contexts (that
provide a binding to all the free variables of a term). This allows
for proofs and programs about open objects and that inspect terms
under binders by keeping track of their contexts. This thesis explores
some extensions and implementation concerns for these ideas.

\section{Logical Framework LF}\label{sec:LF}

The logical framework LF provides the means to represent the syntax
and judgments of deductive systems through the use of a dependently
typed calculus related to Martin-L\"of's system of
arities~\citep{NordstroemPetersonSmith90a}.

In LF, judgment and syntactic categories are represented by types, and
constructors represent respectively the inference rules and terms of
the object language. For example Figure~\ref{fig:stlclf} shows how to
represent the simply typed $\lambda$-calculus (STLC) using the
concrete syntax of the Beluga system. There are two types, one for
each syntactic category (i.e.: \lstinline~tp~ for types and
\lstinline~tm~ for term) and the typing judgment is represented by a
dependent type (\lstinline~oft~). In the representation of $lambda$
expressions in terms, notice how binding is represented by a function
(HOAS). The most interesting case is the \lstinline~oft~ judgment that
relates a term to its type. The judgment contains three constructors
that correspond to the typing rules for applications, abstractions and
the constant respectively. Because of the use of HOAS, there is no
need for a rule for variables, as the typing assumptions for variables
are added to the context by the rule \lstinline~t-lam~.

\begin{figure}
  \begin{lstlisting}[language=Beluga]
LF tp : type =
| b :  tp
| arr : tp -> tp -> tp
;

LF tm : type =
| app : tm -> tm -> tm
| lam : (tm -> tm) -> tm
| c : tm
;

LF oft : tm -> tp -> type =
| t-app : oft M (arr S T) -> oft N S -> oft (app M N) T
| t-lam : ({x:tm} oft x S -> oft (M x) T) ->
          oft (lam M) (arr S T)
| t-c : oft c b
;
  \end{lstlisting}
  \caption[The STLC in LF]{The simply typed $\lambda$-calculus in LF}
  \label{fig:stlclf}
\end{figure}

The logical framework LF is a dependently typed theory related to the
$\lambda P$ vertex of the $\lambda$-cube~\citep{Barendregt:1992}. The
particular presentation we use is often called Canonical LF because it
only allows for normal (i.e.: canonical) forms. Regular substitution
might introduce non-normal forms, so substitution is done in an
``hereditary'' way. Hereditary substitutions continue reducing to
avoid the introduction of non-normal forms. This technique was
introduced in the context of Concurrent LF by \citep{Watkins02tr} and
\citep{Cervesato02tr}. In particular this presentation is based on
work by \citet{HarperLicata:JFP07}.

The syntax is as follows:
\begin{displaymath}
  \begin{array}{lrrl}
    \mbox{Kinds} & K & \bnfas & \code{type} \bnfalt \pit x A K\\
    \mbox{Base Types} & P & \bnfas & \const a \bnfalt \app P M\\
    \mbox{Types} & A, B & \bnfas & P \bnfalt \pit x A B\\
    \mbox{Normal Terms} & M & \bnfas & \lam x M \bnfalt R\\
    \mbox{Neutral Terms} & R & \bnfas & \const c \bnfalt x \bnfalt \app R M\\
    \mbox{Contexts} & \Psi & \bnfas & \cdot \bnfalt \Psi,x\oft A\\
    \mbox{Signature} & \Sigma & \bnfas & \cdot \bnfalt \Sigma, \const a \oft K \bnfalt
                                         \Sigma, \const c \oft A \\
  \end{array}
\end{displaymath}

Kinds classify the types and are either $\code{type}$ for
non-dependent kinds (used to represent our syntactic categories) and
$\pit x A K$ for judgments. Base types are either atomic types
$\const a$ or type constructors applied to terms $\app P M$. Types
classify terms and are: base types $P$, or $\pit x A B$ a dependent
function space, when $x$ does not appear in $B$ we write $A\to B$ to
indicate the simply typed function space. Terms are split between
normal and neutral terms, with the objective of preventing $\beta$
reducible terms. So normal terms contain function abstractions
$\lam x M$ and neutral terms. And neutral terms are constructors
$\const c$, variables $x$ bound in abstractions or dependent
functions, and applications of neutral terms to normal terms
$\app R M$. Finally, contexts $\Psi$ store typing assumptions for
bound variables, and the signature $\Sigma$ contains the user
definitions (that represent the object language as in
Figure~\ref{fig:stlclf}).

Well typed LF terms are defined by these judgments:
\begin{itemize}
\item Well formed signatures, contexts and kinds and types are given by:
  \begin{itemize}
  \item $\boxed{\vdash \Sigma\ \code{sig}}$: $\Sigma$ is a valid signature.
  \item $\boxed{\vdash_\Sigma \Psi \wfctx}$: $\Psi$ is a well formed context in signature $\Sigma$.
  \item $\boxed{\Psi \vdash_\Sigma K \wfkind}$: Kind $K$ is well formed in context $\Psi$.
  \item $\boxed{\Psi \vdash_\Sigma A \wftype}$: Type $A$ is well formed in context $\Psi$.
  \end{itemize}
\item Well kinded base types and well typed terms are given by:
  \begin{itemize}
  \item $\boxed{\Psi \vdash P \synths K}$: Base type $P$ synthesizes kind $K$ in $\Psi$.
  \item $\boxed{\Psi \vdash M \checks A}$: Normal term $M$ checks against type $A$ in $\Psi$.
  \item $\boxed{\Psi \vdash R \synths A}$: Neutral term $P$ synthesizes type $A$ in $\Psi$.
  \end{itemize}
\end{itemize}

We assume signatures are well formed, and we omit them in the rules
because they do not change during typing. Similarly, for contexts, we
remark that the rules only extend contexts with well formed
assumptions, so we assume that one starts with a well-formed context
(empty or otherwise) and the rules preserve that property. We refer to
\citet{Harper93jacm} and \citet{HarperLicata:JFP07} for a more
detailed discussion of these issues.
\begin{figure}
  \centering
  \begin{displaymath}
    \begin{array}{c}
      \multicolumn{1}{l}
      {\boxed{\vdash \Sigma \wfsig}: \mbox{$\Sigma$ is a valid signature.}}\vs

      \infer[\rl{s-empty}]
      {\vdash \cdot \wfsig}
      {}

      \quad

      \infer[\rl{s-type}]
      {\vdash \Sigma, \const c \oft A \wfsig}
      {\vdash \Sigma \wfsig
      & \vdash_\Sigma A \wftype}

      \vs

      \infer[\rl{s-con}]
      {\vdash \Sigma, \const a \oft K \wfsig}
      {\vdash \Sigma \wfsig
      & \cdot \vdash_\Sigma K \wfkind}

      \vs

      \multicolumn{1}{l}
      {\boxed{\vdash_\Sigma \Psi \wfctx}: \mbox{$\Psi$ is a well formed context in signature $\Sigma$.}}\vs

      \infer[\rl{c-empty}]
      {\vdash_\Sigma \cdot \wfctx}
      {}

      \quad

      \infer[\rl{c-hyp}]
      {\vdash_\Sigma \Psi, x\oft A\wfctx}
      {\vdash_\Sigma \Psi \wfctx
      & \cdot \vdash A \wftype}
    \end{array}
  \end{displaymath}
  \caption{Well formed LF signatures and contexts}
  \label{fig:wflf}
\end{figure}

\begin{figure}
  \centering
  \begin{displaymath}
    \begin{array}{c}
      \multicolumn{1}{l}
      {\boxed{\Psi \vdash_\Sigma K \wfkind}: \mbox{Kind $K$ is well formed in context $\Psi$.}}\vs

      \infer[\rl{k-type}]
      {\Psi\vdash_\Sigma \code{type} \wfkind}
      {}

      \quad

      \infer[\rl{k-pi}]
      {\Psi\vdash_\Sigma \pit x A K \wfkind}
      {\Psi\vdash_\Sigma A \wftype
      & \Psi,x\oft A\vdash_\Sigma K \wfkind}

      \vs

      \multicolumn{1}{l}
      {\boxed{\Psi \vdash_\Sigma A \wftype}: \mbox{Type $A$ is well formed in context $\Psi$.}}\vs

      \infer[\rl{t-base}]
      {\Psi\vdash_\Sigma P \wftype}
      {\Psi\vdash P \synths \code{type}}

      \quad

      \infer[\rl{t-pi}]
      {\Psi\vdash_\Sigma \pit x A B \wftype}
      {\Psi\vdash_\Sigma A \wftype
      & \Psi,x\oft A\vdash_\Sigma B\wftype}
    \end{array}
  \end{displaymath}
  \caption{Well formed LF types and kinds}
  \label{fig:wftypkind}
\end{figure}

\begin{figure}
  \centering
  \begin{displaymath}
    \begin{array}{c}
      \multicolumn{1}{l}
      {\boxed{\Psi \vdash P \synths K}: \mbox{Base type $P$ synthesizes kind $K$ in $\Psi$.}}\vs

      \infer[\rl{p-con}]
      {\Psi\vdash \const a \synths K}
      {a\oft K \in \Sigma}

      \quad

      \infer[\rl{p-app}]
      {\Psi\vdash \app P M \synths \hsub k {\erased A} M x K}
      {\Psi\vdash P \synths \pit x A K
      & \Psi\vdash M\checks A}

      \vs

      \multicolumn{1}{l}
      {\boxed{\Psi \vdash M \checks A}: \mbox{Normal term $M$ checks against type $A$ in $\Psi$.}}\vs

      \infer[\rl{t-lam}]
      {\Psi\vdash \lam x M \checks \pit x A B}
      {\Psi,x\oft A\vdash M \checks B}

      \quad

      \infer[\rl{t-neu}]
      {\Psi\vdash R \checks A}
      {\Psi\vdash R \synths A}

      \vs

      \multicolumn{1}{l}
      {\boxed{\Psi \vdash R \synths A}: \mbox{Neutral term $P$ synthesizes type $A$ in $\Psi$.}}\vs

      \infer[\rl{t-con}]
      {\Psi\vdash \const c \synths A}
      {\const c \oft A \in \Sigma}

      \quad

      \infer[\rl{t-var}]
      {\Psi\vdash x \synths A}
      {x \oft A \in \Psi}

      \vs

      \infer[\rl{t-app}]
      {\Psi\vdash \app R M \synths \hsub a {\erased A} M x B}
      {\Psi\vdash R \synths \pit x A B
      & \Psi\vdash M \checks A}
    \end{array}
  \end{displaymath}
  \caption{Typing LF}
  \label{fig:typlf}
\end{figure}

Figures~\ref{fig:wflf}, \ref{fig:wftypkind}, and~\ref{fig:typlf} show
the deduction rules for each judgment. They are needed to establish
when a term is well formed and well typed. Logics and deductive
systems are represented by the canonical forms (i.e.: normal forms) of
LF terms.
There are at least
two approaches to compare the normal forms:
\begin{itemize}
\item Defining a type directed equality to compare terms up-to
  $\beta$- reduction and $\eta$-expansion\citep{Harper03tocl}.
\item Defining Canonical LF where all terms are in canonical form, and
  using hereditary substitutions~\citep{Watkins02tr} to preserve this
  invariant. \citet{HarperLicata:JFP07} explain this approach and show
  a proof of adequacy for these encodings. As we already mentioned,
  this is the approach we present here.
\end{itemize}

\begin{figure}
  \centering
  \begin{displaymath}
    \begin{array}{rcl}
      \erased{\pit x A B} & = & \erased{A} \to \erased{B} \\
      \erased{\app P M} & = & \erased P \\
      \erased{\const a} & = & \const a
    \end{array}
  \end{displaymath}
  \caption{Type Approximations}
  \label{fig:typapprox}
\end{figure}

\filbreak

We define hereditary substitutions for kinds, types and terms. This
requires defining six operations:
\begin{itemize}
\item $\boxed{\hsub k \alpha M x \, K = K'}$: Hereditary substitution in kinds.
\item $\boxed{\hsub a \alpha M x \, A = A'}$: Hereditary substitution in types.
\item $\boxed{\hsub p \alpha M x \, P = P'}$: Hereditary substitution in base types.
\item $\boxed{\hsub m \alpha M x \, M' = M''}$: Hereditary substitution in normal terms.
\item $\boxed{\hsub r \alpha M x \, R = M' \oft \alpha'}$: Hereditary
  substitution in a neutral term producing a normal term.
\item $\boxed{\hsub r \alpha M x \, R = R'}$: Hereditary substitution in a neutral term producing a neutral term.
\end{itemize}
The superscript in the substitution indicates the syntactic domain
(and it can often be omitted) where it applies, and the subscript is a
type approximation of the resulting types that guides the process and
the termination argument. The type approximation is defined in
Figure~\ref{fig:typapprox}. The computation rules for the hereditary
substitution in types are presented in Figure~\ref{fig:heredsubst1}
and Figure~\ref{fig:heredsubst2} presents the rules for terms. The
definition is similar to regular substitution, but it critically
differs when applying a substitution to a neutral term produces a
normal term. In this situation, if the neutral term is an application
then the substitution needs to continue until a normal form is reached
(otherwise an unrepresentable $\beta$-reduction would be created),
notice how the type approximations get progressively smaller.

\begin{figure}
  \centering
\newcommand{\tsub}[1]{\ensuremath{\hsub #1 \alpha M x}}
\begin{displaymath}
  \begin{array}{rlcl}
    \multicolumn{4}{l}{
    \boxed{\hsub k \alpha M x \, K = K'} : \text{Hereditary substitution in kinds.}
    }\vs

    \tsub k & \code{type} & = & \code{type}\\

    \tsub k & \pit y A K & = & \pit y {(\tsub a A)} {(\tsub k K)}\\
    & & & \text{with } x \neq y \text{ and } y \notin FV(M)\vs

    \multicolumn{4}{l}{
    \boxed{\hsub a \alpha M x \, A = A'} : \text{Hereditary substitution in types.}
    }\vs

    \tsub a & P & = & \tsub p P\\
    \tsub a & \pit y A B & = & \pit y {(\tsub a A)} {(\tsub a B)} \\
    & & & \text{with } x \neq y \text{ and } y \notin FV(M)\vs

    \multicolumn{4}{l}{
    \boxed{\hsub p \alpha M x \, P = P'} : \text{Hereditary substitution in base types.}
    }\vs

    \tsub p & \const a & = & \const a\\
    \tsub p & \app P M' & = & \app{(\tsub p P)} {(\tsub m M')}\vs
  \end{array}
\end{displaymath}
  \caption{Hereditary Substitutions on Types}
  \label{fig:heredsubst1}
\end{figure}

\begin{figure}
  \centering
\newcommand{\tsub}[1]{\ensuremath{\hsub #1 \alpha M x}}
\begin{displaymath}
  \begin{array}{rlcl}
    \multicolumn{4}{l}{
    \boxed{\hsub m \alpha M x \, M' = M''} : \text{Hereditary substitution in normal terms.}
    }\vs

    \tsub m & \lam y M' & = & \lam y {(\tsub m M')} \\
    & & & \text{with } x \neq y \text{ and } y \notin FV(M)\\
    \tsub m & R & = & M' \text{ where } \tsub r R = M' \oft \alpha'\vs

    \multicolumn{4}{l}{
    \boxed{\hsub r \alpha M x \, R = M' \oft \alpha'} : \text{H. subst. in a neutral producing a normal term.}
    }\vs

    \tsub r & x & = & M \oft \alpha\\
    \tsub r & \app R M_1 & = & \hsub m {\alpha_1} {(\tsub m M_1)} y M_2 \oft \alpha_2 \\
            & & & \text{when } \tsub r R = \lam y M_2 \oft \alpha_1 \to \alpha_2\vs

    \multicolumn{4}{l}{
    \boxed{\hsub r \alpha M x \, R = R'} : \text{H. subst. in a neutral producing a neutral term.}
    }\vs

    \tsub r & y & = & y \quad \text{with } x \neq y\\
    \tsub r & \app R M' & = & \app {R'} {(\tsub m M')}\\
    & & & \text{when } \tsub r R = R'

  \end{array}
\end{displaymath}
  \caption{Hereditary Substitutions on Terms}
  \label{fig:heredsubst2}
\end{figure}

\section[Programming with Proofs: The Beluga System]
{Programming with Proofs:\\\quad\quad\quad The Beluga System}\label{sec:belugalang}

The Beluga~\citep{Pientka:IJCAR10} system\footnote{Available at:
  \url{https://github.com/Beluga-lang/Beluga}} is a proof
assistant/programming language that manipulates LF terms and inductive
data to implement proofs the meta-theory of formal systems specified
in LF, and to implement programs that manipulate these terms, shining
examples are translators and compilers, for example a type safe
closure conversion and hoisting implementation~\citep{Belanger:CPP13},
or normalization by evaluation for the simply typed
$\lambda$-calculus~\citep{Cave:POPL12}.

Beluga is a dependently typed programming language where programs
directly correspond to first-order logic proofs over a specific
domain. More specifically, a proof by cases (and more generally by
induction) corresponds to a (total) functional program with dependent
pattern matching. We hence separate the language of programs from the
language of our specific domain about which we reason. The language is
similar to indexed type systems (see \citep{Zenger:TCS97,XI99popl});
however, unlike these aforementioned systems, Beluga's index domain is
much richer (contextual LF terms) and it allows pattern matching on
index objects, i.e. we support case-analysis on objects in our domain.

\lstset{language=Beluga}

Figure~\ref{fig:typeuniq} shows the type uniqueness theorem for the
calculus from Figure~\ref{fig:stlclf}. The property we want to
establish is that if a term $M$ has a derivation for type $S$ and one
for type $T$ then $S$ and $T$ are necessarily equal. To be able to
state the theorem first we define \lstinline!eq-tp! to formally say
when two types are equal. And then we define a context schema (that is
the classifier or ``type'' of contexts) where we say that each entry
contains a variable together with its type derivation. The actual
proof is a total recursive function (following the Curry-Howard
correspondence) where the statement of the theorem is the type of the
function, and the proof is its implementation. In this case, the type
is:
\begin{lstlisting}
  (g : ctx) [g |- oft M S[]] -> [g |- oft M T[]] ->
    [|-eq-tp S T]
\end{lstlisting}

That can be read as: give a context \lstinline!g! where each variable
is stored together with its type derivation, and a proof that term
\lstinline!M! has type \lstinline!S! and \lstinline!T! (where the
square brackets around the types \lstinline!T! and \lstinline!S! mean
that they are closed objects and do not depend on the context) the
function shows that both types are equal. We see in
Figure~\ref{fig:typeuniq} that to implement the proof, the function
uses case analysis and well-founded recursion to implement induction,
recursive function calls on smaller arguments to appeal to the
induction hypothesis and let expressions, or irrefutable pattern
matching to implement inversion. The application case follows by
inversion and the invoking the induction hypothesis, the case for the
constructor (i.e.: \lstinline!t-c!) is the base case and it just needs
to perform inversion on the other derivation. The more interesting
cases are \lstinline!t-lam! for functions and the variable case
(\lstinline!#p.u!). The case for $\lambda$ derivation uses an
inductive call done in an extended context where we have to be careful
extend the context appropriately, which is what we do in the
substitution we apply to the variables \lstinline!D! and
\lstinline!E!. The resulting recursive call is:
\begin{lstlisting}
unique [g, b:block x:tm, u:oft x _ |- D[..,<<b.x,b.u>>]]
       [g, b                       |- E[..,<<b.x,b.u>>]]
\end{lstlisting}
Finally, the variable case also follows by inversion using the
information stored in the context.

\begin{figure}
  \centering
\begin{lstlisting}
LF eq-tp: tp -> tp -> type =
% two types are equal only if they are the same
| refl: eq-tp T T;

% the context stores together the variable x
% and its type derivation
schema ctx = some [t:tp] block (x:tm, u:oft x t);

% Thm: If a term M has types S and also T, then S = T
rec unique : (g : ctx) [g |- oft M S[]] -> [g |- oft M T[]] -> [|- eq-tp S T] =
/ total d (unique _ _ _ _ d _) / % requires well founded recursive calls and totality
fn d => fn e => case d of

| [g |- t-app D1 D2] =>
  % by inversion on the other derivation
  let [g |- t-app E1 E2] = e in
  let [|- refl] = unique [g |- D1] [g |- E1] in % by i.h.
  [|- refl]

| [g |- t-c] =>
  % by inversion
  let [g |- t-c] = e in
  [|- refl]

| [g |- t-lam \x.\u. D] =>
  let [g |- t-lam \x.\u.E] = e in % by inversion
  let [|- refl] =
      % by i.h. in the extended ctx
      unique [g, b:block x:tm, u:oft x _ |- D[..,b.x,b.u]]
             [g, b |-  E[..,b.x,b.u]]
  in
  [|- refl]

| [g |- #p.u] => % a variable in the context
  let [g |- #p.u] = e in % the type derivation in the context
  [|- refl];
\end{lstlisting}
  \caption{Type Uniqueness Proof}
  \label{fig:typeuniq}
\end{figure}

\begin{figure}
  \centering
  \begin{tikzpicture}
    \matrix [row sep = 0mm]{
      \node [draw, text width=5cm, align=center, text height=0.6cm, text depth=0.5cm] {LF Specifications}; \\
      \node [text width=5cm, align=center, text height=0.3cm, pattern=crosshatch, pattern color=lightgray] {Contextual Objects}; \\
      \node [draw, text width=5cm, align=center, text height=0.6cm, text depth=0.5cm] {Compuational Language};\\
    };
  \end{tikzpicture}
  \caption{The Beluga Language}
  \label{fig:beldiagram}
\end{figure}
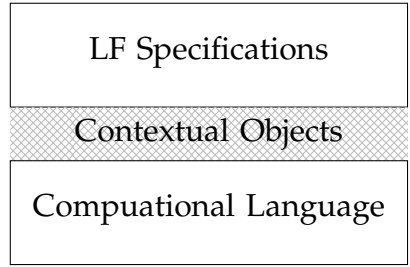
\begin{center}
\end{center}

The Beluga language reasons about \emph{specifications} in the logical
framework LF, by embedding them using \emph{contextual objects} in a
functional programming language. The idea, illustrated in
Figure~\ref{fig:beldiagram}, is that proofs are represented by
programs, as in the proofs as programs methodology \citep{Howard80},
and the specifications are embedded using contextual
objects~\citep{Nanevski:ICML05} and contextual types that combine the
type of LF objects with the context they are valid in.

\begin{figure}
  \begin{displaymath}
    \begin{array}{lrcl}
      \mbox{Contextual Objects} & C & \bnfas & \hat\Psi\vdash M \bnfalt \Psi \\
      \mbox{Contextual Types} & U & \bnfas & \Psi \vdash A \bnfalt \Psi \vdash \#A \bnfalt G\\
      \mbox{Contextual Schemas} & G & \bnfas &
          \exists \overrightarrow{(x:A)}\mathrel{.} B \bnfalt G + \exists \overrightarrow{(x:A)}\mathrel{.} B \\
      \mbox{Substitutions} & \sigma & \bnfas & \cdot \bnfalt \code{id} \bnfalt \sigma, M\\
      \mbox{Terms w/Meta vars.} & M & \bnfas & \dots \bnfalt \asisub u \sigma \bnfalt \asisub {\#p} \sigma \\
      \mbox{Meta Substitutions} & \theta & \bnfas &  \cdot \bnfalt \theta, C/X\\
      \mbox{Meta Context} & \Delta & \bnfas & \cdot \bnfalt \Delta, X\oft U\\
    \end{array}
  \end{displaymath}
  \caption{Contextual Objects}
  \label{fig:ctxobjs}
\end{figure}

One of the crucial features of Beluga is that it handles open objects
(i.e.: objects with free variables bound in a context). The technical
device to achieve this is contextual objects and types and it is based
on contextual modal type theory~\citep{Nanevski:ICML05}. A contextual
object is a term together with the context that describes all the free
variables of the term. In a sense, one could say that there are no
free variables, just variables bound in some context. As a
consequence, LF terms and types are combined with a context (in the
case of terms, the context just describes the names of the assumptions
that appear on the type, this is meant to support
$\alpha$-conversions and it is denoted with the hatted contexts as in $\hat\Psi$).

Figure~\ref{fig:ctxobjs} shows the syntax of the contextual objects,
together with the addition of meta-variables to be able to describe
incomplete LF terms. We use $X$ to represent meta-variables, and $u$
and $\#p$ when talking about a specific one. Moreover, we use
$\asisub u \sigma$ where the substitution expresses how $u$ relates to
its context. The case for $\#p$ is similar, but we use the sharp sign
to represent parameter variables that are a kind of meta-variable that
can only be instantiated by variables from the context $\Psi$. We
write contextual terms together with a context where types have been
erased and only names remain. The typing rule for contextual term
indicated that when a term: $\hat\Psi\vdash M$ is matched against a
contextual type: $\Psi\vdash A$, typing continues with the LF typing
judgment using the adequate contexts, namely:
$\Delta;\Psi\vdash M \checks A$.

Finally, contexts are also first-class objects, and they are
classified by contextual schemas\footnote{Contextual schemas play the
  role of types for contexts.}. Context schemas are formed by elements
$\exists \overrightarrow{(x\oft A)}\mathrel{.} B$ where we say that
assumptions are of type $B$. If $B$ is a type family its indices are
provided by the existential. When schemas describe contexts that may
contain assumptions of different type, they appear as sum types.
Similarly one could, as the implementation has, add products to
express contexts that grow by blocks of assumptions.
Figure~\ref{fig:ctxobjtp} shows the typing for contextual terms and
contexts.

\begin{figure}
  \begin{displaymath}
    \begin{array}{c}
      \multicolumn{1}{l}
      {\boxed{\Delta \vdash C \oft U}: \mbox{$C$ is of type $U$ in $\Delta$}}\vs

      \infer
      {\Delta\vdash (\hat\Psi\vdash M) \oft (\Psi\vdash A)}
      {\Delta;\Psi\vdash M\checks A}

      \quad

      \infer
      {\Delta\vdash\cdot\oft G}
      {}

      \vs

      \infer
      {\Delta\vdash \Psi,x\oft A \oft G}
      {\Delta\vdash \Psi \oft G
      & \mbox{exists}\ \oarr{(x\oft B')}\mathrel{.} B\in G
        \ \mbox{and}\ \Delta;\Psi\vdash\oarr M\oft \oarr B
        \ \mbox{s.t.}\ A = \asisub B {\oarr M / \oarr x}}
    \end{array}
  \end{displaymath}
  \caption{Typing of Contextual Objects}
  \label{fig:ctxobjtp}
\end{figure}

With the machinery to represent LF objects together with their context
we can describe now the reasoning language. This language can be seen
as as a term assignment of first-order logic with inductive types. Its
syntax is presented in Figure~\ref{fig:reasonlang}. It supports
inductive type families~\citep{Dybjer:1994} where the indices come
from the index domain and two kinds of functions, dependent function
space from contextual types to computational types, an arrow type for
computational functions, a box type to represent contextual-LF objects
and finally, the type for fully applied type families from the context.

The terms contain abstractions and applications for each function
space and notable pattern matching through the case expressions to
analyze inductive data and contextual objects.

Finally, the signature stores the definition of types and constructors
and recursive functions. Beluga supports general recursion, but it
implements a termination checker~\citep{Pientka:TLCA15} to ensure that
functions and therefore proofs are well founded.

\begin{figure}
  \centering
\begin{displaymath}
  \begin{array}{lrcl}
    \mbox{Kinds} & K & \bnfas &  \pit X U K \bnfalt \code{ctype} \\

    \mbox{Types} & T & \bnfas & \pit X U T \bnfalt
                                T_1\to T_2 \bnfalt \ibox U \bnfalt \app {\const a} {\vec C} \vs

    \mbox{Expressions} & E & \bnfas & \fne x E
                                      \bnfalt \mlame X E \bnfalt \app {E_1} E_2\\
                 & & &  \bnfalt \app{E_1} {\ibox C} \bnfalt
                       \ibox C \bnfalt
                       \casee E {\vec{B}} \bnfalt E \oft T
                       \bnfalt x \bnfalt \const{c}\\

    \mbox{Branches} & \vec{B} & \bnfas & B \bnfalt
      (B \mid \vec{B})\\

    \mbox{Branch} & B & \bnfas & \Pi\Delta;\Gamma. Pat \oft\theta \mapsto E\\

    \mbox{Pattern} & Pat & \bnfas & x \bnfalt \ibox C \bnfalt \app{\const c} {\oarr{Pat}} \vs

    \mbox{Signature} & \Sigma & \bnfas & {\const c}\oft T \mid {\const a} \oft K \mid \rece f T E\vs

    \mbox{Context} & \Gamma & \bnfas & \cdot \bnfalt \Gamma, x \oft T \\


  \end{array}
\end{displaymath}
  \caption{Reasoning Language}
  \label{fig:reasonlang}
\end{figure}






%% file: reconstruction/rmacros.tex
\newcommand{\wvec}[1]{\overrightarrow{#1}}

\newcommand{\ofth}{\ensuremath{\mathrel{\hat{:}}}} 

\newcommand{\boxto}{\ensuremath{{\to}}}

\newcommand{\recon}{\ensuremath{\rightsquigarrow}}
\newcommand{\reconChk}[1]{\ensuremath{: {#1}\rightsquigarrow}}
\newcommand{\reconBBranch}[1]{\ensuremath{\mid{#1}\overset\square\rightsquigarrow}}

\newcommand{\rn}[1]{\texttt{#1}}

\newcommand{\abstr}{\ensuremath{\mathrel{\searrow}}}

\newcommand{\fv}{\ensuremath{_f}}

\newcommand{\vnor}{\ensuremath{\mathrel{\vdash_N}}}
\newcommand{\vneu}{\ensuremath{\mathrel{\vdash_n}}}

\newcommand{\vdashfree}{\ensuremath{\vdash_{\mbox{f}}}}

\newcommand{\B}{B}
\newcommand{\yux}[2]{#1\,#2} 

\newcommand{\piesrc}[2]{\ensuremath{\{#1\}\,#2}} 
\newcommand{\pie}[2]{\ensuremath{\Pi^e #1.\,#2}}
\newcommand{\piei}[2]{\ensuremath{\Pi^i #1.\,#2}}
\newcommand{\pieboth}[2]{\ensuremath{\Pi^{\{e,i\}} #1.\,#2}}
\newcommand{\pieg}[2]{\ensuremath{\Pi^* #1.\,#2}}

\newcommand{\ctype}{\mbox{\lstinline!ctype!}}
\newcommand{\itype}{\mbox{\lstinline!type!}}

\newcommand{\elabto}{\ensuremath{\rightsquigarrow}}

\newcommand{\subst}[1]{\ensuremath{[#1]}}
\newcommand{\asub}[2]{\ensuremath{\subst{#1}#2}}

\newcommand{\hole}[1]{\ensuremath{?#1}}
\newcommand{\hsubst}[1]{\ensuremath{\llbracket#1\rrbracket}}
\newcommand{\ahsub}[2]{\ensuremath{\hsubst{#1}#2}}
\newcommand{\ids}[1]{\ensuremath{\code{id}(#1)}}
\newcommand{\hs}[2]{\ensuremath{?#1\subst{#2}}}
\newcommand{\rclo}[2]{\ensuremath{\Lbag#1\mathrel{;}#2\Rbag}}
\newcommand{\cns}{\ensuremath{\mathrel{::}}}
\newcommand{\lowers}{\ensuremath{\mathrel{\searrow}}}
\newcommand{\raises}{\ensuremath{\mathrel{\nearrow}}}


\newcommand{\spineret}{\mathrel{\rangle}}


\newcommand{\ift}[3]{\code{if}\,#1\,\code{then}\,#2\,\code{else}\,#3}

\definecolor{light-gray}{gray}{0.95}

\newcommand{\genHole}{\mathsf{genHole}\;}
\newcommand{\elsyn}{\rl{el-syn}}
\newcommand{\elvar}{\rl{el-var}}
\newcommand{\elconst}{\rl{el-const}}
\newcommand{\elfn}{\rl{el-fn}}
\newcommand{\elmlam}{\rl{el-mlam}}
\newcommand{\elmlami}{\rl{el-mlam-i}}
\newcommand{\elcase}{\rl{el-case}}
\newcommand{\eltypcase}{\rl{el-case-index-obj}}
\newcommand{\elimp}{\rl{el-impl}}
\newcommand{\elimpd}{\rl{el-impl-done}}
\newcommand{\elapp}{\rl{el-app}}
\newcommand{\elmapp}{\rl{el-mapp}}
\newcommand{\elmappi}{\rl{el-mapp-underscore}}
\newcommand{\elann}{\rl{el-annotated}}
\newcommand{\elrec}{\rl{el-rec}}
\newcommand{\elbox}{\rl{el-box}}
\newcommand{\eltyp}{\rl{el-typ}}
\newcommand{\elkind}{\rl{el-kind}}

\newcommand{\elbranch}{\rl{el-branch}}
\newcommand{\elsubst}{\rl{el-subst}}

\newcommand{\elpvar}{\rl{el-pvar}}
\newcommand{\elpindex}{\rl{el-pindex}}
\newcommand{\elpsyn}{\rl{el-psyn}}
\newcommand{\elpcon}{\rl{el-pcon}}
\newcommand{\elpann}{\rl{el-pann}}
\newcommand{\elspempty}{\rl{el-sp-empty}}
\newcommand{\elspcmp}{\rl{el-sp-cmp}}
\newcommand{\elspex}{\rl{el-sp-explicit}}
\newcommand{\elspim}{\rl{el-sp-implicit}}



\newcommand{\ep}{\ensuremath{\epsilon}}

\newcommand{\steps}{\ensuremath{\mathrel{\Downarrow}}}


\newcommand{\ltyped}{\emph{Typed}\xspace}
\newcommand{\luntyped}{\emph{Untyped}\xspace}


\newcommand{\onestar}{\ensuremath{(\dagger)}}
\newcommand{\twostars}{\ensuremath{(\ddagger)}}


%% file: reconstruction/reconstruction.tex
\section{Introduction}

Dependently typed programming languages allow programmers to express a
rich set of properties and statically verify them via type checking.
To make programming with dependent types practical, these systems
provide a source language where programmers can omit (implicit)
arguments which can be reasonably easy inferred and elaborate the
source language into a well-understood core language, an idea going
back to \citet{Pollack90}. However, this elaboration is rarely
specified formally for dependently typed languages which support
recursion and pattern matching. For Agda, a full dependently typed
programming language based on Martin L{\"o}f type theory, Norell
\citeyearpar[Chapter 3]{Norell:phd07} describes a bi-directional type
inference algorithm, but does not treat the elaboration of recursion
and pattern matching. For the fully dependently typed language Idris,
\citet{Brady:JFP13} describes the elaboration between source and
target, but no theoretical properties such as soundness are
established. A notable exception is~\citet{Asperti:2012} that
describes a sound bi-directional elaboration algorithm for the
Calculus of (Co)Inductive Constructions (CCIC) implemented in Matita.

In this chapter, we concentrate on a computational language with index
types. Specifically, our source language is inspired by the Beluga
language~\citep{Pientka:CADE15,
  Pientka:POPL08,Pientka:IJCAR10,Cave:POPL12} (presented in
Section~\ref{sec:belugalang} where we specify formal systems in the
logical framework LF \citep{Harper93jacm} (our index language) and
write proofs about LF objects as total recursive functions using
pattern matching.

More generally, our language may be viewed as a smooth extension of
simply typed languages, like Haskell or OCaml to nested dependent
pattern matching. Moreover, taking advantage of the separation between
types and terms, it is easy to allow impure programs, for example to
allow non-termination, partial functions, and polymorphism. All this
while reaping some of the benefits of dependent
types. 

Viewing our language through the Curry-Howard correspondence, it closely
corresponds to a proof term assignment for first-order logic over a
specific domain. Our dependent pattern matching construct corresponds
to case-analysis on predicates in first-order logic. Pattern matching
on index objects corresponds to case-analysis on an object from the
domain. Writing total functions in this language corresponds to
inductive proofs in first-order logic over a given domain. This domain
is kept somewhat abstract, but in the case of Beluga it is the logical
framework LF. The reconstruction algorithm for LF is presented
in~\citep{Pientka:JFP13}. However, the logical framework LF is not a
programming language and it does not contain features like pattern
matching that we consider in this work. 

The main contribution of this chapter is the design of a source
language for dependently typed programs where we omit implicit
arguments together with a sound bi-directional elaboration algorithm
from the source language to a fully explicit core language. This
language supports dependent pattern matching without requiring type
invariant annotations, and dependently-typed case expressions can be
nested as in simply-typed pattern matching. Throughout our
development, we will keep the index language abstract and state
abstractly our requirements such as decidability of equality and
typing. There are many interesting choices of index languages. For
example choosing arithmetic would lead to a DML~\citep{Xi:JFP} style
language ; choosing an authorization logic would let us manipulate
authorization certificates (similar to \emph{Aura}~\citep{Jia:2008});
choosing LF style languages (like Contextual LF
\citep{Nanevski:ICML05}) we obtain Beluga; choosing substructural
variant of it like CLF~\citep{Watkins02tr} we are in principle able to
manipulate and work with substructural specifications.

A central question when elaborating dependently typed languages is
what arguments may the programmer omit. In dependently-typed systems
such as Agda or Coq, the programmer declares constants of a given
(closed) type and labels arguments that can be freely omitted when
subsequently using the constant. Both, Coq and Agda, give the user the
possibility to locally override the implicit arguments and provide
instantiations explicitly.

In contrast, we follow here a simple, lightweight recipe which comes
from the implementation of the logical framework \emph{Elf}
\citep{Pfenning89lics} and its successor \emph{Twelf}
\citep{Pfenning99cade}: \emph{programmers may leave some index
  variables free when declaring a constant of a given type;
  elaboration of the type will abstract over these free variables at
  the outside; when subsequently using this constant, the user must
  omit passing arguments for those index variables which were left
  free in the original declaration.} Following this recipe,
elaboration of terms and types in the logical framework has been
described in~\citet{Pientka:JFP13}. Here, we will
consider a dependently typed functional programming language which
supports pattern matching on index objects.

The key challenge in elaborating recursive programs which support
case-analysis is that pattern matching in the dependently typed
setting refines index arguments and hence refines types. In contrast
to systems such as Coq and Agda, where we must annotate
case-expressions with an invariant, i.e. the type of the scrutinee,
and the return type, our source language does not require such
annotations. Instead we will infer the type of the scrutinee and for
each branch, we infer the type of the pattern and compute how the
pattern refines the type of the scrutinee. This makes our source
language lightweight and closer in spirit to simply-typed functional
languages. Our elaboration of source expressions to target expressions
is type-directed, inferring omitted arguments and producing a closed
well-typed target program. Finally, we prove soundness of our
elaboration, i.e. if elaboration succeeds our resulting program type
checks in our core language. Our framework provides post-hoc
explanation for elaboration found in the programming and proof
environment, Beluga~\citep{Pientka:IJCAR10}, where we use as the index
domain terms specified in the logical framework LF
\citep{Harper93jacm}.

We describe reconstruction as follows: We first give the grammar of our
source language. Showing example programs, we explain informally
what elaboration does. We then revisit our core language, describe the
elaboration algorithm formally and prove soundness. 


%% file: reconstruction/description.tex
\section{Source Language}

We consider here a dependently typed language where types are indexed
by terms from an index domain. Our language is similar to Beluga a
dependently typed programming environment where we can embed LF
objects into computation-level types and computation-level programs
which analyze and pattern match on LF objects. However, in our
description, as in for example~\citet{Cave:POPL12}, we will keep the
index domain abstract, but only assume that equality in the index
domain is decidable and unification algorithms exist. This will allow
us to focus on the essential challenges when elaborating a dependently
typed language in the presence of pattern matching.

The syntax of our source language that allows programmers to omit some
arguments is as follows:
\begin{displaymath}
  \begin{array}{lr@{~~}c@{~~}l}
    \mbox{Kinds} & k & \bnfas & \ctype \bnfalt \piesrc {X\oft u} k\\

    \mbox{Atomic Types} & p & \bnfas & \yux {\const a} {\wvec{\ibox c}} \\

    \mbox{Types} & t & \bnfas & p \bnfalt \ibox u \bnfalt \piesrc {X\oft u} t \bnfalt t_1\to t_2 \\[0.5em]

    \mbox{Expressions} & e & \bnfas & \fne x e \bnfalt \mlame X e \bnfalt x
    \bnfalt \const{c} \bnfalt \ibox c \bnfalt \\
      & & & \yux {e_1} e_2 \bnfalt \yux{e_1} \ibox c \bnfalt
      e\;\_\bnfalt \casee e {\vec{b}} \bnfalt e \oft t\\

    \mbox{Branches} & \vec{b} & \bnfas & b \bnfalt (b \mid \vec{b})\\

    \mbox{Branch} & b & \bnfas & pat \mathrel \mapsto e\\

    \mbox{Pattern} & pat & \bnfas & x \bnfalt \ibox c \bnfalt \yux{\const c} \oarr{pat} \bnfalt pat\oft t \\[0.5em]
     \mbox{Declarations} & d & \bnfas & \rece f t e \mid \const{c}\oft t \mid \const{a}\oft k\\

  \end{array}
\end{displaymath}

As a convention we will use lowercase $c$ to refer to index level
objects, lowercase $u$ for index level types, and upper case letters
$X, Y$ for index-variables. Index objects can be embedded into
computation expressions by using a box modality written as $\ibox{c}$.
Our language supports functions ($\fne x e$), dependent functions
($\mlame X e$), function application ($\yux{e_1} e_2$), dependent
function application ($\yux{e_1}\ibox c$), and case-expressions. We
also support writing underscore ($\_$) instead of providing explicitly
the index argument in a dependent function application ($e\;\_$). Note
that we are overloading syntax: we write $e\;\ibox{c}$ to describe the
application of the expression $e$ of type $\ibox{u}\to t$ to the
expression $\ibox{c}$; we also write $e\;\ibox{c}$ to describe the
dependent application of the expression $e$ of type $\{X\oft u\}t$ to
the (unboxed) index object $c$.This ambiguity can be easily resolved
using type information. Note that in our language the dependent
function type and the non-dependent function type do not collapse,
since we can only quantify over objects of our specific domain instead
of arbitrary propositions (types).

We may write type annotations anywhere in the program ($e \oft t$
and in patterns $pat\oft t$); type annotations are particularly useful to
make explicit the type of a sub-expression and name index variables occurring in
the type. This allows us to resurrect index variables which are kept
implicit. In patterns, type annotations are useful since they provide hints to
type elaboration regarding the type of pattern variables.

A program signature $\Sigma$ consists of kind declarations
($\const a\oft k$), type declarations ($\const c \oft t$) and
declarations of recursive functions ($\rece f t e$). This can be
extended to allow mutual recursive functions in a straightforward way.

One may think of our source language as the language obtained after
parsing where for example let-expressions have been translated into
case-expression with one branch.

Types for computations include non-dependent function types
($t_1 \to t_2$) and dependent function types ($\{X\oft u\} t$); we can
also embed index types into computation types via $\ibox u$ and
indexed computation-level types by an index domain written as
$\yux {\const a} \wvec{\ibox c}$. We also include the grammar for
computation-level kinds which emphasizes that computation-level types
can only be indexed by terms from an index domain $u$. We write
$\ctype$ (that reads as ``computational type'') for the base kind,
since we will use $\itype$ for kinds of the index domain.

We note that we only support one form of dependent function type
$\{X\oft u\}t$; the source language does not provide any means for
programmers to mark a given dependently typed variable as implicit as
for example in Agda. Instead, we will allow programmers 
to leave some index variables occurring in computation-level types
free; elaboration will then infer their types and abstract over them
explicitly at the outside. The programmer must subsequently omit
providing instantiation for those ``free'' variables. We will explain
this idea more concretely below.

\subsection{Well-formed Source Expressions}

\begin{sidewaysfigure}
  \begin{displaymath}
    \begin{array}{c}
      \multicolumn{1}{l}{\boxed{\vdash  d\wf}~~\mbox{Declaration $d$ is well-formed}}\vs

      \infer[\rl{wf-rec}]
      {\vdash\rece f t e \wf}
      {\cdot;\; f\vdash e\wf &
        \cdot\vdashfree t\wf}
\quad

      \infer[\rl{wf-types}]
      {\vdash c : t \wf}
      {\cdot\vdashfree t\wf}

\quad

      \infer[\rl{wf-kinds}]
      {\vdash a : k \wf}
      {\cdot\vdashfree k\wf}

\vs

      \multicolumn{1}{l}{\boxed{\delta;\gamma\vnor e\wf}~~
                               \mbox{Normal expression $e$ is well-formed in context $\delta$ and $\gamma$}}
\vs

      \infer[\rl{wf-fn}]
      {\delta;\gamma\vnor\fne x e\wf}
      {\delta;\gamma,x\vnor e\wf}

      \quad

      \infer[\rl{wf-mlam}]
      {\delta;\gamma\vnor\mlame X e\wf}
      {\delta, X;\gamma\vnor e\wf}

      \vs

      \infer[\rl{wf-box}]
      {\delta;\gamma\vnor\ibox c\wf}
      {\delta\vdash c\wf}

      \quad

      \infer[\rl{wf-case}]
      {\delta;\gamma\vnor\casee e {\vec b}\wf}
      {\delta;\gamma\vneu e\wf &
        \mbox{for all $b_n$ in $\vec b$ .}\ \delta;\gamma\vdash b_n\wf}

      \quad

      \infer[\rl{wf-neu}]
      {\delta;\gamma\vnor e\wf}
      {\delta;\gamma\vneu e\wf}

      \vs

      \multicolumn{1}{l}{\boxed{\delta;\gamma\vneu
          e\wf}~~\mbox{Neutral expression $e$ is well-formed in context $\delta$ and $\gamma$}}\vs

      \infer[\rl{wf-ann}]
      {\delta;\gamma\vneu e\oft t\wf}
      {\delta;\gamma\vnor e\wf &
        \delta\vdash t\wf}

      \quad

      \infer[\rl{wf-app}]
      {\delta;\gamma\vneu\yux{e_1}{e_2}\wf}
      {\delta;\gamma\vneu e_1\wf &
        \delta;\gamma\vnor e_2\wf}

      \vs

      \infer[\rl{wf-appe}]
      {\delta;\gamma\vneu\yux{e}{\ibox c}\wf}
      {\delta;\gamma\vneu e\wf &
        \delta \vdash c\wf}

      \quad

      \infer[\rl{wf-apph}]
      {\delta;\gamma\vneu\yux{e_1}\_\wf}
      {\delta;\gamma\vneu e_1\wf}



      \vs

      \multicolumn{1}{l}{\boxed{\delta;\gamma\vdash
          pat\mapsto e\wf}~~\mbox{Branch is well-formed in $\delta$ and $\gamma$}}
      \vs

      \infer[\rl{wf-branch}]
      {\delta;\gamma\vdash pat\mapsto e\wf}
      {\delta' ;\gamma' \vdash pat\wf
      & \delta,\delta';\gamma,\gamma'\vneu e\wf
      & \text{ where } \delta', \gamma' \text{ introduce fresh vars}}

      \vs

      \multicolumn{1}{l}{\boxed{\delta ; \gamma \vdash pat\wf}~~
        \mbox{Pattern $pat$ is well-formed in $\delta$ and $\gamma$}}
      \vs

      \infer[\rl{wf-p-var}]
      {\delta ; x \vdash x\wf}
      {}

      \quad

      \infer[\rl{wf-p-con}]
      {\delta_1,\ldots, \delta_n ; \gamma_1, \ldots, \gamma_n \vdash\yux{\const c}{\wvec{Pat}}\wf}
      {\mbox{for all $p_i$ in $\wvec{Pat}$.}~~\delta_i ; \gamma_i \vdash p_i\wf}

      \vs

      \infer[\rl{wf-p-i}]
      {\delta ; \cdot \vdash\ibox c\wf}
      {\delta \vdash c\wf}

      \quad

      \infer[\rl{wf-p-ann}]
      {\delta ; \gamma \vdash pat\oft t\wf}
      {\delta ; \gamma \vdash pat\wf &
        \delta \vdash t\wf}
    \end{array}
  \end{displaymath}
  \caption{Well-formed Source Expressions}
  \label{fig:wf}
\end{sidewaysfigure}

\begin{figure*}[bht]
  \begin{displaymath}
    \begin{array}{c}
      \multicolumn{1}{l}
      {\boxed{\delta\vdash k\wf}~~\mbox{
          Kind $k$ is well-formed and closed with respect to $\delta$}}\vs

      \infer{\delta\vdash\ctype\wf}{}

      \quad

      \infer{\delta\vdash\piesrc{X\oft u} k\wf}{\delta\vdash u\wf & \delta,X\oft u\vdash k\wf}\vs

      \multicolumn{1}{l}
      {\boxed{\delta\vdash t\wf}~~\mbox{
          Type $t$ is well-formed and closed with respect to $\delta$}}\vs

      \infer{\delta\vdash {\const a} \wvec{\ibox c}\wf}{\delta;\vdash c_i\wf & \mbox{for all $c_i$ in $\wvec c$}}

      \quad

      \infer{\delta\vdash \ibox u\wf}{\delta\vdash u\wf}

      \vs

      \infer{\delta\vdash\piesrc{X\oft u} t\wf}{\delta\vdash u\wf & \delta,X\oft u\vdash t\wf}

      \quad

      \infer{\delta\vdash t_1\to t_2\wf}{\delta\vdash t_1\wf & \delta\vdash t_2\wf}
    \end{array}
  \end{displaymath}
  \caption{Well-formed Kinds and Types}
  \label{fig:wftk}
\end{figure*}

Before elaborating source expressions, we state when a given source
expression is accepted as a well-formed expression. In particular, it
will highlight that free index variables are only allowed in
declarations when specifying kinds and declaring the type of constants
and recursive functions. We use $\delta$ to describe the list of index
variables and $\gamma$ the list of program variables. We rely on two judgments
from the index language:
\vspace{-0.2cm}
\begin{displaymath}
\begin{array}{ll}
\delta\vdash~~ c\wf   & \mbox{Index object $c$ is well formed and} \\
                      & \mbox{ closed with respect to $\delta$}   \\
\delta\vdashfree c\wf & \mbox{Index object $c$ is well formed with respect to $\delta$} \\
                      & \mbox{and may contain free index variables}   \\
\end{array}
\end{displaymath}
\vspace{-0.15cm}

We describe declaratively the well-formedness of declarations and
source expressions in Fig.~\ref{fig:wf}. The distinction between
normal and neutral expressions forces a type annotation where a
non-normal program would occur. The normal vs. neutral term
distinction is motivated by the bidirectional type-checker presented
in Section~\ref{sec:typtarget}. The rules for well-formed types and
kinds are given in Figure~\ref{fig:wftk}.

In branches, pattern variables from $\gamma$ must occur linearly while
we put no such requirement on variables from our index language listed
in $\delta$. The judgment for well formed patterns synthesizes
contexts $\delta$ and $\gamma$ that contain all the variables bound in
the pattern (this presentation is declarative, but algorithmically the
two contexts result from the well-formed judgment). Notice that the
rules \rl{wp-p-i} and \rl{wp-con} look similar but they operate on
different syntactic categories and refer to the judgment for
well-formed index terms provided by the index level language. They
differ in that the one for patterns synthesizes the $\delta$ context
that contains the meta-variables bound in the pattern.

\subsection{Some Example Programs}

We next illustrate writing programs in our language and explain the
main ideas behind elaboration. We use Beluga syntax in our examples,
in these examples we want to focus on the elaboration of computations
not the index language.

\subsubsection{Translating Untyped Terms to Intrinsically Typed Terms}

We implement a program to translate a simple language with numbers,
booleans and some primitive operations to its intrinsically typed
counterpart. This illustrates declaring an index domain, using index
computation-level types, and explaining the use and need to pattern
match on index objects. We translate a source language that we will
call \luntyped into its typed counter part, we call this
representation \ltyped. This translation is basically a
type-checker for terms written in \luntyped.
\filbreak

We first define the syntax of \luntyped using the recursive
datatype \lstinline!UTm!. Note the use of the keyword
\lstinline!ctype! to define a computation-level recursive data-type.

\begin{lstlisting}
inductive UTm : ctype =
| UNum  : Nat -> UTm
| UPlus : UTm -> UTm -> UTm
| UTrue : UTm
| UFalse: UTm
| UNot  : UTm -> UTm
| UIf   : UTm -> UTm -> UTm -> UTm
;
\end{lstlisting}

Terms can be of type \lstinline!nat! for numbers or \lstinline!bool!
for booleans. Our goal is to define \ltyped our language of typed
terms using a computation-level type family \lstinline!Tm! which is
indexed by objects \lstinline!nat! and \lstinline!bool! which are
constructors of our index type \lstinline!tp!. Note that
\lstinline!tp! is declared as having the kind \lstinline!type! which
implies that this type lives at the index level and that we will be
able to use it as an index for computation-level type families.
\begin{lstlisting}
LF tp : type =
| nat  : tp
| bool : tp
;
\end{lstlisting}

\filbreak
Using indexed families we can now define the type \lstinline!Tm! that
specifies only type correct terms of the language \ltyped, by indexing
terms by their type using the index level type \lstinline![|- tp]!.
The square brackets define a box for index language terms and types
and because Beluga's index language contains contextual types the
empty turnstyle indicates that types are closed.
\begin{lstlisting}
inductive Tm : [|- tp] -> ctype =
| Num       : Nat -> Tm [|- nat]
| Plus      : Tm [|- nat] -> Tm [|- nat] -> Tm [|- nat]
| True      : Tm [|- bool]
| False     : Tm [|- bool]
| Not       : Tm [|- bool] -> Tm [|- bool]
| If        : Tm [|- bool] -> Tm [|- T] -> Tm [|- T] -> Tm [|- T]
;
\end{lstlisting}

When the \lstinline!Tm! family is elaborated, the free variable
\lstinline!T! in the \lstinline!If! constructor will be abstracted
over by an implicit $\Pi$-type, as in the Twelf~\citep{Pfenning99cade}
tradition. Because \lstinline!T! was left free by the programmer, the
elaboration will add an implicit quantifier; when we use the constant
\lstinline!If!, we now must omit passing an instantiation for \lstinline!T!. For
example, we must write \lstinline!(If True (Num 3) (Num 4))! and elaboration
will infer that \lstinline!T! must be \lstinline!nat!.

One might ask how we can provide the type explicitly - this is
possible indirectly by providing type annotations. For example:
\begin{lstlisting}
If e (e1:TM[|- nat]) e2
\end{lstlisting}
will fix the type of \lstinline!e1! to be \lstinline!Tm [|- nat]!.

\filbreak
Our goal is to write a program to translate an untyped term
\lstinline!UTm!  to its corresponding typed
representation. Because this operation might fail for ill-typed
\lstinline!UTm! terms we need an option type to reflect the
possibility of failure.

\begin{lstlisting}
inductive TmOpt : ctype =
| None : TmOpt
| Some : {T : [|- tp]} Tm [|- T] -> TmOpt
;
\end{lstlisting}

A value of type \lstinline!TmOpt! will either be empty
(i.e. \lstinline!None!) or some term of type \lstinline!T!. We chose to make
\lstinline!T! explicit here by quantifying over it explicitly using the
curly braces. When returning a \lstinline!Tm! term, the program must now provide
the instantiation of \lstinline!T! in addition to the actual term.

So far we have declared types and constructors for our language. These
declarations will be available in a global signature. The next step is
to declare a function that will take \luntyped terms into \ltyped
terms if at all possible. Notice that for the function to be type
correct it has to respect the specification provided in the
declaration of the type \lstinline!Tm!. We only show a few interesting
cases below.

\begin{lstlisting}
rec typecheck : UTm -> TmOpt =
fn e => case e of
| UNum n => Some [|- nat] (Num n)
| UNot e => (case typecheck e of
  | Some [|- bool] x => Some [|- bool] (Not x)
  | other => None)
| UIf c e1 e2 =>  (case typecheck c of
  | Some [|- bool] c' => (case (typecheck e1 , typecheck e2) of
    | (Some [|- T] e1' , Some [|- T] e2') =>
      Some [|- T] (If c' e1' e2')
    | other => None)
  | other => None)
% ... the cases for UPlus, UTrue and UFalse are similar
;
\end{lstlisting}

In the \lstinline!typecheck! function the cases for numbers, plus,
true and false are completely straightforward. The case for negation
(i.e. constructor \lstinline!UNot!) is interesting because we need to
pattern match on the result of type-checking the sub-expression
\lstinline!e! to match its type to \lstinline!bool! otherwise we
cannot construct the intrinsically typed term, i.e. the constructor
\lstinline!Not! requires a boolean term, this requires matching on
index level terms. Additionally the case for \lstinline!UIf! is also
interesting because we not only need a boolean condition but we also
need to have both branches of the \lstinline!UIf! term to be of the
same type. Again we use pattern matching on the indices to verify that
the condition is of type \lstinline!bool! but notably we use
non-linear pattern matching to ensure that the type of the branches
coincides. Therefore, by using non-linear patterns we can force two
meta-variables to match against the same term. Notably, note that
\lstinline!If! has an implicit argument (the type \lstinline!T!) which
will be inferred during elaboration and the fact that it is used in
both branches implies that in an if expression the ``then'' branch and
the ``else'' branch are of the same type.

In the definition of type \lstinline!TmOpt! we chose to explicitly
quantify over \lstinline!T!, however another option would have been to
leave it implicit.  When pattern matching on \lstinline!Some e!, we
would need to resurrect the type of the argument \lstinline!e! to be
able to inspect it and check whether it has the appropriate type. We
can employ type annotations, as shown in the code below, to constrain
the type of \lstinline!e!.

\begin{lstlisting}
| UIf c e1 e2 =>  (case typecheck c of
  | Some (c' <<: [|- bool]>>) => (case (typecheck e1, typecheck e2) of
    | (Some (e1' <<: [|- T]>>) , Some (e2' <<: [|- T]>>) =>
      Some (If c' e1' e2')
    | other => None)
  | other => None)
\end{lstlisting}

\filbreak
In this first example there is not much to elaborate. The missing
argument in \lstinline!If! and the types of variables in patterns are
all that need to be elaborated.

\subsubsection{Type-preserving Evaluation}\label{sec:tpeval}

Our previous program used dependent types sparingly; most notably
there were no dependent types in the type declaration given to the
function \lstinline!typecheck!. We now discuss the implementation of
an evaluator, which evaluates type correct programs to values of the
same type, to highlight writing dependently typed functions. Because
we need to preserve the type information, we index the values by their
types in the following manner:

\begin{lstlisting}
inductive Val : [|- tp] -> ctype =
| VNum  : Nat -> Val [|- nat]
| VTrue : Val [|- bool]
| VFalse: Val [|- bool]
;
\end{lstlisting}

We define a type preserving evaluator below:

\begin{lstlisting}
rec eval : Tm [|- T] -> Val [|- T] = fn e => case e of
| Num n => VNum n
| Plus e1 e2 => (case (eval e1 , eval e2) of
  | (VNum x , VNum y) => VNum (add x y))
| Not e => (case eval e of
  | VTrue => VFalse
  | VFalse => VTrue)
| If e e1 e2 =>  (case eval e of
  | VTrue => eval e1
  | VFalse => eval e2)
| True => VTrue
| False => VFalse
;
\end{lstlisting}

\filbreak
First, we specify the type of the evaluation function as:
\begin{lstlisting}
Tm[|- T] -> Val [|- T]
\end{lstlisting}
where \lstinline!T! remains free. The user provided type has a free
variable (i.e. \lstinline!T!) that elaboration will abstract over
after inferring its type. The type of the variable should be fixed by
the places where it is used. Elaboration will therefore abstract over
it in the outside by adding it as an implicit parameter (introduced by
$\Pi^i$, which is an abstraction that we indicate as implicit since
the user will not provide instantiations for and that elaboration will
reconstruct). We then elaborate the body of the function against
\lstinline!$\Pi^i$T:|- tp. Tm [|- T] -> Val [|- T]!. It will first
need to introduce the appropriate dependent function abstraction in
the program before we introduce the non-dependent function $\fne x e$.
Moreover, we need to infer omitted arguments in the pattern in
addition to inferring the type of pattern variables in the
\lstinline!If! case. Since \lstinline!T! was left free in the type
given to \lstinline!eval!, we must also infer the omitted argument in
the recursive calls to \lstinline!eval!. Finally, we need to keep
track of refinements the pattern match induces: our scrutinee has type
\lstinline!Tm [|- T]!; pattern matching against \lstinline!Plus e1 e1!
which has type \lstinline!Tm [|- nat]! refines \lstinline!T! to
\lstinline!nat!.

\subsubsection{A Certifying Evaluator}\label{sec:certeval}

So far in our examples, we have used a simply typed index language. We
used our index language to specify natural numbers, booleans, and a
tagged enumeration that contained labels for the \lstinline!bool! and
\lstinline!nat! types. In this example we go one step further, and use
a dependently typed specification, in fact we take advantage of LF as our index
level language as used in Beluga. Using LF at the index language we
specify the simply-typed lambda calculus and its operational semantics
in Figure~\ref{fig:certeval}. These rules provide a call by name
operational semantic chosen for no reason other than to save one
evaluation rule compare to call by value. Using these specifications
we write a recursive function that returns the value of the program
together with a derivation tree that shows how the value was computed.
This example requires dependent types at the index level and
consequently the elaboration of functions that manipulate these
specifications has to be more powerful.

\begin{figure}
  \centering
\begin{displaymath}
  \begin{array}{lr@{~~}c@{~~}l}
    \mbox{Types} & T & \bnfas & \top \bnfalt T_1\to T_2 \\[0.5em]

    \mbox{Terms} & M, N & \bnfas & () \bnfalt x \bnfalt \lambda x\oft T. M \bnfalt M N\\[.5em]

    \mbox{Context} & \Gamma & \bnfas & \cdot \bnfalt \Gamma, x\oft T\\[.5em]

    \mbox{Values} & V & \bnfas & \top \bnfalt \lambda x\oft T. M
  \end{array}
\end{displaymath}
  \begin{displaymath}
    \begin{array}{c}
      \multicolumn{1}{l}{\boxed{\Gamma\vdash M\oft T}~~\mbox{Term $M$
          has type $T$ in context $\Gamma$}}\vs

      \infer{\Gamma\vdash()\oft\top}{}

      \quad

      \infer{\Gamma\vdash x\oft T}{x\oft T \in \Gamma}

      \quad

      \infer{\Gamma\vdash\lambda x\oft T_1. M\oft T_2} {\Gamma,x\oft T_1\vdash M\oft T_2}

      \vs

      \infer{\Gamma\vdash M N\oft T} {\Gamma\vdash M \oft T_1\to T & \Gamma\vdash N\oft T_1}

      \vs

      \multicolumn{1}{l}{\boxed{M\Downarrow V}~~\mbox{Term $M$
          steps to value $V$}}\vs

      \infer{\top \steps \top}{}

      \quad

      \infer{\lambda x\oft T.M\steps\lambda x\oft T.M}{}

      \vs

      \infer
      {M N \steps N'}
      {M\steps \lambda x:T.M' & [N/x]M'\steps N' }

    \end{array}
  \end{displaymath}
\caption{Example: A Simply-typed $\lambda$-calculus}
  \label{fig:certeval}
\end{figure}

As in the previous example, we define the types of terms of our
language using the index level language. As opposed to the type
preserving evaluator, in this case we define our intrinsically typed
terms also using the index level language (which will be LF for this
example). We take advantage of LF to represent binders in
$\lambda$-terms, and use dependent types to represent well-typed terms
only.

\begin{lstlisting}
LF tp : type =
| unit : tp
| arr : tp -> tp -> tp
;

LF term : tp -> type =
| one : term unit
| lam : (term A -> term B) -> term (arr A B)
| app : term (arr A B) -> term A -> term B
;
\end{lstlisting}

These datatypes represent an encoding of well typed terms of a
simply-typed lambda calculus with \lstinline!unit! as a base type.
Using LF we can also describe what constitutes a value and a big-step
operational semantics. We use the standard technique in LF to
represent binders with the function space (usually called Higher Order
Abstract Syntax, HOAS~\citep{Pfenning88pldi}) and type families to only
represent well-typed terms, thus this representation combines the
syntax for terms with the typing judgment from
Figure~\ref{fig:certeval}.

\begin{lstlisting}
LF value : tp -> type =
| v-one : value unit
| v-lam : (term A -> term B) -> value (arr A B)
;

LF big-step : term T -> value T -> type =
| e-one : big-step one v-one
| e-lam : big-step (lam M) (v-lam M)
| e-app : big-step M (v-lam M') ->
          big-step (M' N) N' ->
          big-step (app M N) N'
;
\end{lstlisting}

The \lstinline!value! type simply states that \lstinline!one! and
lambda terms are values, and the type \lstinline!big-step! encodes the
operational semantics where each constructor corresponds to one of the
rules in Figure~\ref{fig:certeval}. The constructors \lstinline!e-one!
and \lstinline!e-lam! simply state that \lstinline!one! and lambdas
step to themselves. On the other hand rule \lstinline!e-app! requires
that in an application, the first term evaluates to a lambda
expression (which is always the case as the terms are intrinsically
well typed) and then it performs the substitution and continues
evaluating the term to a value. Note how the substitution is performed
by an application as we reuse the LF function space for binders as
typically done with HOAS.


To implement a certifying evaluator we want the \lstinline!eval!
function to return a value and a derivation tree that shows how we
computed this value. We encode this fact in the \lstinline!Cert!
data-type that encodes an existential or dependent pair that combines a
value with a derivation tree.

\begin{lstlisting}
inductive Cert : [|- term T] -> ctype =
| Ex : {N: [|- value T]} [|- big-step M N] -> Cert [|- M]
;
\end{lstlisting}

In the \lstinline!Ex! constructor we have chosen to explicitly
quantify over \lstinline!N!, the value of the evaluation, and left the
starting term \lstinline!M! implicit. However another option would
have been to leave both implicit, and use type annotations when
pattern matching to have access to both the term and its value.

Finally the evaluation function simply takes a term and returns a
certificate that contains the value the terms evaluates to, and the
derivation tree that led to that value.

\begin{lstlisting}
rec eval : {M : [|- term T]} Cert [|- M] =
mlam M => case [|- M] of
| [|- one] => Ex [|- v-one] [|- e-one]
| [|- lam (\x.M)] => Ex [|- v-lam (\x.M)] [|- e-lam]
| [|- app M N] =>
  let Ex [|- v-lam (\x. M')] [|- D]  = eval [|- M] in
  let Ex [|- N'][|- D'] = eval [|- M'[N]] in
  Ex [|- N'] [|- e-app D D']
;
\end{lstlisting}

Elaboration of \lstinline!eval! starts by the type
annotation. Inferring the type of variable \lstinline!T! and
abstracting over it, resulting in:
\begin{lstlisting}
  <<$\Pi^i$T:[|- tp].>> {M : [|- term T]} Cert [|- M]
\end{lstlisting}
The elaboration proceeds with the body, abstracting over the inferred
dependent argument with \lstinline!mlam T => ...! When elaborating the
case expression, the patterns in the index language will need
elaboration. In this work we assume that each index language comes
equipped with an appropriate notion of elaboration (described in
\citep{Pientka:JFP13} for the logical framework LF). For example,
index level elaboration will abstract over free variables in
constructors and the pattern for lambda terms becomes
\lstinline![lam <<A B>> (\x. M)]! when the types for parameters and body
are added.
Additionally, in order to keep the core language as lean as possible
we desugar \lstinline!let! expressions into \lstinline!case!
expressions. For example, in the certifying evaluator, the following
code from \lstinline!eval!:
\begin{lstlisting}
  let Ex [|- v-lam M'] [|- D]  = eval [|- M] in
  let Ex [|- N'][|- D'] = eval [|- M' [N]] in
  Ex [|- N'] [|- e-app D D']
\end{lstlisting}
is desugared into:
\begin{lstlisting}
 (case eval [|- M] of | Ex [|- v-lam M'] [|- D] =>
    (case eval [|- M' [N]] of
    | Ex [|- N'][|- D'] => Ex [|- N'] [|- \e-app D D']))
\end{lstlisting}

We will come back to this example and discuss the fully elaborated
program in the next section.


%% file: reconstruction/target.tex
\begin{figure}
  \centering
\begin{displaymath}
  \begin{array}{lrcl}
    \mbox{Kinds} & K & \bnfas &  \ctype \bnfalt \pie {X\oft U} K
    \bnfalt \piei {X\oft U} K\\


    \mbox{Types} & T & \bnfas & \pie {X\oft U} T \bnfalt \piei {X\oft U} T \\
    & & & P \bnfalt U \bnfalt T_1\to T_2 \bnfalt \const a \bnfalt \yux T C\\[1em]

    \mbox{Expressions} & E & \bnfas & \fne x E
      \bnfalt \mlame X E \\
      & & & \bnfalt \yux {E_1} E_2 \bnfalt \yux{E_1} C \bnfalt
      C \bnfalt\\
      & & & \casee E {\vec{B}} \bnfalt x \bnfalt E \oft T
      \bnfalt \const{c}\\

    \mbox{Branches} & \vec{B} & \bnfas & B \bnfalt
      (B \mid \vec{B})\\

    \mbox{Branch} & B & \bnfas & \Pi^i\Delta;\Gamma. Pat \oft\theta \mapsto E\\

    \mbox{Pattern} & Pat & \bnfas & x \bnfalt C \bnfalt \yux{\const c}
    \oarr{Pat}
\\[1em]

    \mbox{Declarations} & D & \bnfas & {\const c}\oft T \mid {\const a} \oft K \mid \rece f T E
\\[1em]

    \mbox{Context} & \Gamma & \bnfas & \cdot \bnfalt \Gamma, x \oft T \\

    \mbox{Index-Var-Context} & \Delta & \bnfas & \cdot \bnfalt \Delta, X\oft U\\

    \mbox{Refinement} & \theta & \bnfas & \cdot \bnfalt \theta, C/X
    \bnfalt \theta, X/X \\
  \end{array}
\end{displaymath}

  \caption{Target Language}
  \label{fig:target}
\end{figure}

\section{Target Language}
The \emph{target language} is similar to the computational language
described in~\citet{Cave:POPL12} which has a well developed
meta-theory including descriptions of coverage
\citep{Dunfield:coverage08, JacobRao:2017} and termination
\citep{Pientka:TLCA15}. The target language (see
Fig.~\ref{fig:target}), which is similar to our source language, is
indexed by fully explicit terms of the index level language; we use
$C$ for fully explicit index level objects, and $U$ for elaborated
index types; index-variables occurring in the target language will be
represented by capital letters such as $X, Y$. Moreover, we rely on a
substitution which replaces index variables $X$ with index objects.
The main difference between the source and target language is in the
description of branches. In each branch, we make the type of the
pattern variables (see context $\Gamma$) and variables occurring in
index objects (see context $\Delta$) explicit. We associate each
pattern with a refinement substitution $\theta$ which specifies how
the given pattern refines the type of the scrutinee.

\subsection{Typing of the Target Language}\label{sec:typtarget}

\begin{sidewaysfigure}
  \centering\scriptsize
  \begin{displaymath}
    \begin{array}{c} 
      \multicolumn{1}{l}{\boxed{\vdash  D \wf}
        ~~\mbox{Target declaration $D$ is well-formed}}\vspace{.5em}\\

      \infer[\rl{t-rec}]
      {\vdash\rece f T E \wf}
      {\cdot \vdash T \checks \ctype & \cdot~;~ f\oft T \vdash E\checks T}

\qquad
      \infer[\rl{t-type}]
      {\vdash {\const c} \oft T \wf}
      {\cdot \vdash T \checks \ctype}

\qquad
      \infer[\rl{t-kind}]
      {\vdash {\const a} \oft K \wf}
      {\cdot \vdash K \checks \mathsf{kind}}

\\[1em]
      \multicolumn{1}{l}{\boxed{\Delta\vdash T : K}
        ~~\mbox{$T$ is a well-kinded type of kind $K$}}\vspace{.5em}\\

      \infer[\rl{k-arr}]
      {\Delta\vdash T_1 \to T_2 \oft \ctype}
      {\Delta\vdash T_1 \oft \ctype & \Delta\vdash T_2 \oft \ctype}

      \quad

      \infer[\rl{k-pi}]
      {\Delta\vdash \pieboth {X\oft U} T \oft \ctype}
      {\Delta\vdash U \oft \itype & \Delta,X\oft U\vdash T \oft \ctype}

      \vs

      \infer[\rl{k-idx}]
      {\Delta\vdash U\oft\ctype}
      {\Delta\vdash U \oft \itype}

      \quad

      \infer[\rl{k-app}]
      {\Delta\vdash \yux T C \oft \pieboth {X\oft U} K}
      {\Delta\vdash C \checks U & \Delta,X\oft U\vdash T \oft K}

      \quad

      \infer[\rl{k-con}]
      {\Delta\vdash \const a \oft K}
      {\Sigma(\const a) = K}

      \\[1em]
      \multicolumn{1}{l}{\boxed{\Delta;\Gamma\vdash E\synths T}
        ~~\mbox{$E$ synthesizes type $T$}}\vspace{.5em}\\

      \infer[\rl{t-app}]
      {\Delta;\Gamma\vdash\yux {E_1}{E_2}\synths T}
      {\Delta;\Gamma\vdash E_1 \synths S\to T &
        \Delta;\Gamma\vdash E_2 \checks S}

      \qquad

      \infer[\rl{t-app-index}]
      {\Delta;\Gamma\vdash\yux E C\synths \asub{C/X}T}
      {\Delta;\Gamma\vdash E\synths\pieg{X\oft U}T ~~* = \{i,e\} &
        \Delta\vdash C\oft U}

      \vs

      \infer[\rl{t-const}]
      {\Delta;\Gamma\vdash \const c \synths T}
      {\Sigma(\const c) = T}

      \quad

      \infer[\rl{t-var}]
      {\Delta;\Gamma\vdash x \synths T}
      {\Gamma(x) = T}

      \qquad

      \infer[\rl{t-ann}]
      {\Delta;\Gamma\vdash E\oft T\synths T}
      {\Delta;\Gamma\vdash E\checks T}

      \qquad

      \vs 
      \multicolumn{1}{l}{\boxed{\Delta;\Gamma\vdash E\checks
          T}~~\mbox{$E$ type checks against type $T$}}\vs

      \infer[\rl{t-syn}]
      {\Delta;\Gamma\vdash E \checks T}
      {\Delta;\Gamma\vdash E \synths T}

      \qquad

      \infer[\rl{t-fn}]
      {\Delta;\Gamma\vdash (\fne x E) \checks T_1\to T_2}
      {\Delta;\Gamma,x\oft T_1\vdash E \checks T_2}

      \vs

      \infer[\rl{t-mlam}]
      {\Delta;\Gamma\vdash (\mlame X E) \checks \pieg{X\oft U} T}
      {\Delta,X\oft U;\Gamma\vdash E \checks T ~~~ * = \{i,e\}}

      \qquad

      \infer[\rl{t-case}]
      {\Delta;\Gamma\vdash\casee E {\wvec{\B}}\checks T}
      {\Delta;\Gamma\vdash E\synths S &
        \Delta;\Gamma\vdash \wvec{\B}\checks S\to T}

      \vs 
      \multicolumn{1}{l}{\boxed{\Delta;\Gamma \vdash\Pi\Delta';\Gamma'. Pat
          \oft\theta\mapsto E\checks T}~~\mbox{Branch $\B = \Pi\Delta';\Gamma'. Pat
          \oft\theta\mapsto E$
          checks against type $T$}}\vs

      \infer[\rl{t-branch}]
      {\Delta;\Gamma\vdash
        \Pi\Delta';\Gamma'. Pat \oft\theta\mapsto E
        \checks S\to T}
      {\Delta'\vdash\theta\oft\Delta &
        \Delta';\Gamma'\vdash Pat\checks\asub{\theta}{S} &
        \Delta';\asub{\theta}\Gamma,\Gamma'\vdash E\checks\asub{\theta}{T}}

      \vs 
      \multicolumn{1}{l}{\boxed{\Delta;\Gamma\vdash Pat\checks
          T}~~\mbox{Pattern $Pat$ checks against $T$}}\vs

      \infer[\rl{t-pindex}]
      {\Delta;\Gamma\vdash C\checks U}
      {\Delta\vdash C\checks U}

      \qquad

      \infer[\rl{t-pvar}]
       {\Delta;\Gamma\vdash x \checks T}
      {\Gamma(x) = T}

      \qquad

       \infer[\rl{t-pcon}]
       {\Delta;\Gamma\vdash\yux{\const{c}}{\wvec{Pat}}\checks S}
       {\Sigma(\const c)=T &
        \Delta;\Gamma\vdash\wvec{Pat}\checks T\spineret S}

      \vs

      \multicolumn{1}{l}{\boxed{\Delta;\Gamma\vdash \wvec{Pat}\checks
          T\spineret S}~~\mbox{Pattern spine $\wvec{Pat}$ checks against $T$ and
          has result type $S$}}
      \vs

      \infer[\rl{t-spi}]
      {\Delta;\Gamma\vdash C\,\wvec{Pat}\checks\pie{X\oft U} T\spineret S}
      {\Delta \vdash C \checks  U &
        \Delta;\Gamma\vdash\wvec{Pat}\checks[C/X]T\spineret S}

      \quad

      \infer[\rl{t-sarr}]
      {\Delta;\Gamma\vdash Pat\,\wvec{Pat}\checks T_1\to T_2\spineret S}
      {\Delta;\Gamma\vdash Pat\checks T_1 &
        \Delta;\Gamma\vdash\wvec{Pat}\checks T_2\spineret S}

      \quad

      \infer[\rl{t-snil}]
      {\Delta;\Gamma\vdash \cdot \checks S\spineret S}
      {}

    \end{array}
  \end{displaymath}
  \caption{Typing of Computational Expressions}
  \label{fig:comptyp}
\end{sidewaysfigure}

The kinding and typing rules for our core language are given in
Fig.~\ref{fig:comptyp}. We use a bidirectional type
system~\citep{Pierce:2000} for the target language which is similar to
the one in~\citet{Cave:POPL12} but we simplify the presentation by
omitting recursive types. Instead we assume that constructors together
with their types are declared in a signature $\Sigma$. We choose a
bi-directional type-checkers because it minimizes the need for
annotations by propagating known typing information in the checking
phase (judgment $\Delta;\Gamma\vdash E\checks T$) and inferring the
types when it is possible in the synthesis phase (judgment
$\Delta;\Gamma\vdash E \synths T$).

We rely on the fact that our index domain comes with rules which check
that a given index object is well-typed. This is described by the
judgment: $\Delta \vdash C : U$.

We check the introductions forms for functions $\fne x e$ and
dependent functions $\mlame x e$ against their respective types.
Dependent functions check against both $\pie {X\oft U} T$ and
$\piei {X\oft U} T$ where types are annotated with $e$ for explicit
quantification and $i$ for implicit quantification filled in by
elaboration. Their corresponding eliminations, application
$\yux {E_1} {E_2}$ and dependent application $\yux E \ibox C$,
synthesize their type. We rely in this rule on the index-level
substitution operation and we assume that it is defined in such a way
that normal forms are preserved\footnote{In Beluga, this is for
  example achieved by relying on hereditary
  substitutions\citep{Cave:POPL12}.}.

To type-check a case-expressions $\casee E {\vec{B}}$ against $T$, we synthesize
a type $S$ for $E$ and then check each branch against $S \to T$.
A branch $\Pi \Delta';\Gamma'.Pat \oft \theta\mapsto E$ checks against $S \to
T$, if: 1) $\theta$ is a refinement substitution mapping all index variables
declared in $\Delta$ to a new context $\Delta'$, 2)
the pattern $Pat$ is compatible with the type $S$ of the scrutinee,
i.e. $Pat$ has type $\asub{\theta}S$, and the body $E$ checks against
$\asub{\theta}T$ in the index context $\Delta'$ and the program context
$\asub{\theta}\Gamma, \Gamma_i$. Note that the refinement substitution
effectively performs a context shift.

We present the typing rules for patterns in spine format which will simplify our
elaboration and inferring types for pattern variables. We start checking a
pattern against a given type and check index objects and variables against the
expected type. If we encounter $\yux{\const{c}}{\wvec{Pat}}$ we look up the type $T$
of the constant $\const{c}$ in the signature and continue to check the spine
$\wvec{Pat}$ against $T$ with the expected return type $S$. Pattern spine typing
succeeds if all patterns in the spine $\wvec{Pat}$ have the corresponding type
in $T$ and yields the return type $S$.

\filbreak
\subsection{Elaborated Examples}

In Section~\ref{sec:certeval} we give an evaluator for a simply typed
lambda calculus that returns the result of the evaluation together
with the derivation tree needed to compute the value. The elaborated
version of function \lstinline~eval~ is:

\begin{lstlisting}
rec eval : <<$\piei$ T:[|- tp].>> {M : [|- term T]} Cert <<[|- T]>>[|- M] =
<<mlam T => >>mlam M => case [|- M] of
| <<$\Pi^i$.;.>> [|- one]<<: unit/T>> => Ex <<[|- unit]>> [|- v-one] [|- e-one]
| <<$\Pi^i$.T1:[|- tp], T2:[|- tp], M : [x:T1|-term T2];.>>
  [|- lam <<T1 T2>> (\x.M)] <<: arr T1 T2/ T>> =>
  Ex <<[|- arr T1 T2]>>
     [|- v-lam <<T1 T2>> (\x. M)]
     [|- e-lam <<(\x. M)>>]
| <<$\Pi^i$.T1:[|- tp], T2:[|- tp],
  M:[|- term (arr T1 T2)], N:[|- term T1];.>>
  [|- app <<T1 T2>> M N]<<: T2/T>> =>
    (case eval <<[|- arr T1 T2]>> [|- M] of
    | <<$\Pi^i$.T1:[|- tp],T2:[|- tp],M':[x:T1|-term T2],
      D: [|- big-step (arr T1 T2) M'
            (v-lam (\x.M))];.>>
      Ex <<[|- arr T1 T2]>>[|- v-lam <<(arr T1 T2)>> M'][|- D]<<:.>>=>
      (case eval <<[|- T2]>> [|- M' [N]] of
     | <<$\Pi^i$.T2:[|- tp], N':[|- val T2],
       D':[|- big-step T2 (M' [N]) N'];.>>
       Ex <<[|- T2]>> [|- N'][|- D']<<:.>> =>
         Ex <<[|- T2]>> [|- N']
            [|- e-app <<M M' N N'>> D D']))
;
\end{lstlisting}

To elaborate a recursive declaration we start by reconstructing the
type annotation given to the recursive function. In this case the user
left the variable \lstinline!T! free which becomes an implicit
argument and we abstract over this variable with
\lstinline!<<$\piei$ T:[|- tp]>>! marking it implicit. Notice however how
the user explicitly quantified over \lstinline!M! this means that
callers of \lstinline!eval! have to provide the term \lstinline!M!
while parameter \lstinline !T! will be omitted and inferred at each
calling point. Next, we elaborate the function body given the fully
elaborated type. We therefore add the corresponding abstraction\\
\lstinline!<<mlam T=> >>! for the implicit parameter.

Elaboration proceeds recursively on the term. We reconstruct the
case-expression, considering first the scrutinee \lstinline![M]! and we
infer its type to be \lstinline![term T]!. We elaborate the branches
next. Recall that each branch in the source language consists of a
pattern and a body. Moreover, the body can refer to any variable in
the pattern or variables introduced in outer patterns. However, in the
target language branches abstract over the context $\Delta;\Gamma$ and
add a refinement substitution $\theta$. The body of the branch refers
to variables declared in the branch contexts only.  In each branch, we
list explicitly the index variables and pattern variables.  For
example in the branch for \lstinline![lam M]! we added \lstinline!T1!
and \lstinline!T2! to the index context $\Delta$ of the branch,
index-level reconstruction adds these variables to the
pattern. 
The refinement substitution moves terms from the outer context to the
branch context, refining the appropriate index variables as expressed
by the pattern. For example in this branch, the substitution refines
the type \lstinline!<<[T]>>!  to the type \lstinline!<<[arr T1 T2]>>!. And in
the \lstinline!<<[one]>>!  branch it refines the type
\lstinline!<<[T]>>! to \lstinline!<<[unit]>>!.

As we mentioned before, elaboration adds an implicit parameter to the
type of function \lstinline!eval!, and the user is not allowed to
directly supply an instantiation for it. Implicit parameters have to be
inferred by elaboration. In the recursive calls to \lstinline!eval!, we add
the parameter that represents the type of the term being evaluated.

The output of the elaboration process is a target language term that
can be type checked with the rules from Figure~\ref{fig:comptyp}.

If elaboration fails it can either be because the source level program
describes a term that would be ill-typed when elaborated, or in some
cases, elaboration fails because it cannot infer all the implicit
parameters. This might happen if unification for the index language is
undecidable, as is for example the case for contextual LF.  In this
case, annotations are needed when the term falls outside the strict
pattern fragment where unification is decidable; this is rarely a
problem in practice. For other index languages where unification is
decidable, we do not expect such annotations to be necessary.






%% file: reconstruction/elaboration.tex
\section{Description of Elaboration}

Elaboration of our source-language to our core target language is
guided by the expected target type using a bi-directional algorithm.
Recall that we mark in the target type the arguments which are
implicitly quantified (e.g.: $\piei {X \oft U} T$). This annotation is
added when we elaborate a source type with free variables. If we check
a source expression against $\piei {X \oft U} T$ we insert the
appropriate mlam-abstraction in our target. When we switch between
synthesizing a type $S$ for a given expression and checking an
expression against an expected type $T$, we will rely on unification
to make them equal. A key challenge is how to elaborate
case-expressions where pattern matching a dependently typed expression
of type $\tau$ against a pattern in a branch might refine the type
$\tau$. Our elaboration is parametric in the index domain, hence we
keep our definitions of holes, instantiation of holes and unification
abstract and only state their definitions and properties.

\begin{figure*}
  \centering
  \begin{displaymath}
  \begin{array}{c}
    \infer[\eltyp]
    {\vdash {\const c} : t \recon \piei {(\Delta_i,\ahsub{\ep}\Delta)}\ahsub{\ep}T}
    {\cdot;\cdot\mid\cdot \vdash t \recon T / \Theta ; \Delta; \cdot &
      \Delta_i \vdash \ep : \Theta}

    \vs

    \infer[\elkind]
    {\vdash {\const a} : k \recon \piei {(\Delta_i,\ahsub{\ep}\Delta)}\ahsub{\ep}K}
    {\cdot;\cdot\mid\cdot \vdash k \recon K / \Theta ; \Delta; \cdot &
      \Delta_i \vdash \ep : \Theta}



    \vs

    \infer[\elrec]
    {\ \vdash \rece f t e \recon
      \rece f {\piei {(\Delta_i,\ahsub{\ep}\Delta)} \ahsub{\ep}T} E }
    {
    \begin{array}{c}
        \cdot; \cdot\mid\cdot \vdash t \recon T/\Theta;\Delta;\cdot \\
        \Delta_i \vdash \ep : \Theta\\
         \cdot; f\oft \piei {\Delta_i,\ahsub{\ep}\Delta} T \vdash
             \rclo e \cdot \reconChk{\piei {(\Delta_i,\ahsub{\ep}\Delta)} \ahsub{\ep}T} E/\cdot ; \cdot
    \end{array}
    }

  \end{array}
  \end{displaymath}
  \caption{Elaborating Declarations}
  \label{fig:eldec}
\end{figure*}

\subsection{Elaboration of Index Objects}
To elaborate a source expression, we insert holes for omitted index
arguments and elaborate index objects which occur in it. We
characterize holes with contextual objects as
in~\citep{Nanevski:ICML05, Pientka:TOCL09}. Contextual objects encode
the dependencies on the context that the hole might have. We hence
make a few requirements about our index domain. We assume:

\begin{enumerate}
\item A function $\genHole (\hole Y \oft \Delta\vdash U)$ that generates a term
  standing for a hole of type $U$ in the context $\Delta$, i.e. its
  instantiation may refer to the index variables in $\Delta$. If the
  index language is first-order, then we can characterize holes for
  example by meta-variables \citep{Nanevski:ICML05}. If our index
  language is higher-order, for example if we choose contextual LF as
  in Beluga, we characterize holes using meta$^\textbf{2}$-variables as
  described in~\citet{Boespflug:LFMTP11}. As is common in these meta-variable
  calculi, holes are associated with a delayed substitution $\theta$
  which is applied as soon as we know what $Y$ stands for.

\filbreak
\item A typing judgment for guaranteeing that index objects with holes are
  well-typed:
\[
\begin{array}{ll}
\Theta ; \Delta \vdash C \oft U  & \mbox{Index object $C$ has index type $U$ in
  context $\Delta$}\\
& \mbox{and all holes in $C$ are declared in $\Theta$}
\end{array}
\]

where $\Theta$ stores the hole typing assumptions:

\[
\begin{array}{lrcl}
  \mbox{Hole Context} & \Theta & \bnfas &  \cdot
  \bnfalt \Theta,\hole X\oft \Delta\vdash U
\end{array}
\]

\item Unification algorithm which finds the most general unifier for
  two index objects. In Beluga, we rely on the higher-order
  unification; more specifically, we solve eagerly terms which fall
  into the pattern fragment \citep{Miller91iclp,Dowek96jicslp} and
  delay others~\citep{Abel:TLCA11}. A most general unifier exists if
  all unification constraints can be solved. Our elaboration relies on
  unifying computation-level types which in turn relies on unifying
  index-level terms; technically, we in fact rely on two unification judgments: one finding
  instantiations for holes in $\Theta$, the other finding most general
  instantiations for index variables defined in $\Delta$ such that two
  index terms become equal. We use the first one during elaboration
  when unifying two computation-level types; the second one is used when
  computing the type refinement in branches.

\[
  \begin{array}{rcll}
      \Theta;\Delta & \vdash & C_1 \doteq C_2 /\Theta';\rho
& \quad \mbox{where:}~
      \Theta'\vdash\rho\oft\Theta \\
      \Delta & \vdash & C_1 \doteq C_2 /\Delta';\theta
& \quad \mbox{where:}~
      \Delta'\vdash\theta\oft\Delta
  \end{array}
\]

where $\rho$ describes the instantiation for holes in $\Theta$. If unification
succeeds, then we have $\ahsub\rho {C_1} =\ahsub\rho {C_2}$ and
$\asub{\theta}C_1 = \asub{\theta}C_2$ respectively.

\item Elaboration of index objects themselves. If the index language is simply
  typed, the elaboration has nothing to do; however, if as in Beluga, our index
  objects are objects described in the logical framework LF, then we need to
  elaborate them and infer omitted arguments following
  \citep{Pientka:JFP13}. 
\end{enumerate}

\subsubsection{Index Language Elaboration}

In this chapter we focus on the elaboration of the computational
language, therefore we assume the existence of the counter part for the
index domain. In this section we summarize the requirements on the
index domain. 







\subsubsection*{Well-typed Index Objects (target)}

\[
\begin{array}{ll}
\Delta \vdash C : U  & \mbox{Index object $C$ has index type $U$ in context $\Delta$}
\end{array}
\]

Substitution $C/X$ in an index object $C'$ is defined as $[C/X] C'$.

\subsubsection*{Well-typed Index Objects with Holes}

\[
\begin{array}{ll}
\Theta ; \Delta \vdash C : U  & \mbox{Index object $C$ has index type $U$ in
  context $\Delta$}\\
& \mbox{and all holes in $C$ are well-typed wrt $\Theta$}
\end{array}
\]

\[
  \begin{array}{lrcl}
    \mbox{Hole types} &  & \bnfas & \Delta\vdash U \\
    \mbox{Hole Contexts} & \Theta & \bnfas & \cdot \bnfalt \Theta,\hole{X}\oft \Delta\vdash U\\
    \mbox{Hole Inst.} & \rho & \bnfas & \cdot \bnfalt \rho,\hat\Delta\vdash C/\hole X \\
  \end{array}
\]

When we insert hole variables for omitted arguments in a given context
$\Delta$, we rely on the abstract function
$\genHole (\hole Y : \Delta\vdash U)$ which returns an index term containing
a new hole variable. When instantiating a hole we only need the names
of the variables in the context for $\alpha$-conversion, we represent
this as $\hat\Delta$. We assume the existence of an operation
$\ids\Theta$ that computes the identity substitution on the hole
context $\Theta$ by replacing each variable with itself.

\[
\begin{array}{lcl}
\genHole (\hole Y:\Delta\vdash U) & =  & C
\qquad \mbox{where $C$ describes a hole. }
\end{array}
\]


\subsubsection*{Unification of index objects}
The notion of unification that elaboration needs depends on the
index level language. As we mentioned, we require that equality on our index
domain is decidable; for elaboration, we also require that there is a decidable
unification algorithm which makes two terms equal. In fact, we need two forms:
one which allows us to infer instantiations for holes and another which unifies
two index objects finding most general instantiations for index variables such
that the two objects become equal. We
use the first one during elaboration, the second one is used to make two index
objects equal as for example during matching.

\[
  \begin{array}{rcll}
      \Theta;\Delta & \vdash & C_1 \doteq C_2 /\Theta';\rho
& \quad \mbox{where:}~
      \Theta'\vdash\rho\oft\Theta \\
      \Delta & \vdash & C_1 \doteq C_2 /\Delta';\theta
& \quad \mbox{where:}~
      \Delta'\vdash\theta\oft\Delta
  \end{array}
\]

where $\rho$ describes the instantiation for holes in $\Theta$. If unification
succeeds, then we have $\ahsub\rho {C_1} =\ahsub\rho {C_2}$ and
$\asub{\theta}C_1 = \asub{\theta}C_2$ respectively.


\subsubsection{Elaboration of index objects}
We describe the elaboration of index objects themselves. There are two
related judgements of elaboration for index objects that we use:

\begin{displaymath}
  \begin{array}{rl@{~~}l}
     \Theta;\Delta & \vdash c : U  ~~~~~~~~~ \recon &  C / \Theta';\Delta';\rho \\
     \Theta;\Delta & \vdash \rclo c \theta  \reconChk U &  C/\Theta';\rho \\

  \end{array}
\end{displaymath}

The first judgment elaborates the index object $c$ by checking it against
$U$. We thread through a context $\Theta$ of holes and a context of index
variables $\Delta$, we have seen so far. The object $c$ however may contain
additional free index variables whose type we infer during elaboration. All variables occurring
in $C$ will be eventually declared with their corresponding type in $\Delta'$.
As we elaborate $c$, we may refine holes and add additional holes. $\rho$
describes the mapping between $\Theta$ and $\Theta'$, i.e. it records
refinement of holes. Finally, we know that $\Delta' = \ahsub{\rho}\Delta,
\Delta_0$, i.e. $\Delta'$ is an extension of $\Delta$. We use the first
judgment in elaborating patterns and type declarations in the signature.

The second judgment is similar to the first, but does not allow free index
variables in $c$. We elaborate $c$ together with a refinement substitution
$\theta$, which records refinements obtained from earlier branches.
When we encounter an index variable, we look up what it is mapped to in $\theta$
and return it. Given a hole context $\Theta$ and a index variable context
$\Delta$, we elaborate an index term $c$ against a given type $U$. The result is
two fold: a context $\Theta'$ of holes is related to the original hole context
$\Theta$ via the hole instantiation $\rho$. We use the second judgment to
elaborate index objects embedded into target expressions.

\subsection{Elaborating Declarations}
We begin our discussion of elaborating source programs in a top-down
manner starting with declarations, the entry point of the
algorithm. Types and kinds in declarations may contain free variables
and there are two different tasks: we need to fill in omitted
arguments, infer the type of free variables and abstract over the free
variables and holes which are left over in the elaborated type /
kind. We rely here on the fact that the index language provides a way
of inferring the type of free variables.

To abstract over holes in a given type $T$, we employ a lifting operation:
$\Delta \vdash \ep : \Theta$ which maps each hole to a fresh index variable.

\[
\begin{array}{l}
\infer{\cdot\vdash \cdot : \cdot}{} \qquad
\infer{
\Delta, X:U~~\vdash \ep,~(\vdash X)/X :\Theta, X:(\vdash U)}
{\Delta \vdash \ep : \Theta}
\end{array}
\]

We require that holes are closed (written as $\vdash U$ and $\vdash X$
respectively where the context associated with a hole is empty);
otherwise lifting fails. \footnote{In practice this seems to be not an
  important restriction, for instance none of Beluga's examples need
  to reconstruct open holes.} In other words, holes are not allowed to
depend on some local meta-variables.

We use double brackets (i.e. $\ahsub{\ep}M$) to represent the
application of the lifting substitutions and hole instantiation
substitutions. We use this to distinguish them from regular
substitutions such as the refinement substitutions in the target
language.

Elaborating declarations requires three judgments. One for constants
and one for kinds to be able to reconstruct inductive type
declarations, and one for recursive functions. These judgments are:

\begin{displaymath}
  \begin{array}{rll@{~~}l}
     \vdash & \const c\oft t &  ~~~~~~~~~ \recon &  T \\
     \vdash & \const a\oft k &  ~~~~~~~~~ \recon  & K \\
     \vdash & \rece f t e    &  ~~~~~~~~~ \recon  & \rece f T E
  \end{array}
\end{displaymath}

The elaboration of declarations succeeds when the result does not
contain holes.

Figure~\ref{fig:eldec} shows the rules for elaborating declarations.
To elaborate a constant declaration ${\const c} : t$ we elaborate the
type $t$ to a target type $T$ where free index variables are listed in
$\Delta$ and the remaining holes in $T$ are described in $\Theta$. We
then lift all the holes in $\Theta$ to proper declarations in
$\Delta_i$ via the lifting substitution $\ep$. The final elaborated
type of the constant $\const c$ is:
$\piei {(\Delta_i, \ahsub{\ep}\Delta)} \ahsub{\ep}T$. Note that both
the free variables in the type $t$ and the lifted holes described in
$\Delta_i$ form the implicit arguments and are marked with $\Pi^i$.
For example in the certifying evaluator from
Section~\ref{sec:certeval}, the type of the constructor \lstinline!Ex!
is reconstructed to:

\begin{lstlisting}
$\Pi^i$T:[|- tp],M:[|- term T].$\Pi^e$N:[|- value T]. [|- big-step T M N]
    -> Cert [|- T][|- M]
\end{lstlisting}

The elaboration of kinds follows the same principle.
Section~\ref{sec:elabkindandtype} explains the details for the
elaboration of types and kinds.

To elaborate recursive function declarations, we first elaborate the type $t$
abstracting over all the free variables and lifting the remaining holes to
obtain $\piei {(\Delta_i, \ahsub{\ep}\Delta)} \ahsub{\ep}T$. Second, we assume
$f$ of this type and elaborate the body $e$ checking it against
$\piei {(\Delta_i, \ahsub{\ep}\Delta)} \ahsub{\ep}T$. We note that we always
elaborate a source expression $e$ together with a possible refinement
substitution $\theta$. In the beginning, $\theta$ will be empty. We describe
elaboration of source expressions in Section~\ref{sec:elabexpr}.

\clearpage
\begin{sidewaysfigure}
  \begin{displaymath}
    \begin{array}{c}
      \multicolumn{1}{l}
      {\boxed{\Theta;\Delta\fv\mid\Delta\vdash k\recon K/\Theta';\Delta\fv';\rho'}~~\mbox{
          Elaborate kind $k$ to target kind $K$}}\vs

      \infer[\rl{el-k-ctype}]
      {\Theta;\Delta\fv\mid\Delta\vdash\ctype\recon\ctype/\Theta;\Delta\fv;\ids\Theta}
      {}

      \vs

      \infer[\rl{el-k-pi}]
      {\Theta;\Delta\fv\mid\Delta\vdash\piesrc{X\oft u} k\recon
        \pie {X\oft(\ahsub{\rho'} U)} K/\Theta'';\Delta\fv'';\rho''\circ\rho'}
      {\Theta;\Delta\fv\mid\Delta\vdash u\recon U/\Theta';\Delta\fv';\rho' &
        \Theta';\Delta\fv'\mid\Delta,X\oft U\vdash k\recon K/\Theta'';\Delta\fv'';\rho''}

      \vs

      \multicolumn{1}{l}
      {\boxed{\Theta;\Delta\fv\mid\Delta\vdash t\recon T/\Theta';\Delta\fv';\rho'}~~\mbox{
          Elaborate type $t$ to target type $T$}}\vs

      \infer[\rl{el-t-arr}]
      {\Theta;\Delta\fv\mid\Delta\vdash t_1\to t_2 \recon
        (\ahsub{\rho''}T_1)\to T_2/\Theta'';\Delta\fv'';\rho''\circ\rho'}
      {\Theta;\Delta\fv\mid\Delta\vdash t_1\recon T_1/\Theta';\Delta\fv';\rho'&
      \Theta';\Delta'\fv\mid\Delta\vdash t_2\recon T_2/\Theta'';\Delta\fv'';\rho''}

    \vs

    \infer[\rl{el-t-idx}]
    {\Theta;\Delta\fv\mid\Delta\vdash\ibox u\recon\ibox U/\Theta';\Delta\fv';\rho'}
    {\Theta;\Delta\fv\mid\Delta\vdash u\recon U/\Theta';\Delta\fv';\rho'}

    \vs

    \infer[\rl{el-t-pi}]
    {\Theta;\Delta\fv\mid\Delta\vdash\piesrc{X\oft u} t \recon
      \pie {X\oft(\ahsub{\rho'}U)}T/\Theta'';\Delta\fv'';\rho''\circ\rho' }
    {\Theta;\Delta\fv\mid\Delta\vdash u\recon U/\Theta';\Delta\fv';\rho' &
      \Theta';\Delta\fv'\mid\Delta,X\oft U\vdash t\recon T/\Theta'';\Delta\fv'';\rho''}

    \vs

    \infer[\rl{el-t-con}]
    {\Theta;\Delta\fv\mid\Delta\vdash\yux{\const{a}}{\wvec{\ibox c}}\recon
      \yux{\const{a}}{\wvec{C}}/\Theta';\Delta\fv';\rho'}
    {\Sigma(\const a) = K &
      \Theta;\Delta\fv\mid\Delta\vdash\wvec{\ibox c}\reconChk K
      \wvec{\ibox C}/\Theta';\Delta\fv';\rho'}

    \vs

    \multicolumn{1}{l}
    {\boxed{\Theta;\Delta\fv\mid\Delta\vdash\wvec{\ibox c}\reconChk K
        \wvec{\ibox C}/\Theta';\Delta\fv';\rho'}~~\mbox{
        Elaborate spine $\wvec{\ibox c}$ checking against kind $K$
        to spine $\wvec{\ibox C}$}}\vs

    \infer[\rl{el-t-sp-explicit}]
    {\Theta;\Delta\fv\mid\Delta\vdash\yux{\ibox c}{\wvec{\ibox c}}\reconChk{\pie{X\oft U} K}
      \yux {(\ahsub{\rho'}\ibox C)} {\wvec{\ibox C}}
      /\Theta'';\Delta\fv'';\rho''\circ\rho'}
    {\Theta;\Delta\fv \mid \Delta\vdash c\reconChk U C/\Theta';\Delta\fv';\rho' &
      \Theta';\Delta\fv'\mid\Delta\vdash\wvec{\ibox c}\reconChk {\asub{C/X} K}
      \wvec{\ibox C}/\Theta'';\Delta\fv'';\rho''}

    \vs

    \infer[\rl{el-t-sp-implicit}]
    {\Theta;\Delta\fv \mid \Delta\vdash\wvec c\reconChk{\piei{X\oft U} K}
      \yux{(\ahsub{\rho'}C)}{\wvec C}
      /\Theta';\Delta\fv';\rho'}
    {\genHole (\hole Y\oft (\Delta\fv,\Delta).U) = C &
      \Theta;\Delta\fv\mid\Delta \vdash \wvec c\reconChk {\asub{C/X} K}
      \wvec C/\Theta';\Delta\fv';\rho'}

    \vs

    \infer[\rl{el-t-sp-empty}]
    {\Theta;\Delta\fv\mid\Delta\vdash\cdot\reconChk\ctype\cdot/\Theta;\Delta\fv;\ids\Theta}
    {}

    \end{array}
  \end{displaymath}
  \caption{Elaborating Kinds and Types in Declarations}
  \label{fig:reckt}
\end{sidewaysfigure}
\clearpage

\subsection{Elaborating Kinds and Types in Declarations}\label{sec:elabkindandtype}

Recall that programmers may leave index variables free in type and kind
decarations. Elaboration must infer the type of the free index variables in
addition to reconstructing omitted arguments. We require that the index language
provides us with the following judgments:

\[
\begin{array}{l}
  \Theta;\Delta\fv\mid\Delta\vdash u\recon U/\Theta';\Delta\fv';\rho'\\
  \Theta;\Delta\vdash\rclo u\theta\recon U/\Theta';\rho'\\
\end{array}
\]

Hence, we assume that the index language knows how to infer the type of free
variables, for example. In Beluga where the index language is LF, we fall back
to the ideas described by \citet{Pientka:JFP13}.

The first judgment collects free variables in $\Delta\fv$ that later in elaboration
will become implicit parameters. The context $\Delta\fv$ is threaded
through in addition to the hole context $\Theta$.

The judgments for elaborating computation-level kinds and types are similar:

\begin{enumerate}
\item $\Theta;\Delta\fv\mid\Delta\vdash k\recon K/\Theta';\Delta\fv';\rho'$
\item $\Theta;\Delta\fv\mid\Delta\vdash t\recon T/\Theta';\Delta\fv';\rho'$
\item $\Theta;\Delta\fv\mid\Delta\vdash\wvec{\ibox c}\reconChk K
        \wvec{\ibox C}/\Theta';\Delta\fv';\rho'$
\end{enumerate}

We again collect free index variables in $\Delta\fv$ which are
threaded through together with the holes context $\Theta$ (see
Figure~\ref{fig:eldec} and Figure~\ref{fig:reckt}). As notation note,
we use $\Delta_f$ to store the free index variables, $\Delta$ for the
bound ones, and $\Theta$ for holes.

\begin{figure}
  \centering
\begin{displaymath}
  \begin{array}{c}
\multicolumn{1}{l}{\boxed{\Theta;\Delta;\Gamma\vdash \rclo e \theta \reconChk T
    E/\Theta';\rho}}:\\
\multicolumn{1}{l}{\mbox{Elaborate source $\rclo e \theta$ to target
    expression $E$ checking against type $T$}}\vs

\infer[\elbox]
{\Theta;\Delta;\Gamma\vdash\rclo{\ibox{c}}\theta\reconChk{\ibox U}\ibox C/\Theta';\rho}
{\Theta;\Delta\vdash\rclo c \theta\reconChk U C/\Theta';\rho}

\vs

    \infer[\elfn]
    {\Theta;\Delta;\Gamma\vdash\rclo{\fne x e}\theta \reconChk{T_1\to T_2} \fne x E/\Theta';\rho }
    {\Theta;\Delta;\Gamma,x\oft T_1\vdash \rclo e\theta \reconChk{T_2} E/\Theta';\rho}

    \vs

    \infer[\elmlami]
    {\Theta;\Delta;\Gamma\vdash\rclo{e}\theta \reconChk{\piei{X\oft U}T} \mlame X E/\Theta';\rho }
    {\Theta;\Delta,X\oft U;\Gamma\vdash \rclo e{\theta, X/X} \reconChk{T} E/\Theta';\rho}

   \vs

    \infer[\elmlam]
    {\Theta;\Delta;\Gamma\vdash\rclo{\mlame X e}\theta \reconChk{\pie{X\oft U}T} \mlame X E/\Theta';\rho }
    {\Theta;\Delta,X\oft U;\Gamma\vdash \rclo e{\theta, X/X} \reconChk{T} E/\Theta';\rho}

    \vs

    \infer[\elcase]
    {\Theta;\Delta;\Gamma\vdash\rclo{\casee e {\vec b}}\theta\reconChk T
      \casee E {\vec B} / \cdot;\rho}
    {\begin{array}{c}
       \Theta;\Delta;\Gamma\vdash \rclo e\theta\recon E\oft S/ \cdot ;\rho \\
      \ahsub\rho\Delta;\ahsub\rho\Gamma\vdash \rclo{\vec b}{\ahsub\rho\theta}\reconChk{S\to \ahsub{\rho}T}\vec B
     \end{array}
    }

\vs

\infer[\elsyn]{\Theta;\Delta;\Gamma\vdash \rclo e\theta \reconChk T
                \ahsub{\rho'} E/\Theta_2;\rho'\circ\rho}
  {\Theta;\Delta;\Gamma\vdash \rclo e\theta \recon E\oft T_1/\Theta_1;\rho
    &
    \Theta_1 ;\ahsub{\rho}{\Delta} \vdash T_1 \doteq \ahsub\rho T/\Theta_2;\rho'}
  \end{array}
\end{displaymath}
  \caption{Elaboration of Expressions (Checking Mode)}
  \label{fig:elexp}
\end{figure}

\subsection{Elaborating Source Expressions}\label{sec:elabexpr}
We elaborate source expressions bidirectionally.  Expressions such as
non-dependent functions and dependent functions are elaborated by
checking the expression against a given type; expressions such as
application and dependent application are elaborated to a
corresponding target expression and at the same time synthesize the
corresponding type. This approach seeks to propagate the typing
information that we know in the checking rules, and in the synthesis
phase, to take advantage of types that we can infer.

\begin{displaymath}
  \begin{array}{ll@{~~}c@{~~}ll}
    \multicolumn{1}{l}{\mbox{Synthesizing:}} &
    \Theta;\Delta;\Gamma & \vdash & \rclo e \theta ~~~~ ~~\recon E\oft T/\Theta';\rho
\\    \multicolumn{1}{l}{\mbox{Checking:}}
&
    \Theta;\Delta;\Gamma & \vdash & \rclo e \theta \reconChk T E ~~~/\Theta';\rho
  \end{array}
\end{displaymath}

We first explain the judgment for elaborating a source expression $e$
by checking it against $T$ given holes in $\Theta$, index variables
$\Delta$, and program variables $\Gamma$. Because of pattern matching,
index variables in $\Delta$ may get refined to concrete index terms.
Abusing slightly notation, we write $\theta$ for the map of free
variables occurring in $e$ to their refinements and consider a source
expression $e$ together with the refinement map $\theta$, written as
$\rclo e \theta$. The result of elaborating $\rclo e \theta$ is a
target expression $E$, a new context of holes $\Theta'$, and a hole
instantiation $\rho$ which instantiates holes in $\Theta$, i.e.
$\Theta' \vdash \rho : \Theta$. The result $E$ has type
$\ahsub{\rho} T$. It is important to notice here, that $\rho$ contains
instances for holes, while $\theta$ contains refinements for
meta-variables (i.e. instances for meta-variables refined by dependent
pattern matching).

The result of elaboration in synthesis mode is similar; we return the
target expression $E$ together with its type $T$, a new context of
holes $\Theta'$ and a hole instantiation $\rho$, s.t.
$\Theta' \vdash \rho : \Theta$. The result is well-typed, i.e. $E$ has
type $T$.

We give the rules for elaborating source expressions in checking
mode in Fig.~\ref{fig:elexp} and in synthesis mode in Fig.~\ref{fig:elexpsyn}.
To elaborate a function (see rule $\elfn$) we simply elaborate the body
extending the context $\Gamma$. There are two cases when we elaborate an
expression of dependent function type. In the rule $\elmlam$, we elaborate
a dependent function $\mlame X e$ against $\pie{X\oft U}T$ by elaborating the
body $e$ extending the context $\Delta$with the declaration $X \oft U$. In the
rule $\elmlami$, we elaborate an expression $e$ against $\piei{X\oft U}T$ by
elaborating $e$ against $T$ extending the context $\Delta$ with  the declaration
$X \oft U$. The result of elaborating $e$ is then wrapped in a dependent
function.

\filbreak
When switching to synthesis mode, we elaborate $\rclo e \theta$ and
obtain the corresponding target expression $E$ and type $T'$ together
with an instantiation $\rho$ for holes in $\Theta$. We then unify the
synthesized type $T'$ and the expected type $\ahsub{\rho}T$ obtaining
an instantiation $\rho'$ and return the composition of the
instantiation $\rho$ and $\rho'$. When elaborating an index object
$\ibox c$ (see rule $\elbox$), we resort to elaborating $c$ in our
indexed language which we assume.

One of the key cases is the one for case-expressions. In the rule
$\elcase$, we elaborate the scrutinee synthesizing a type $S$; we then
elaborate the branches. Note that we verify that $S$ is a closed type,
i.e. it is not allowed to refer to holes. To put it differently, the
type of the scrutinee must be fully known. This is done to keep a type
refinement in the branches from influencing the type of the scrutinee.
The practical impact of this restriction is difficult to quantify,
however this seems to be the case for the programs we want to write. A
justification is that it is not a problem in any of the
examples of the Beluga implementation. For a similar reason, we
enforce that the type $T$, the overall type of the case-expression, is
closed; were we to allow holes in $T$, we would need to reconcile the
different instantiations found in different branches.

We omit the special case of pattern matching on index objects to save
space and because it is a refinement on the $\elcase$ rule where we
keep the scrutinee when we elaborate a branch. We then unify the
scrutinee with the pattern in addition to unifying the type of the
scrutinee with the type of the pattern. In the implementation of
Beluga, case-expressions on computation-level expressions (which do
not need to track the scrutinee and are described in this chapter) and
case-expressions on index objects (which do keep the scrutinee when
elaborating branches) are handled separately.

When elaborating a constant, we look up its type $T_c$ in the
signature $\Sigma$ and then insert holes for the arguments marked
implicit in its type (see Fig.~\ref{fig:elexpsyn}). Recall that all
implicit arguments are quantified at the outside, i.e.
$T_c = \piei {X_n\oft U_n} \ldots \piei {X_1 \oft U_1} S$ where $S$
does not contain any implicit dependent types $\Pi^i$. We generate for
each implicit declaration $X_k\oft U_k$ a new hole which can depend on
the currently available index variables $\Delta$. When elaborating a
variable, we look up its type in $\Gamma$ and because the variable can
correspond to a recursive function with implicit parameters we insert
holes for the arguments marked as implicit as in the constant case.

Elaboration of applications in the synthesis mode threads through the hole
context and its instantiation, but is otherwise straightforward. In each of the
application rules, we elaborate the first argument of the application obtaining
a new hole context $\Theta_1$ together with a hole instantiation $\rho_1$. We
then apply the hole instantiation $\rho_1$ to the context $\Delta$ and $\Gamma$
and to the refinement substitution $\theta$, before elaborating the second part.

\clearpage
\begin{sidewaysfigure}
  \centering\small
  \begin{displaymath}
    \begin{array}{c}
      \multicolumn{1}{l}{\boxed{\Theta ; \Delta \vdash E : T \recon E' : T' / \Theta' }}\\
      \multicolumn{1}{l}{\mbox{Apply $E$ to holes for representing omitted arguments based on $T$
          obtaining a term $E'$ of type $T'$}}\vs
      \infer[\elimp]
      {\Theta ; \Delta \vdash E : \piei {X\oft U}T  \recon E' \oft T'  ~/~\Theta'}
      {\genHole (\hole Y{:}\Delta\vdash U) = C &
        (\Theta,\hole Y\oft \Delta\vdash U) ;~ \Delta
        \vdash E~\ibox{C} : \asub{C/X}T \recon E' \oft T' ~/~\Theta'}
      \vs
      \infer[\elimpd]{\Theta ; \Delta  \vdash E \oft S \recon E \oft S ~/~\Theta}{
        S \neq \piei {X\oft U}T}
      \vs
      \multicolumn{1}{l}{\boxed{\Theta;\Delta;\Gamma\vdash \rclo e \theta \recon
          E\oft T/\Theta';\rho }~~
        \mbox{Elaborate source $\rclo e \theta$ to target $E$ and synthesize type
          $T$}}\vs
      \infer[\elvar]
      {\Theta ; \Delta; \Gamma \vdash \rclo x \theta \recon E' \oft T'
        ~/~\Theta'; \ids{\Theta'}}
      {\Gamma(x) = T
        &
        \Theta;\Delta;\Gamma\vdash x\oft T\recon E'\oft T'/\Theta'}

      \quad

      \infer[\elconst]{\Theta ; \Delta; \Gamma \vdash \rclo {\const c} \theta \recon
        E : T
        ~/~\Theta'; \ids{\Theta'}}{\Sigma (c) = T_c &
        \Theta ; \Delta \vdash \const c : T_c \recon E : T ~/~\Theta'
      }

      \vs

      \infer[\elapp]
      {\Theta;\Delta;\Gamma\vdash\rclo{\yux {e_1}{e_2}}\theta \recon
        \yux{E_1}{E_2} : \ahsub{\rho_2}T ~/~\Theta_2;\rho_2\circ\rho_1}
      {\Theta;\Delta;\Gamma\vdash \rclo{e_1}\theta \recon E_1 \oft S \to T ~/~ \Theta_1;\rho_1
        &
        \Theta_1;\ahsub{\rho_1}\Delta;\ahsub{\rho_1}\Gamma\vdash
        \rclo{e_2}{\ahsub{\rho_1}\theta} \reconChk{\ahsub{\rho_1} S}
        E_2 ~/~\Theta_2;\rho_2
      }

      \vs

      \infer[\elmapp]
      {\Theta;\Delta;\Gamma\vdash\rclo{\yux {e}{\ibox c}}\theta \recon
        \yux{E_1}{\ibox C}\oft \asub{C/X}(\ahsub{\rho_2}T)/\Theta_2;\rho_2\circ\rho_1}
      {
        \Theta;\Delta;\Gamma\vdash \rclo {e} \theta\recon E_1\oft \pie{X\oft U}T/\Theta_1;\rho_1
        &
        \Theta_1;\ahsub{\rho_1}\Delta\vdash \rclo c {\ahsub{\rho_1}\theta}\reconChk U C/\Theta_2;\rho_2
      }
      \vs

      \infer[\elmappi]
      {\Theta;\Delta;\Gamma\vdash\rclo{\yux {e}\_}\theta  \recon E~\ibox{C} ~\oft~
        \asub{C/X}T ~~/ ~~\Theta_1,\hole Y\oft (\ahsub{\rho_1}\Delta).U~;~\rho}
      {\Theta;\Delta;\Gamma\vdash \rclo {e} \theta\recon E \oft \pie{X\oft U}T
        /\Theta_1;\rho &
        \genHole (\hole Y: (\ahsub{\rho}\Delta).U) = C
      }

      \vs

      \infer[\elann]
      {\Theta;\Delta;\Gamma\vdash \rclo{e\oft t}\theta\recon
        (E\oft T)\oft T /\Theta_2;\rho_2\circ\rho_1}
      {
        \Theta;\Delta\vdash \rclo t\theta \recon T/\Theta_1;\rho_1 &
        \Theta_1;\ahsub{\rho_1}\Delta;\ahsub{\rho_1}\Delta\vdash
        \rclo e{\ahsub{\rho_1}\theta}\reconChk T E/\Theta_2;\rho_2
      }

      \vs
      \multicolumn{1}{l}{
        \mbox{Where } \ids \Theta\ \mbox{returns the identity substitution for context $\Theta$ such as:}
        \Theta \vdash \ids\Theta \oft \Theta}

    \end{array}
  \end{displaymath}
  \caption{Elaboration of Expressions (Synthesizing Mode)}
  \label{fig:elexpsyn}
\end{sidewaysfigure}
\clearpage


%% file: reconstruction/elbranch.tex
\subsubsection{Elaborating Branches}

We give the rules for elaborating branches in Fig.~\ref{fig:elpat}.
Recall that a branch $pat \mapsto e$ consists of the pattern $pat$ and
the body $e$. We elaborate a branch under the refinement
$\theta$, because the body $e$ may contain index variables declared
earlier and which might have been refined in earlier branches.

Intuitively, to elaborate a branch, we need to elaborate the pattern
and synthesize the type of index and pattern variables bound inside of
it. In the dependently typed setting, pattern elaboration needs to do
however a bit more work: we need to infer implicit arguments which
were omitted by the programmer (e.g: the constructor \lstinline!Ex!
takes the type of the expression, and the source of evaluation as
implicit parameter \lstinline!Ex [T] [M]...!) and we need to establish
how the synthesized type of the pattern refines the type of the
scrutinee.

Moreover, there is a mismatch between the variables the body $e$ may
refer to (see rule $\rl{wf-branch}$ in Fig.~\ref{fig:wf}) and the
context in which the elaborated body $E$ is meaningful (see rule
$\rl{t-branch}$ in Fig.~\ref{fig:comptyp}).  While our source
expression $e$ possibly can refer to index variables declared prior,
the elaborated body $E$ is not allowed to refer to any index variables
which were declared at the outside; those index variables are replaced
by their corresponding refinements. To account for these additional
refinements, we not only return an elaborated pattern
$\Pi\Delta_r;\Gamma_r. Pat \oft\theta_r$ when elaborating a pattern
$pat$ (see rule $\rl{el-subst}$ in Fig.~\ref{fig:elpat}), but in
addition return a map $\theta_e$ between source variables declared
explicitly outside and their refinements.

\begin{figure}
  \centering\small
  \begin{displaymath}
    \begin{array}{c}
    \multicolumn{1}{l}{\boxed{\Delta;\Gamma\vdash \rclo b\theta \reconChk{S\to T} B}~~
      \mbox{Elaborate source branch $\rclo b \theta$ to branch $B$}}
      \vs

    \infer[\elbranch]
    {\Delta;\Gamma\vdash \rclo{pat \mapsto e}\theta \reconChk{S\to T}
      \Pi\Delta_r;\Gamma_r.{Pat}\oft\theta_r\mapsto E}
      {\begin{array}{c}
        \Delta\vdash pat \reconChk{S} \Pi \Delta_r ; \Gamma_r . Pat: \theta_r \mid \theta_e \\
        \cdot ; \Delta_r;\asub{\theta_r}\Gamma,\Gamma_r\vdash\rclo e{\theta_r\circ\theta,~\theta_e}\reconChk{\asub{\theta_r}{T}} E/\cdot;\cdot
      \end{array}}

\vs

    \multicolumn{1}{l}{\boxed{
    \Delta\vdash pat\reconChk{T} \Pi \Delta_r ; \Gamma_r . Pat: \theta_r \mid \theta_e}}\vs

    \infer[\elsubst]
    {\Delta\vdash pat\reconChk{S}
     \Pi \Delta_r ; \asub{\theta_p}\ahsub{\ep}\Gamma_p .
     \asub{\theta_p}\ahsub{\ep}Pat : \theta_r \mid \theta_e}
    {\begin{array}{c}
       \cdot;\cdot\vdash pat\recon \Pi \Delta_p;\Gamma_p .Pat : S_p/\Theta_p;\cdot \\
       \Delta_p' \vdash \ep : \Theta_p \\
       \ahsub \ep S_p \doteqdot S/\Delta_r;\theta_R
      \end{array}
    }

\vs

\multicolumn{1}{l}{\mbox{where~$\theta_R = \theta_r, \theta_p$~s.t.~
  $\Delta_r \vdash \theta_p : (\Delta'_p,\ahsub{\ep}\Delta_p)$ %
and~$\theta_p = \theta_i, \theta_e$~s.t.~
  $\Delta_r \vdash \theta_i : \Delta'_p$
}}\\[1em]
\end{array}
  \end{displaymath}
   \caption{Branches and Patterns}
  \label{fig:elpat}
\end{figure}

Elaborating a pattern is done in three steps (see rule \elsubst):
\begin{enumerate}
\item First, given $pat$ we elaborate it to a target pattern $Pat$
  together with its type $S_p$ synthesizing the type of index
  variables $\Delta_p$ and the type of pattern variables $\Gamma_p$
  together with holes ($\Theta_p$) which denote omitted
  arguments. This is accomplished by the first premise of the rule
  \elsubst:
  \[
  \cdot;\cdot\vdash pat\recon \Pi \Delta_p;\Gamma_p . Pat : S_1/\Theta_p; \cdot
  \]

  Our pattern elaboration judgment (Figure~\ref{fig:elpattern})
  threads through the hole context and the context of index variables,
  both of which are empty in the beginning. Because program variables
  occur linearly, we do not thread them through but simply combine
  program variable contexts when needed. The result of elaborating
  $pat$ is a pattern $Pat$ in our target language where $\Delta_p$
  describes all index variables in $Pat$, $\Gamma_p$ contains all
  program variables and $\Theta_p$ contains all holes, i.e. most
  general instantiations of omitted arguments. We describe pattern
  elaboration in detail in Section~\ref{sec:elpat}.

\item Second, we abstract over the hole variables in $\Theta_p$ by
  lifting all holes to fresh index variables from $\Delta_p'$. This is
  accomplished by the second premise of the rule \elsubst using the
  lifting substitution $\Delta'_p\vdash \ep \oft \Theta_p$.

\item Finally, we compute the refinement substitution $\theta_r$ which
  ensures that the type of the pattern $\ahsub \ep S_p$ is compatible
  with the type $S$ of the scrutinee. We note that the type of the
  scrutinee could also force a refinement of holes in the
  pattern. This is accomplished by the judgment:
  \[
  \Delta,(\Delta_p',\ahsub \ep\Delta_p) \vdash
  \ahsub \ep S_1 \doteqdot T_1/\Delta_r;\theta_R
  \qquad
  \theta_R = \theta_r,\theta_p
  \]

  We note because $\theta_R$ maps index variables from $\Delta,
  (\Delta_p',\ahsub \ep \Delta_p)$ to $\Delta_r$, it can be split in two
  parts: $\theta_r$ that provides refinements for variables $\Delta$ in the
  type of the scrutinee; $\theta_p$ provides possible refinements of
  the pattern forced by the scrutinee. This can happen, if the
  scrutinee's type is more specific than the type of the pattern.
\end{enumerate}

\subsubsection{Elaborating Patterns} \label{sec:elpat}
Pattern elaboration is bidirectional.
The judgments for elaborating patterns by checking them against a given type and
synthesizing their type are:
\[
\begin{array}{ll@{}c@{}l}
\mbox{Synthesizing:} & \Theta;\Delta & \vdash & pat~~~~~~\recon \Pi \Delta';\Gamma . Pat\oft
T~/~\Theta';\rho
\\
\mbox{Checking:} & \Theta;\Delta & \vdash & pat\reconChk{T} \Pi \Delta';\Gamma
. Pat~~~~~/~\Theta';\rho
\end{array}
\]

As mentioned earlier, we thread through a hole context $\Theta$
together with the hole substitution $\rho$ that relates:
$\Theta'\vdash\rho\oft\Theta$.  Recall that as our examples show
index-level variables in patterns need not to be linear and hence we
accumulate index variables and thread them through as well. Program
variables on the other hand must occur linearly, and we can simply
combine them.  The elaboration rules are presented in
Figure~\ref{fig:elpattern}. In synthesis mode, elaboration returns a
reconstructed pattern $Pat$, a type $T$ where $\Delta'$ describes
the index variables in $Pat$ and $\Gamma'$ contains all program
variables occurring in $Pat$.  The hole context $\Theta'$ describes
the most general instantiations for omitted arguments which have been
inserted into $Pat$. In checking mode, we elaborate $pat$ given a type
$T$ to the target expression $Pat$ and index variable context
$\Delta'$, pattern variable context $\Gamma'$ and the hole context
$\Theta'$.

Pattern elaboration starts in synthesis mode, i.e. either elaborating an
annotated pattern ($e:t$) (see rule \elpann) or
a pattern ${\const c}\;\wvec{pat}$ (see rule \elpcon). To reconstruct patterns that start with a
constructor we first look-up the constructor in the signature $\Sigma$
to get its fully elaborated type $T_c$ and then elaborate the arguments $\wvec{pat}$
against $T_c$. Elaborating the spine of arguments is guided by the type
$T_c$. If $T_c = \piei{X\oft U}T$, then we generate a new hole for the omitted
argument of type $U$. If $T_c = T_1 \to T_2$, then we elaborate the first
argument in the spine $pat\;\wvec{pat}$  against $T_1$ and the remaining
arguments $\wvec{pat}$ against $T_2$. If $T_c = \pie {X\oft U} T$, then we elaborate the first
argument in the spine $\ibox{c}\;\wvec{pat}$  against $U$ and the remaining
arguments $\wvec{pat}$ against $\asub{C/X}T$. When the spine is empty, denoted
by $\cdot$, we simply return the final type and
check that the constructor was fully applied by ensuring that the type $S$
we reconstruct against is either of index level type, i.e. $\ibox U$,
or a recursive type, i.e. $\const{a}\wvec{\ibox C}$.

\begin{sidewaysfigure}
  \centering
\begin{displaymath}
  \begin{array}{c}
\multicolumn{1}{l}{\mbox{Pattern (synthesis mode)}~~~\boxed{
        \Theta;\Delta\vdash pat\recon \Pi \Delta';\Gamma . Pat\oft T~/~\Theta'~;~\rho}}\\[1em]

    \infer[\elpcon]
    {\Theta;\Delta\vdash\yux{\const{c}}{\wvec{pat}}\recon
      \Pi \Delta';\Gamma . \yux{\const{c}}{\wvec{Pat}}\oft S~/~\Theta'; \rho}
    {\Sigma(c)=T &
      \Theta;\Delta\vdash\wvec{pat}\reconChk{T} \Pi \Delta';\Gamma.\wvec{Pat} \spineret
      S ~/~\Theta';\rho}

\qquad

\\[1em]
     \infer[\elpann]
     {\Theta;\Delta\vdash(pat\oft t)\recon \Pi \Delta'';\Gamma  . Pat\oft
       \ahsub{\rho'}T~/~\Theta''; \rho'}
     {\cdot;\cdot \vdash\rclo t {\cdot}\recon
       T/\Theta';\Delta';\cdot &
       (\Theta,\Theta')~;~ (\Delta, \Delta')\vdash pat\reconChk T
       \Pi \Delta'';\Gamma . Pat~/~\Theta''; \rho'}

\\[1em]
\multicolumn{1}{l}{\mbox{Pattern (checking mode)}~~~\boxed{
        \Theta;\Delta\vdash pat\reconChk{T} \Pi \Delta';\Gamma. Pat~/~\Theta';\rho}}\\[1em]
    \infer[\elpvar]
    {\Theta;\Delta\vdash x\reconChk{T}
      \Pi \Delta~;~x\oft T. x ~/~\Theta ; \ids\Theta}
    {}

\qquad

    \infer[\elpindex]
    {\Theta;\Delta\vdash\ibox c\reconChk{\ibox U} \Pi \Delta';\cdot ~.~ [C]/\Theta'~;~\rho}
    {\Theta;\Delta\vdash c\reconChk{U}C/\Theta';\Delta';\rho}

    \vs
    \infer[\elpsyn]
    {\Theta;\Delta\vdash pat\reconChk{T} \Pi \ahsub{\rho'}\Delta';\ahsub{\rho'}\Gamma~.~\ahsub\rho{Pat}~/~\Theta''~;~\rho'\circ\rho}
    {
        \Theta;\Delta\vdash pat\recon
        \Pi \Delta';\Gamma . Pat\oft S~/~\Theta';\rho \qquad
        \Theta';\Delta'\vdash S \doteq \ahsub\rho T~/~\Theta'';\rho'
}

    \vs
\multicolumn{1}{l}{\mbox{Pattern Spines}~~ \boxed{\Theta;\Delta\vdash\wvec{pat}\reconChk{T}
     \Pi \Delta';\Gamma.\wvec{Pat} \spineret S~/~\Theta';\rho}}
\\[1em]
    \infer[\elspempty]
    {\Theta;\Delta\vdash\cdot\reconChk{T} \Pi \Delta;\cdot ~.~\cdot \spineret T~/~\Theta;\ids\Theta}
    {\mbox{either}~T = \ibox U~\mbox{or}~ T = \const{a}\;\wvec{\ibox{C}}}

    \vs

    \infer[\elspcmp] 
    {\Theta;\Delta\vdash pat~~\wvec{pat}\reconChk{T_1\to T_2}
   \Pi \Delta'';(\Gamma,\Gamma') ~.~ (\ahsub{\rho'}Pat)~~\wvec{Pat} \spineret S~/~\Theta'';\rho'\circ\rho}
    {
        \Theta;\Delta\vdash pat\reconChk{T_1}\Pi \Delta';\Gamma . Pat/\Theta';\rho \qquad
        \Theta';\Delta'\vdash\wvec{pat}\reconChk{\ahsub\rho{T_2}}
        \Pi \Delta'';\Gamma' . \wvec{Pat} \spineret S~/~\Theta'';\rho'
}

    \vs

    \infer[\elspex] 
    {\Theta;\Delta\vdash\yux{\ibox c}{\wvec{pat}}\reconChk{\pie{X\oft U}T}
      \Pi \Delta'';\Gamma~.~\yux{(\ahsub{\rho'}\ibox{C})}{\wvec{Pat}}
      \spineret S
      ~/~\Theta'';\rho'\circ\rho}
    {
        \Theta;\Delta\vdash c\reconChk{U} C/\Theta';\Delta';\rho \quad
        \Theta';\Delta' \vdash\wvec{pat}\reconChk{\asub{C/X}{\ahsub\rho T}}
        \Pi \Delta'';\Gamma . \wvec{Pat} \spineret S~/~\Theta'';\rho'
}

    \vs

    \infer[\elspim]
    {\Theta;\Delta\vdash\wvec{pat}\reconChk{\piei{X\oft U}T}
    \Pi \Delta';\Gamma . (\ahsub{\rho}C) ~~\wvec{Pat} \spineret S~/~\Theta';\rho}
    {
      \genHole (\hole Y\oft\Delta.U) = C \qquad
        \Theta,\hole Y\oft{\Delta.U};\Delta\vdash\wvec{pat}\reconChk{\asub{C/X}T}
        \Pi \Delta';\Gamma . \wvec{Pat} \spineret S~/~\Theta';\rho
}
  \end{array}
\end{displaymath}
  \caption{Elaboration of Patterns and Pattern Spines}
  \label{fig:elpattern}
\end{sidewaysfigure}

For synthesizing the patterns with a type annotation, first we
elaborate the type $t$ in an empty context using a judgment that
returns the reconstructed type $T$, its holes and index variables
(contexts $\Theta'$ and $\Delta'$). Once we have the type we
elaborate the pattern checking against the type $T$.

To be able to synthesize the type of pattern variables and return it, we check variables
against a given type $T$ during elaboration (see rule \elpvar). For index level
objects, rule \elpindex\ we defer to the index level elaboration that the index
domain provides\footnote{Both, elaboration of pattern variables and of index objects can
be generalized by for example generating a type skeleton in the rule \elsubst\ given
the scrutinee's type. This is in fact what is done in the implementation
of Beluga.}.
Finally, when elaborating a pattern against a given type it is possible to
switch to synthesis mode using rule \elpsyn, where first we elaborate the
pattern synthesizing its type $S$ and then we make sure that $S$ unifies
against the type $T$ it should check against.


%% file: reconstruction/soundness.tex
\section{Soundness of Elaboration}

We establish soundness of our elaboration: if from a well-formed
source expression, we obtain a well-typed target expression $E$ which
may still contain some holes then $E$ is well-typed for any ground
instantiation of these holes. In fact, our final result of elaborating
a recursive function and branches must always return a closed
expression.

\begin{thm}[Soundness]~
  \begin{enumerate}
  \item If $\Theta;\Delta;\Gamma\vdash \rclo e \theta \reconChk T E/\Theta_1;\rho_1$
    then for any grounding hole instantiation $\rho_g$ s.t.
    $\cdot \vdash \rho_g : \Theta_1$ and $\rho_0 = \rho_g \circ \rho_1$, we have:\\
    $\ahsub{\rho_0}\Delta;\ahsub{\rho_0}\Gamma\vdash \ahsub{\rho_g}E \checks \ahsub{\rho_0}T$.

  \item If $\Theta;\Delta;\Gamma\vdash \rclo e \theta \recon E\oft T/\Theta_1;\rho_1$
    then for any grounding hole instantiation $\rho_g$ s.t. $\cdot \vdash \rho_g :
    \Theta_1$ and $\rho_0 = \rho_g\circ\rho_1$, we have:\\
    $\ahsub{\rho_0}\Delta;\ahsub{\rho_0}\Gamma\vdash \ahsub{\rho_g}E \synths\ahsub{\rho_g}T$.

  \item If $\Delta;\Gamma\vdash \rclo {pat\mapsto e} \theta\reconChk {S\to T}
    \Pi\Delta';\Gamma' . Pat : \theta' \mapsto E$
    then $\Delta;\Gamma\vdash\Pi\Delta';\Gamma' . Pat : \theta' \mapsto E\checks S\to T$.
  \end{enumerate}
\end{thm}

To establish the soundness of elaboration of case-expressions and
branches, we rely on pattern elaboration which abstracts over the
variables in patterns as well as over the holes which derive from the
instantiations inferred for omitted arguments. We abstract over these
holes using a lifting substitution $\ep$. The proofs for this theorem
and related lemmas are in Appendix~\ref{sec:completesoundnessproof}).

\begin{lem}[Pattern elaboration]\mbox{\\} 
  \begin{enumerate}
  \item If $\Theta;\Delta\vdash pat\recon\Pi\Delta_1;\Gamma_1 . Pat\oft T/\Theta_1;\rho_1$ and
    $\ep$ is a ground lifting substitution, such as $\Delta_i\vdash\ep\oft\Theta_1$  then:\\
    $\Delta_i,\ahsub{\ep}\Delta_1;\ahsub{\ep}\Gamma_1\vdash
    \ahsub{\ep}Pat\checks\ahsub{\ep}T$.

  \item If $\Theta;\Delta\vdash pat\reconChk{T}\Pi\Delta_1;\Gamma_1 . Pat/\Theta_1;\rho_1$ and
    $\ep$ is a ground lifting substitution, such as $\Delta_i\vdash\ep\oft\Theta_1$ then:\\
    $\Delta_i,\ahsub{\ep}\Delta_1;\ahsub{\ep}\Gamma_1\vdash
    \ahsub{\ep}Pat\checks\ahsub{\ep}\ahsub{\rho_1} T$.

  \item If $\Theta;\Delta\vdash\wvec{pat}\reconChk{T}\Pi\Delta_1;\Gamma_1 . \wvec{Pat}\spineret S/\Theta_1;\rho_1$ and
    $\ep$ is a ground lifting substitution, such as $\Delta_i\vdash\ep\oft\Theta_1$  then:\\
    $\Delta_i,\ahsub{\ep}\Delta_1;\ahsub{\ep}\Gamma_1\vdash
    \ahsub{\ep}\wvec{Pat}\checks\ahsub{\ep}\ahsub{\rho_1} T\spineret\ahsub{\ep}S$.
  \end{enumerate}
\end{lem}



%% file: reconstruction/related.tex
\section{Related Work}

Our language contains indexed families of types that are related to
Zenger's work~\citep{Zenger:TCS97} and the Dependent ML
(DML)~\citep{Xi:JFP} and Applied Type System (ATS)~\citep{Xi03:ATS,
  Chen05:CombiningProgramsWithProofs}. The objective in these systems
is: a program that is typable in the extended indexed type system is
already typable in ML. By essentially erasing all the type annotations
necessary for verifying the given program is dependently typed, we
obtain a simply typed ML-like program. In contrast, our language
supports pattern matching on index objects. Our elaboration, in
contrast to the one given in \citet{Xi:JFP}, inserts omitted
arguments producing programs in a fully explicit dependently typed
core language. This is different from DML-like systems which treat
\emph{all} index arguments as implicit and do not provide a way for
programmers to manipulate and pattern match directly on index objects.
Allowing users to explicitly access and match on index arguments
changes the game substantially.

Elaboration from implicit to explicit syntax for dependently typed
systems has first been mentioned by \citet{Pollack90} although no
concrete algorithm to reconstruct omitted arguments was given.
\citet{Luther:IJCAR01} refined these ideas as part of the TYPELab
project. He describes an elaboration and reconstruction for the
calculus of constructions without treating recursive functions and
pattern matching. There is in fact little work on elaborating
dependently-typed source language supporting recursion and pattern
matching. For example, the Agda bi-directional type inference
algorithm described in~\citet{Norell:phd07} concentrates on a core
dependently typed calculus enriched with dependent pairs, but omits
the rules for its extension with recursion and pattern
matching\footnote{\citet{Norell:phd07} contains extensive discussions
  on pattern matching and recursion, but the chapter on elaboration
  does not discuss them.}. Idris, a dependently typed language
developed by \citet{Brady:JFP13} uses a different technique. Idris
starts by adding holes for all the implicit variables and it tries to
instantiate these holes using unification. However, the language uses
internally a tactic based elaborator that is exposed to the user who
can interactively fill the holes using tactics. He does not prove
soundness of the elaboration, but conjectures that given a type
correct program its elaboration followed by a reverse elaboration
produces a matching source level program.

A notable example, is the work by \citet{Asperti:2012} on describing a
bi-directional elaboration algorithm for the Calculus of (Co)Inductive
Constructions (CCIC) implemented in Matita. Their setting is very
different from ours: CCIC is more powerful than our language since the
language of recursive programs can occur in types and there is no
distinction between the index language and the programming language
itself. Moreover in Matita, we are only allowed to write total
programs and all types must be positive. For these reasons their
source and target language is more verbose than ours and refinement,
i.e. the translation of the source to the target, is much more complex
than our elaboration. The difference between our language and Matita
particularly comes to light when writing case-expressions. In Matita
as in Coq, the programmer might need to supply an invariant for the
scrutinee and the overall type of the case expression as a type
annotation. Each branch then is checked against the type given in the
invariant. Sometimes, these annotations can be inferred by using
higher-order unification to find the invariant. In contrast, our
case-expressions require no type annotations and we refine each branch
according to refinement imposed by the pattern in each branch. The
refinement is computed with help from higher-order unification. This
makes our source and target language more light-weight and closer to a
standard simply typed functional language.

Finally, refinement in Matita may leave some holes in the final
program which then can be refined further by the user using for
example tactics. We support no such interaction; in fact, we fail, if
holes are left-over and the programmer is asked to provide more
information.

Agda, Matita and Coq require users to abstract over all variables
occurring in a type and the user statically labels arguments the user
can freely omit. To ease the requirement of declaring all variables
occurring in type, many of these systems such as Agda supports simply
listing the variables occurring in a declaration without the
type. This however can be brittle since it requires that the user
chose the right order. Moreover, the user has the possibility to
locally override the implicit arguments mechanism and provide
instantiations for implicit arguments explicitly. This is in contrast
to our approach where we guide elaboration using type annotations and
omit arguments based on the free variables occurring in the declared
type, similarly to Idris which abstracts and makes implicit all the
free variables in types.

This work is also related to type inference for Generalized Algebraic
Data Types (i.e: GADTs) such as~\citep{Schrijvers:2009}. Here the
authors describe an algorithm where they try to infer the types of
programs with GADTs when the principal type can be inferred and
requiring type annotations for the cases that lack a principal type or
it can not be inferred. This is in contrast to our system which always
requires a type annotation at the function level. On the other hand
our system supports a richer variety of index languages (some index
languages can be themselves dependently typed as with Contextual LF in
Beluga). Moreover we support pattern matching on index terms, a
feature that is critical to enable reasoning about objects from the
index level. Having said that, the approach to GADTs
from~\citep{Schrijvers:2009} offers interesting ideas for future work,
first making the type annotations optional for cases when they can be
inferred, and providing a declarative type systems that helps the
programmer understand when will the elaboration succeed to infer the
types.


%% file: reconstruction/conclusion.tex
\section{Conclusion}

In this chapter we describe a surface language for writing dependently
typed programs where we separate the language of types and index
objects from the language of programs. Total programs in our language
correspond to first-order inductive proofs over a specific index
domain where we mediate between the logical aspects and the
domain-specific parts using a box modality. Our programming language
supports indexed data-types, dependent pattern matching and
recursion. Programmers can leave index variables free when declaring
the type of a constructor or recursive program as a way of stating
that arguments for these free variables should be inferred by the
type-directed elaboration. This offers a lightweight mechanism for
writing compact programs which resemble their ML counterparts and
information pertaining to index arguments can be omitted. In
particular, our handling of case-expressions does not require
programmers to specify the type invariants the patterns and their
bodies must satisfy and case expressions can be nested due to the
refinement substitutions that mediate between the context inside and
outside a branch. Moreover, we seamlessly support nested pattern
matching inside functions in our surface and core languages (as
opposed to languages such as Agda or Idris where the former supports
pattern matching lambdas that are elaborated as top-level functions
and the latter only supports simply typed nested pattern matching).

The proposed case-expression can be nested and does not require
annotating its return type. Notably, it refines all the variables in
the context, this allows the programmer to write pattern matching on
dependent types using variables from the context that have been
refined by the current branch. We think this is a powerful feature for
the users of our language, nonetheless it is at the expense of some
complexity in the theory of the language, namely the elaboration of
terms together with a substitution that contains the current
refinements.

To guide elaboration and type inference, we allow type
annotations which indirectly refine the type of sub-expressions; type
annotations in patterns are also convenient to name index variables which do not
occur explicitly in a pattern.

We prove our elaboration sound, in the sense that if elaboration
produces a fully explicit term, this term will be well-typed. Finally, our
elaboration is implemented in Beluga, where we use as the index domain
contextual LF, and has been shown practical (see for example
the implementation of a type-preserving compiler \citep{Belanger:CPP13}). We
believe our language presents an interesting point in the design space
for dependently typed languages in general and sheds light into
how to design and implement a dependently typed language where we have
a separate index language, but still want to support pattern matching
on these indices.

%% file: babybel/bmacros.tex

\newcommand{\todo}[1]{\footnote{#1}}
\newenvironment{metanote}{\begin{quote}\message{note!}[\begingroup\it}%
                         {\endgroup]\end{quote}}

\setlength{\parskip}{0cm}
\setlength{\parindent}{1em}


\newcommand{\addref}{\xspace}

\newcommand{\many}[1]{\overrightarrow{#1}}
 \newcommand{\eval}[1]{\ensuremath{\ulcorner #1 \urcorner}}




\newcommand{\boxd}[1]{\ensuremath{\square\,#1}}
\newcommand{\prodt}[2]{\ensuremath{#1\times #2}}
\newcommand{\termbox}[1]{\ensuremath{\{#1\}}}
\newcommand{\shift}[1]{\ensuremath{\uparrow^{#1}}}
\newcommand{\sub}[3]{\ensuremath{[#1]}}
\newcommand{\bclo}[1]{\ensuremath{[#1]}}
\newcommand{\evalsto}[1]{\ensuremath{\bclo{#1}\Downarrow}}
\newcommand{\dotsub}[3]{\ensuremath{#1,#2/#3}}
\newcommand{\HPsi}{\ensuremath{\hat{\Psi}}}
\newcommand{\mvar}[1]{\ensuremath{\code{'#1}}}
\newcommand{\pvar}[1]{\ensuremath{\code{#1}}}
\newcommand{\metavar}[2]{\ensuremath{\app{\mvar{#1}}{\sub{#2}{}{}}}} 
\newcommand{\pmetavar}[1]{\ensuremath{\code{'#1}}}
\newcommand{\pparvar}[1]{\ensuremath{\code{\##1}}}
\newcommand{\ppparvar}[1]{\ensuremath{\code{\#\##1}}}
\newcommand{\wkn}[2]{\ensuremath{\code{wkn}_{\hat #1}^{\hat #2}}}
\newcommand{\lookup}[0]{\ensuremath{\code{lookup}}}

\newcommand{\babybel}[0]{Babybel\xspace}
\newcommand{\debruijn}{de~Bruijn\xspace}

\newcommand{\constr}[1]{C}

\newcommand{\pair}[2]{\ensuremath{(#1,\, #2)}}
\newcommand{\lete}[3]{\ensuremath{\code{let}\, #1 = #2\mathrel{\code{in}} #3}}
\newcommand{\letpe}[4]{\ensuremath{\lete {\pair #1 #2} #3 #4}}
\newcommand{\fune}[3]{\ensuremath{\code{fun}\, #1(#2) = #3}}
\newcommand{\branche}[2]{\ensuremath{\mid pat_{#1} \mapsto #2}}
\newcommand{\cbranche}[3]{\ensuremath{\mid [#1 \vdash #2] \mapsto #3}}
\newcommand{\matche}[2]{\ensuremath{\code{match}\, #1\,\code{with}\ #2}}
\newcommand{\cmatche}[2]{\ensuremath{\code{cmatch}\, #1\,\code{with}\ #2}}
\newcommand{\branches}{\ensuremath{\many b}}
\newcommand{\cbranches}{\ensuremath{\many c}}
\newcommand{\unitt}{\ensuremath{\code{unit}}}
\newcommand{\unite}{\ensuremath{()}}
\newcommand{\n}[1]{\code{#1}} 


\newcommand{\tapp}[2]{\ensuremath{#1[#2]}}
\newcommand{\tp}[2]{\ensuremath{\tapp{\n{#1}}{#2}}}
\newcommand{\aconstr}[3]{\ensuremath{\app{\tapp {\constr #1} {#2}} {#3}}}
\newcommand{\fixe}[3]{\ensuremath{\code{fix}\, #1\oft #2 = #3}}
\newcommand{\lame}[1]{\ensuremath{\lambda #1\mathrel{.}}}
\newcommand{\Lame}[1]{\ensuremath{\Lambda #1\mathrel{.}}}
\newcommand{\alle}[1]{\ensuremath{\forall{#1}\mathrel{.}}}
\newcommand{\sml}{\ensuremath{\text{Core-ML}}\xspace}
\newcommand{\smlhoas}{\ensuremath{\text{Core-ML}^{\text{hoas}}}\xspace}
\newcommand{\smlgadt}{\ensuremath{\text{Core-ML}^{\text{gadt}}}\xspace}
\newcommand{\smlg}{\smlgadt}
\newcommand{\smlgsf}{\ensuremath{\smlg + \text{SF}}\xspace}
\newcommand{\wfj}{\,\code{wf}}
\newcommand{\ptj}[2]{\ensuremath{#1\oft #2\mathrel{\downarrow}}} 
\newcommand{\embe}[1]{\ensuremath{#1}}



\newcommand{\unif}{\doteq}
\newcommand{\nunif}{\,\cancel{\doteq}\,}

%% file: babybel/introduction.tex
\section{Introduction}

Writing programs that manipulate other programs is a common activity
for a computer scientist, either when implementing interpreters,
writing compilers, or analyzing phases for static analysis. This is so
common that we have programming languages that specialize in writing
these kinds of programs. In particular, ML-like languages are
well-suited for this task thanks to recursive data types and pattern
matching. However, when we define syntax trees for realistic input
languages, there are more things on our wish list: we would like
support for representing and manipulating variables and tracking their
scope; we want to compare terms up-to $\alpha$-equivalence (i.e. the
renaming of bound variables); we would like to avoid implementing
capture avoiding substitutions, which is tedious and error-prone. ML
languages typically offer no high-level abstractions or support for
manipulating variables and the associated operations on abstract
syntax trees.

Over the past decade, there have been several proposals to add support
for defining and manipulating syntax trees into existing programming
environments. For example: FreshML~\citep{Shinwell:ICFP03}, the
related system Romeo~\citep{Stansifer:2014}, and
C$\alpha$ml~\citep{Pottier:2006} use Nominal Logic~\citep{Pitts:2003}
as a basis and the Hobbits library for Haskell~\citep{Westbrook:2011}
uses a name based formalism. In this chapter, we show how to extend an
existing (functional) programming language to define abstract syntax
trees with variable binders based on higher-order abstract syntax
(HOAS) (sometimes also called $\lambda$-trees~\citep{Miller:1999}).
Specifically, we allow programmers to define object languages in the
simply-typed $\lambda$-calculus where programmers use the intentional
function space of the simply typed $\lambda$-calculus to define
binders (as opposed to the extensional function space of ML). Hence,
HOAS representations inherit $\alpha$-renaming from the simply-typed
$\lambda$-calculus and we can model object-level substitution for HOAS
trees using $\beta$-reduction in the underlying simply-typed
$\lambda$-calculus. We further allow programmers to express whether a
given sub-tree in the HOAS tree is closed by using the necessity
modality of S4 \citep{Davies:ACM01}. This additional expressiveness is
convenient to describe closed abstract syntax trees.

This work follows the work of HOAS representations in the logical
framework LF~\citep{Harper93jacm}. On the one hand we restrict it to
the simply-typed setting to integrate it smoothly into existing
simply-typed functional programming languages such as OCaml, and on
the other hand we extend its expressiveness by allowing programmers to
distinguish between closed and open parts of their syntax trees. As we
analyze HOAS trees, we go under binders and our sub-trees may not
remain closed. To model the scope of binders in sub-trees we pair a
HOAS tree together with its surrounding context of variables following
ideas from Beluga \citep{Pientka:POPL08,Nanevski:ICML05}. In addition,
we allow programmers to pattern match on such contextual objects, i.e.
an HOAS tree together with its surrounding context. In essence, this
work shows that the fundamental ideas underlying Beluga
\citep{Pientka:IJCAR10,Pientka:CADE15} can be smoothly integrated into
existing programming environments.

%
%

In this chapter the contribution is two-fold: First, we present a
general methodology for adding support for HOAS tree definitions and
first-class contexts to an existing (simply-typed) programming
language. In particular, programmers can define simply-typed HOAS
definitions in the syntactic framework (SF) based on modal S4
following~\citep{Nanevski:ICML05,Davies:ACM01}. In addition,
programmers can manipulate and pattern match on well-scoped HOAS trees
by embedding HOAS objects together with their surrounding context into
the programming language using contextual types
\citep{Pientka:POPL08}. The result is a programming language that can
express computations over open HOAS objects. We describe our technique
abstractly and generically using a language that we call \sml. In
particular, we show how \sml with first-class support for HOAS
definitions and contexts can be translated in into a language \smlgadt
that supports Generalized Abstract Data Types (GADTs) using a deep
(first-order) embedding of SF and first-class contexts (see
Fig.~\ref{fig:diagram} for an overview). We further show that our
translation preserves types.

\begin{figure}
  \centering
\begin{tikzpicture}[node distance=0.5cm, auto]
\tikzset{
    mynode/.style={rectangle,draw=black, inner sep=0.3em, text centered, align=center},
    myarrow/.style={->, >=latex', shorten >=1pt, thick},
    myarrow2/.style={->, >=latex', thick, double equal sign distance},
    mylabel/.style={text width=7em, text centered}
}

\node[mynode] (sml) {~\\\sml\\~};
\node[mynode, below right=of sml] (ctyp) {+ Contextual\\Types};
\node[mynode, below left=of ctyp] (sf) {Syntactic\\Framework};
\node[mynode, right=of ctyp] (smlct) {~\\\smlgadt\\~};

\draw[myarrow] (sml.east) -| (ctyp.north);
\draw[myarrow] (sf.east) -| (ctyp.south);
\draw[myarrow2] (ctyp.east) -- (smlct.west);

\end{tikzpicture}
  \caption{Adding Contextual Types to ML}
  \label{fig:diagram}
\end{figure}
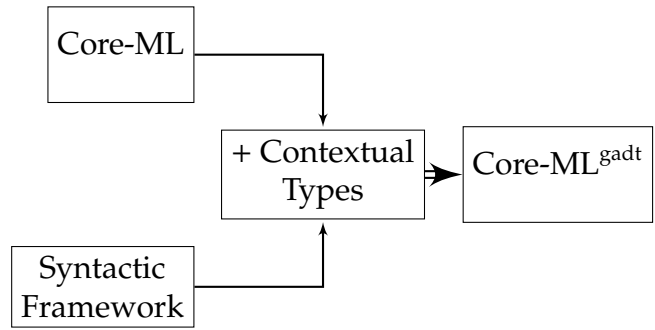

Second, we show how this methodology can be realized in OCaml by
describing our prototype \babybel \footnote{available at
  \url{www.github.com/fferreira/babybel/}}. In our implementation of
\babybel we take advantage of the sophisticated type system, in
particular GADTs, that OCaml provides to ensure our translation is
type-preserving. By translating HOAS objects together with their
context to a first-order representation in OCaml with GADTs we can
also reuse OCaml's first-order pattern matching compilation allowing
for a straightforward compilation. Programmers can also exploit
OCaml's impure features such as exceptions or references when
implementing programs that manipulate HOAS syntax trees.
The \babybel prototype includes implementations of a type-checker, an
evaluator, closure conversion (shown in Section~\ref{sec:cloconv}
together with a variable counting example and a syntax desugaring
examples) and a continuation passing style translation. These examples
demonstrate that this approach allows programmers to write programs
that operate over abstract syntax trees in a manner that is safe and
effective.

Finally, notice that \babybel's language extensions do not have to be
OCaml specific. The same approach could be implemented in Haskell, and
other typed functional programming languages.





%% file: babybel/example.tex
\section{Main Ideas} \label{sec:examples}
\lstset{language=Babybel}

In this section, we show some examples that illustrate the use of
\babybel, our proof of concept implementation where we embed the
syntactic framework SF inside OCaml. To smoothly integrate SF into
OCaml, \babybel defines a PPX filter (a mechanism for small syntax
extensions for OCaml). In particular, we use attributes and quoted strings to
implement our syntax extension. 






\filbreak
\subsection{Example Programs}
\subsubsection{Removing Syntactic Sugar}\label{sec:synsugar}
In this example, we describe the compact and elegant implementation of
a compiler phase that de-sugars programs functional programs with
let-expressions by translating them into function applications. We
first specify the syntax of a simple functional language that we will
transform. To do this we embed the syntax specification using this
tag:
\begin{lstlisting}
[@@@signature {def| ... |def}]
\end{lstlisting}
Inside the \lstinline!@@@signature! block we will embed our SF specifications.

Our source language is defined using the type \lstinline!tm!. It consists of constants
(written  as \lstinline!cst!),  pairs (written as
\lstinline!pair!), functions (built using \lstinline!lam!),
applications (built using \lstinline!app!), and let-expressions.

\begin{lstlisting}
[@@@signature {def|
tm : type.
cst     : tm.
pair    : tm -> tm -> tm.
lam     : (tm -> tm) -> tm.
fst     : tm -> tm.
snd     : tm -> tm.
letpair : tm -> (tm -> tm -> tm) -> tm.
letv    :  tm -> (tm -> tm) -> tm.
app     : tm -> tm -> tm.
|def}]
\end{lstlisting}

 Our
definition of the source language
exploits HOAS using the function space of our syntactic framework SF
to represent binders in our object language.
For example, the constructor \lstinline!lam!  takes as an
argument a term of type \lstinline!tm -> tm!. Similarly, the definition of
let-expressions models variable binding by falling back to the function
space of our meta-language, in our case the syntactic framework SF. As a
consequence, there is no constructor for variables in our syntactic
definition and moreover we can reuse the substitution operation from
the syntactic framework SF to model substitution in our object language.
This avoids building up our own infrastructure for variables bindings.

We now show how to simplify programs written in our source language by
replacing uses of \lstinline!letpair! in terms with projections, and
uses of \lstinline!letv! by $\beta$ reduction. 

\begin{center}
  \lstinline!letv!\ $M$ \lstinline!(\x.!$N$\lstinline!)!
   $\equiv$
  N\lstinline![!$M$\lstinline!/x]!
  \\
  \lstinline!letpair!\ $M$\ \lstinline!(\x.\y.! $N$\lstinline!)!
  $\equiv$
  N\lstinline![(fst!\ $M$ \lstinline!)/x,(snd!\ $M$\lstinline!)/y]!\\
\end{center}

To implement this simplification phase we implement an OCaml
program \lstinline!rewrite!: it analyzes the structure of our terms,
calls itself on the sub-terms, and eliminates the use of the
let-expressions into simpler constructs. 
As we traverse terms, our sub-terms may not remain
closed. For simplicity, we use the same language as source and target
for our transformation.
We therefore specify the type of the function
\lstinline!rewrite! using contextual types pairing the type
\lstinline!tm! together with a context \lstinline!$\gamma$! in which the term
is meaningful inside  the tag \lstinline![@type .... ]!.

\begin{lstlisting}
rewrite[@type $\gamma$.[$\gamma$ |- tm]->[$\gamma$ |- tm]]
\end{lstlisting}

The type can be read: for all contexts \lstinline!$\gamma$!, given a
\lstinline!tm! object in the context \lstinline!$\gamma$!, we return a
\lstinline!tm! object in the same context.
In general, contextual types associate a context and a type in the syntactic
framework SF. For example if we want to specify a term in the empty
context we would write \lstinline![ |- tm]!  or for a term that
depends on some context with at least one variable and potentially
more we would write \lstinline![$\gamma$,x:tm |- tm]!.

\filbreak
We now implement the function \lstinline!rewrite! by pattern matching on the
structure of a contextual term. In \babybel, contextual terms are written
inside boxes (\lstinline!{t|...|t}!) and contextual patterns inside
(\lstinline!{p|...|p}!).
\begin{lstlisting}
let rec rewrite[@type $\gamma$.[$\gamma$ |- tm]->[$\gamma$ |- tm]]
= function
| {p| cst |p} -> {t|cst|t}
| {p| pair 'm 'n |p} ->
    let mm, nn = rewrite m, rewrite n
    in {t|pair 'mm 'nn|t}
| {p| fst 'm |p} ->
    let mm = rewrite m in {t|fst 'mm|t}
| {p| snd 'm |p} ->
    let mm = rewrite m in {t|snd 'mm|t}
| {p| app 'm 'n |p} ->
    let mm,nn = rewrite m, rewrite n
    in {t|app 'mm 'nn|t}
| {p| lam (\x. 'm) |p} ->
    let mm = rewrite m in {t|lam (\x. 'mm)|t}
| {p| #x |p} -> {t|#x|t}
| {p| letpair 'm (\f.\s. 'n) |p} ->
      let mm = rewrite m in
      rewrite {t|'n [snd 'mm;fst 'mm]|t}
| {p| letv 'm (\x. 'n) |p} -> rewrite {t|'n['m]|t}
\end{lstlisting}

Note that we are pattern matching on potentially open terms. Although
we do not write the context \lstinline!$\gamma$! explicitly, in
general patterns may mention their context (i.e.:
\lstinline!{p|_ |- cst|p}!\footnote{The underscore means that there
  might be a context but we do not bind any variable for it because
  the term does not explicitly mention them and contexts are not
  available at run-time.}. As a guiding principle, we may omit writing
contexts, if they do not mention variables explicitly and are
irrelevant at run-time.
Inside patterns or terms, we specify incomplete terms using quoted
variables (e.g.: \lstinline!'n!). Quoted variables are an 'unboxing'
of a computational expression inside the syntactic framework SF. The
quote signals that we are mentioning a computational variable inside
SF. Quoted variables can depend on all the variables in scope in scope
where they are defined. For example, in the pattern:
\begin{lstlisting}
{p| letpair 'm (\f.\s. 'n) |p}
\end{lstlisting}
 the variable \lstinline!'n!
may depend on bound variables \lstinline!f! and \lstinline!s!.

The computationally interesting cases are the let-expressions. For
them, we perform the rewriting according to the two rules given
earlier. The syntax of the substitutions puts in square brackets the
terms that will be substituted for the variables. We consider contexts
and substitutions ordered, this allows for efficient implementations
and more lightweight syntax (e.g.: substitutions omit the name of the
variables because contexts are ordered). Importantly, the substitution
is an operation that is eagerly applied and not part of the
representation. Consequently, the representation of the terms remains
normal and substitutions cannot be written in patterns. We come back
to this design decision later. When pattern matching on open terms,
one needs to match against variables, this is accomplished with the
pattern for variables from the context that we denote as
\lstinline!{p| #x |p}!.

To translate contextual SF objects and contexts, \babybel takes
advantage of OCaml's advanced type system. In particular, we use
Generalized Abstract Data Types~\citep{ChH03Pha,Xi03:guarded} to index
types with the contexts in which they are valid. Type indices, in
particular contexts, are then erased at run-time. When the contexts
are relevant at run-time, we need to provide a term to explicitly
represent the contexts. In Section~\ref{sec:cloconv} there is an
example of this.

\subsubsection{Finding the Path to a Variable}\label{sec:path}
In this example, we compute the path to a specific variable in an abstract
syntax tree describing a lambda-term. This will show how to specify
particular context shapes, how to pattern match on variables, how to
manage our contexts, and how the \babybel extensions interact
seamlessly with OCaml's impure features. 
For this example, we concentrate on the fragment of terms that
consists only of abstractions and application which we repeat here.

\begin{lstlisting}
[@@@signature {def|
tm : type.
app : tm -> tm -> tm.
lam : (tm -> tm) -> tm.
|def}]
\end{lstlisting}

To find the first occurrence of a particular variable in the HOAS
tree, we use backtracking that we implement using the user-defined
OCaml exception \lstinline!Not_found!. To model the path to a
particular variable occurrence in the HOAS tree, we define an OCaml
data type \lstinline!step! that describes the individual steps we take
and finally model a \lstinline!path! as a list of individual steps.

\begin{lstlisting}
exception Not_found
type step
= Here  (*the path ends here*)
| AppL  (*take left on app*)
| AppR  (*take right on app*)
| InLam (*go inside the body of the term*)
type path = step list
\end{lstlisting}

The main function \lstinline!path_aux! takes as input a term that lives in a
context with at least one variable and returns a path to the occurrence
of the top-most variable or an empty list, if the variable is not
used. Its type is:
\begin{lstlisting}
[@type $\gamma$. [$\gamma$, x : tm |- tm] -> path].
\end{lstlisting}
We again quantify over all contexts
\lstinline!$\gamma$! and require that the input term is meaningful in a context
with at least one variable. This specification simply excludes closed
terms since there would be no top-most variable. Note also how
we mix in the type annotation to this function both contextual types
and OCaml data types.

\begin{lstlisting}
let rec path_aux [@type $\gamma$.[$\gamma$, x:tm |- tm] -> path]
= function
| {p|_, x |- x|p}-> [Here]
| {p|_, x |- #y|p}-> raise Not_found
| {p|_, x |- lam (\y. 'm)|p}->
      InLam::(path_aux{t|_,x,y |-'m[_;y;x]|t})
| {p|_, x |- app 'm 'n|p}->
      try AppL::(path_aux m)
      with _ -> AppR::(path_aux n)
\end{lstlisting}

All patterns in this example make the context explicit, as we pattern
match on the context to identify whether the variable we encounter
refers to the top-most variable declaration in the context. The
underscore simply indicates that there might be more variables in the
context. The first case, matches against the bound variable
\lstinline!x!. As mentioned before, the second case has a special
pattern with the sharp symbol that matches against any variable in the
context \lstinline!_, x!. Because of the first pattern if it had been
\lstinline!x! it would have matched the first case. Therefore, it
simply raises the exception to backtrack to the last choice we had.
The case for abstractions is interesting, since we have to go under
its binder and the variable that we are looking for is no longer the
top most declaration of the context. Hence we must apply a
substitution to swap the two variables at the top of the context.

The case for lambda expressions is interesting because the recursive
call happens in an extended context. Furthermore, in order to keep the
variable we are searching for on top, we need to swap the two top-most
variables. For that purpose, we apply the \lstinline![_ ; y; x]!
substitution. In this substitution the underscore stands for the
identity on the rest of the context, or more precisely, the
appropriate shift in our internal representation that uses \debruijn
indices. Once elaborated, this substitution becomes
\lstinline![^2 ; y; x]! where the shift by two is because we are
swapping variables as opposed to instantiating them with closed terms.

The final case is for applications. We first look on the left side and
if that raises an exception we catch it and search again on the right.
We again use quoted variables (e.g.: \lstinline!'m!) to bind and refer
to ML variables in patterns and terms of the syntactic framework and
more generally be able to describe incomplete terms.

\begin{lstlisting}
let get_path [@type $\gamma$.[$\gamma$, x:tm |- tm] -> path]
= fun t -> try path_aux t with _ -> []
\end{lstlisting}
The \lstinline!get_path! function has the same type as the
\lstinline!path_aux! function. It simply handles the exception and
returns an empty path in case that variable \lstinline!x! is not found
in the term.

\subsubsection{Closure Conversion} \label{sec:cloconv}

In the final example, we describe the implementation of a naive
algorithm for closure conversion for untyped $\lambda$-terms
following~\citep{Cave:POPL12}. We take advantage of the syntactic
framework SF to represent source terms (using the type family
\lstinline!tm!) and closure-converted terms (using the type family
\lstinline!ctm!). In particular, we use SF's closed modality box to
ensure that all functions in the target language are closed (this is
written with curly braces: \lstinline!{}!). This is impossible when we
simply use LF as the specification framework for syntax as in
\citep{Cave:POPL12}. We omit here the definition of lambda-terms, our
source language, that was given in the previous section and
concentrate on the target language \lstinline!ctm!.
\begin{figure}
  \centering
\begin{lstlisting}
[@@@signature {def|
ctm: type.  % closed term
btm: type.  % binder term
env : type. % environment

capp   : ctm -> ctm -> ctm.
clam   : {btm} -> ctm.
clo    : ctm -> env -> ctm.

embed  : ctm -> btm.
bind   : (ctm -> btm) -> btm.

empty  : env.
dot    : env -> ctm -> env.
|def}]
\end{lstlisting}
  \caption{Closure Converted Language}
  \label{fig:clolang}
\end{figure}

Concretely, \lstinline!tm! is the type of our input language.
Applications are again represented by the constructor \lstinline!app! that
takes two terms, the first represents the function and the second its
parameter.

One option would be to simply use the same representation for the
target language that we used for the source language (i.e.: the
untyped $\lambda$-calculus). A better option is to use a new
representation that highlights that functions in the target language
do no depend on their environments. To this effect we declare in
Figure~\ref{fig:clolang} two types, \lstinline!ctm! to represent terms
after the conversion and \lstinline!btm! to represent the bodies of
functions. We take advantage of the expressive power of the
specification framework SF to define the closed bodies of the
functions \lstinline!btm! and the converted terms \lstinline!ctm!.

Applications in the target language are defined using the constructor
\lstinline!capp! and simply take two target terms to form an
application. But functions (constructor \lstinline!clam!) take a
\lstinline!btm! object wrapped in \lstinline!{}! braces. This means
that the object inside the braces is closed. The curly braces denote
the internal closed modality of the syntactic framework. As the
original functions may depend on variables in the environment, we need
closures where we pair a function with an environment that provides
the appropriate instances for variables. We define our own environment
explicitly, because they are part of the target language and the
built-in substitution is an operation on terms that is eagerly
computed away. Inside the body of the function, we need to bind all
the variables from the environment that the body uses such that later
we can instantiate them applying the substitution. This is achieved by
defining multiple bindings using constructors \lstinline!bind! and
\lstinline!embed! inside the term.

When writing a function that translates between representations, the
open terms depend on contexts that store assumptions of different
representations. Therefore, it is often the case that one needs to
relate these contexts. In our example here we define a context
relation that keeps the input and output contexts in sync using a GADT
data type \lstinline!rel! in OCaml where we model contexts as types.
The relation statically checks correspondence between contexts, but it
is also available at run-time (i.e.~ after type-erasure). It states
that for each variable in the source contexts there is a corresponding
one in the target context.
\filbreak
\begin{lstlisting}
type (_ , _) rel  =
    Empty : ([.], [.]) rel
  | Both  : ([$\gamma$], [$\delta$]) rel ->
      ([$\gamma$, x:tm], [$\delta$, y:ctm]) rel

exception Error of string

let rec lookup
  [@type $\gamma$ $\delta$.[$\gamma$ |- tm]->($\gamma$, $\delta$) rel->[$\delta$ |- ctm]] =
fun t -> function
| Both r' -> begin match t with
  | {p| _,x |- x |p} -> {t|_,x |- x|t}
  | {p| _,x |- ##v |p} -> let v1 = lookup {t|#v|t} r'
                      in {t|_, x |- 'v1 [_]|t}
  | _ -> raise Error (``Term that is not a variable'')
| Empty -> raise Error (``Term is not a variable'')
end
\end{lstlisting}

The function \lstinline!lookup! searches for related variables in the
context relation. If we have a source context
\lstinline!$\gamma$,x:tm! and a target context
\lstinline!$\delta$,y:ctm!, then we consider two variable cases: In
the first case, we use matching to check that we are indeed looking
for the top-most variable \lstinline!x! and we simply return the
corresponding target variable. If we encounter a variable from the
context, written as \lstinline!##v!, then we recurse in the smaller
context stripping off the variable declaration \lstinline!x!. Note
that \lstinline!##v! denotes a variable from the context
\lstinline!_!, that is not \lstinline!x!, while \lstinline!#v!
describes a variable from the context \lstinline!_, x!, i.e. it could
be also \lstinline!x!. The recursive call returns the corresponding
variable \lstinline!v1! in the target context that does not include
the variable declaration \lstinline!x!. We hence need to weaken
\lstinline!v1! to ensure it is meaningful in the original context. We
therefore associate \lstinline!'v1! with the identity substitution for
the appropriate context, namely: \lstinline![_]!. In this case, it
will be elaborated into a one variable shift in the internal syntax
(i.e.: $\uparrow^1$). The last case returns an exception whenever we
are trying to look up in the context something that is not a variable.

As we cannot express at the moment in the type annotation that the
input to the lookup function is indeed only a variable from the
context $\gamma$ and not an arbitrary term, we added another fall-through
case for when the context is empty. In this case the input term cannot
be a variable, as it would be out of scope.

Finally, we implement the function \lstinline!conv! which takes an
untyped source term in a context \lstinline!$\gamma$! and a relation of
source and target variables, described by \lstinline!($\gamma$, $\delta$) rel! and
returns the corresponding target term in the target context
\lstinline!$\delta$!.

\filbreak
\begin{lstlisting}
let rec close [@type $\gamma$ $\delta$. ($\gamma$, $\delta$) rel->[$\delta$ |- btm]->[btm]]
= fun r m -> match r with
| Empty -> m
| Both r -> close r {t|bind (\x. 'm)|t}

let rec envr [@type $\gamma$ $\delta$. ($\gamma$, $\delta$)rel->[$\delta$ |- env]]
= fun r -> match r with
| Empty -> {t|empty|t}
| Both r ->
  let s = envr r in {t|_, x |- dot ('s[_]) x|t}

let rec conv [@type $\gamma$ $\delta$.($\gamma$, $\delta$)rel->[$\gamma$ |-tm]->[$\delta$ |-ctm]]
= fun r m -> match m with
| {p| lam (\x. 'm)  |p} ->
  let mc = conv (Both r) m in
  let mb = close r {t|bind(\x. embed 'mc)|t}
  in let s = envr r in {t|clo (clam {'mb}) 's|t}
| {p|#x|p} -> lookup {t|#x|t} r
| {p|app 'm 'n|p} -> let mm, nn = conv r m, conv r n in
                 {t|capp 'mm 'nn|t}
\end{lstlisting}




The core of the translation is defined in functions \lstinline!conv!,
\lstinline!envr!, and \lstinline!close!. The main function is
\lstinline!conv!. It is implemented by recursion on the source
term. There are three cases: i) source variables simply get translated
by looking them up in the context relation, ii) applications just get
recursively translated each term in the application, and iii)
lambda expressions are translated recursively by
converting the body of the expression in the extended context (notice the
recursive call with \lstinline!Both r!)  and then turning
the lambda expression into a closure.

In the first step we generate the closed body by the function
\lstinline!close! that adds the multiple binders (constructors
\lstinline!bind! and \lstinline!embed!) and generates the
closed term.  Note that the return type \lstinline![btm]! of \lstinline!close! guarantees that
the final result is indeed a closed term, because we omit the context.
For clarity, we could have written \lstinline![ |- btm]!.

Finally, the function \lstinline!envr!
computes the substitution (represented by the type
\lstinline!env!) for the closure.



The implementation of closure conversion shows how to enforce closed
terms in the specification, and how to make contexts and their
relationships explicit at run-time using OCaml's GADTs. We believe it
also illustrates well how HOAS trees can be smoothly manipulated and
integrated into OCaml programs that may use effects.


%% file: babybel/smallml.tex
\section[\sml: A Small Functional Language]{\sml : A Functional Language with Pattern Matching and Data Types} \label{sec:sml}

We now introduce \sml, a functional language based on ML with pattern
matching and data types. In Section~\ref{sec:mlct} we will extend this
language to also support contextual types and terms in our syntactic
framework SF.

We keep the language design of \sml minimal in the interest of
clarity. However, our prototype implementation which we describe in
Section~\ref{sec:impl} supports interaction with all of OCaml's
features such as exceptions, references and GADTs.

\begin{displaymath}
  \begin{array}{lrcll}
    \mbox{Types} & \tau & \bnfas & D \bnfalt \tau_1 \to \tau_2 \\[.2em]

    \mbox{Expressions} & e & \bnfas & \embe i \bnfalt \fune f x e \bnfalt\\[.2em]
    & & & \lete x {i} {e} \bnfalt
                  \matche i \branches\\[.2em]

    \mbox{Neutral Exp.} & i & \bnfas & \app i e \bnfalt \app {\constr k} {\many{e}} \bnfalt
                                       x \bnfalt e\oft\tau\\[.2em]

    \mbox{Patterns}  & pat & \bnfas & \app {\constr{k}} {\many {pat}} \bnfalt x \\[.2em]

    \mbox{Branches} & b & \bnfas & \branche {} e \\[.2em]
    \mbox{Contexts} & \Gamma & \bnfas & \cdot \bnfalt \Gamma, x\oft \tau\\
    \mbox{Signature} & \Xi & \bnfas & \cdot \bnfalt \Xi, D \oft\code{type} \bnfalt \Xi, \constr k \oft \many {\tau} \to D\\
  \end{array}
\end{displaymath}

In \sml, we declare data-types by adding type formers ($D$) and type
constructors ($\constr k$) to the signature ($\Xi$). Constructors must
be fully-applied. In addition all functions are named and recursive.
The language supports pattern matching with nested patterns where
patterns consist of just variables and fully applied constructors. We
assume that all patterns are linear (i.e. each variable occurs at most
once) and that they are covering.
\begin{figure}
  \centering
\begin{displaymath}
  \begin{array}{c}
    \multicolumn{1}{l}{\boxed{\Gamma \vdash e \checks \tau}: \mbox{Expression $e$ checks against type $\tau$ in context $\Gamma$}}\\[.5em]

    \infer[\rl{t-rec}]
    {\Gamma \vdash \fune f x e \checks \tau \to \tau'}
    {\Gamma, f: \tau\to\tau', x: \tau \vdash e \checks \tau'}

    \quad

    \infer[\rl{t-let}]
    {\Gamma \vdash \lete x {i} {e} \checks \tau}
    {\Gamma \vdash i \synths \tau'
    & \Gamma,x\oft \tau' \vdash e \checks \tau}

    \vs

    \infer[\rl{t-match}]
    {\Gamma \vdash \matche i \branches \checks \tau}
    {\Gamma \vdash i \synths \tau'
    & \forall b_k\in\many b \mathrel{.} \Gamma \vdash b_k \checks \tau' \to \tau}

    \vs

    \infer[\rl{t-emb}]
    {\Gamma\vdash\embe i \checks \tau}
    {\Gamma\vdash i \synths \tau' & \tau=\tau'}

    \vs

    \multicolumn{1}{l}{\boxed{\Gamma \vdash i \synths \tau}: \mbox{Neutral expr. $i$ synthesizes type $\tau$ in context $\Gamma$}}\\[.5em]

    \infer[\rl{t-ann}]
    {\Gamma\vdash e\oft\tau \synths \tau}
    {\Gamma\vdash e \checks \tau}

    \quad

    \infer[\rl{t-app}]
    {\Gamma\vdash \app{i}{e} \synths \tau}
    {\Gamma\vdash i \synths \tau'\to\tau
    & \Gamma\vdash e \checks \tau'}

    \vs

    \infer[\rl{t-var}]
    {\Gamma \vdash x \synths \tau}
    {\Gamma(x) = \tau}

    \vs

    \infer[\rl{t-constr}]
    {\Gamma \vdash \app {\constr k}  {\many e} \synths D}
    {\Xi(\constr k) = \many{\tau} \to D
    & \forall \tau_i \in \many{\tau}\mathrel{.} \forall e_i \in \many{e} \mathrel{.} \Gamma \vdash e_i \checks \tau_i}

    \vs
    \multicolumn{1}{l}{\boxed{\Gamma \vdash\branche {} e \checks \tau_1\to\tau_2}: \mbox{Branch checks against $\tau_1$ and $\tau_2$ in $\Gamma$}}\\[.5em]

    \infer[\rl{t-branch}]
    {\Gamma \vdash \branche {} e \checks \tau'\to\tau}
    {\vdash\ptj {pat} {\tau'} {\Gamma'}
    & \Gamma,\Gamma' \vdash e \checks \tau}

    \vs
    \multicolumn{1}{l}{\boxed{\vdash pat \oft \tau \downarrow \Gamma} : \mbox{Pattern $pat$ is of type $\tau$ and binds variables in context $\Gamma$}}\\[.5em]

    \infer[\rl{t-pat-var}]
    {\vdash \ptj x \tau {x\oft \tau}}
    {}

    \vs

    \infer[\rl{t-pat-con}]
    {\vdash\ptj {\app{\constr k}{\many {pat}}} D {\Gamma_1,...,\Gamma_i}}
    {\Xi(\constr k) = \many{\tau}\to D
    & ~~\forall \tau_i \in \many{\tau}\mathrel{.} \forall pat_i \in \many{pat} \mathrel{.}~~ \vdash\ptj {pat_i} {\tau_i} {\Gamma_i}}
  \end{array}
\end{displaymath}
  \caption{\sml Typing Rules}
  \label{fig:smltyp}
\end{figure}

The bi-directional typing rules for \sml have access to a signature
$\Xi$ and are standard (see Fig.~\ref{fig:smltyp}) and the signature
remains unchanged throughout the rules.

There are three things to notice in the rules:
\begin{enumerate}
\item The rule \rl{t-constr} ensures that the constructor $\constr k$
  is fully applied.

\item When type checking pattern matching in rule \rl{t-match} the
  branches are checked against the type that the pattern needs to have
  (i.e.: $\tau'$) and the type the body needs to have (i.e.: $\tau$).
  The use of the arrow in this context is an overloading of syntax to
  mean the type of the pattern and the type of the body.

\filbreak
\item Finally, in the rule for branches(\rl{t-branch}) we check the
  body in the context $\Gamma$ extended with the context from the
  pattern of the branch (context $\Gamma'$).
\end{enumerate}
\begin{figure}
  \centering
\begin{displaymath}
  \begin{array}{lrcll}
    \mbox{Values} & v & \bnfas & \app {\constr k} {\many v} \bnfalt
                                 (\fune f x e)\bclo\rho\\
    \mbox{Environments} & \rho & \bnfas & \cdot \bnfalt \rho, v/x\\
    \mbox{Closures} & L & \bnfas & e \bclo\rho \\
  \end{array}
\end{displaymath}
\begin{displaymath}
  \begin{array}{c}
    \multicolumn{1}{l}{\boxed{e\evalsto\rho v}: \mbox{Expression $e$ evaluates to value $v$ in environment $\rho$}}\\[.5em]
    \infer[\rl{e-app}]
    {(\app {i}{e})\evalsto\rho v}
    {i\evalsto\rho (\fune f x e')\bclo{\rho'}
    & e\evalsto\rho v'
    & e' \evalsto{\rho',v'/x} v}

    \vs

    \infer[\rl{e-var}]
    {x \evalsto\rho v}
    {\rho(x) = v}

    \quad

    \infer[\rl{e-let}]
    {(\lete x i {e}) \evalsto\rho v}
    {i \evalsto\rho v_1
    & e \evalsto{\rho,v_1/x} v}

    \vs

    \infer[\rl{e-match-fail}]
    {(\matche i {{\branche 1 {e_b}}\cons\branches}) \evalsto\rho v}
    {i \evalsto\rho v_1
    & \vdash {pat_1} \nunif v
    & (\matche i \branches) \evalsto\rho v}

    \vs

    \infer[\rl{e-match-succ}]
    {(\matche i {{\branche 1 {e_b}}\cons\branches})\evalsto\rho v}
    {i \evalsto\rho v_1
    & \hat\Gamma\vdash {pat_1} \unif v /\rho'
    & e_b \evalsto{\rho,\rho'} v}

    \vs

    \infer[\rl{e-constr}]
    {(\app{\constr k}{\many e}\,) \evalsto{\rho} \app{\constr k}{\many v}}
    {\many e\, \evalsto\rho \many v}

    \quad

    \infer[\rl{e-ann}]
    {e\oft \tau \evalsto\rho v}
    {e \evalsto\rho v}

    \quad

    \infer[\rl{e-emb}]
    {\embe i \evalsto\rho v}
    {i \evalsto\rho v}

    \vs

    \infer[\rl{e-fun}]
    {(\fune f x e) \evalsto\rho (\fune f x e)\bclo\rho}
    {}
  \end{array}
\end{displaymath}
  \caption{\sml Big-Step Operational Semantics}
  \label{fig:smlsos}
\end{figure}

In Figure~\ref{fig:smlsos} we define the operational semantics using
an environment based approach. We define:
\begin{itemize}
\item values, that are either constructors applied to other values or
  recursive functions.
\item environments that assign a value to each variable in the context.
\item closures that represent pending computations in the
  environment. We write $e\bclo\rho$ for a closure that consists of
  expression $e$ in context $\rho$.
\end{itemize}
We use an environment based operational semantics to avoid having to
define substitutions at the level of \sml and later, substitution of
quoted variables. This agrees with \babybel that in its implementation
we do not need to define quoted variable substitution as we reuse
OCaml's substitution so this style of semantics seems fitting.

The evaluation judgment $e\evalsto\rho v$ means that the expression
$e$ in environment $\rho$ evaluates in a big step to value $v$. The
more interesting rules are the ones for pattern matching that use
first-order matching as defined in Fig.~\ref{fig:fomatch}. As
traditional, in a \code{match} expression, patterns are matched branch
by branch until one matches and then the body is executed in the
extended environment that resulted from the matching. 

We characterize both matching and failure to match. Successful
matching is defined by the judgment $\Gamma \vdash pat \unif v/\rho$
that means that the pattern $pat$ with unification variables in
$\Gamma$ and the value $v$ match producing the substitution $\rho$.
Failure to match is defined by the judgment $\vdash pat \nunif v$ that
states that the pattern $pat$ with unification variables in $\Gamma$
and the value $v$ cannot be matched. Since $v$ is always ground we do
not need unification but matching. We also will not address the
details of pattern matching compilation but merely state that it is
possible to implement it in an efficient manner using decision
trees~\citep{Augustsson:FPCA85}.

\begin{figure}
  \centering
\begin{displaymath}
  \begin{array}{c}
    \multicolumn{1}{l}{\boxed{\hat\Gamma \vdash pat \unif v/\rho} : \mbox{Value $v$ matches $pat$ producing substitution $\rho$ }}\\[.5em]
    \infer[\rl{m-v}]
    {x \vdash x \unif v / (\cdot,v/x)}
    {}

    \vs

    \infer[\rl{m-c}]
    {\hat\Gamma \vdash \app{\constr k}{\many{pat}} \unif \app{\constr k}{\many v} / \rho_0,\dots\rho_{n-1} }
    {n = |\many{pat}|
    & \hat\Gamma = \hat\Gamma_0,\dots, \hat\Gamma_{n-1}
    & \forall i < n . \hat\Gamma_i \vdash \ pat_i \unif v_i /\rho_i}

    \vs

    \multicolumn{1}{l}{\boxed{\vdash pat \nunif v} : \mbox{Value $v$ does not match pattern $pat$}}\\[.5em]

    \infer[\rl{f-c}]
    {\vdash \app{\constr k}{\many{pat}} \nunif \app{\constr m}{\many{pat}}}
    {k \neq m}

    \quad

    \infer[\rl{f-f}]
    {\vdash pat \nunif \fune f x e}
    {}

    \vs

    \infer[\rl{f-r}]
    {\vdash \app{\constr k}{\many{pat}} \nunif \app{\constr k}{\many v}}
    {n = |\many{pat}| & \exists i < n .\ pat_i \nunif v_i }
  \end{array}
\end{displaymath}
  \caption{First-order Matching}
  \label{fig:fomatch}
\end{figure}


%% file: babybel/sf.tex
\section{A Syntactic Framework}\label{sec:sf}

In this section we describe the Syntactic Framework (SF) based on the
modal logic S4~\citep{Davies:ACM01}. Our framework characterizes only
normal forms. All computation is delegated to the ML layer, that will
perform pattern matching and substitutions on terms.

\subsection{The definition of SF}

The Syntactic Framework (SF) is a simply typed $\lambda$-calculus
based on S4 where the type system forces all variables to be of base
type, and all constants declared in a signature $\Sigma$ to be fully
applied. This simplifies substitution as variables of base type cannot
be applied to other terms, and in consequence, there is no need for
hereditary substitution in the specification language. Finally, the
syntactic framework supports the box type to describe closed terms
\citep{Pfenning01mscs}. It can also be viewed as a restricted version of
the contextual modality in~\citep{Nanevski:ICML05} which could be an
interesting extension to our work.

Having closed objects enforced at the specification level is not
strictly necessary. However, being able to state that some objects are
closed in the specification has two distinct advantages: first, the
user can specify some objects as closed so their contexts are always
empty. This removes the need for some unnecessary substitutions.
Second, it allows us to encode more fine-grained invariants and is
hence an important specification tool (i.e. when implementing closure
conversion in Section~\ref{sec:cloconv}).
\begin{displaymath}
  \begin{array}{lrcll}
    \mbox{Types} & A,B & \bnfas & \const a \bnfalt A \to B \bnfalt \boxd A\\
    \mbox{Terms} & M,N & \bnfas & \app {\const c} {\many M} \bnfalt
                                  \lam x M \bnfalt \termbox{M} \bnfalt x \\
    \mbox{Contexts} & \Psi,\Phi & \bnfas & \cdot \bnfalt \Psi, x\oft \const a\\
    \mbox{Signature} & \Sigma & \bnfas & \cdot \bnfalt \Sigma, \const a \oft K \bnfalt \Sigma,\, \const c \oft A
  \end{array}
\end{displaymath}

Fig.~\ref{fig:sftyp} shows the typing rules for the syntactic
framework. Note that constructors always are fully applied (as per
rule \rl{t-con}), and that all variables are of base type as enforced
by rules \rl{t-var} and \rl{t-lam}.

The specification framework's terms are manipulated by the computational language, in
the resulting system any concrete use of terms in the syntactic
framework will be done by pattern matching and the application of
substitutions in the computational language therefore we will not have
any elimination forms in the syntactic framework.



\begin{figure}
  \centering
\begin{displaymath}
  \begin{array}{c}
    \multicolumn{1}{l}{\boxed{\Psi \vdash M \oft A} : \mbox{$M$ has type $A$ in context $\Psi$}}\\[.5em]

    \infer[\rl{t-lam}]
    {\Psi \vdash \lam x M \oft \const a \to A}
    {\Psi, x\oft \const a \vdash M \oft A}

    \quad

    \infer[\rl{t-box}]
    {\Psi \vdash \termbox{M} \oft \boxd A}
    {\cdot \vdash M \oft A}

    \quad

    \infer[\rl{t-var}]
    {\Psi \vdash x \oft \const a}
    {\Psi(x) = \const a}

\vs

    \infer[\rl{t-con}]
    {\Psi \vdash \app {\const c}{\many M} \oft \const a}
    {\Sigma(\const c) = A & \Psi \vdash \many M\oft A/\const a}

\vs

    \multicolumn{1}{l}{\boxed{\Psi \vdash \many{M} \oft A / B} :
      \mbox{spine $\many{M}$ checks against  type $A$ and has target type $B$}}\\[.5em]

    \quad

    \infer[\rl{t-sp-em}]
    {\Psi\vdash \cdot \oft \const a /\const a}
    {}

\quad

    \infer[\rl{t-sp}]
    {\Psi\vdash \app N {\many M} \oft A \to B/\const a}
    {\Psi \vdash N \oft A & \Psi\vdash\many M\oft B/\const a}

  \end{array}
\end{displaymath}
  \caption{Syntactic Framework Typing}
  \label{fig:sftyp}
\end{figure}



\subsection{Contextual Types}

We use contextual types~\citep{Nanevski:ICML05} to embed possibly open
SF objects in \sml and ensure that they are well-scoped. We use
contextual types where we pair the type $A$ of an SF object together
with its surrounding context $\Psi$ in which it makes sense. This
follows the design of Beluga \citep{Pientka:POPL08,Cave:POPL12}.
\begin{displaymath}
  \begin{array}{lrcll}
    \mbox{Contextual Types} & U & \bnfas & [\Psi \vdash A]\\
    \mbox{Type Erased Contexts} & \hat\Psi & \bnfas & \cdot \bnfalt \hat\Psi, x\\
    \mbox{Contextual Objects} & C & \bnfas & [\hat\Psi \vdash M]\\
  \end{array}
\end{displaymath}

Contextual objects, written as $[\hat{\Psi} \vdash M]$ pair the term
$M$ with the variable name context $\hat\Psi$ to allow for
$\alpha$-renaming of variables occurring in $M$. Note how the
$\hat\Psi$ context just corresponds to the context with the typing
assumptions erased.

When we embed contextual objects in a programming language we want to
refer to variables and expressions from the ambient language, in order
to support incomplete terms. Following
\citep{Nanevski:ICML05,Pientka:POPL08}, we extend our syntactic
framework SF with two ideas: first, we have incomplete terms with
meta-variables to describe holes in terms. As in Beluga, there are two
different kinds: {\textit{quoted variables}} $\mvar u$ represent a
hole in the term that may be filled by an arbitrary term. In contrast,
\emph{parameter variables} represent a hole in a term that
may be filled only with some bound variable from the context.
Concretely, a parameter variable may be $\pparvar x$ and describe any
concrete variable from a context $\Psi$. We may also want to restrict
what bound variables a parameter variable describes. For example, if
we have two sharp signs (i.e. $\ppparvar x$) the top-most variable
declaration is excluded. Intuitively, the number of sharp signs, after
the first, in front of \lstinline!x! correspond to a weakening (or in
de Bruijn lingo the number of shifts). Second, substitution operations
allow us to move terms from one context to another.

We hence extend the syntactic framework SF with quoted variables,
parameter variables and closures, written as $M\sub\sigma\Phi\Psi$.
We annotate the substitution with its domain and range to simplify the
typing rule, however our prototype omits these typing annotations and
lets type inference infer them.

\begin{displaymath}
  \begin{array}{lrcll}
    \mbox{Parameter Variables} & \pvar v & \bnfas & \pparvar x \bnfalt \ppparvar x\\
    \mbox{Terms} & M & \bnfas & \dots \bnfalt \mvar u \bnfalt \pvar v
\bnfalt M \sub\sigma\Phi\Psi\\
    \mbox{Substitutions} & \sigma & \bnfas & \cdot \bnfalt \dotsub
    \sigma M x\\
    \mbox{Ambient Ctx.} & \Gamma & \bnfas & \dots \bnfalt \Gamma, \code u \oft [\Psi \vdash \const a] \\
  \end{array}
\end{displaymath}

In addition, we extend the context $\Gamma$ of the ambient language
\sml to keep track of assumptions that have a contextual type.

Finally, we extend the typing rules of the syntactic framework SF to
include quoted variables, parameter variables, closures, and
substitutions. We keep all the previous typing rules for SF from
Section~\ref{sec:sf}  where we thread through the ambient $\Gamma$,
but the rules remain unchanged otherwise.

\begin{displaymath}
  \begin{array}{c}
    \multicolumn{1}{l}{\boxed{\Gamma;\Psi \vdash_v \pvar v \oft \const a}:
      \mbox{Parameter Variable \pvar v   has type $\const a$ in contexts
        $\Psi$ and $\Gamma$}}
\vs
    \infer[\rl{t-pvar-v}]
    {\Gamma;\Psi \vdash_v \pparvar x \oft \const a}
    {\Gamma(x) = [\Psi \vdash \const a]}

    \quad

    \infer[\rl{t-pvar-\#}]
    {\Gamma;\Psi,y\oft\_ \vdash_v \pparvar v \oft \const a}
    {\Gamma;\Psi \vdash_v \pvar v \oft \const a}

\vs

    \multicolumn{1}{l}{\boxed{\Gamma;\Psi \vdash M \oft A}: \mbox{Term
        $M$ has type $A$ in contexts $\Psi$ and $\Gamma$}}\\[.5em]
    \infer[\rl{t-qvar}]
    {\Gamma;\Psi \vdash \mvar u \oft \const a}
    {\Gamma(\code{u}) = [\Psi \vdash \const a]}

    \quad

    \infer[\rl{t-pvar}]
    {\Gamma;\Psi \vdash \pvar v \oft \const a}
    {\Gamma;\Psi \vdash_v \pvar v \oft \const a}

   \vs

    \infer[\rl{t-sub}]
    {\Gamma;\Psi \vdash M\sub\sigma\Phi\Psi \oft A}
    {\Gamma;\Psi \vdash \sigma \oft \Phi
    & \Gamma; \Phi \vdash M \oft A}

\vs

    \multicolumn{1}{l}{\boxed{\Gamma;\Psi \vdash \sigma \oft \Psi'}:
      \mbox{Substitution $\sigma$ from $\Psi'$ to$\Psi$ in the amb. ctx. $\Gamma$}} \vs

    \infer[\rl{t-empty-sub}]
    {\Gamma;\Psi \vdash \cdot \oft \cdot}
    {}

\quad

    \infer[\rl{t-dot-sub}]
    {\Gamma;\Psi \vdash \dotsub \sigma M x \oft (\Psi', x \oft \const a)}
    {\Gamma;\Psi \vdash \sigma \oft \Psi' & \Gamma;\Psi \vdash M \oft \const a}
  \end{array}
\end{displaymath}

The rules for quoted and parameter variables (\rl{t-qvar} and
\rl{t-pvar} respectively) might seem very restrictive as we can only
use a meta-variable of type $\Psi \vdash \const a$ in the same context
$\Psi$. As a consequence meta-variables often occur as a closure
paired with a substitution (i.e.: $\metavar u \sigma$). This leads to
the following admissible rule:
\begin{displaymath}
\begin{array}{c}
    \infer[\rl{t-qvar-adm}]
    {\Gamma; \Psi \vdash \metavar u \sigma \oft \const a}
    {\Gamma(\code{u}) = [\Phi \vdash \const a] & \Delta;\Psi \vdash \sigma \oft \Phi}
\end{array}
\end{displaymath}


We stress
that $M \sub\sigma\Phi\Psi$ it is an operation that is applied and not part of the syntax of
our terms.

Note that when we compile SF objects, this substitution will be
eagerly applied in order to keep terms in the syntactic framework in
normal form.

Substitution is straightforward to define. We write it here as a
prefix, to stress that it is an operation that is applied and not part
of the term. Substitutions ($\sub {\_}{\_}{\_}\_$) is defined as:
\begin{displaymath}
  \begin{array}{rcl}
    \sub\sigma\Phi\Psi (\lam x M) & = & \lam x {(M \sub{\dotsub \sigma x x} {\Phi,\_}\Psi)} \\
    \sub\sigma\Phi\Psi \, \termbox M & = & \termbox M \\
    \sub\sigma\Phi\Psi \, \mvar u & = & \mvar u\\
    \sub\sigma\Phi\Psi \, \pparvar x & = & \pparvar x\\
    \sub\sigma\Phi\Psi\, x & = &  \lookup\ x\ \sigma \\
    \sub\sigma {\Psi'}\Psi (M \sub\phi \Phi{\Psi'}) & = & \sub{\sigma \circ \phi}\Phi\Psi \, M\\
    \sub\sigma\Phi\Psi (\app {\const c} {\vec M}) & = & \app {\const c} {\vec N} \quad \text{where}\ {N_i = (\sub\sigma\Phi\Psi {M_i})}\\
    \sub\sigma\cdot\cdot \termbox M & = & \termbox M
  \end{array}
\end{displaymath}
We distinguish the actual operation of applying a substitution from
the explicit substitution in a term, by writing the substitution in a
prefix position.

Composing substitutions is defined as:
\begin{displaymath}
  \begin{array}{rrccl}
    \sigma & \circ & \dotsub \phi M x & = & \dotsub {(\sigma \circ \phi)} {(\sub \sigma{\_}{\_} \, M)} x \\
    \sigma & \circ & \cdot & = & \sigma\\
  \end{array}
\end{displaymath}

The lookup of variables in substitutions is defined as:

\begin{displaymath}
  \begin{array}{rrrcl}
    \lookup & x & \dotsub \sigma M x  & = & M \\

    \lookup & x & \dotsub \sigma M y  & = & \lookup\ x\ \sigma \\
  \end{array}
\end{displaymath}

Note that looking up in an empty substitution is not defined as it is
ill-typed.




The next step is to define the embedding of this framework in a
programming language that will provide the computational power to
analyze and manipulate contextual objects.


%% file: babybel/mlct.tex
\filbreak
\section{\sml with Contextual Types} \label{sec:mlct}

To embed contextual SF objects into \sml,  we extend the syntax of
\sml as follows:

\begin{displaymath}
  \begin{array}{lrcl}
    \mbox{Types} & \tau & \bnfas & \dots \bnfalt [\Psi \vdash \const a]\\

    \mbox{Expressions} & e & \bnfas & \dots \bnfalt [\hat\Psi \vdash M]
                          \bnfalt \cmatche e \cbranches\\

    \mbox{Patterns}  & pat & \bnfas & \dots \bnfalt [\hat\Psi \vdash R]\\

    \mbox{Contextual Branches} & c & \bnfas & \dots \bnfalt \cbranche \Psi R e \\
  \end{array}
\end{displaymath} 

In particular, we allow programmers to use the expression:
\begin{displaymath}
\cmatche e \cbranches
\end{displaymath}
to directly pattern match on the syntactic structures they define in
SF.

\subsection{SF Objects as SF Patterns}

We allow programmers to analyze SF objects directly via pattern
matching. The grammar of SF patterns follows the grammar of SF
objects.
\begin{displaymath}
  \begin{array}{lrcl}
    \mbox{SF Parameter Pattern} & \pvar w & \bnfas & \pparvar p \bnfalt \ppparvar {w}\\
    \mbox{SF Patterns} & R & \bnfas & \lam x R \bnfalt \termbox{R} \bnfalt x \bnfalt
                                   \app{\const c}{\many R} \bnfalt
    \pmetavar u \bnfalt \pvar w \\
  \end{array}
\end{displaymath}

However, there is an important restriction: closures are not allowed
in SF patterns. Intuitively this means that all quoted variables are
associated with the identity substitution and hence depend on the
entire context in which they occur. Parameter variables may be
associated with weakening substitutions. This allows us to easily
infer the type of quoted variables and parameter variables as we type
check a pattern. This is described by the judgment:
\[
\boxed{\Psi\vdash\ptj R A \Gamma}: \mbox{Pattern $R$ has type $A$ in $\Psi$ and binds $\Gamma$}
\]

\begin{figure}
  \centering
\begin{displaymath}
  \begin{array}{c}
    \multicolumn{1}{l}{\boxed{\Psi\vdash_v \ptj {\pvar w} {\const a}{\Gamma}}:
      \mbox{Parameter Pattern $\pvar w$ has type $\const a$ in $\Psi$ and binds $\Gamma$}}
\vs

    \infer[\rl{tp-pvar}]
    {\Psi \vdash_v \ptj {\pparvar p} {\const a}{p:[\Psi \vdash \const a]}}{}

    \quad

    \infer[\rl{tp-pvar-\#}]
    {\Psi,y\oft\_ \vdash_v \ptj {\pparvar w}{\const a}{\Gamma}}
    {\Psi \vdash_v \ptj {\pvar w}{\const a}{\Gamma}}

\vs

    \multicolumn{1}{l}{\boxed{\Psi\vdash\ptj R A \Gamma}: \mbox{Pattern $R$ has type $A$ in $\Psi$ and binds $\Gamma$}}\vs

    \infer[\rl{tp-lam}]
    {\Psi\vdash\ptj{\lam x R} {\const a \to A} \Gamma}
    {\Psi,x\oft\const a\vdash\ptj R A \Gamma}

    \quad

    \infer[\rl{tp-box}]
    {\Psi\vdash\ptj{\termbox R}{\boxd A} \Gamma}
    {\cdot\vdash\ptj R A \Gamma}

    \vs

    \infer[\rl{tp-mvar}]
    {\Psi\vdash\ptj{\pmetavar u} {\const a} {u\oft[\Psi\vdash \const a]}}
    {}

\vs

    \infer[\rl{tp-constr}]
    {\Psi\vdash\ptj{\app{\const c}{\many R}} {\const a} {\Gamma}}
    {\Sigma(\const c) = A & \Psi\vdash\ptj {\many R} {A/\const a} \Gamma}

    \vs

    \infer[\rl{tp-pvar}]
    {\Psi \vdash\ptj{\pvar w} {\const a} {\Gamma}}
    {\Psi \vdash_v \ptj{\pvar w} {\const a} {\Gamma}}
\quad

    \infer[\rl{tp-var}]
    {\Psi\vdash\ptj{x} {\const a} {\cdot}}
    {\Psi(x) = \const a}

    \vs

    \multicolumn{1}{l}{\boxed{\Psi\vdash\ptj {\many M}{A / B}{\Gamma}}:
    \mbox{Pat. Spine $\many{R}$ has type $A$ and target $B$ and binds $\Gamma$}}\vs

   \infer[\rl{tp-sp-em}]
   {\Psi\vdash \ptj \cdot {\const a/\const a} \cdot}
   {}

   \quad

   \infer[\rl{tp-sp}]
    {\Psi\vdash\ptj {\app N {\many M}} {A \to B/\const a} {\Gamma,\Gamma'}}
    {\Psi\vdash\ptj N A \Gamma & \Psi\vdash\ptj {\many M} {B/\const a} {\Gamma'}}
  \end{array}
\end{displaymath}
  \caption{Typing Rules for SF Patterns}
  \label{fig:mlctp}
\end{figure}

Figure~\ref{fig:mlctp} shows the typing rules for SF patterns. They
closely follow the typing of SF terms. The more interesting ones are
the parameter patterns as they illustrate the built-in weakening.





Further, the matching algorithm for SF patterns degenerates to simple
first-order matching \citep{PientkaPfenning:CADE03} and can be defined
straightforwardly. 
However, it is worth considering the matching rules for parameter
patterns. As matching will only consider well-typed terms, we know
that in the rules $\rl{m-pv}$ and $\rl{m-pv-\#}$ the variable $x$ is
well-typed in the context $\hat\Psi$.

\begin{displaymath}
  \begin{array}{c}
    \multicolumn{1}{l}{\boxed{\Gamma;\hat\Psi \vdash_v w \unif x/\rho} :
      \mbox{Param. Patt. $w$ matches var. $x$ from $\hat\Psi$
        producing $\rho$.}}
\vs

    \infer[\rl{m-pv}]
    {\code p\oft [\Psi \vdash A];\hat\Psi \vdash_v \pparvar p \unif x/\cdot,[\hat\Psi \vdash x]/\code p}
    {}

\vs

    \infer[\rl{m-pv-\#}]
    {\Gamma;{\hat\Psi,y} \vdash_v \pparvar w \unif x/ \rho}
    {x \neq y & \Gamma; \hat\Psi \vdash_v w \unif x / \rho}

\vs
    \multicolumn{1}{l}{\boxed{\Gamma;\hat\Psi \vdash R \unif M/\rho} :
      \mbox{$M$ matches pattern $R$ with vars. in $\hat\Psi$
        producing $\rho$.}}
\vs
    \infer[\rl{m-$\lambda$}]
    {\Gamma;\hat\Psi \vdash \lam x R \unif \lam x M / \rho}
    {\Gamma;\hat\Psi,x \vdash R \unif M / \rho}

    \quad

    \infer[\rl{m-bv}]
    {\cdot;\hat\Psi \vdash x \unif x / \cdot}
    {}
 \vs
    \infer[\rl{m-box}]
    {\Gamma;\hat\Psi \vdash \termbox{R} \unif \termbox{M} / \rho}
    {\Gamma;\cdot \vdash R \unif M / \rho}

    \vs

    \infer[\rl{m-cc}]
    {\Gamma;\hat\Psi \vdash \app {\const c} {\many R} \unif \app {\const c} {\many M}/\rho_0,\dots,\rho_n}
    {\text{for all } R_i \in \many R \text{ such as } \Gamma;\hat\Psi \vdash R_i \unif M_i/\rho_i}

    \vs

    \infer[\rl{m-cv}]
    {\code u\oft [\Psi \vdash A];\hat\Psi \vdash \pmetavar u \unif M/\cdot,[\hat\Psi \vdash M]/\code u}
    {}

\quad

    \infer[\rl{m-pv}]
    {\Gamma ; \hat\Psi \vdash \pvar w \unif x/\rho}
    {\Gamma ; \hat\Psi \vdash_v \pvar w \unif x/\rho}

  \end{array}
\end{displaymath}

Finally, it has another important consequence: closures only appear in
the branches of case-expressions. As \sml has a call-by-value
semantics, we know the instantiations of quoted variables and
parameter variables when they appear in the body of a case-expression
and all closures are eliminated by applying the substitution eagerly.
Given these conditions the matching operation remains first-order. An
alternative way of explaining this is that in fact we have a
degenerate case of the pattern fragment~\citep{Miller91iclp} of
higher-order unification where all unification variables are applied
to different variables and all the variables in the context. In this
situation we only need first order
matching~\citep{PientkaPfenning:CADE03}, which enables efficient
pattern matching compilation.



\subsection{Typing Rules for \sml with Contextual Types}
We now add the following typing rules for contextual objects and
pattern matching to the typing rules of \sml:
\begin{displaymath}
  \begin{array}{c}
    \infer[\rl{t-ctx-obj}]
    {\Gamma \vdash [\hat\Psi \vdash M] \checks [\Psi \vdash \const a]}
    {\Gamma;\Psi \vdash M \oft \const a}

    \\[.7em]

    \infer[\rl{t-cm}]
    {\Gamma \vdash \cmatche i \branches \checks \tau}
    {\Gamma \vdash i \synths [\Psi \vdash \const a]
    & \forall b\in\many b\mathrel{.}\Gamma \vdash b \checks [\Psi \vdash \const a] \to \tau}

    \\[.7em]

    \infer[\rl{t-cbranch}]
    {\Gamma \vdash [\Psi \vdash R] \mapsto e \checks [\Psi \vdash \const a] \to \tau }
    {\Psi\vdash R \oft \const a \downarrow \Gamma'
    & \Gamma,\Gamma' \vdash e \checks \tau}
  \end{array}
\end{displaymath}

The typing rule for contextual objects (rule \rl{t-ctx-obj}) simply
invokes the typing judgment for contextual objects. Notice, that we
need the context $\Gamma$ when checking contextual objects, as they
may contain quoted variables from $\Gamma$.

Extending the operational semantics to handle contextual SF objects is
also straightforward.

Additionally, we need rules to evaluate the matching of contextual
types that largely follow the rules for pattern matching
\rl{e-match-succ} and \rl{e-match-fail} but use the matching operation
defined for contextual objects.

Finally, the operational semantics needs to be extended with rules to
support the new constructs. The rules, in Figure~\ref{fig:exsos} for
contextual terms simply take the terms apart and re-build them,
except for quoted variables that are looked up in the environment(rule
\rl{ec-qvar}) and for terms with substitutions where the substitution
operation is immediately applied(rule \rl{ec-sub}). 
The rules for the new pattern matching construct are analogous to the
ones from \sml.

\begin{figure}
  \centering
\begin{displaymath}
  \begin{array}{c}
    \multicolumn{1}{l}{\boxed{e\evalsto\rho v}: \mbox{Expression $e$ evaluates to value $v$ in environment $\rho$}}\\[.5em]
    \infer[\rl{ec-lam}]
    {[\hat\Psi \vdash \lam x M]\evalsto\rho [\hat\Psi \vdash \lam x N]}
    {[\hat\Psi,x \vdash M] \evalsto\rho [\hat\Psi,x \vdash N]}

    \vs

    \infer[\rl{ec-box}]
    {[\hat\Psi \vdash \termbox M] \evalsto\rho [\hat\Psi \vdash \termbox N]}
    {[\cdot \vdash \termbox M] \evalsto\rho [\cdot \vdash \termbox N]}

    \vs

    \infer[\rl{ec-var}]
    {[\hat\Psi \vdash x] \evalsto\rho [\hat\Psi \vdash x]}
    {}

    \vs

    \infer[\rl{ec-constr}]
    {[\hat\Psi \vdash \app {\const c} {\many M}] \evalsto\rho [\hat\Psi \vdash \app {\const c} {\many N}]}
    {\forall M\in\many M.\ M\evalsto\rho N}

    \vs

    \infer[\rl{ec-qvar}]
    {\mvar u \evalsto\rho M}
    {\rho(\code u) = M}

    \quad

    \infer[\rl{ec-pvar}]
    {\pparvar x \evalsto\rho y}
    {\rho(\code x) = y}

    \quad

    \infer[\rl{ec-ppvar}]
    {\ppparvar x \evalsto{\dotsub \rho M z} y}
    {\rho(\code x) = y}

    \vs

    \infer[\rl{ec-sub}]
    {M\sub\sigma\Phi\Psi\evalsto \rho N}
    {(\sub\sigma\Phi\Psi M) \evalsto\rho N}

    \vs

    \infer[\rl{ec-mf}]
    {(\cmatche e {{\cbranche \Psi R {e_b}}\,\cbranches})\evalsto\rho v}
    {e \evalsto\rho [\hat\Psi \vdash M]
    & \Psi\vdash R \nunif M
    & (\cmatche e \cbranches) \evalsto\rho v}

    \vs

    \infer[\rl{ec-ms}]
    {(\cmatche e {{\cbranche \Psi R {e_b}}\,\cbranches})\evalsto\rho v}
    {e \evalsto\rho [\hat\Psi \vdash M]
    & \Gamma';\Psi\vdash R \unif M /\rho'
    & e_b \evalsto{\rho,\rho'} v}

  \end{array}
\end{displaymath}
  \caption{Extended Operational Semantics}
  \label{fig:exsos}
\end{figure}


%% file: babybel/gadt.tex
\section{\sml with GADTs}

So far, in this chapter, we reviewed how to support contextual types
and contextual objects in a standard functional programming language.
This allows us to define syntactic structures with binders and
manipulate them with the guarantee that variables will not escape
their scopes. This brings some of the benefits of the Beluga system to
mainstream languages focusing on writing programs instead of proofs. A
naive implementation of this language extension requires augmenting
the type checker and operational semantics of the host language. This
is a rather significant task -- especially if it includes implementing
a compiler for the extended language.  However, we can take advantage
of the powerful type-system in modern functional languages to make the
implementation more straight forward. Concretely, we use Generalized
Abstract Data Types (GADTs) as a target to our translation. GADTs are
a generalization of algebraic data-types that allow types to be
indexed by other types. This mild form type dependency is enough to
implement the SF. In the literature, GADTs have been introduced
several times under different names: phantom
types~\citep{Cheney:2003}, guarded recursive
types~\citep{Xi03:guarded} and equality types~\citep{Sheard:2008}.
In this section, we describe how
to embed \sml with contextual types in a functional language with
GADTs, called \smlg, based on $\lambda_{2,G\mu}$
by~\citet{Xi03:guarded}. The choice of this target language is
motivated by the fact that it is close to what realistic typed
languages already offer (e.g.: OCaml and Haskell) and it directly
lends itself to an implementation.

\begin{displaymath}
  \begin{array}{lrcl}
    \mbox{Signatures} & \Sigma & \bnfas & \cdot \bnfalt \Sigma, D\oft (*,\dots,*)\to * \bnfalt
    {\constr k}\oft\alle {\many{\alpha}} \tau \to \tapp D{\many\tau}\\
    \mbox{Types} & \tau & \bnfas & \tapp D {\many\tau} \bnfalt \alle \alpha\tau \bnfalt \tau_1\to\tau_2\bnfalt
    \alpha\bnfalt\tau_1\times\tau_2 \\
    \mbox{Expressions} & e & \bnfas & x \bnfalt \aconstr k{\many \tau} e \bnfalt \fixe f \tau e \bnfalt
                      \app {e_1}{e_2} \bnfalt \pair {e_1}{e_2} \bnfalt \lame x e\\
                      & & & \bnfalt \lete x{e_1}{e_2} \bnfalt
                      \matche e {\many b}\\
    & & & \bnfalt\Lame\alpha e \bnfalt
                      \tapp e\tau \bnfalt \pair {e_1}{e_2}\\
    \mbox{Branch} & b & \bnfas & pat \mapsto e \\
    \mbox{Pattern} & pat & \bnfas & x \bnfalt \aconstr k{\many\alpha}{pat}\bnfalt\pair{pat_1}{pat_2}\\
    \mbox{Exp. Ctx.} & \Gamma & \bnfas & \cdot \bnfalt \Gamma, x\oft \tau\\
    \mbox{Type Ctx.} & \Delta & \bnfas & \cdot \bnfalt \Delta,\alpha \bnfalt \Delta, \tau_1\equiv\tau_2
  \end{array}
\end{displaymath}

\smlg contains polymorphism and GADTs, which makes it a good ersatz
OCaml that is still small and easy to reason about. GADTs are
particularly convenient, since they allow us to track invariants about
our objects in a similar fashion to dependent types. Compared to \sml,
\smlg's signatures now store type constants and constructors that are
parametrized by other types. We show the typing judgments for the
language in Fig.~\ref{fig:smlgadt}. The term language is similar to
\sml, with the addition of the usual terms for supporting abstraction
over types (i.e. $\Lame \alpha e$) and type applications (i.e.
$\tapp e \tau$).

The language is type-checked with the two typing judgments, one for
terms and one for patterns:
\begin{itemize}
\item $\boxed{\Delta;\Gamma\vdash e\oft\tau}$ : expression $e$ is of type $\tau$
  in typing context $\Delta$ and expression context $\Gamma$.
\item $\boxed{\Delta_o\vdash\ptj{pat}\tau\Delta;\Gamma}$ : $pat$ is of
  type $\tau$ and binds type variables in $\Delta$ and variables in
  $\Gamma$.
\end{itemize}

\begin{sidewaysfigure}
  \centering
\begin{displaymath}
  \begin{array}{c}
    \multicolumn{1}{l}
    {\boxed{\Delta;\Gamma\vdash e\oft\tau}: \mbox{$e$ is of type $\tau$ in typing context $\Delta$ and expression context $\Gamma$.}} \vs

    \infer[\rl{g-con}]
    {\Delta;\Gamma\vdash \aconstr k{\many \tau}e \oft \tapp{D}{\many\tau}}
    {\Sigma(\constr k) = \alle{\many{\alpha}} \tau_1 \to \tapp D{\many\alpha} &
                             \Delta;\Gamma\vdash e \oft \tapp{\tau_1}{\many\tau} & \Delta\vdash \many{\tau} \wfj}

    \quad

    \infer[\rl{g-app}]
    {\Delta;\Gamma\vdash\app {e_1}{e_2}\oft\tau_2}
    {\Delta;\Gamma\vdash e_1\oft \tau_1\to\tau_2
     & \Delta;\Gamma\vdash e_2\oft\tau_1}

    \vs

    \infer[\rl{g-pair}]
    {\Delta;\Gamma\vdash\pair {e_1}{e_2}\oft\tau_1\times\tau_2}
    {\Delta;\Gamma\vdash e_1\oft\tau_1
      & \Delta;\Gamma\vdash e_1\oft\tau_2}

    \quad

    \infer[\rl{g-var}]
    {\Delta;\Gamma\vdash x\oft\tau}
    {\Gamma(x) = \tau}

    \quad

    \infer[\rl{g-fix}]
    {\Delta;\Gamma\vdash \fixe f \tau e \oft \tau}
    {\Delta;\Gamma,f\oft\tau\vdash e\oft\tau}

    \vs

    \infer[\rl{g-lam}]
    {\Delta;\Gamma\vdash\lame x e\oft\tau_1\to\tau_2}
    {\Delta;\Gamma,x\oft\tau_1\vdash e\oft\tau_2}

    \quad

    \infer[\rl{g-tapp}]
    {\Delta;\Gamma\vdash\tapp e{\many\tau_1}\oft\tapp{\tau}{\many\tau_1}}
    {\Delta;\Gamma\vdash e\oft \alle {\many\alpha}{\tau} & \Delta;\Gamma\vdash\many\tau_1\wfj}

    \quad

    \infer[\rl{g-Lam}]
    {\Delta;\Gamma\vdash\Lame \alpha e\oft\alle\alpha\tau}
    {\Delta,\alpha;\Gamma\vdash e\oft\tau}

    \vs

    \infer[\rl{g-let}]
    {\Delta;\Gamma\vdash\lete x{e_1}{e_2}\oft\tau}
    {\Delta;\Gamma\vdash e_1\oft\tau_1 & \Delta;\Gamma,x\oft\tau_1\vdash e_2\oft \tau}

    \quad

    \infer[\rl{g-match}]
    {\Delta;\Gamma\vdash\matche e{\many b}\oft\tau}
    {\Delta;\Gamma\vdash e\oft\tau_1 & \text{for all }i.\Delta;\Gamma\vdash b_i\oft\tau_1\to\tau}

    \\[.7em]

    \infer[\rl{g-branch}]
    {\Delta;\Gamma\vdash pat\mapsto e\oft \tau_1\to\tau_2}
    {\Delta\vdash\ptj{pat}\tau_1\Delta';\Gamma'
    & \Delta,\Delta';\Gamma,\Gamma'\vdash e\oft\tau_2 }

    \\[.7em]

    \multicolumn{1}{l}
    {\boxed{\Delta_o\vdash\ptj{pat}\tau\Delta;\Gamma}: \mbox{$pat$ is of type $\tau$ and binds variables in $\Delta$ and $\Gamma$}}\\[.5em]

    \infer[\rl{gp-var}]
    {\Delta_0\vdash\ptj x \tau \cdot;x\oft\tau}
    {}

   \\[.7em]

    \infer[\rl{gp-pair}]
    {\Delta_0\vdash\ptj{\pair{pat_1}{pat_2}}{\tau_1\times\tau_2}{\Delta_1,\Delta_2;\Gamma_1,\Gamma_2}}
    {\Delta_0\vdash\ptj{pat_1}{\tau_1}\Delta_1;\Gamma_1
    & \Delta_0\vdash\ptj{pat_2}{\tau_2}\Delta_2;\Gamma_2}

    \\[.7em]

    \infer[\rl{gp-con}]
    {\Delta_o\vdash\ptj{\aconstr k{\many\alpha} pat} {\tapp D {\many{\tau_2}}} \many{\alpha},\many{\tau_1}\equiv\many{\tau_2},\Delta;\Gamma}
    {\Sigma(\constr k) = \alle{\many{\alpha}}\tau\to\tapp D{\many{\tau_1}}
   & \Delta_0,\many{\alpha},\many{\tau_1}\equiv\many{\tau_2}\vdash\ptj{pat}\tau\Delta;\Gamma}
  \end{array}
\end{displaymath}
\caption{The Typing of \smlg}
  \label{fig:smlgadt}
\end{sidewaysfigure}

The expressive power of pattern matching is greatly enhanced by the
presence of a limited form of dependent types (i.e. types that depend
on other types).

Of particular interest are the type equalities introduced in the type
context in the \rl{gp-con} rule. Let's consider an example to
understand this (a complete discussion is fully developed
by~\citet{Xi03:guarded}). As an example, we define the type of vectors
as lists indexed by their length (i.e. the classical dependent types
example).
\begin{displaymath}
\begin{array}{rcl}
  \Sigma & = & \n{z} \oft *, \n{s}\oft * \to *,\n{vec}\oft(*, *)\to *,\\
         & & \n{Nil}\oft \alle\beta\tapp{\n{vec}}{\n{z},\beta},\\
         & & \n{Cons}\oft\alle{\alpha,\beta}(\beta\times\tapp{\n{vec}}{\alpha,\beta})\to\tapp{\n{vec}}{\tapp {\n{s}}\alpha, \beta}\\
\end{array}
\end{displaymath}
In the signature we need to define the type level encoding of natural
numbers: types \n{z} and \n{s}, and the type of vectors \n{vec} that
are indexed by their length and the type of their elements. Finally,
we define the constructor \n{Nil} of empty vectors of length zero,
thus the first index is the type \n{z}. And finally the constructor
\n{Cons} that puts together an element of type $\beta$ and a vector
of length $\alpha$ to form a vector of length $\alpha + 1$, that is
$\tp{s}\alpha$ in \smlg.

Vectors are interesting because they make many functions type safe by
relying on the more expressive type system. As an example, in
Figure~\ref{fig:zipsmlg} shows the function \n{zip} that joins two
vectors of the same length and produces a vector of the same lengths
containing pairs of values taken from the original vectors:

\begin{figure}
  \centering
\begin{displaymath}
\begin{array}{l}
  \code{fix}\,\n{zip}\oft \alle{\alpha,\beta_1,\beta_2}\\
  {\quad (\tapp{\n{vec}}{\alpha,\beta_1} \times \tapp{\n{vec}}{\alpha,\beta_2})
  \to\tapp{\n{vec}}{\alpha,\pair{\beta_1}{\beta_2}}} = \\[.2em]
  \Lame{\alpha,\beta_1,\beta_2} {\lame v {\matche v {}}}\\[.2em]
  |\ (\tapp{\n{Nil}}{\beta_1},\tapp{\n{Nil}}{\beta_2})\mapsto \tapp{\n{Nil}}{\beta_1\times\beta_2}\\[.2em]
  |\ (\tapp{\n{Cons}}{\alpha_1,\beta_1}\pair{x}{xs}, \tapp{\n{Cons}}{\alpha_2,\beta_2}\pair{y}{ys}\mapsto\\[.2em]
  \quad \tapp{\n{Cons}}{\alpha_1,\beta_1\times\beta_2}\pair{\pair x y} {\app{\tapp {\n{zip}}{\alpha_1,\beta_1,\beta_2}{\pair{xs}{ys}}}}
\end{array}
\end{displaymath}
  \caption{Zip in \smlg}
  \label{fig:zipsmlg}
\end{figure}

In this case, type refinement happens for example in the second
branch, where in the body we have the choice of referring to the
length as $\alpha_1$ or $\alpha_2$, knowing they have to be the same.
In fact, when typing the pattern in rule \rl{gp-con} context $\Delta$
is extended with $\alpha\equiv\alpha_1$ and $\alpha\equiv\alpha_2$
from which is easy to conclude that $\alpha_1\equiv\alpha_2$.
Additionally, the type constraint that the vectors are of the same
length makes it obvious that the diagonal cases (those when one vector
has more elements than the other) are impossible.

Since \smlg has strong type separation the operational semantics, in
Figure~\ref{fig:smlgsos}, is similar to the semantics for \sml, after
all, type information is irrelevant at run-time. The interested reader
can find the meta-theory in~\citep{Xi03:guarded}.
We define values and environments as:
\begin{displaymath}
  \begin{array}{lrcll}
    \mbox{Values} & v & \bnfas & \app {\aconstr k{\many\tau} v} \bnfalt \pair{v_1}{v_2}\bnfalt (\lame x e)\bclo{\theta;\rho} \bnfalt (\Lame \alpha e)\bclo{\theta;\rho} \\
    \mbox{Environments} & \theta;\rho & \bnfas & \cdot ;\cdot \bnfalt \theta;(\rho, v/x)\bnfalt (\theta,\tau/\alpha);\rho\\
  \end{array}
\end{displaymath}
And with them the semantics are implemented with two judgments, one to
evaluate expressions and the other to compute matching, the matching
operation simply augments the environment inside of the branches in
case expressions:
\begin{itemize}
\item $\boxed{e\evalsto{\theta;\rho} v}$ : Expression $e$ evaluates to value $v$ in environment $\rho$.
\item $\boxed{pat \unif v\backslash\theta;\rho}$ : Value $v$ matches pattern $pat$ and produces env. $\theta;\rho$.
\end{itemize}

\begin{figure*}
  \centering
\begin{displaymath}
  \begin{array}{c}
    \multicolumn{1}{l}{\boxed{e\evalsto{\theta;\rho} v}: \mbox{Expression $e$ evaluates to value $v$ in envirnoment $\rho$}}\\[.5em]
    \infer[\rl{ge-app}]
    {(\app {e_1}{e_2})\evalsto{\theta;\rho} v}
    {e_1\evalsto{\theta;\rho} (\lame x e)\bclo{\theta';\rho'}
    & e_2\evalsto{\theta;\rho} v'
    & e \evalsto{\theta';\rho',v'/x} v}

    \vs

    \infer[\rl{ge-tapp}]
    {\tapp e \tau\evalsto{\theta;\rho} v}
    {e\evalsto{\theta;\rho}\Lame\alpha{e'}
    & e'\evalsto{\theta,\tau/\alpha;\rho} v}

    \vs

    \infer[\rl{ge-var}]
    {x \evalsto{\theta;\rho} v}
    {\rho(x) = v}

    \quad

    \infer[\rl{ge-let}]
    {(\lete x {e_1} {e_2}) \evalsto{\theta;\rho} v}
    {e_1 \evalsto{\theta;\rho} v_1
    & e_2 \evalsto{\theta;\rho,v_1/x} v}

    \vs

    \infer[\rl{ge-match}]
    {(\matche e {{\branche 1 {e_b}}\cons\branches})\evalsto{\theta;\rho} v}
    {e \evalsto{\theta;\rho} v_1
    & {pat_1} \unif v \backslash\theta';\rho'
    & e_b \evalsto{\theta,\theta';\rho,\rho'} v}

    \vs

    \infer[\rl{ge-constr}]
    {\aconstr k{\many\tau} e \evalsto{\theta;\rho} \aconstr k {\many\tau} v }
    {e \evalsto{\theta;\rho} v}

    \quad

    \infer[\rl{ge-lam}]
    {(\lame x e) \evalsto{\theta;\rho} (\lame x e)\bclo{\theta;\rho}}
    {}

    \vs

    \infer[\rl{ge-Lam}]
    {(\Lame \alpha e) \evalsto{\theta;\rho} (\Lame \alpha e)\bclo{\theta;\rho}}
    {}

    \vs

    \infer[\rl{ge-pair}]
    {\pair{e_1}{e_2}\evalsto{\theta;\rho} \pair{v_1}{v_2}}
    {e_1 \evalsto{\theta;\rho} v_1
    & e_2 \evalsto{\theta;\rho} v_2}

    \quad

    \infer[\rl{ge-fix}]
      {\fixe f t e \evalsto{\theta;\rho} v}
      {e \evalsto{\theta;(\rho,\fixe f t e/f)} v}

    \\[.7em]

    \multicolumn{1}{l}
    {\boxed{pat \unif v\backslash\theta;\rho}: \mbox{Value $v$ matches pattern $pat$ and produces env. $\theta;\rho$.}}\\[.5em]

    \infer[\rl{gm-v}]
    {x\unif v \backslash\cdot;(v/x)}
    {}

    \quad

    \infer[\rl{gm-p}]
    {(\pair{pat_1}{pat_2}\unif\pair{v_1}{v_2} \backslash \theta_1,\theta_2;\rho_1,\rho_2}
    {pat_1\unif v_1 \backslash \theta_1;\rho_1
    & pat_2\unif v_2 \backslash \theta_2;\rho_2}

    \vs

    \infer[\rl{gm-c}]
    {\aconstr k {\many\alpha} {pat} \unif \aconstr k {\many\tau} v \backslash (\many\tau/\many\alpha,\theta);\rho}
    {v\unif pat / \theta;\rho}

  \end{array}
\end{displaymath}
  \caption{\smlg Big-Step Operational Semantics}
  \label{fig:smlgsos}
\end{figure*}


%% file: babybel/implementation.tex
\section[Deep Embedding of SF]{Deep Embedding of SF into \smlg} \label{sec:ic}

We now show how to translate objects and types defined in the
syntactic framework SF into \smlg using a deep embedding. We take
advantage of the advanced features of \smlg's type system to fully
type-check the result. Our representation of SF objects and types is
inspired by~\citep{Benton:12} but uses GADTs instead of full dependent
types. We add the idea of typed context shifts instead of renamings to
represent weakening. This is necessary to be able to completely erase
types at run-time.

To ensure SF terms are well-scoped and well-typed, we define SF types
in \smlg and index their representations by their type and
context. The following types are only used as indices for
GADTs. Because of that, they do not have any term constructors.
\begin{displaymath}
  \begin{array}{rcl}
    \Sigma & = & \n{base}\oft * \to *,  \n{arr}\oft(*, *)\to *, \n{boxed}\oft * \to *,\\
    & & \n{prod}\oft (*,*)\to *, \n{unit}\oft *\\
  \end{array}
\end{displaymath}

We define three type families, one for each of SF's type
constructors. It is important to note the number of type parameters
they require. Base types take one parameter: a type from the
signature. Function types simply have an input and output
type. Finally, boxes contain just one type.

Terms are also indexed by the contexts in which they are valid. To
this effect, we define two types to statically represent contexts.
Analogously to the representation of types, these two types are only
used during type-checking and there will be no instances at run-time.
The type \n{nil} represents an empty context and thus has no
parameters. And the constructor \n{cons} has two parameters the first
one is the rest of the context and the second one is the type of the
top-most variable so far.
\begin{displaymath}
  \begin{array}{rcl}
    \Sigma & = & \dots,\n{nil}\oft *, \n{cons}\oft (*, *) \to * \\
  \end{array}
\end{displaymath}

We show the the encoding well-typed SF objects and
types in Fig.~\ref{fig:sfc}. Every declaration is parametrized with the type of constructors
that the user defined inside of the \lstinline!@@@signature!
blocks.

\begin{figure}
  \footnotesize
  \centering
  \begin{displaymath}
    \begin{array}{rcl}
      \Sigma & = & \dots, \n{var}\oft(*,*)\to *,\\[.2em]
             & &\n{Top}\oft\alle{\gamma,\alpha} \tp{var}{\tp{cons}{\gamma, \alpha},\alpha},\\[.2em]
             & & \n{Pop}\oft\alle{\gamma,\alpha,\beta} \tp{var}{\gamma,\alpha} \to
                 \tp{var}{\tp{cons}{\gamma, \beta}, \alpha},\\[.2em]
             & & \n{sftm}\oft (*, *)\to *, \n{sp}\oft (*, *, *)\to *\\[.2em]
             & & \n{Lam} \oft\alle{\gamma,\alpha,\tau} \tp{sftm}{\tp{cons}{\gamma,\tp{base}\alpha},\tau} \to
                 \tp{sftm}{\gamma,\tp{arr}{\tp{base}\alpha,\tau}} ,\\[.2em]
             & & \n{Var} \oft\alle{\gamma,\alpha} \tp{var}{\gamma,\alpha}\to
                 \tp{sftm}{\gamma,\tp{base}\alpha} ,\\[.2em]
             & & \n{Box} \oft\alle{\gamma,\tau} \tp{sftm}{\cdot,\tau}\to\tp{sftm}{\gamma,\tau} ,\\[.2em]
             & & \n{C} \oft\alle{\gamma,\tau,\alpha} \tp{con}{\tau,\alpha}\times \tp{sp}{\gamma,\tau} \to
                 \tp{sftm}{\gamma, \tp{base}\alpha} ,\\[.2em]
             & & \n{Empty}\oft \alle{\gamma,\tau} \tp{sp}{\gamma,\tau,\tau},\\[.2em]
             & & \n{Cons}\oft \alle{\gamma,\tau_1,\tau_2,\tau_3} \tp{sftm}{\gamma, \tau_1} \times
                 \tp{sp}{\gamma, \tau_2, \tau_3} \to \tp{sp}{\gamma,\tp{arr}{\tau_1,\tau2},\tau_3},\\[.2em]
             & & \n{shift}\oft(*,*)\to *,\\[.2em]
             & & \n{Id} \oft\alle{\gamma} \tp{shift}{\gamma,\gamma} ,\\[.2em]
             & & \n{Suc} \oft\alle{\gamma,\delta,\alpha} \tp{shift}{\gamma,\delta} \to \tp{shift}{\tp{cons}{\gamma,\tp{base}\alpha},\delta} ,\\[.2em]
             & & \n{sub}\oft(*,*)\to *,\\[.2em]
             & & \n{Shift} \oft\alle{\gamma,\delta} \tp{shift}{\gamma,\delta}\to\tp{sub}{\gamma,\delta} ,\\[.2em]
             & & \n{Dot} \oft\alle{\gamma,\delta,\tau} \tp{sub}{\gamma,\delta} \times \tp{sftm}{\gamma,\tau} \to \tp{sub}{\gamma,\tp{cons}{\delta,\tau}}\\[.2em]
    \end{array}
  \end{displaymath}
  \caption{Syntactic Framework Definition}
  \label{fig:sfc}
\end{figure}

The specification takes the form of the type
$\n{con} \oft (*, *) \to *$, where $\n{con}$ is the name of a
constructor from the signature indexed by the type of its parameters
and the base type they produce. So, all the SF definitions the user
makes add constructors, in our closure conversion example some
constructors could be for instance: \lstinline!capp!,
\lstinline!clam!, \lstinline!clo!.

Variables and terms are indexed by two types, the first parameter is
always their context and the second is their type. The type \n{var}
represents variables with two constructors: \n{Top} represents the
variable that was introduced last in the context and if \n{Top}
corresponds to the \debruijn index 0 then the constructor \n{Pop}
represents the successor of the variable that it takes as parameter.
It is interesting to consider the parameters of these constructors.
\n{Top} is simply indexed by its context and type (variables $\gamma$
and $\alpha$ respectively). On the other hand, \n{Pop} requires three
type parameter: the first $\gamma$ represents a context, $\alpha$ the
resulting type of the variable, and $\beta$ the type of the extension
of the context. These parameter make it so that if we apply the
constructor \n{Pop} to a variable of type $\alpha$ in context $\gamma$
we obtain a variable of type $\alpha$ in the context $\gamma$ extended
with type $\beta$. \lstinline!Top! and \lstinline!Pop! contexts cannot
be empty because they refer to existing variables. The constructor
\lstinline!Top! is indexed by a context with at least one variable,
and its type is the same as the top variable. On the other hand, the
constructor \lstinline!Pop! is also indexed by a non-empty context but
its type corresponds to the type of the variable that it takes as a
parameter.

As mentioned, terms described by the type family \n{sftm} are indexed by
their context and their type. It is interesting to check in some
detail how the indices of the term constructors follow the typing
rules from Fig.~\ref{fig:sftyp}. The constructor for lambda terms
(\n{Lam}), extends the context $\gamma$ with base type $\alpha$ and
then it produces a term in $\gamma$ of function type from the base
type $\alpha$ to the type of the body $\tau$. The constructor for
boxes simply forces its body to be closed by using the context type
\n{nil}. The constructor \n{Var} simply embeds variables as terms.
Finally the \n{C} constructor has two parameters, one is the name of
the constructor from the user's definitions that constrains the type
of the second parameter, the other is the term of the appropriated
type.

The definition of substitution is a modified presentation of the
substitution for well-scoped \debruijn indices, as for example
presented in~\citep{Benton:12}. We define two types, \n{sub} and
\n{shift} indexed by two contexts, the domain and the range of the
substitutions. Substitutions are either a shift (constructor
\n{Shift}) or the combination of a term for the top-most variable and
the rest of the substitution (constructor \n{Dot}).

Our implementation differs from~\citet{Benton:12} in the
representation of renamings. Benton et.al define substitutions and
renamings, the latter as a way of representing shifts. However to
compute a shift, they need the context that they use to index the
data-types. Hence, contexts are not erasable during run-time. As we do
want contexts to be erasable at run-time, we cannot use
renamings. Instead, we replace renamings with typed shifts (defined in
type \n{shift}), that encode how many variables we are shifting
over. This is encoded in the indices of shifts.

Variables that use \debruijn indices and shifts ultimately correspond
to natural numbers, as they encode how many binders we need to
traverse, or how many variables a shift adds. However, we use
specialized data types because we explicitly carry the typing
information in the indices. Finally, we omit the function implementing
the substituti0on as it is standard. We will simply mention that we
implement a function \code{apply\_sub} of type:

\begin{displaymath}
\alle{\gamma,\delta,\tau}{\tp{sftm}{\gamma,\tau} \to
  \tp{sub}{\gamma,\delta} \to \tp{sftm}{\delta,\tau}}
\end{displaymath}

That applies a substitution moving a term from context $\gamma$ to
context $\delta$.

\filbreak
The finding the path to a variable example from
Section~\ref{sec:path}, contains a signature that defines the untyped
$\lambda$-calculus. This would result in a new type, and two new
constructors for the existing type \lstinline!con!. We will prefix the
definitions with \lstinline!def_! to create unique names:
\begin{itemize}
\item A new type: $\n{def\_tm} \oft *$
\item A constructor for applications:\\
  $\n{def\_app} \oft  \tp{con} {\tp{base}{\n{def\_tm}}, \tp{base}{\n{def\_tm}} \times \tp{base}{\n{def\_tm}}}$
\item A constructor for abstractions:\\
  \hspace*{-1.5em}$\n{def\_lam} \oft\tp{con} {\tp{base}{\n{def\_tm}}, \tp{arr}{\tp{base}{\n{def\_tm}}, \tp{base}{\n{def\_tm}}}} $
\end{itemize}
To implement the signature we declare a new type for the terms of the
defined language ($\n{def\_tm}$), and we extend the type of SF constructors, with as
many constructors as our signature had ($\n{def\_app}$ and
$\n{def\_lam}$).

\section[From Contextual Types to GADTs]{From \sml with Contextual Types  to \smlg}\label{sec:trans}

In this section, we translate \sml with contextual types into the
lower level \smlg. Because our embedding of the syntactic framework SF
in \smlg is intrinsically typed, there is no need to extend the
type-checker to accommodate contextual objects. Further, recall that
we restricted quoted variables and parameter variables such that the
matching operation remains first order. In addition, as our deep
embedding uses a representation with canonical names (namely \debruijn
indices), we are able to translate pattern matching into \smlg's
pattern matching; thus there is no need to extend the operational
semantics of the language.

The translation we describe in this section provides the footprint of
an implementation of it to directly generate OCaml code, as \smlg is
essentially a subset of OCaml. It therefore shows how to extend a
functional programming language such as OCaml with the syntactic
framework with minimal impact on OCaml's compiler.

Our main translation of \sml to \smlg uses the following main operations:
\begin{displaymath}
  \begin{array}{rl}
    \eval{\tau},\eval{\Xi},\eval{\Gamma}:& \mbox{Translate types, signatures and contexts.}\\
    \eval{e}_{\Gamma\vdash\tau}:& \mbox{Type directed translation of expressions.}\\
    \eval{pat}_{\Gamma\vdash\tau}^{\Gamma'}:& \mbox{Translates patterns and outputs $\Gamma'$ the}\\
    & \mbox{context of the bound variables.}\\
  \end{array}
\end{displaymath}

We begin by translating SF types and contexts into \smlg types. These
types are used to index terms in the implementation of SF:
\begin{displaymath}
  \begin{array}{lrcl}
\mbox{SF Types:} & \eval{\const a \to A} & = & \tp{arr}{\const a, \eval A} \\[.2em]
    & \eval{\boxd A} & = & \tp{boxed}{\eval A}\\[.2em]
    & \eval{\const a} & = & \const a\\[.7em]

\mbox{SF Contexts:} & \eval{\cdot} & = & \tp{nil}{} \\[.2em]
 & \eval{\Psi, x \oft \const a} & = & \tp{cons}{\eval{\Psi}, \const a}  \\[.2em]
  \end{array}
\end{displaymath}

The translation of SF terms is directed by their contextual type
$\Psi\vdash A$, because it needs the context to perform the
translation of names to \debruijn indices and the types to
appropriately index the terms.
\begin{displaymath}
  \begin{array}{lrcl}
\mbox{SF Terms:} & \eval{\lam x M}_{\Psi\vdash \const a \to A} & = & \tp{Lam}{\tp{cons}{\eval{\Psi},\const a},\eval A} \eval{M}_{\Psi,\const a\vdash A}\\[.2em]
& \eval{\termbox M}_{\Psi\vdash \boxd A} & = & \tp{Box}{\eval\Psi,\eval A}{\eval{M}_{\cdot\vdash A}}\\[.2em]
 & \eval{x}_{\Psi\vdash\const a} & = & \tp{Var}{\eval\Psi,\const a}{\eval{x}_\Psi^v} \\[.2em]
 & \eval{\app{\const c} {\many M}}_{\Psi\vdash\const a} & = & \tp{C}{\eval\Psi,\eval A,\const a}{(\const c, \eval{\many M}_{\Psi\vdash A\downarrow\const a})} \\
  &                                              & & \text{with } \Sigma(\const c) = A\\[.2em]
  & \eval{M\sub\sigma\Phi\Psi}_{\Psi\vdash A} & = & \code{apply\_sub}\eval{M}_{\Phi\vdash A} \eval{\sigma}_{\Psi\vdash\Phi}\\[.2em]
  & \eval{\mvar u}_{\Psi\vdash A} & = & u\\[.4em]
  & \eval{\pparvar v}_{\Psi\vdash\const a} & = & \tp{Var}{\eval\Psi,\const a}{\eval{\code v}_{\Psi\vdash\const a}}\\[.2em]
                 & \eval{\app N {\many M}}_{\Psi\vdash A\to B\downarrow\const a} & = & \tp{Cons}{\eval\Psi, \tp{arr}{\eval A, \eval B},\const a}\\
    & & & ~~~ {(\eval{N}_{\Psi\vdash A},\eval{\many M}_{\Psi\vdash B\downarrow\const a})}\\[.2em]
  & \eval{\cdot}_{\Psi\vdash\const a\downarrow\const a} & = & \tp{Empty}{\eval{\Psi},\const a,\const a}\\[.2em]

    \mbox{SF Vars.:} & \eval{x}_{\Psi, x\oft\const a\vdash\const a} ^v & = & \tp{Top}{\eval\Psi,\const a}{}\\[.2em]
    & \eval{y}_{\Psi, x\oft\const b\vdash\const a} ^v & = & \tp{Pop}{\eval\Psi,\const a,\const b} \eval{y}_{\Psi\vdash\const a}^v \\[.2em]

  \mbox{Param. Vars:}  & \eval{\code x}_{\Psi\vdash\const a} & = & x\\[.2em]
  & \eval{\pparvar v}_{\Psi,y\oft\const a'\vdash\const a} & = & {\tp{Pop}{\eval\Psi,\const a,\const a'} \eval{\code v}_{\Psi\vdash\const a}}\\[.2em]
  \end{array}
\end{displaymath}

There are are three kinds of variables in the syntactic framework SF:
bound variables, quoted variables and parameter variables. Each kind
requires a different translation strategy. Bound variables are
translated into \debruijn indices where the numbers are encoded using
the constructors \n{Top} and \n{Pop}. Quoted variables are simply
translated into the \smlg variables they quote. And finally the
parameter variables are translated into a \n{Var} constructor to
indicate that the resulting expression is an SF variable, and the
shifts (indicated by extra '\n{\#}') are translated to applications of
the constructor \n{Pop}.

Notice how substitutions are not part of the representation. They are
translated to the eager application of \code{apply\_sub}, an OCaml
function that performs the substitution. Before we call
\code{apply\_sub} we translate the substitution. This amounts to
generating the right shift for empty substitutions and otherwise
recursively translating the terms and the substitution.


We also need to translate SF patterns into \smlg expressions with the
right structure. The special cases are:
\begin{itemize}
\item Variables are translated to \debruijn indexes.
\item Quoted variables simply translate  to \smlg variables.
\item Parameter variables translate to a
  pattern that matches only variables by specifying the
  \n{Var} constructor.
\end{itemize}

The translation of patterns follows the same line as the translation
of terms, however, we do not use the indices of type variables in
\smlg patterns. This is indicated by writing an underscore.

Additionally, the translation of substitutions amounts to generating the right shift for
empty substitutions and otherwise recursively translating the terms and the substitution:
\begin{displaymath}
  \begin{array}{rcl}
    \eval{\cdot}_{\Psi}^{\Psi} & = & \tp{Shift}{\eval\Psi,\eval\Psi} (\tp{Id}{\eval\Psi,\eval\Psi})\\[.2em]
    \eval{\cdot}_{\Psi,x\oft\const a}^{\Phi} & = & \tp{Shift} {\eval{\Psi,x\oft\const a},\eval\Phi}\\
    & & \quad (\tp{Suc}{\eval{\Psi,x\oft\const a},\eval\Phi} {\eval{\cdot}_\Psi^\Phi})\\[.2em]
    \eval{\sigma,x/M}_{\Psi}^{\Phi,x\oft A} & = & \tp{Dot} {\eval\Psi,\eval\Phi,\const a} (\eval{\sigma}_\Psi^\Phi, \eval{M}_{\Phi\vdash\const a})\\[.2em]
  \end{array}
\end{displaymath}

Our Babybel prototype (see also our examples from the beginning)
exploits the power GADTs to model context relations and relies on
OCaml's type reconstruction engine \citep{Garrigue:2013} to infer the
omitted indices.

\begin{displaymath}
  \begin{array}{lrcl}
\mbox{SF Patterns:} & \eval{\lam x R}_{\Psi\vdash \const a \to A} ^\Gamma& = & \tp{Lam}{\_,\_,\_} \eval{R}_{\Psi,\const a\vdash A}^\Gamma\\[.2em]
 & \eval{\termbox R}_{\Psi\vdash \boxd A}^\Gamma & = & \tp{Box}{\_, \_}{\eval{R}_{\cdot\vdash A}^\Gamma}\\[.2em]
 & \eval{x}_{\Psi\vdash\const a}^\cdot & = & \tp{Var}{\_,\_}{\eval{x}_\Psi^p} \\[.2em]
 & \eval{\app{\const c} {\many R}}_{\Psi\vdash\const a}^\Gamma & = & \tp{C}{\_,\_,\_}{\eval{\many R}_{\Psi\vdash A\downarrow \const a}^\Gamma}\\
    & & &  \hspace{1cm} \text{with } \Sigma(\const c) = A \\[.2em]
 & \eval{\mvar u}_{\Psi\vdash A}^{u\oft\eval{\Psi\vdash A}} & = & u\\[.2em]
 & \eval{\pparvar x}_{\Psi\vdash\const a}^{x\oft\eval{\Psi\vdash A}} & = & \tp{Var}{\_,\_}{x}\\[.2em]
 & \eval{\ppparvar x}_{\Psi,y\oft\_\vdash\const a}^{x\oft\eval{\Psi,y\oft\_\vdash A}} & = & \tp{Var}{\_,\_}({\tp{Pop}{\_,\_,\_} x})\\[.2em]

 & \eval{\app R {\many {R'}}}_{\Psi\vdash A\to B\downarrow\const a}^{\Gamma,\Gamma'} & = &
      \tp{Cons}{\_,\_,\_,\_}(\eval{R}_{\Psi\vdash A}^\Gamma, \eval{\many{R'}}_{\Psi\vdash B\downarrow\const a}^{\Gamma'} )\\[.2em]
 & \eval{\cdot}_{\Psi\vdash\const a\downarrow\const a}^{\cdot} & = & \tp{Empty}{\_,\_}\\[.7em]
\mbox{SF Variables:} & \eval{x}_{\Psi, x\oft\const a} ^p & = & \tp{Top}{\_,\_}\\[.2em]
& \eval{y}_{\Psi, x\oft\const b} ^p & = & \tp{Pop}{\_,\_,\_}{\eval{y}_\Psi^p} \\[.2em]
  \end{array}
\end{displaymath}



\begin{figure}
  \centering
\begin{displaymath}
  \begin{array}{rcl}
    \multicolumn{3}{l}{\boxed{\text{Translating Types:}}}\vs
    \eval{D} & = & \tapp D {} \\
    \eval{\tau_1\to\tau_2} & = & \eval{\tau_1} \to \eval{\tau_2}\\
    \eval{[\Psi\vdash A]}& = & \eval{A} \\
    \eval{\many\tau} & = & \eval{\tau_1},\dots,\eval{\tau_n}\vs

    \multicolumn{3}{l}{\boxed{\text{Translating Signatures:}}}\vs
    \eval{\cdot} & = & \cdot \\
    \eval{\Xi, D} & = & \eval{\Xi}, D\oft *\\
    \eval{\Xi, \constr k\oft \many\tau \to D} & = & \eval{\Xi}, \constr k \oft \eval{\many\tau} \to \tapp D{} \vs

    \multicolumn{3}{l}{\boxed{\text{Translating Contexts:}}}\vs
    \eval{\cdot} & = & \cdot \\
    \eval{\Gamma, x \oft \tau} & = & \eval{\Gamma}, x\oft \eval\tau \vs

    \multicolumn{3}{l}{\boxed{\text{Translating SF Signatures:}}}\vs
    \eval{\cdot} & = & \cdot \\
    \eval{\Sigma, \const a\oft \code{type}} & = & \eval{\Sigma}, \const a\oft * \\
    \eval{\Sigma, \const c\oft A \to \const a} & = & \eval{\Sigma}, \const c\oft \eval{A}\to \const a\vs

    \multicolumn{3}{l}{\boxed{\text{Translating SF Contexts:}}}\vs
    \eval{\cdot} & = & \tp{nil}{} \\
    \eval{\Psi, x \oft \const a} & = & \tp{cons}{\eval{\Psi}, \const a}  \\
  \end{array}
\end{displaymath}

  \caption{Translating Types, Signatures, and Contexts}
  \label{fig:trantypsigctx}
\end{figure}

Figure~\ref{fig:trantypsigctx} shows the translation of types,
signatures and contexts for the computational language and the index
language.

The translation of \sml expressions into \smlg directly follows the
structure of programs in \sml and is type directed to fill in the
required types for the \smlg representation. Figure~\ref{fig:transml}
contains the translation.

\begin{figure}
  \centering
\begin{displaymath}
  \begin{array}{rcl}
    \multicolumn{3}{l}{\boxed{\text{Translating expressions:}}}\vs
    \eval{x}_{\Gamma\vdash\tau} & = & x \\
    \eval{\app{\constr k}{\many e}}_{\Gamma\vdash D} & = & \aconstr k {~} {\eval{\many e}_{\Gamma\vdash\many\tau}} \text{ with } \Xi(\constr k) = \many\tau\to D\\
    \eval{\fune f x e}_{\Gamma\vdash\tau_1\to\tau_2} & = & \fixe f {\eval{\tau_1\to\tau_2}} {\lame x {\eval{e}_{\Gamma,x\oft\tau_1\vdash\tau_2}}}\\
 \eval{\app i e}_{\Gamma\vdash\tau} & = & \app{\eval{i}_{\Gamma\vdash\tau_1\to\tau}}{\eval{e}_{\Gamma\vdash\tau_1}}\\
                                & & \text{with } \Gamma\vdash i\synths \tau_1\to\tau\\
    \eval{\lete x i e}_{\Gamma\vdash\tau} & = & \lete x  {\eval{i}_{\Gamma\vdash\tau_1}} {\eval{e}_{\Gamma,x\oft\tau_1\vdash\tau}}\\
                                & & \text{with } \Gamma\vdash i\synths \tau_1\\
    \eval{\matche i \branches}_{\Gamma\vdash\tau} & = & \matche {\eval{i}_{\Gamma\vdash\tau_1}} {\eval{\branches}_{\Gamma\vdash\tau_1\to\tau}}\\
                                & & \text{with } \Gamma\vdash i \synths\tau_1\\
    \eval{e_1,\dots,e_n}_{\Gamma\vdash\many\tau} & = & \eval{e_1}_{\Gamma\vdash\tau_1},\dots,\eval{e_n}_{\Gamma\vdash\tau_n}\\
    \eval{[\hat\Psi\vdash M]}_{\Gamma\vdash[\Psi\vdash A]} & = & \eval{M}_{\Psi\vdash A} \\
    \eval{\cmatche i \cbranches}_{\Gamma\vdash\tau} & =& \matche {\eval{i}_{\Gamma\vdash[\Psi\vdash A]}} {\eval{\cbranches}_{\Gamma\vdash[\Psi\vdash A]\to \tau}}\\
                                & & \text{with } \Gamma\vdash i \synths [\Psi\vdash A]\vs

    \multicolumn{3}{l}{\boxed{\text{Translating branches and patterns:}}}\vs

\eval{pat\mapsto e}_{\Gamma\vdash\tau_1\to\tau_2} & = & \eval{pat}_{\Gamma\vdash\tau_1}^{\Gamma'} \mapsto \eval{e}_{\Gamma,\Gamma'\vdash \tau_2}\vs

 \eval{x}_{\Gamma\vdash\tau}^{x\oft\tau} & = & x \\[.2em]

    \eval{\app {\constr{k}} {\many {pat}}}_{\Gamma\vdash D}^{\Gamma'} & = & \aconstr k {}
    {\eval{\many{pat}}_{\Gamma\vdash\many\tau}^{\Gamma'}}\\[.2em]

  & & \hspace{1cm} \text{with } \Xi(\constr k)=\many\tau\to D\vs

   \multicolumn{3}{l}{\boxed{\text{Translating branches:}}}\vs
   \eval{[\Psi\vdash R]\mapsto e}_{\Gamma\vdash[\Psi\vdash A]\to\tau} & = &
     \eval{R}_{\Psi\vdash A}^{\Gamma'} \mapsto \eval{e}_{\Gamma,\Gamma'\vdash \tau} \\
  \end{array}
\end{displaymath}
  \caption{Translating Computational Expressions}
  \label{fig:transml}
\end{figure}


Finally we show that the translation from \sml with contextual types into \smlg preserves types.

\begin{thm}[Main]$\;$
 \begin{enumerate}
 \item If $\Gamma\vdash e\checks \tau$ then $\cdot;\eval\Gamma\vdash\eval{e}_{\Gamma\vdash\tau} \oft \eval\tau$.
\item If $\Gamma\vdash i \synths \tau$ then $\cdot;\eval\Gamma\vdash\eval{i}_{\Gamma\vdash\tau} \oft \eval\tau$.
  \end{enumerate}
\end{thm}

Our result follows form mutual induction on the typing derivations and
relies on several lemmas that deal with the other judgments and
context lookups. Appendix~\ref{sec:bbproofs} contains the proofs.

\begin{lem}[Ambient Context]\label{lem:ctx}
If $\Gamma(\code{u}) = [\Psi \vdash \const a]$ then\\
$ \eval{\Gamma} (u) =  \tp{sftm}{ \eval{\Psi}, \const a}$.
\end{lem}

\begin{lem}[Terms]\label{lem:tm}$\;$
  \begin{enumerate}
  \item   If $\Gamma;\Psi\vdash M \oft A$ then $\cdot;\eval\Gamma\vdash\eval{M}_{\Psi\vdash A} \oft \eval{\Psi\vdash A}$.
  \item If $\Gamma ; \Psi \vdash \sigma \oft \Phi$ then $\cdot ; \eval{\Gamma} \vdash \eval{\sigma}_ {\Psi \vdash \Phi} : \eval{\Psi\vdash\Phi}$
  \end{enumerate}
\end{lem}

\begin{lem}[Pat.]\label{lem:pat}
  If $\vdash\ptj {pat} \tau \Gamma$ then $\cdot\vdash\ptj{\eval{pat}_{\Psi\vdash A}^\Gamma}{\eval \tau}{\Gamma}$.
\end{lem}

\begin{lem}[Ctx. Pat.]\label{lem:ctxpat}
  If $\Psi\vdash\ptj R A \Gamma$ then $\cdot\vdash\ptj{\eval{R}_{\Psi\vdash A}^\Gamma}{\eval{\Psi \vdash A}}{\Gamma}$.
\end{lem}

Given our set-up, the proofs for the are straightforward by induction
on the structure of $\Gamma$ and the typing derivation for
Lemma~\ref{lem:tmproof}, \ref{lem:patproof}, \ref{lem:ctxpatproof}.


%% file: babybel/poc.tex
\section{A Proof of Concept Implementation}\label{sec:impl}

In this section, we describe the implementation\footnote{Available at
  \url{www.github.com/fferreira/babybel/}} of \babybel which uses the
ideas described in this chapter. One major difference is that Babybel
translates OCaml programs that use syntax extensions for contextual SF
types and terms and translates them into pure OCaml with GADTs. In
fact, even our input OCaml programs may use GADTs to for example
describe context relations on SF contexts (see also our examples from
Sec.\ref{sec:examples}).

The presence of GADTs in our source language also means that we can
specify precise types for functions where we can quantify over
contexts. Let's revisit some of the types of the programs that we
wrote earlier in Sec.\ref{sec:examples}:
\begin{itemize}
\item \lstinline!rewrite: $\gamma$. [$\gamma$ |- tm]->[$\gamma$ |-  tm]!:
  we implicitly quantify over all contexts \lstinline!$\gamma$!
  and then we take a potentially open term and return another term in
  the same context. These constraints imposed in the types are due to
  being able to index types with types thanks to GADTs.
\item \lstinline!get_path: $\gamma$. [$\gamma$,x:tm |- tm]->path!: In
  this case we quantify over all contexts, but the input of the
  function is some term in a non-empty context. On the left of the
  turn-style our type mandates that there is at least one assumption
  of type \lstinline!tm!. Notice how the return type is not inside a
  box(i.e.: square braces) because it is a regular OCaml recursive
  type and not a contextual type. We can mix and match as necessary.
\item \lstinline!conv: $\gamma$ $\delta$.($\gamma$, $\delta$) rel->[$\gamma$ |- tm]->[$\delta$ |- ctm]!:\\
  This final example shows that we can also use the contexts to index
  regular OCaml GADTs. In this function we are translating between
  terms in different representations. Naturally, their contexts
  contain assumptions of different type. To be able to translate
  between these different context representations, it is necessary to
  establish a relation between these contexts. So we need to define a
  special OCaml type (i.e.:\lstinline!rel!) that relates variable to
  variable in each contexts.
\end{itemize}

By embedding the SF in OCaml using contextual types, we can combine and use the
impure features of OCaml. Our example, in Section~\ref{sec:path} takes
advantage of them in our implementation of backtracking with exceptions. Additionally,
performing I/O or using references works seamlessly in the prototype.


The presence of GADTs in our target language also makes the actual
implementation of Babybel simpler than the theoretical description, as
we take advantage of OCaml's built-in type reconstruction. In addition
to GADTs, our implementation depends on several OCaml extensions. We
use Attributes from Section~7.18 of the reference
manual~\citep{Ocaml:2016} and strings to embed the specification of
our signature. We use quoted strings from Section~7.20 to implement
the boxes for terms (\lstinline!{t|...|t}!) and patterns
(\lstinline!{p|...|p}!). All these appear as annotations in the
internal Abstract Syntax Tree in the compiler implementation. To
perform the translation (based on Section~\ref{sec:trans}) we define a
PPX rewriter as discussed in Section~23.1 of the OCaml manual. In our
rewriter, we implement a parser for the framework SF and translate all
the annotations using our embedding.

The Babybel prototype has been used to implement several case studies
that expand from the examples at the beginning of this chapter. The
Babybel distribution contains the following case studies besides some
smaller programs used as a test suite.
\begin{itemize}
\item The example that removes syntactic sugar from Section~\ref{sec:synsugar}.
\item The example that finds the path to a variable from Section~\ref{sec:path}.
\item The closure conversion example from Section~\ref{sec:cloconv}.
\item Inferring types for MiniML.
\item An evaluator for an untyped MiniML.
\item Comparing $\lambda$-terms up to alpha renaming.
\item Translating the $\lambda$-calculus to a CPS form.
\end{itemize}

Currently, implementing these programs was straightforward and not
having to think about the representation of binders was liberating.
Along with the examples there are a couple of compiler phases, it would be
interesting to implement a small compiler combining them and
generating some code. Much remains to be done, but the prototype is
already usable\footnote{Given that type-checking is done by OCaml on
  the translation, the error messages are far from perfect, but
  otherwise it is easy to use.}.


%% file: babybel/conclusion.tex
\section{Related Work}

\babybel and the syntactic framework SF are derived from ideas that
originated in proof checkers based on the logical framework LF such as
the Twelf system~\citep{Pfenning99cade}. In the same category are the
proof and programming languages
Delphin~\citep{Poswolsky:DelphinDesc08} and
Beluga~\citep{Pientka:CADE15} that offer a computational language on
top of the LF. In many ways, this chapter and \babybel are a
distillation of Beluga's ideas applied to a mainstream programming
language. As a consequence, we have shown that we can get some of the
benefits from Beluga at a much lower cost, since we do not have to
build a stand-alone system or extend the compiler of an existing
language to support contexts, contextual types and objects.

Our approach of embedding an LF specification language into a host
language is in spirit related to the systems Abella~\citep{Gacek:JAR12}
and Hybrid~\citep{Felty12} that use a two-level approach. In these
systems we embed the specification language (typically hereditary
harrop formulas) in first-order logic (or a variant of it).  While our
approach is similar in spirit, we focus on embedding SF specifications
into a programming language instead of embedding it into a proof
theory. Moreover, our embedding is type preserving by construction.


There are also many approaches and tools that specifically add support
for writing programs that manipulate syntax with binders -- even if
they do not necessarily use HOAS. FreshML~\citep{Shinwell:ICFP03} and
C$\alpha$ml~\citep{Pottier:2006} extend OCaml's data types with the
ideas of names and binders from nominal logic~\citep{Pitts:2003}. In
these system, name generation is an effect, and if the user is not
careful variables may extrude their scopes. Purity can be enforced by
adding a proof system with a decision procedure that statically
guarantees that no variable escapes its scope~\citep{Pottier:LICS07}.
This adds some complexity to the system. We feel that \babybel's
contextual types offer a simpler formalism to deal with bound
variables. On the other hand, \babybel's approach does not try to
model variables that do not have lexical scope, like top-level
definitions. Another related language is Romeo~\citep{Stansifer:2014}
that uses ideas from Pure FreshML to represent binders. Where our
system statically catches variables escaping their scope, the Romeo
system either throws a run time exception or uses an SMT solver to
prove the absence of scoping issues. The Hobbits for
Haskell~\citep{Westbrook:2011} system is implemented in a similar way
to ours, using quasi-quoting but they use a different formalism based
on the concepts of names and freshness. In general, parametric HOAS
(PHOAS)~\citep{Chlipala:ICFP08} reuses the extensional (as in black
box functions) function space of the language, while our approach
introduces a distinction between an intensional and extensional
function space. A particular approach is described
in~\citep{WashburnWeirich:JFP06} where the authors propose a library
that uses catamorphisms to compute over a parametric HOAS
representation. This is a powerful approach but requires a different
way of implementing recursive functions.

The separate intensional function space from SF allows us to
model binding and supports pattern matching. The extentional function
space allows us to write recursive functions. Following essentially
work started by Despeyroux, Schuermann, and
Pfenning~\citep{Despeyroux97} and then later in the work on Beluga, we
use a modality to embed syntactic objects characterized using the
intensional function space into our programs. This is done while
preserving the host language's (ML) extensional function space that
allows for computation.

\section{Conclusion}

In this chapter, we describe the syntactic framework SF (a simply
typed variant of the logical framework LF with the necessity modality
from S4) and explain the embedding of SF into a functional programming
language using contextual types. This gives programmers the ability to
write programs that manipulate abstract syntax trees with binders
while knowing at type checking time that no variables extrude their
scope. We also show how to translate the extended language back into a
first order representation. For this, we use \debruijn indices and
GADTs to implement the SF in \smlg. Important characteristics of the
embedding are that it preserves the phase separation, making types
(and thus contexts) erasable at run-time. This allows pattern matching
to remain first-order and thus it is possible to compile with the
traditional algorithms.

Finally, we describe \babybel an implementation of these ideas that
embeds SF in OCaml using Contextual Types. The embedding is flexible
enough that we can take advantage of the more powerful type system in
OCaml to make the extension more powerful. We use GADTs in our
examples to express more powerful invariants (e.g. that the
translation preserves the context).




%% file: orca/omacros.tex

\newcommand{\sn}[1]{\ensuremath{\mathsf{#1}}} 


\newcommand{\eqa}{\ensuremath{\mathrel{\equiv_\alpha}}} 
\newcommand{\eqas}{\ensuremath{\mathrel{\equiv_{\alpha,\sigma}}}} 
\newcommand{\rew}{\ensuremath{\mathrel{\leadsto_\sigma}}}
\newcommand{\stepm}{\ensuremath{\mathrel{\leftrightsquigarrow^*}}} 

\newcommand{\eqctx}{\ensuremath{\mathrel{\overset{ctx}\equiv}}}
\newcommand{\eqkind}{\ensuremath{\mathrel{\overset{kind}\equiv}}}
\newcommand{\eqtyp}{\ensuremath{\mathrel{\equiv}}}


\newcommand{\set}[1]{\ensuremath{\code{set}_{#1}}}
\newcommand{\lfkind}{\ensuremath{\star}\xspace}
\newcommand{\opit}[3]{\ensuremath{(#1 \oft #2)\mathrel{\to} #3}}
\newcommand{\mopit}[3]{\ensuremath{{\overrightarrow{(#1 \oft #2)}}\mathrel{\to} #3}}
\newcommand{\opits}[3]{\ensuremath{(\widehat{#1} \oft #2)\mathrel{\twoheadrightarrow} #3}}
\newcommand{\mopits}[3]{\ensuremath{\overrightarrow{(\widehat{#1} \oft #2)}\mathrel{\twoheadrightarrow} #3}}
\newcommand{\ctx}{\code{ctx}}
\newcommand{\elam}[2]{\ensuremath{\lambda #1. #2}}
\newcommand{\eemb}[1]{\ensuremath{\underline{#1}}}
\newcommand{\eip}[1]{\ensuremath{.{#1}}}
\newcommand{\eimp}{\code{impossible}\ }
\newcommand{\eprog}[3]{\ensuremath{\eapp #1 #2 \mathrel{=} #3}}
\newcommand{\evar}[1]{#1}
\newcommand{\easub}{\ensuremath{\mathrel{:=}}}
\newcommand{\esub}[2]{\ensuremath{[#2\easub #1]}} 
\newcommand{\sesub}[2]{\ensuremath{\{#2\easub #1\}}} 
\newcommand{\essub}[2]{\ensuremath{\app{\{#1\}}{#2}}}
\newcommand{\ecssub}[2]{\ensuremath{\app{\{\{#1\}\}}{#2}}}
\newcommand{\etsub}[2]{\ensuremath{\app{[#1]}{#2}}}
\newcommand{\edot}[3]{\ensuremath{#1;#3\easub #2}}

\newcommand{\etts}{\ensuremath{\triangleright}} 
\newcommand{\ebox}[2]{\ensuremath{[\hat{#1}\mathrel{\etts}{#2}]}}
\newcommand{\etbox}[2]{\ensuremath{[{#1}\vdash{#2}}]}
\newcommand{\boxt}[2]{\etbox{#1}{#2}}

\newcommand{\eapp}[2]{\ensuremath{#1\,#2}}
\newcommand{\iskind}{\ \code{is kind}}
\newcommand{\isctx}{\ \code{is ctx}}
\newcommand{\isectx}{\ \code{is erased}}
\newcommand{\isvalid}{\ \code{valid}}


\newcommand{\pconst}[2]{\ensuremath{\app{\aconst{#1}}{\vec{#2}}}}
\newcommand{\pinac}[1]{\ensuremath{.#1}}
\newcommand{\spconst}[2]{\ensuremath{\sconst{#1}\,\,\vec{#2}}}
\newcommand{\psnoc}[2]{\ensuremath{#1,\,\avar{#2}}}


\newcommand{\aprod}{\ensuremath{\mathrel{\otimes}}}
\newcommand{\ato}{\ensuremath{\twoheadrightarrow}}
\newcommand{\aann}[2]{\ensuremath{#1^{#2}}}
\newcommand{\aapp}[2]{\ensuremath{#1\mathrel{\text{\textquotesingle}}#2}}
\newcommand{\aconst}[1]{\ensuremath{\mathbf {#1}}}
\newcommand{\sconst}[1]{\ensuremath{\widehat{\mathbf {#1}}}}
\newcommand{\ac}[1]{\ensuremath{\mathbf {#1}}}
\newcommand{\afor}[2]{\ensuremath{#1:=\,#2}}
\newcommand{\adot}[3]{\ensuremath{#1;\afor {\widehat{#2}}{#3}}}
\newcommand{\alam}[1]{\ensuremath{\widehat\lambda#1.}}
\newcommand{\ashift}[1]{\ensuremath\uparrow^{#1}}
\newcommand{\ashiftone}{\ensuremath\uparrow}
\newcommand{\aclo}[2]{\ensuremath{#1[#2]}}
\newcommand{\aempty}{\ensuremath{^{\wedge}}}
\newcommand{\aectx}{\ensuremath{\emptyset}}
\newcommand{\asnoc}[3]{\ensuremath{#1,\,\avar{#2}\oft #3}}

\newcommand{\naid}[1]{\ensuremath{\code{id}_{#1}}}
\newcommand{\aid}{\ensuremath{\code{id}}}
\newcommand{\avar}[1]{\ensuremath{\widehat{#1}}}
\newcommand{\acomp}{\ensuremath{\mathrel{\circ}}}

\newcommand{\aunbox}[2]{\ensuremath{\lfloor{#1}\rfloor_{#2}}}
\newcommand{\acunbox}[1]{\aunbox{#1}{}}

\newcommand{\alen}[1]{\ensuremath{|#1|}}

\newcommand{\aeq}{\ensuremath{\mathrel{\equiv}}}
\newcommand{\aleq}[1]{\ensuremath{\mathrel{\overset{#1}{\sim}}}}
\newcommand{\wfc}{\ensuremath{\ \text{wfc}}}


\newcommand{\annlam}[2]{\elam{{#1}^{#2}}}
\newcommand{\annif}[3]{\ensuremath{\code{if}\, #1 \,\code{then}\, #2 \,\code{else}\, #3}}
\newcommand{\anncase}[3]{\ensuremath{\code{case}\, #1\,\code{inl}\, x \mapsto #2\,\mid\code{inr}\, x \mapsto #3}}
\newcommand{\anncasepr}[3]{\ensuremath{\code{case}\, #1\,\code{inl \_}\, \mapsto #2\,\mid\code{inr \_}\, \mapsto #3}}
\newcommand{\anninl}[1]{\code{inl}^{#1}\,}
\newcommand{\anninr}[1]{\code{inr}^{#1}\,}


\newcommand{\Var}{\ensuremath{\mathcal{V}}\xspace}
\newcommand{\Exp}{\code{Exp}\xspace}
\newcommand{\Neu}{\code{Neu}\xspace}
\newcommand{\CanC}{\ensuremath{\mathcal{C}}\xspace}
\newcommand{\Can}[2]{\ensuremath{\CanC_{#1,#2}}}
\newcommand{\inT}[3]{\ensuremath{\llbracket #1 \rrbracket _{#2,#3} }}
\newcommand{\SN}{\ensuremath{\code{SN}}}


\newcommand{\whnff}{\code{whnf}\xspace}
\newcommand{\appf}{\code{app}}
\newcommand{\lookupf}{\code{lookup}}
\newcommand{\envf}{\code{env}}
\newcommand{\shiftf}{\code{shift}}


%% file: orca/intro.tex
\section{Introduction}

This is the last chapter\footnote{This is the last chapter before the
  concluding remarks and conclusion.}, but in a sense it is not the
end but the beginning of what we hope will be a fruitful research
path. This chapter presents the first steps on answering an important
question in this field, namely how to integrate contextual types and
the logical framework LF with a reasoning language that provides full
dependent types. Here, we develop a system that extends the ideas of
Chapter~\ref{chp:babybel} to a system that combines the logical
framework LF~\citep{Harper93jacm} with an extended Martin-L\"of style
type theory~\citep{Martin-Loef84a}
(MLTT). 
As in Babybel, the embedding uses contextual
types~\citep{Nanevski:ICML05} to mediate between both layers and allow
the system to effortlessly represent open objects.

Furthermore, this system can be seen also as an extension of the
theory of Beluga~\citep{Pientka:CADE15} into full dependent types. The
resulting system supports the definition of abstract syntax with
binders and the use of substitutions. This simplifies proofs about
systems with binders by providing for free the substitution lemma and
lemmas about the equational theory of substitutions. As in Beluga, we
mediate between specifications and computations via contextual types.
However, unlike Beluga, we can embed and use computations directly in
contextual objects and types, hence we allow the arbitrary mixing of
specifications and computations. Moreover, dependent types allow for
reasoning about proofs in addition to reasoning about LF
specifications. This resulted in the Orca language, with its prototype
implementation available at \url{https://github.com/orca-lang/orca}.





From an expressivity perspective, Orca and Beluga differ as when using the
former the user can reason about specifications \textbf{and}
computational functions. While the latter cannot reason about
computations and instead needs to specify the computations in a
relational style in LF.


The focus of this chapter is to describe the theory that is
implemented by the Orca language (where the prototype already
implements some extensions to the theory presented here) and to
describe some ideas for an appropriate surface language for this
theory. In particular, we will discuss how the prototype tries to
alleviate the mediation between the two main syntactic categories (the
specifications in LF and the reasoning language) by performing what we
call a box inference pass on the surface language. Much remains to be
done. As we discuss in the future work section a particularly
important next step is to work on the meta-theory of this language.
This will settle the answer to the question how to combine LF
specifications with dependent types, while this chapter just hints at
the answer.

\section{Example: Translating Boolean Types}\label{sec:tranbool}

Let's consider a simple example of translating between two languages
with a different representation of boolean types. Starting from a
simply typed $\lambda$-calculus with boolean type:

\begin{displaymath}
  \begin{array}{rlcl}
    \text{Types} & t & \bnfas & \code{bool} \bnfalt t \to t'\\
    \text{Expressions} & e & \bnfas & \code{tt} \bnfalt \code{ff} \bnfalt
                                \annif e {e'}{e''}\bnfalt \annlam x t e
                                \bnfalt \app e {e'} \bnfalt x\\
    \text{Context} & \Gamma & \bnfas & \cdot \bnfalt \Gamma,x\oft t
  \end{array}
\end{displaymath}

We want to translate into a language that lacks booleans, but instead
it has unit and sums:
\begin{displaymath}
  \begin{array}{rlcl}
    \text{Types} & T & \bnfas & \code{1} \bnfalt T + T'  \bnfalt T \to T'\\
    \text{Expressions} & E & \bnfas & \code{()} \bnfalt \anninl T E \bnfalt \anninr T E\\
    & & \bnfalt & (\anncase E {E'}{E''})\\
    & & \bnfalt & \annlam x T E \bnfalt \app E {E'} \bnfalt x\\
    \text{Context} & \Delta & \bnfas & \cdot \bnfalt \Delta,x\oft T
  \end{array}
\end{displaymath}

The typing for both languages is as one would expect. For completeness,
Figure~\ref{fig:srctyp} contains the typing rules for the source
language and Figure~\ref{fig:tgttyp} contains the typing for the target language.

The idea for the translation is that type \code{bool} from the source
language can be represented by type $\code{1} + \code{1}$ in the
target language, and that there is a mechanical way to translate
expressions in a way that preserves types.

\begin{figure}
  \centering
  \begin{displaymath}
    \begin{array}{c}
      \multicolumn{1}{l}{\boxed{\Gamma\vdash e \oft t} : \text{$e$ is of type $t$ in context $\Gamma$}}\vs

      \infer[\rl{t-tt}]{\Gamma\vdash \code{tt} \oft \code{bool}}{}

      \quad

      \infer[\rl{t-ff}]{\Gamma\vdash \code{ff} \oft \code{bool}}{}

      \vs

      \infer[\rl{t-if}]
      {\Gamma\vdash \annif e {e'}{e''} \oft t}
      {\Gamma\vdash e \oft \code{bool}
      & \Gamma\vdash e' \oft t
      & \Gamma\vdash e'' \oft t}

      \vs

      \infer[\rl{t-lam}]
      {\Gamma\vdash \annlam x t e \oft t \to t'}
      {\Gamma,x\oft t\vdash e \oft t'}

      \quad

      \infer[\rl{t-app}]
      {\Gamma\vdash \app e {e'} \oft t}
      {\Gamma\vdash e\oft t' \to t
      & \Gamma\vdash e' \oft t'}

      \vs

      \infer[\rl{t-var}] {\Gamma\vdash x \oft t} {x\oft t \in \Gamma}
    \end{array}
  \end{displaymath}
  \caption{The Typing of the Source Language}
  \label{fig:srctyp}
\end{figure}

\begin{figure}
  \centering
  \begin{displaymath}
    \begin{array}{c}
      \multicolumn{1}{l}{\boxed{\Delta\vdash E \oft T} : \text{$E$ is of type $T$ in context $\Delta$}}\vs

      \infer[\rl{T-unit}]{\Delta\vdash \code{()} \oft \code{1}}{}

      \quad

      \infer[\rl{T-var}] {\Delta\vdash x \oft T} {x\oft T \in \Delta}

      \vs

      \infer[\rl{T-inl}]
      {\Delta\vdash \anninl T E \oft T' + T}
      {\Delta\vdash E \oft T'}

      \quad

      \infer[\rl{T-inr}]
      {\Delta\vdash \anninr T E \oft T + T'}
      {\Delta\vdash E \oft T'}

      \vs

      \infer[\rl{T-case}]
      {\Delta\vdash \anncase E {E'}{E''} \oft T}
      {\Delta\vdash E \oft T' + T''
      & \Delta,x\oft T'\vdash E' \oft T
      & \Delta,x\oft T''\vdash E'' \oft T}

      \vs

      \infer[\rl{T-lam}]
      {\Delta\vdash \annlam x T E \oft T \to T'}
      {\Delta,x\oft T\vdash E \oft T'}

      \quad

      \infer[\rl{T-app}]
      {\Delta\vdash \app E {E'} \oft T}
      {\Delta\vdash E\oft T' \to T
      & \Delta\vdash E' \oft T'}
    \end{array}
  \end{displaymath}
  \caption{The Typing of the Target Language}
  \label{fig:tgttyp}
\end{figure}

The translation of types is a function from types in the source
language to types in the target language and it is defined as follows:

\begin{displaymath}
  \begin{array}{rcl}
    \tran{\code{bool}} & = & \code{1} + \code{1} \\
    \tran{t \to t'} & = & \tran t \to \tran {t'}
  \end{array}
\end{displaymath}

Finally, the translation of expressions depends on the translation of
types and it is as follows:
\begin{displaymath}
  \begin{array}{rcl}
    \tran{\code{tt}} & = & \anninl {\code{unit}} \code{()} \\
    \tran{\code{ff}} & = & \anninr {\code{unit}} \code{()} \\
    \tran{\annif e {e'} {e'}} & = & \anncasepr {\tran e} {\tran {e'}} {\tran {e''}} \\
    \tran{\annlam x t e} & = & \annlam x {\tran t} {\tran e} \\
    \tran{\app e {e'}} & = & \app {\tran e} {\tran {e'}} \\
    \tran{x} & = & x \\
  \end{array}
\end{displaymath}

This translation can be implemented as a function that computes the
expression in the target language. Moreover, it would be interesting to
establish some properties of the translation. One such property is the
fact that the translation is type preserving. Notice, that because
both languages have different types, we say that the idea of type
preservation is \emph{up-to} the notion of translated types.

Let's consider how to formalize this in Orca in the most direct way.
This is as opposed to formalizing an existing proof. The idea is to
minimize the lemmas and complexity of the formalization. To this end,
let's specify the source and the target languages in an intrinsically
typed way. That is, using dependent types in the logical framework LF
to easily restrict the expressions to those that are well-typed.

\lstset{language=Orca}
\begin{lstlisting}
spec s-tp : * where
| bool : s-tp
| arr : s-tp ->> s-tp ->> s-tp

spec s-exp : s-tp ->> * where
| app : (s : s-tp) ->> (t : s-tp) ->>
        s-exp (arr s t) ->> s-exp s ->> s-exp t
| lam : (s : s-tp) ->> (t : s-tp) ->>
        (s-exp s ->> s-exp t) ->> s-exp (arr s t)
| tt : s-exp bool
| ff : s-exp bool
| if : (t : s-tp) ->>
       s-exp bool ->> s-exp t ->> s-exp t ->> s-exp t
\end{lstlisting}

In this code, we define two new types families, first \lstinline!s-tp!
for types and for expressions \lstinline!s-exp!. Notice how the
expressions form a type family where expressions are indexed by their
types (i.e.: the function from \lstinline!s-tp! to \lstinline!*! that
represents the base kind for LF specifications). The declaration of
constructors is straightforward even if the lack of implicit
parameters makes it verbose. Specifically, applications take two
expressions, the first one a function from \lstinline!s! to
\lstinline!t! and a parameter of appropriate type to produce an
expression of type \lstinline!t!. The constructor \lstinline!lam!
builds an abstraction from an expression with one bound variable
(using HOAS we represent the binder using the LF function space) to
produce an expression of function type. Then the two values of boolean
type, and finally expressions of boolean type are eliminated with
\lstinline!if! expressions. It is straightforward to see how these
definitions follow from the typing rules in Figure~\ref{fig:srctyp}.
Notice how, for LF specifications (declared with \lstinline!spec!), we
use a special function space \lstinline!(x : $\alpha$) ->> $\beta$! to
indicate the LF function space. Also, as it is commonly done, for non
dependent functions we use \lstinline!$\alpha$ ->> $\beta$! when
$\alpha$ does not occur in $\beta$.

\begin{lstlisting}
spec t-tp : * where
| tunit : t-tp
| tsum : t-tp ->> t-tp ->> t-tp
| tarr : t-tp ->> t-tp ->> t-tp

spec t-exp : t-tp ->> * where
| tapp : (s : t-tp) ->> (t : t-tp) ->>
         t-exp (tarr s t) ->> t-exp s ->> t-exp t
| tlam : (s : t-tp) ->> (t : t-tp) ->>
         (t-exp s ->> t-exp t) ->> t-exp (tarr s t)
| tone : t-exp tunit
| tinl : (s : t-tp) ->> (t : t-tp) ->>
         t-exp s ->> t-exp (tsum s t)
| tinr : (s : t-tp) ->> (t : t-tp) ->>
         t-exp t ->> t-exp (tsum s t)
| tcase : (s : t-tp) ->> (t : t-tp) ->> (r : t-tp) ->>
          t-exp (tsum s t) ->> (t-exp s ->> t-exp r) ->>
                             (t-exp t ->> t-exp r) ->>
          t-exp r
\end{lstlisting}

The target language is similar to the source language, and it is
represented in a analogous way. Notice, how we model the binders in
the branches of the case expression with the logical framework's
function space.

The translation of types is represented with a function from source
types to target types. We use the keyword \lstinline!def! to introduce
a top level function definition with pattern matching. It is
implemented as follows:
\begin{lstlisting}
def tran-tp : (|- s-tp) -> (|- t-tp) where
| bool => tsum tunit tunit
| (arr s t) => tarr (tran-tp s) (tran-tp t)
\end{lstlisting}

First notice the type specification after the colon in the first line
where the turn-styles specify that these are specification types and
that they are closed types (after all these two languages types have
no binders in them). Then the implementation of the function goes by
pattern matching on the argument and producing the appropriate target
type according to the definition. In Orca, definitions are elaborated
in a type directed way to allow the user to omit the boxes when
possible. This is explained further when we describe the prototype in
Section~\ref{sec:orcaprototype} and we show the result of the
reconstruction of boxes for the translation function in
Figure~\ref{fig:tranelab}.

We proceed similarly to implement the translation of expressions.
However, let's consider the type for the translation function first. A
first idea could be:
\begin{lstlisting}
def tran : (st : |- s-tp) -> (|- s-exp st) ->
             (|- t-exp (tran-tp st)) where
...
\end{lstlisting}

This type indicates that given a closed source expression of type
\lstinline!st! we can produce a closed target expression whose type
will be the translation of type \lstinline!st!. Notice how this is a
place we interleave specifications and computations as the resulting
expression's type is the result of actually calling the function
\lstinline!tran-tp! on the source expression. In a system without
computation in types like Beluga, we need to implement this using a
relation and this requires a couple of lemmas to show that the
relation is actually a function. As a reference
Appendix~\ref{chp:belugatran} contains a Beluga implementation of this
example. However, this type is not strong enough as the recursive
calls may go under binders. Thus we need to be able to translate open
expressions. We need to generalize this type to arbitrary contexts and
add a context relation that expresses that the source context and
target context grow in unison. Context relations are required because
so far, Orca contexts can contain assumptions of any kind. For this we
define a binary relation on contexts as follows:
\begin{lstlisting}
data rel: ctx -> ctx -> set where
| empty : rel 0 0
| cons : (g h : ctx) (t : |- s-tp) ->
         rel g h ->
         rel (g, x: s-exp t)
             (h, y: t-exp (tran-tp t))
\end{lstlisting}

In this relation the \lstinline!empty! constructor states that empty
contexts are related, and the \lstinline!cons! constructor shows that
if we have two related contexts we can extend both contexts with a
fresh related assumption.

\filbreak
With all this in place, we can write the translation function:
\begin{lstlisting}
def tran : (g h : ctx) -> rel g h ->
           (st : |- s-tp) -> (g |- s-tm st) ->
           (h |- t-tm (tran-tp st))
where
| g h r t (app s[^] .t m n) =>
    tapp (tran-tp s) (tran-tp t)
         (tran g h r (arr s t) m) (tran g h r s n)
| g h r (arr s[^] t[^]) (lam .s .t m) =>
    tlam (tran-tp s) (tran-tp t)
         (\x. tran (g,x:s-tm s) (h,x:t-tm (tran-tp s))
                   (cons g h s r) t (m x))
| g h r bool tt =>
    tinl tunit tunit tone
| g h r bool ff =>
    tinr tunit tunit tone
| g h r t (if .t b e1 e2) =>
    tcase tunit tunit (tran-tp t) (tran g h r bool b)
                      (\x. tran g h r t e1)
                      (\x. tran g h r t e2)
(* when translating variables it is either the top variable *)
| ._ ._ (cons g h .t r) t (g, x: s-tm t :> x) =>
    (h, x: t-tm (tran-tp t) :> x)
(* or a variable deeper in the context (hence the shift) *)
| ._ ._ (cons g h t r) s (v[^1]) =>
    tran g h r s v
\end{lstlisting}

The function \lstinline!tran! has type:
\begin{lstlisting}
(g h : ctx) -> rel g h -> (st : |- s-tp) -> <<(g |- s-exp st)>> ->
     <<(h |- t-exp (tran-tp st))>>
\end{lstlisting}

This function transforms expressions of type:
\lstinline!<<(g |- s-exp st)>>! into expressions of type
\lstinline!<<(h |- t-exp (tran-tp st))>>!, where it computes the type
of the source expression by embedding the computation
\lstinline!tran-tp st! in the resulting type. The computation proceeds
by pattern matching, in the patterns we write \lstinline!._! for
inaccessible patterns that the implementation fills in automatically.
Inaccessible patterns are a concept from Agda that simplifies
implementing pattern matching, their value is forced by the other
variables, for example in the pattern for lambda expressions
\lstinline!g h r (arr s[^] t[^]) (lam ._ ._ m)! the two inaccessible
patterns are forced by the arrow type to be \lstinline!s! and
\lstinline!t! respectively. In the same pattern, we want to indicate
that the types of expressions are closed, we annotate \lstinline!s!
and \lstinline!t! with the empty substitution (i.e.:\lstinline![^]!)
that prevents them from using variables from the context. We build
contextual types by combining a term with a context (e.g.:
\lstinline!g |- s-exp st!), for contextual terms we use a different
syntax to make it easier to disambiguate between types and terms,
contexts and terms are combined in the following way:
\lstinline!g, x: s-tm t :> x!.

This is a simple example, but even in its simplicity it showcases the
power of the Orca language. Even if the Orca prototype does not offer many of
the features present in Beluga, this example shows a hint to what
is coming in the future. In the next section, we peek behind the
curtains into Orca's type theory.


%% file: orca/calculus.tex
\section{Orca's Core Calculus}\label{sec:coreorca}

Orca's core calculus is composed of two separate languages, a
\emph{reasoning} and \emph{computation} language on one side and a
\emph{specification} language on the other side. We refer to the
former as the computation language when we want to emphasize that we
can write programs with it, and as the reasoning language when trying
to emphasize that we can write proofs with it.

The syntax for the reasoning language is a fully dependently typed
theory, extended with contextual types and terms (in red in the
syntax), it is exactly as one would expect a type theory to be. As is
the case with Agda~\citep{Norell:phd07} and Idris~\citep{Brady:JFP13},
we do not have recursion in the calculus, but we do top-level
recursive definitions that allow for defining inductive functions. We
say contextual objects are boxed specifications as a reference to the
contextual modality in contextual modal type
theory~\citep{Nanevski:ICML05}. The type for contexts is a
simplification that avoids discussing classifiers for contexts (i.e.:
context schemas). This is a limitation of the current presentation of
the theory, and it makes discussing totality and coverage difficult.
In short, it would be necessary to add schemas to be able to describe
a coverage checking algorithm.

\begin{displaymath}
  \begin{array}{rlcll}
    \multicolumn{5}{c}{\boxed{\text{Reasoning Language}}}\vs
    \text{Terms} & E, S, T & \bnfas & \set n & \text{A universe of level $n$}\\
    & & \bnfalt & \opit x S T & \text{A function type}\\
    & & \bnfalt & \elam x E & \text{An abstraction} \\
    & & \bnfalt & \eapp E {E'} & \text{An application} \\
    & & \bnfalt & x & \text{A variable} \\
    & & \bnfalt & \aconst c & \text{A constant} \\
    & & \bnfalt & \ebox \Psi M & \text{\color{red}A contextual term}\\
    & & \bnfalt & \etbox \Psi \alpha & \text{\color{red}A contextual type}\\
    & & \bnfalt & \ctx & \text{\color{red}The type of contexts}
  \end{array}
\end{displaymath}

Conversely, the specification language is a version of the logical
framework LF, but instead of meta-variables, it supports a direct
embedding of computations together with a substitution, to unbox the
result of computation into a specification object. Of course, the
computation needs to produce a result of specification type.

\begin{displaymath}
  \begin{array}{rlcll}
    \multicolumn{5}{c}{\boxed{\text{Specification Language}}}\vs

    \text{Kinds} & K & \bnfas & \lfkind & \text{The base kind}\\
    & & \bnfalt & \opits x \alpha K & \text{A kind family} \vs

    \text{Families} & \alpha, \beta & \bnfas & \aconst a & \text{A base type} \\
    & & \bnfalt & \opits x \alpha \beta & \text{A type family} \\
    & & \bnfalt & \app \alpha M & \text{A family instance}\vs

    \text{Objects} & M, N & \bnfas & \alam x M & \text{An abstraction}\\
    & & \bnfalt & \aapp M N & \text{An application} \\
    & & \bnfalt & \avar x & \text{A variable} \\
    & & \bnfalt & \sconst c & \text{A constant} \\
    & & \bnfalt & \aunbox E \sigma & \text{\color{red} An unboxed computation}\vs

    \text{Substitutions} & \sigma & \bnfas & \aempty & \text{An empty substitution} \\
    & & \bnfalt & \naid n & \text{An identity substitution,}\\
    & & & & \text{with weakening} \\
    & & \bnfalt & \adot \sigma x M & \text{A substitution extension} \vs

    \text{Context} & \Psi & \bnfas & \aectx & \text{Empty context} \\
    & & \bnfalt & \acunbox x & \text{\color{red} An unboxed context variable} \\
    & & \bnfalt & \asnoc \Psi x \alpha & \text{A context extension}\vs

    \text{Erased Context} & \hat\Psi & \bnfas & \aectx & \text{An empty context}\\
    & & \bnfalt & \acunbox x & \text{\color{red} An unboxed context variable}\\
    & & \bnfalt & \psnoc{\hat\Psi} x &  \text{A variable declaration}
  \end{array}
\end{displaymath}

Finally, we define contexts for computational variables, that as in
Chapter~\ref{chp:babybel} will be the meta-variables for incomplete
objects, and a signature to define types and constants for computation
terms and for the specification language.

\begin{displaymath}
  \begin{array}{rlcll}
    \multicolumn{5}{c}{\boxed{\text{Signature and Computational Context}}}\vs

    \text{Context} & \Gamma & \bnfas & \cdot & \text{Empty context}\\
    & & \bnfalt & \Gamma,x\oft T & \text{An assumption} \vs

    \text{Signature} & \Sigma & \bnfas & \cdot & \text{Empty signature}\\
    & & \bnfalt & \Sigma,\const c \oft T & \text{A comp. constant} \\
    & & \bnfalt & \Sigma,\const a \oft K & \text{A spec. type} \\
    & & \bnfalt & \sigma,\sconst c \oft \alpha & \text{A spec. constant} \\
  \end{array}
\end{displaymath}

We present a Martin-L\"of style type theory with an infinite hierarchy
of universes extended with contextual types $[\Psi \vdash \alpha]$
which represent a specification type $\alpha$ in an open context
$\Psi$. Specification types are classified by a single universe
$\star$, together with an intensional function space
$\opits x \alpha \beta$ and constants. Substitutions can be an empty
substitution $\aempty$ which weakens closed objects, an identity
substitution $\naid n$ that weakens the context with $n$ elements (we
write $\aid = \naid 0$ as syntactic sugar), or a substitution extended
with an object for a variable $\adot \sigma x M$. Contexts are either
empty $\aectx$ or a context extended with a new assumption
$\asnoc \Psi x A$.

Type theory terms and specification objects are typed with two
different judgments $\Gamma \vdash E : T$ and
$\Gamma; \Psi \vdash M: \beta$, respectively. For instance, we type
the type theory functions and specifications functions in the
following way:
\begin{displaymath}
  \begin{array}{c}
    \infer[\rl{t-fun}]
    {\Gamma\vdash\elam x E \oft \opit x S T}
    {\Gamma,x\oft S\vdash E \oft T}

\quad

    \infer[\rl{s-lam}]
    {\Gamma;\Psi\vdash \alam x M \oft \opits x \alpha \beta}
    {\Gamma;\Psi,\widehat{x}\oft\alpha \vdash M \oft \beta}
  \end{array}
\end{displaymath}

The two function spaces differ by the contexts they act on and that
the computational function space is extensional (i.e.: we can compute
with these functions) while the specification function space is
intensional (i.e.: we can inspect these functions with pattern
matching) in the sense of \citet{Pfenning:LICS01}. Type theory
$\lambda$-abstractions introduce variables in the computational
context $\Gamma$ while specification $\widehat\lambda$-abstractions
use the specification context $\Psi$. The $\beta$-rule for each
$\lambda$-abstraction uses its own substitution operation for its
corresponding class of variables denoted respectively $\{\theta\}E$
and $[\sigma] M$. The definition for the substitution operations is
presented in Figures~\ref{fig:compsubst}, \ref{fig:compsubstspecs},
\ref{fig:compsubstctx}, and \ref{fig:lfsubstspecs}. This is in
contrast to $\aunbox E\sigma$ for the closure that embeds a
computation of boxed type into a specification and is defined by the
following introduction rule:
\[
  \infer[\rl{s-clo}]
  {\Gamma;\Psi\vdash \aunbox E \sigma \oft [\sigma]\alpha}
  {\Gamma;\Psi\vdash \sigma \oft \Phi & \Gamma\vdash E \oft \etbox \Phi \alpha}
\]

\begin{figure}
  \centering
  \begin{displaymath}
    \begin{array}{lcl}
      \multicolumn{3}{c}{\boxed{\text{Substitution in the reasoning language}}} \vs
      \essub\theta {\set n} & = & \set n\\
      \essub\theta {\opit x S T} & = & \opit x {\essub\theta S} {\essub{\edot\theta x x} T} \text{ for $x$ free in $\theta$} \\
      \essub\theta {\elam x E} & = & \elam x {(\essub{\edot\theta x x} E)} \text{ for $x$ free in $\theta$}\\
      \essub\theta {\eapp E {E'}} & = & \eapp {(\essub\theta E)} {(\essub\theta {E'})}\\
      \essub\theta {x} & = & E \text{ if } x \easub E \text{ is in $\theta$}\\
      \essub\theta {x} & = & x \text{ if } x \text{ is not in $\theta$}\\
      \essub\theta {\aconst c} & = & \aconst c\\
      \essub\theta {\ebox \Psi M} & = & [{\essub\theta{\hat\Psi}}\mathrel{\etts} {\essub\theta M}]\\
      \essub\theta {\etbox \Psi \alpha} & = & \etbox {\essub\theta \Psi} {\essub\theta\alpha}
    \end{array}
  \end{displaymath}
  \caption{Computation Substitution}
  \label{fig:compsubst}
\end{figure}

\begin{figure}
  \centering
  \begin{displaymath}
    \begin{array}{lcl}
      \multicolumn{3}{c}{\boxed{\text{Substitution in objects of the specification language}}} \vs
      \essub\theta {\alam x M} & = & \alam x {(\essub\theta M)}\\
      \essub\theta {\aapp M N} & = & \aapp {(\essub\theta M)} {(\essub\theta N)}\\
      \essub\theta {\avar x} & = & \avar x\\
      \essub\theta {\sconst c} & = & \sconst c\\
      \essub\theta {\aunbox E \sigma} & = & \aunbox {(\essub\theta E)}{(\essub\theta\sigma)}\vs
      \multicolumn{3}{c}{\boxed{\text{Substitution in families of the specification language}}} \vs
      \essub\theta {\aconst a} & = & \aconst a\\
      \essub\theta {\opits x \alpha \beta} & = & \opits x {(\essub\theta\alpha)} {(\essub\theta\beta)}\\
      \essub\theta {\app \alpha M} & = & \app {(\essub\theta\alpha)}{(\essub\theta M)} \\
    \end{array}
  \end{displaymath}
  \caption{Computation Substitution in Specifications}
  \label{fig:compsubstspecs}
\end{figure}

\begin{figure}
  \centering
  \begin{displaymath}
    \begin{array}{lcl}
      \multicolumn{3}{c}{\boxed{\text{Substitution in the specification contexts}}} \vs
      \essub\theta {\aectx} & = & \aectx\\
      \essub\theta {\asnoc\Psi x \alpha} & = & \asnoc {(\essub\theta\Psi)} x {(\essub\theta\alpha)}\vs
      \multicolumn{3}{c}{\boxed{\text{Substitution in the specification substitutions}}} \vs
      \essub\theta {\aempty} & = & \aempty\\
      \essub\theta {\naid n} & = & \naid n\\
      \essub\theta {\adot \sigma x M} & = & \adot {(\essub\theta\sigma)} x {(\essub\theta M)}\\
    \end{array}
  \end{displaymath}
  \caption{Computation Substitution in Contexts and Substitutions}
  \label{fig:compsubstctx}
\end{figure}

\begin{figure}
  \centering
  \begin{displaymath}
    \begin{array}{lcl}
      \multicolumn{3}{c}{\boxed{\text{LF Substitution in objects}}} \vs
      \etsub\sigma {\alam x M} & = & \alam x {(\etsub{\adot\sigma x {\avar x}} M)}\\
      \etsub\sigma {\aapp M N} & = & \aapp {(\etsub\sigma M)} {(\etsub\sigma N)}\\
      \etsub\sigma {\avar x} & = & M \quad\quad \text{ if $\avar x \easub M$ is in $\sigma$}\\
      \etsub\sigma {\avar x} & = & \avar x \quad\quad \text{ if $\avar x$ is not in $\sigma$}\\
      \etsub\sigma {\sconst c} & = & \sconst c\\
      \etsub\sigma {\aunbox E {\sigma'}} & = & \aunbox {E}{(\etsub\sigma{\sigma'})}\vs
      \multicolumn{3}{c}{\boxed{\text{LF Substitution in families}}} \vs
      \etsub\sigma {\aconst a} & = & \aconst a\\
      \etsub\sigma {\opits x \alpha \beta} & = & \opits x {(\etsub\sigma\alpha)} {(\etsub{\adot\sigma x {\avar x}}\beta)}
                                                 \text{ with $\avar x$ fresh in $\sigma$}\\
      \etsub\sigma {\app \alpha M} & = & \app {(\etsub\sigma\alpha)}{(\etsub\sigma M)} \vs
      \multicolumn{3}{c}{\boxed{\text{LF Substitution in kinds}}} \vs
      \etsub\sigma {\lfkind} & = & \lfkind \\
      \etsub\sigma {\opits x \alpha K} & = & \opits x {\etsub\sigma\alpha} {(\etsub{\adot\sigma x {\avar x}} K)}
                                             \text{ with $\avar x$ fresh in $\sigma$}\\
    \end{array}
  \end{displaymath}
  \caption{Specification Substitution}
  \label{fig:lfsubstspecs}
\end{figure}

In addition to the description of the term calculus we present a
prototype implementation\footnote{Available
  at:\url{http://github.com/orca-lang/orca}} for our type theory which
supplements the calculus with recursive functions and Agda-style
dependent pattern matching \citep{Norell:phd07} extended to allow
matching on specifications including specification-level
$\lambda$-abstractions and thus abstracting over binders.






In conclusion, we present a theory that allows for embedding
contextual LF specifications into a fully dependently typed language
that simplifies proofs about structures with syntactic binders (such
as programming languages and logics). Moreover, we have a prototype,
the Orca system, in which we implement some example proofs.

\subsection{The Typing of Programs}

\begin{figure}
  \centering
\begin{displaymath}
  \begin{array}{c}
    \multicolumn{1}{l}
    {\boxed{\Gamma\vdash E \oft T}\mathrel{:}\mbox{$E$ is of type $T$ in context $\Gamma$}}\vs

    \infer[\rl{t-set}]
    {\Gamma \vdash \set n \oft \set {(n+1)}}
    {}

    \quad

    \infer[\rl{t-var}]
    {\Gamma\vdash x \oft T}
    {x\oft T \in \Gamma}

    \quad

    \infer[\rl{t-con}]
    {\Gamma\vdash \aconst c \oft T}
    {\aconst c \oft T \in \Sigma}

    \vs

   \infer[\rl{t-pi-1}]
    {\Gamma\vdash \opit x S T \oft \set 0}
    {\Gamma\vdash S \oft \set n & \Gamma,x\oft S\vdash T \oft \set 0}

   \vs

   \infer[\rl{t-pi-2}]
    {\Gamma\vdash \opit x S T \oft \set {(\text{max}\ n\ (m+1))}}
    {\Gamma\vdash S \oft \set n & \Gamma,x\oft S\vdash T \oft \set {(m+1)}}

   \vs

    \infer[\rl{t-fun}]
    {\Gamma\vdash\elam x E \oft \opit x S T}
    {\Gamma,x\oft S\vdash E \oft T}

    \quad

    \infer[\rl{t-app}]
    {\Gamma\vdash \eapp E {E'} \oft \sesub {E'} x T}
    {\Gamma\vdash E \oft \opit x S T & \Gamma\vdash E' \oft S}

   \vs

    \infer[\rl{t-conv}]
    {\Gamma\vdash E \oft T}
    {\Gamma\vdash E \oft S & \Gamma\vdash S \equiv T}

    \vs

    \infer[\rl{t-box}^*]
    {\Gamma\vdash \ebox \Psi M \oft \etbox \Psi \alpha}
    {\Gamma\vdash \Psi \isctx & \Gamma;\Psi\vdash M \oft \alpha}

    \vs

    \infer[\rl{t-spec}^*]
    {\Gamma\vdash \etbox \Psi \alpha \oft \set 0}
    {\Gamma\vdash \Psi \isctx & \Gamma;\Psi\vdash \alpha \oft \lfkind}
  \end{array}
\end{displaymath}
  \caption{Typing for Computations}
  \label{fig:ocomptyp}
\end{figure}

Type checking depends on these mutually dependent typing judgments:
\begin{itemize}
\item $\boxed{\Gamma\vdash E \oft T}$ : $E$ is of type $T$ in context $\Gamma$.
\item $\boxed{\Gamma\vdash \Psi \isctx}$ : $\Psi$ is a valid specification context in ctx. $\Gamma$.
\item $\boxed{\Gamma;\Psi\vdash \alpha \iskind}$ : $\alpha$ is a syntactic type in ctx. $\Gamma$ and spec. ctx. $\Psi$.
\item $\boxed{\Gamma;\Psi\vdash \alpha \oft K}$ : $\alpha$ is a syntactic type in ctx. $\Gamma$ and spec. ctx. $\Psi$.
\item $\boxed{\Gamma;\Psi\vdash M \oft \alpha}$ : $M$ is of type $\alpha$ in ctx. $\Gamma$ and spec. ctx. $\Psi$.
\item $\boxed{\Gamma;\Psi\vdash \sigma \oft \Phi}$ : $\sigma$ transports types from spec. context $\Phi$ to ctx. $\Psi$.
\end{itemize}
and these definitional equality judgments:
\begin{itemize}
\item $\boxed{\Gamma\vdash E \equiv E' \oft T}$ : $E$ is equal to $E'$
  at type $T$ in context $\Gamma$.

\item $\boxed{\Gamma;\Psi\vdash M \equiv N \oft \alpha}$ : $M$ is
  equal to $N$ at type $\alpha$ in contexts $\Gamma$ and $\Psi$.

\item $\boxed{\Gamma;\Psi\vdash K\equiv K'}$ : $K$ is equal to $K'$ in
  contexts $\Gamma$ and $\Psi$.

\item $\boxed{\Gamma;\Psi\vdash \alpha\equiv\beta \oft K}$ : $\alpha$
  is equal to $\beta$ at kind $K$ in contexts $\Gamma$ and
  $\Psi$.

\item $\boxed{\Gamma\vdash \Psi \equiv \Phi \oft \ctx}$ : $\Psi$ is
  equal to $\Phi$ in context $\Gamma$.

\item $\boxed{\Gamma\vdash \hat\Psi \equiv \hat\Phi \oft \ctx}$ :
    Erased context $\hat\Psi$ is equal to $\hat\Phi$ in context $\Gamma$.

\end{itemize}

The typing for computations appears in Figure~\ref{fig:ocomptyp} and
the typing for specifications appears in Figures~\ref{fig:ospectypI}
and~\ref{fig:ospectypII}. While rules for computational equality
appear in Figure~\ref{fig:oeqrulescomp} and for specifications in
Figures~\ref{fig:oeqrulesspec}, \ref{fig:oeqtypekind}, and
\ref{fig:oeqctx}. Rules for equivalences and congruence rules are
omitted as they are as expected.

Most rules are as expected, however the interesting ones appear with a
star in their names. These rules, reproduced below deal with the
embedding of specifications in computations and vice versa.
\begin{displaymath}
  \begin{array}{c}
    \infer[\rl{t-box}^*]
    {\Gamma\vdash \ebox \Psi M \oft \etbox \Psi \alpha}
    {\Gamma\vdash \Psi \isctx & \Gamma;\Psi\vdash M \oft \alpha}

    \vs

    \infer[\rl{s-unbox}^*]
    {\Gamma;\Psi\vdash \aunbox E \sigma \oft \etsub \sigma \alpha}
    {\Gamma;\Psi\vdash \sigma \oft \Phi & \Gamma\vdash E \oft \etbox \Phi \alpha}
  \end{array}
\end{displaymath}

The rule \rl{t-box} embeds an object together with its context in a
computation (of contextual type). And the \rl{s-unbox} rule allows
specifications to refer to computations, in particular this is
essential to represent incomplete terms where the role of
meta-variables is taken by computational variables inside an unbox
construction.

\begin{figure}
  \centering
\begin{displaymath}
  \begin{array}{c}
    \multicolumn{1}{l}
    {\boxed{\Gamma\vdash \Psi \isctx}\mathrel{:}\mbox{$\Psi$ is a valid specification context in ctx. $\Gamma$.}}\vs

    \infer[\rl{c-empty}]
    {\Gamma\vdash \aectx \isctx}
    {}

    \quad

    \infer[\rl{c-hyp}]
    {\Gamma\vdash \asnoc \Psi x \alpha \isctx}
    {\Gamma\vdash \Psi \isctx & \Gamma;\Psi\vdash \alpha\oft \lfkind}

    \vs

    \infer[\rl{c-ctx-var}]
    {\Gamma \vdash \acunbox x \isctx}
    {x\oft \ctx  \in \Gamma}

    \vs

    \multicolumn{1}{l}
    {\boxed{\Gamma\vdash \hat\Psi \isectx}\mathrel{:}\mbox{$\hat\Psi$ is a valid erased context in ctx. $\Gamma$.}}\vs

    \infer[\rl{cp-empty}]
    {\Gamma\vdash \aectx \isectx}
    {}

    \quad

    \infer[\rl{cp-hyp}]
    {\Gamma\vdash \psnoc \Psi x \isectx}
    {\Gamma\vdash \Psi \isectx}

    \vs

    \infer[\rl{cp-ctx-var}]
    {\Gamma \vdash \acunbox x \isectx}
    {x \oft \ctx \in \Gamma}

    \vs

    \multicolumn{1}{l}
    {\boxed{\Gamma;\Psi\vdash K \iskind}\mathrel{:}\mbox{$\alpha$ is a syntactic kind in ctx. $\Gamma$ and spec. ctx. $\Psi$.}}\vs

   \infer[\rl{s-kind}]
    {\Gamma;\Psi \vdash \lfkind \iskind }
    {}

    \vs

    \infer[\rl{s-pi-k}]
    {\Gamma;\Psi \vdash \opits x \alpha K \iskind}
    {\Gamma;\Psi \vdash \alpha \oft \lfkind & \Gamma;\Psi,x \oft \alpha\vdash K \iskind}

    \vs

    \multicolumn{1}{l}
    {\boxed{\Gamma;\Psi\vdash \alpha \oft K}\mathrel{:}\mbox{$\alpha$ is a syntactic type in ctx. $\Gamma$ and spec. ctx. $\Psi$.}}\vs

    \infer[\rl{s-base}]
    {\Gamma;\Psi\vdash \const a \oft K}
    {\const a \oft K \in \Sigma}

    \quad

    \infer[\rl{s-fam}]
    {\Gamma;\Psi\vdash \opits x \alpha \beta \oft \lfkind}
    {\Gamma;\Psi\vdash \alpha \oft \lfkind & \Gamma; \asnoc \Psi x \alpha\vdash \beta \oft \lfkind}

    \vs

    \infer[\rl{s-ins}]
    {\Gamma;\Psi\vdash \app \alpha M \oft \esub M x K}
    {\Gamma;\Psi\vdash \alpha\oft \opits x \beta K & \Gamma;\Psi\vdash M \oft \beta}

    \vs

    \infer[\rl{s-eqk}]
    {\Gamma;\Psi\vdash \alpha\oft K}
    {\Gamma;\Psi\vdash \alpha\oft K' & \Gamma;\Psi\vdash K \equiv K'}
  \end{array}
\end{displaymath}
  \caption{Typing Rules for Specifications (I)}
  \label{fig:ospectypI}
\end{figure}

\begin{figure}
  \centering
\begin{displaymath}
  \begin{array}{c}
    \multicolumn{1}{l}
    {\boxed{\Gamma;\Psi\vdash M \oft \alpha}\mathrel{:}\mbox{$M$ is of type $\alpha$ in ctx. $\Gamma$ and spec. ctx. $\Psi$.}}\vs

    \infer[\rl{s-lam}]
    {\Gamma;\Psi\vdash \alam x M \oft \opits x \alpha \beta}
    {\Gamma;\asnoc \Psi x \alpha \vdash M \oft \beta}

    \vs

    \infer[\rl{s-app}]
    {\Gamma;\Psi\vdash \aapp M N \oft \esub N x \beta}
    {\Gamma;\Psi\vdash M \oft \opits x \alpha \beta & \Gamma;\Psi\vdash N \oft \alpha}

    \vs

    \infer[\rl{s-var}]
    {\Gamma;\Psi\vdash \avar x\oft \alpha}
    {\avar x\oft \alpha \in \Psi}

    \quad

    \infer[\rl{s-con}]
    {\Gamma;\Psi\vdash \sconst c \oft \alpha}
    {\sconst c \oft \alpha \in \Sigma}

    \vs

    \infer[\rl{s-unbox}^*]
    {\Gamma;\Psi\vdash \aunbox E \sigma \oft \etsub \sigma \alpha}
    {\Gamma;\Psi\vdash \sigma \oft \Phi & \Gamma\vdash E \oft \etbox \Phi \alpha}

    \vs

    \infer[\rl{s-eq}]
    {\Gamma;\Psi\vdash M \oft \alpha}
    {\Gamma;\Psi\vdash M\oft \alpha' & \Gamma;\Psi\vdash \alpha \equiv \alpha'\oft \lfkind}

    \vs

    \multicolumn{1}{l}
    {\boxed{\Gamma;\Psi\vdash \sigma \oft \Phi}\mathrel{:}\mbox{$\sigma$ transports types from spec. context $\Phi$ to ctx. $\Psi$.}}\vs

    \infer[\rl{s-empty}]
    {\Gamma;\Psi\vdash \aempty \oft \aectx}
    {}

    \quad

    \infer[\rl{s-id}]
    {\Gamma;\Psi,\Psi'\vdash \naid n \oft \Psi}
    {|\Psi'| =  n}

    \vs

    \infer[\rl{s-ex}]
    {\Gamma;\Psi\vdash \adot \sigma x M \oft \asnoc \Phi x \alpha}
    {\Gamma;\Psi\vdash \sigma \oft\Phi & \Gamma;\Psi\vdash M \oft [\sigma] \alpha}
  \end{array}
\end{displaymath}
  \caption{Typing Rules for Specifications (II)}
  \label{fig:ospectypII}
\end{figure}

\subsection {Equality}

Due to the presence of fully dependent types Orca's computational
language does not have a distinction between type checking and
evaluation (as evaluation may happen at type check time). As usual,
the conversion rule (rule \rl{t-conv} in Figure~\ref{fig:ocomptyp}) is
where there may be some computation to compare two types, that may
not be yet in normal form. Equality (and computation), as typing, is
shown in several judgments, one for the computational/reasoning
language and one for each syntactic category of the specification
language.

Again equality is defined in Figure~\ref{fig:oeqrulescomp} for
computations and for specifications in Figures~\ref{fig:oeqrulesspec},
\ref{fig:oeqtypekind}, and \ref{fig:oeqctx}.

\begin{figure}
  \centering
\begin{displaymath}
  \begin{array}{c}
    \multicolumn{1}{l}
    {\boxed{\Gamma\vdash E \equiv E' \oft T}\mathrel{:}\mbox{$E$ is equal to $E'$ at type $T$ in context $\Gamma$}.}\vs

    \infer[\rl{e-beta}]
    {\Gamma\vdash \eapp {(\elam x E)} E' \equiv \essub {x \easub {E'}} E \oft \essub {x \easub {E'}} T}
    {\Gamma,x\oft S \vdash E\oft T
   & \Gamma\vdash E' \oft S }

     \vs

     \infer[\rl{e-box-unbox}]
     {\Gamma\vdash \ebox \Psi {\aunbox E \aid} \equiv E \oft \etbox \Psi \alpha}
     {\Gamma\vdash E \oft \etbox \Psi\alpha}

     \vs

     \infer[\rl{e-box-1}]
     {\Gamma\vdash \ebox\Psi M \equiv \ebox\Phi{M'} \oft \etbox \Psi \alpha}
     {\Gamma\vdash\hat\Psi\equiv\hat\Phi\oft\ctx & \Gamma;\Psi\vdash M \equiv M' \oft \alpha}

     \vs

    \infer[\rl{e-box-2}]
    {\Gamma\vdash \etbox \Psi \alpha \equiv \etbox \Phi \beta \oft \set 0}
    {\Gamma\vdash \Psi \equiv \Phi \oft \ctx
   & \Gamma;\Psi\vdash \alpha \equiv \beta \oft \lfkind}




  \end{array}
\end{displaymath}
  \caption{Equality Rules for Computations}
  \label{fig:oeqrulescomp}
\end{figure}

\begin{figure}
  \centering
\begin{displaymath}
  \begin{array}{c}
    \multicolumn{1}{l}
    {\boxed{\Gamma;\Psi\vdash M \equiv N \oft \alpha}\mathrel{:}
        \mbox{$M$ is equal to $N$ at type $\alpha$ in contexts $\Gamma$ and $\Psi$}.}\vs

     \infer[\rl{es-beta}]
     {\Gamma;\Psi\vdash \aapp{(\alam x M)} N \equiv (\{\hat x \easub N\}M)  \oft (\{\hat x\easub N\}\beta)}
     {\Gamma;\asnoc \Psi x \alpha \vdash M \oft \beta
     & \Gamma;\Psi \vdash N \oft \alpha}

     \vs

    \infer[\rl{es-eta}]
    {\Gamma;\Psi\vdash \alam x {\app M x}  \equiv M \oft \opits x \alpha \beta}
    {\Gamma;\Psi\vdash M\oft \opits x \alpha \beta}

    \vs

     \infer[\rl{es-unbox}]
     {\Gamma ; \Psi \vdash \aunbox{\ebox \Phi M}{\sigma} \equiv \etsub \sigma M \oft \etsub \sigma \alpha}
     {\Gamma ; \Psi \vdash \sigma : \Phi & \Gamma ; \Phi \vdash M : \alpha}
  \end{array}
\end{displaymath}
  \caption{Equality Rules for Specification Terms}
  \label{fig:oeqrulesspec}
\end{figure}

\begin{figure}
  \centering
\begin{displaymath}
  \begin{array}{c}
    \multicolumn{1}{l}
    {\boxed{\Gamma;\Psi\vdash K\equiv K'}\mathrel{:}
        \mbox{$K$ is equal to $K'$ in contexts $\Gamma$ and $\Psi$}.}\vs

    \infer[\rl{ek-base}]
    {\Gamma;\Psi\vdash \lfkind\equiv \lfkind}
    {}

    \quad

    \infer[\rl{}]
    {\Gamma;\Psi\vdash \opits x \alpha K \equiv \opits x \beta {K'}}
    { \Gamma;\Psi\vdash \alpha\equiv\beta\oft\lfkind
  & \Gamma;\asnoc\Psi x \alpha\vdash K \equiv K'}

    \vs

    \multicolumn{1}{l}
    {\boxed{\Gamma;\Psi\vdash \alpha\equiv\beta \oft K}\mathrel{:}
        \mbox{$\alpha$ is equal to $\beta$ at kind $K$ in contexts $\Gamma$ and $\Psi$}.}\vs

    \infer[\rl{ef-base}]
    {\Gamma;\Psi\vdash \aconst a \equiv \aconst a \oft K}
    {}

    \vs

    \infer[\rl{ef-pi}]
    {\Gamma;\Psi\vdash \opits x \alpha \beta \equiv \opits x {\alpha'} {\beta'} \oft \lfkind}
    {\Gamma;\Psi\vdash \alpha\equiv\alpha'\oft\lfkind
  & \Gamma;\asnoc\Psi x \alpha\vdash \beta\equiv\beta'\oft\lfkind}

    \vs

    \infer[\rl{ef-ins}]
    {\Gamma;\Psi\vdash \app \alpha M \equiv \app \beta N\oft [x\easub M]K}
    {\Gamma;\Psi\vdash \alpha \equiv \beta \oft \opits x \delta K
  & \Gamma;\Psi\vdash M\equiv N\oft \delta}

  \end{array}
\end{displaymath}
  \caption{Equality Rules for Specification Types and Kinds}
  \label{fig:oeqtypekind}
\end{figure}

\begin{figure}
  \centering
\begin{displaymath}
  \begin{array}{c}
    \multicolumn{1}{l}
    {\boxed{\Gamma\vdash \Psi \equiv \Phi \oft \ctx}\mathrel{:}
    \mbox{$\Psi$ is equal to $\Phi$ in context $\Gamma$.}}\vs

    \infer[\rl{ec-empty}]
    {\Gamma\vdash \aectx \equiv \aectx \oft \ctx}
    {}

    \quad

    \infer[\rl{ec-ctx}]
    {\Gamma\vdash \acunbox x \equiv \acunbox x \oft \ctx}
    {}

    \vs

    \infer[\rl{ec-cons}]
    {\Gamma\vdash \asnoc\Psi x \alpha \equiv \asnoc\Phi x \beta \oft\ctx}
    {\Gamma\vdash \Psi\equiv\Phi\oft\ctx
  & \Gamma;\Psi\vdash \alpha\equiv\beta \oft \lfkind}

    \vs

    \multicolumn{1}{l}
    {\boxed{\Gamma\vdash \hat\Psi \equiv \hat\Phi \oft \ctx}\mathrel{:}
    \mbox{Erased context $\hat\Psi$ is equal to $\hat\Phi$ in context $\Gamma$.}}\vs

    \infer[\rl{ee-empty}]
    {\Gamma\vdash \aectx \equiv \aectx \oft \ctx}
    {}

    \quad

    \infer[\rl{ee-ctx}]
    {\Gamma\vdash \acunbox x \equiv \acunbox x \oft \ctx}
    {}

    \vs

    \infer[\rl{ee-cons}]
    {\Gamma\vdash \psnoc\Psi x \equiv \psnoc\Phi x \oft\ctx}
    {\Gamma\vdash \Psi\equiv\Phi\oft\ctx}
  \end{array}
\end{displaymath}
  \caption{Equality Rules for Contexts}
  \label{fig:oeqctx}
\end{figure}

\section{Definitions and Pattern Matching}

In Section~\ref{sec:coreorca}, we extended a type theory with
contextual types and their introduction forms (i.e.: contextual
terms). However, to make use of these definitions it is necessary to
have an elimination form. To that effect in these section we extend
the language with top-level definitions by pattern matching. These
definitions take apart values by using simultaneous pattern matching
in the style of Agda~\citep{Norell:phd07}.
\begin{displaymath}
  \begin{array}{rlcll}
    \multicolumn{5}{c}{\boxed{\text{Patterns for Definitions}}}\vs

    \text{Pattern} & P & \bnfas & \pconst c P & \text{A fully applied constant}\\
    & & \bnfalt & x & \text{A variable} \\
    & & \bnfalt & \pinac P & \text{An inaccessible pattern}\\
    & & \bnfalt & \ebox \Psi {\mathcal{P}} & \text{\color{red} A contextual pattern}\vs

    \text{Spec. Pattern} & \mathcal{P} & \bnfas & \spconst c {\mathcal{P}} & \text{A fully applied constant}\\
    & & \bnfalt & \alam x {\mathcal{P}} & \text{An abstraction}\\
    & & \bnfalt & \avar x & \text{A bound variable}\\
    & & \bnfalt & \pinac {\mathcal{P}}& \text{An inaccessible pattern}\\
    & & \bnfalt & \aunbox x {\bar\sigma} & \text{\color{red} An unboxed meta-variable}\vs

    \text{Pattern Subst.} & \bar\sigma & \bnfas & \aempty & \text{An empty substitution} \\
    & & \bnfalt & \naid n & \text{An identity substitution,}\\
    & & & & \text{with weakening} \\
  \end{array}
\end{displaymath}

The syntax for patterns is straightforward, we highlight in red where
computations embed specifications and vice-versa. Notice, how in
patterns only variables can be unboxed as it is not clear what it
would mean to pattern match against a computation and not a value.
Furthermore, the substitution is weaker as only empty and identity
with weakening substitutions are allowed. Because we use fully applied
constructors for patterns we define the patterns $\pconst c P$ and
$\spconst c {\mathcal P}$ that are applied to a spine of parameters.
The type-checking of pattern spines uses its own typing judgment where
we check a spine against a type that produces as a result the final
type of the constant applied to the pattern. Figure~\ref{fig:opattyp}
shows the typing of patterns.

\begin{figure}
  \centering
  \begin{displaymath}
    \begin{array}{c}
      \multicolumn{1}{l}
      {\boxed{\Gamma\vdash P \oft T}\mathrel{:}\mbox{$P$ is of type $T$ in context $\Gamma$}.}\vs

      \infer[\rl{tp-con}]
      {\Gamma\vdash \pconst c P \oft T}
      {\aconst c \oft S \in \Sigma
      & \Gamma\vdash \vec P \oft S \spineret T}

      \quad

      \infer[\rl{tp-var}]
      {\Gamma\vdash x \oft T}
      {x\oft T \in \Gamma}

      \vs

      \infer[\rl{tp-inac}]
      {\Gamma\vdash \pinac P \oft T}
      {\Gamma\vdash P \oft T}

      \quad

      \infer[\rl{tp-box}]
      {\Gamma\vdash\ebox \Psi{\mathcal{P}} \oft \etbox \Psi \alpha}
      {\Gamma;\Psi\vdash \mathcal P \oft \alpha}

      \vs

      \multicolumn{1}{l}
      {\boxed{\Gamma\vdash \overrightarrow{P} \oft S \spineret T}\mathrel{:}
        \mbox{Spine $\overrightarrow P$ is of type $S$ and produces type $T$ in context $\Gamma$}.}\vs

      \infer[\rl{tp-spine}]
      {\Gamma\vdash \app P {\vec P} \oft \opit x S {S'} \spineret T}
      {\Gamma\vdash P\oft S & \Gamma\vdash{\vec P} \oft \essub {x \easub P} S' \spineret T}

      \quad

      \infer[\rl{tp-nil}]
      {\Gamma\vdash \cdot \oft T \spineret T}
      {}

      \vs

      \multicolumn{1}{l}
      {\boxed{\Gamma;\Psi\vdash \mathcal P \oft \alpha}\mathrel{:}
        \mbox{$\mathcal P$ is of type $\alpha$ in contexts $\Gamma$ and $\Psi$}.}\vs

      \infer[\rl{sp-const}]
      {\Gamma;\Psi\vdash \spconst c {\mathcal P} \oft \etsub \sigma \beta}
      {\sconst c\oft \mopits x \alpha \beta \in \Sigma
      & \Gamma;\Psi\vdash \vec{\mathcal P} \oft \vec\alpha
      & \sigma = \overrightarrow{\afor x {\mathcal P_i}}}

      \vs

      \infer[\rl{sp-lam}]
      {\Gamma;\Psi\vdash \alam x {\mathcal P} \oft \opits x \alpha \beta}
      {\Gamma;\asnoc \Psi x \alpha \vdash \mathcal P \oft \beta}

      \quad

      \infer[\rl{sp-var}]
      {\Gamma;\Psi\vdash \avar x \oft \alpha}
      {\avar x\oft \alpha \in \Psi}

      \vs

      \infer[\rl{sp-inac}]
      {\Gamma;\Psi\vdash \pinac{\mathcal P} \oft \alpha}
      {\Gamma;\Psi\vdash \mathcal P \oft \alpha}

     \quad

     \infer[\rl{sp-unbox}]
     {\Gamma;\Psi\vdash \aunbox x \sigma \oft \etsub \sigma \alpha}
     {\Gamma;\Phi\vdash x \oft \alpha
     & \Gamma;\Psi\vdash \sigma \oft \Phi}

     \vs

      \multicolumn{1}{l}
      {\boxed{\Gamma;\Psi\vdash \overrightarrow{\mathcal P} \oft \alpha \spineret \beta}\mathrel{:}
        \mbox{Spine $\overrightarrow {\mathcal P}$ is of type $\alpha$ and produces type $\beta$ in $\Gamma$ and $\Psi$}.}\vs

      \infer[\rl{sp-spine}]
      {\Gamma;\Psi\vdash \app {\mathcal P} {\vec {\mathcal P}} \oft \opits x \alpha {\alpha'} \spineret \beta}
      {\Gamma;\Psi\vdash \mathcal P\oft \alpha & \Gamma\vdash{\vec {\mathcal P}} \oft \esub {\mathcal P} x {\alpha'} \spineret \beta}

      \quad

      \infer[\rl{sp-nil}]
      {\Gamma;\Psi\vdash \cdot \oft \beta \spineret \beta}
      {}
    \end{array}
  \end{displaymath}
  \caption{Typing Rules for Patterns}
  \label{fig:opattyp}
\end{figure}

Having defined patterns we move on to define top-level pattern
matching definitions:
\begin{displaymath}
  \begin{array}{ll}
    \text{Type annotation} & \aconst f \oft T \\
    \text{Equations} & \Gamma_1\mathrel{.} \eapp {\aconst f} {\vec P_1} = E_1\\
    & \multicolumn{1}{c}{\vdots}\\
    & \Gamma_n\mathrel{.} \eapp {\aconst f} {\vec P_n} = E_n\\
  \end{array}
\end{displaymath}

A definition by pattern matching introduces a new constant and its
type followed by its pattern matching equations. Each equation
contains a context that binds all the variables used in the pattern, a
set of pattern and the right hand side of the equation that contains
the term that performs the computation when this equation is matched.
For each equation, patterns use every variable from the context
exactly once. Because non-linear patterns are unavoidable in dependent
pattern matching, all non-linear occurrences of a variable must be in
inside inaccessible patterns (where the value of the pattern is forced
by the type and thus is not really a discrimination pattern). This is
not enforced in the typing rules and has to be checked by the
implementation. Inaccessible patterns provide a nice way of
implementing the computation of pattern matching by using
discrimination trees, this is mostly a practical concern at this
point. The constant defined in a function by pattern matching may be
used inside the right hand side of any equations. As functions may
appear in types and be executed during type-checking, they should be
total. In this chapter we do not discuss how to check termination and
coverage to establish totality, but rather leave that as future work.
However the user is required to provide functions where pattern
matching is covering and that recursive calls are
terminating. 

A valid pattern matching definition extends the signature and the
reduction behaviour of the system where each of the clauses becomes a
new reduction rule.

The typing of equations is done using the typing judgment in
Figure~\ref{fig:patmatchtyp} that depends on the typing of patterns
presented in Figure~\ref{fig:opattyp}.

\begin{figure}
  \centering
  \begin{displaymath}
    \begin{array}{c}
      \multicolumn{1}{l}
      {\boxed{\vdash \Gamma_i\mathrel{.}\eapp{\aconst f}{\vec{P_i}} = E_i\isvalid}}:
      \text{The equation is valid}\vs

      \infer[\rl{t-pat}]
      {\vdash \Gamma_i\mathrel{.}\eapp{\aconst f}{\vec P} = E_1 \isvalid}
      {\Gamma_i\vdash \eapp{\aconst f}{\vec P_1} \oft T
      & \Gamma_i\vdash E_i \oft T'
      & \Gamma_i\vdash T \equiv T'}
    \end{array}
  \end{displaymath}
  \caption{Typing Equations}
  \label{fig:patmatchtyp}
\end{figure}


Definitions by pattern matching extend the signature $\Sigma$ with a
new constant (i.e.: the name of the function) and possibly multiple
pattern matching equations that become new computation rules. We
extend the syntax of the system as follows to accommodate for
equations:
\begin{displaymath}
  \begin{array}{rlcll}
    \multicolumn{5}{c}{\boxed{\text{Extended Signature and Simultaneous Substitutions}}}\vs

    \text{Signature} & \Sigma & \bnfas & \dots \\
    & & \bnfalt & \Sigma, \Gamma_i\mathrel{.}\eapp{\aconst f}{\vec P_i} = E_i & \text{Pattern matching equation} \vs

    \text{Substitution} & \theta & \bnfas & \cdot & \text{An empty substitution}\\
    & & \bnfalt & \edot \theta  E x & \text{A substitution extension}
  \end{array}
\end{displaymath}

Definitions by pattern matching compute by extending computation using
the equations in the definition in the following form:
\begin{displaymath}
  \infer[\rl{e-equ}]
  {\Gamma\vdash \eapp{\aconst f}{\vec E} \equiv \essub \theta {E_i} \oft \essub \theta T}
  {\aconst f \oft \mopit x S T \in \Sigma
  & \Gamma\vdash\vec E\oft \vec S
  & \Gamma_1\mathrel{.}\eapp{\aconst f}{\vec P_i} = E_i \in \Sigma
  & \Gamma\vdash \vec P_i \doteq \vec E / \theta}
\end{displaymath}

This requires the computation of higher-order matching (i.e.:
operation $\Gamma\vdash P \doteq E/\theta$) between patterns and the
terms passed as parameter to the function. A sensible implementation,
(and it is indeed the Orca implementation does), tries the equations
in the order they where defined. This is done so as to respect the
user's specified order in the case of overlapping patterns.


%% file: orca/prototype.tex
\section{The Prototype Implementation} \label{sec:orcaprototype}

As seen in the example from Section~\ref{sec:tranbool}, the ideas from
Section~\ref{sec:coreorca} have been implemented in a prototype
language called Orca (following the marine mammal tradition of
languages like Beluga and Delphin~\citep{Poswolsky:DelphinDesc08}).
The OCaml implementation is available on
\href{https://github.com/orca-lang/orca}{Github}\footnote{\url{https://github.com/orca-lang/orca}}.

\begin{figure}
  \centering
\tikzstyle{block} = [ rectangle
                    , draw
                    , fill=blue!20
                    , text width=5em
                    , text centered
                    , rounded corners
                    , minimum height=4em
                    , auto
                    , node distance=3cm
                    , text width=2.5cm
                    , font=\sffamily
                    ]

\tikzstyle{line} = [ draw
                   , -latex
                   ]

\tikzstyle{dottedline} = [ draw
                         , -latex
                         , dashdotted
                         , color=black!75
                         ]

\begin{tikzpicture}
  \node[block] (parser) {Parser};
  \node[block, right of=parser] (form) {Well-formed Check};
  \node[block, right of=form] (recon)  {Type Reconstruction};
  \node[block, right of=recon] (ast) {Final AST};

  \node[block, below of=recon, fill=blue!10] (impl) {Implicit Parameters};
  \node[block, below of=ast, fill=blue!10] (check) {Small Type Checker};

  \draw[line] (parser) -- (form);
  \draw[line] (form) -- (recon);
  \draw[line] (recon) -- (ast);

  \draw[dottedline] (recon) -- (impl);
  \draw[dottedline] (impl) -- (ast);

  \draw[dottedline] (ast) -- (check);

  \draw [decorate,decoration={brace,amplitude=5pt,raise=2mm}]
  (parser.85) -- (ast.95) node [black,midway,yshift=6mm]
  {\small Implemented};

  \draw [decorate,decoration={brace,amplitude=5pt,raise=2mm, mirror}]
  (impl.275)--(check.265)  node  [black,midway,yshift=-0.3cm,text width=3cm, align=center,anchor=north] (EHF)
  {\small Planned};
\end{tikzpicture}
  \caption{The Orca Pipeline}
  \label{fig:orcapipe}
\end{figure}
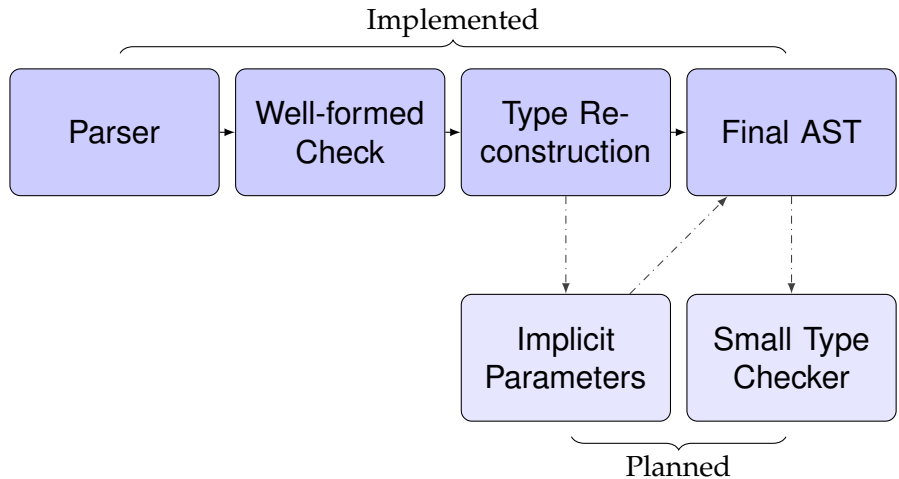

Figure~\ref{fig:orcapipe} shows the pipeline for the prototype. The
first line shows the existing pipeline, and the phases in the second
line show planned features that have not been implemented yet. The
first two phases form the front end of the language. First, the parser
produces the first version of the AST (Abstract Syntax Tree).
Subsequently, the AST generated by the parser  is checked for some well
formed properties:
\begin{itemize}
\item It disambiguates between constructors, variables, and bound variables.
\item It checks the scope of variables and definitions.
\item It transforms the bound variables of the specification framework
  to \debruijn indices.
\end{itemize}

The idea behind this design is to have a very permissive
\emph{parser}, and a second pass that removes ambiguities and can
report accurate error messages\footnote{Even if the current
  implementation does not exploit this to have great error messages}.
The language does not enforce restrictions on naming conventions for
user defined identifiers. Constructors, variable names and definitions
can be composed of capitals, symbols and characters from any writing
system supported by the Unicode standard (with some restriction for
keywords and the symbols used in the grammar itself). The
\emph{well-formed check} phase checks that variables and definitions
respect their scopes and prepares an intermediate AST that is
susceptible to type reconstruction and checking. As a side note, this
phase is the perfect place to add mix-fix
parsing~\citep{Danielsson:2011} in the style of the Agda language.
As a remark, the lack of enforced conventions for names is done for
flexibility, but for large developments it would be important to
develop such a convention. We expect an ad-hoc convention to emerge
once larger Orca examples are implemented.

The third pass, type reconstruction, generates the final version of
the AST with all the type checking information included. This pass
performs two main tasks:
\begin{itemize}
\item It disambiguates the computation and
specification terms by inserting boxes and un-boxes appropriately.
\item It simultaneously type checks user input.
\end{itemize}

\begin{figure}
  \centering
\begin{lstlisting}[basicstyle=\scriptsize\ttfamily]
def tran : (g h : ctx) -> rel <<[>>g<<]>> <<[>>h<<]>> ->
           (st : |- s-tp) -> (g |- s-exp <<'>> st) ->
           (h |- t-exp <<' [>>(tran-tp st)<<][^]>>)
where
| g h r t <<[g :> >>(app s[^] .t<<[^]>> m n)<<]>> =>
    <<[h :> >>tapp <<' [>>(tran-tp s)<<][^] ' [>>(tran-tp t)<<][^]>>
      <<' [>>(tran g h r <<[g :> >>(arr <<' [>>s<<][^] ' [>>t<<][^]>>)<<]>> m)<<]>>
      <<' [>>(tran g h r s n)<<]]>>

| g h r <<[:> >>(arr <<'>> s[^] <<'>> t[^])<<]>>
        <<[g :>>>(lam <<'>> .s<<[^] '>> .t<<[^] '>> m)<<]>> =>
    <<[h :> >>tlam <<' [>>(tran-tp s)<<][^] ' [>>(tran-tp t)<<][^] '>>
          (\x. <<[>>tran (g, x:s-exp s[^])
                     (h, x:t-exp <<[>>tran-tp s<<][^]>>)
                     (cons g h s r)
                     t
                     (<<[g, x:s-exp s[^] :> >>m<<[^1] ' >> x<<]>>)
                <<][^1;x]>>)<<]>>

| g h r <<[:> >> bool<<] [g :> >>tt<<]>> =>
  <<[h :> >>tinl <<'>> tunit <<'>> tunit <<'>> tone <<]>>

| g h r <<[:> >> bool<<] [g :> >>ff<<]>> =>
  <<[h :> >>tinr <<'>> tunit <<'>> tunit <<'>> tone <<]>>

| g h r t <<[g :> >>(if <<'>> .t<<[^]>> <<'>> b <<'>> e1 <<'>> e2)<<]>> =>
    <<[h :> >>tcase <<'>> tunit <<'>> tunit <<' [>>(tran-tp t)<<][^] '>>
          <<[>>(tran g h r bool b)<<] '>>
          (\x. <<[>>tran g h r t e1<<][^1]>>) <<'>>
          (\x. <<[>>tran g h r t e2<<][^1]>>)<<]>>

| .<<[g, x:s-exp ' t[^]]>> .<<[h, x:t-exp ' [tran-tp t][^]]>>
  (cons g h .t r) t (g, x: s-exp t<<[^]>> :> x) =>
      (h, x: t-exp <<[>>(tran-tp t)<<][^]>> :> x)

| .<<[g, x:s-exp ' t[^]]>> .<<[h, x:t-exp ' [tran-tp t][^]]>>
  (cons g h t r) s (v[^1]) =>
    <<[>>tran g h r s v<<][^1]>>
\end{lstlisting}
  \caption{Boolean Translation After Box Inference}
  \label{fig:tranelab}
\end{figure}

Following the spirit from Chapter~\ref{chp:recon}, the boxes are
reconstructed by a type directed algorithm that by construction
produces well typed expressions as a result. The AST produced by this
phase, is complete, and can be type-checked without having to infer or
complete any type information. A small type-checker would make the
trusted computing base of the system small. Similarly, having implicit
parameter reconstruction would make writing programs easier to write.
Both aspects are planned additions to the Orca prototype.


%% file: orca/conclusion.tex
\section{Related Work}

The idea of embedding a specification framework in a computational
$\lambda$-calculus using a modality was first presented
in~\citep{Despeyroux97} where they embed a simply-typed version of LF
in a simply typed lambda calculus extended with a box modality based
on that from S4. As mentioned, in Orca the reasoning framework is the
logical framework LF~\citep{Harper93jacm} and the computational
framework is a dependently typed system~\citep{Martin-Loef79a}. This
presents a very first step towards answering the long standing
question of combining HOAS and dependent types.

Moreover, Orca can be seen as an extension of the Beluga
language \citep{Pientka:CADE15}. As such it uses the logical framework
LF to represent syntax and judgments using higher-order abstract
syntax (HOAS). However, it extends Beluga's first order reasoning
language to a dependently typed system where computations can be
embedded in specifications. The Twelf system~\citep{Pfenning99cade}
also implements LF where they implement computation using logic
programming instead of a functional framework. The Delphin
system~\citep{Schuermann:ESOP08} manipulates LF objects using
functional programs in a similar way to Beluga but without support for
inductive types or contextual types.

Other approaches that use a two-level system with specifications and
reasoning are the Abella prover~\citep{Baelde:AbellaTutorial} and
Hybrid framework~\citep{Felty12}, these two systems provide a
specification logic on top of a proof-theoretic reasoning logic.
Abella is a standalone proof assistant, while Hybrid can be
implemented on existing proof assistants, implementations for Coq and
Isabelle exist.

Traditionally, implementing HOAS in systems without a specification
framework is troublesome because of the positivity restriction for
inductive types. To overcome this problem, weaker forms of HOAS use an
abstract type to represent variables and obtain a strictly positive
type. Examples of this are the work of~\citet{Chlipala:ICFP08}
and~\citet{Despeyroux:1994:HAS:645708.664171}. Weaker forms of HOAS
are convenient but they do not get substitution from $\beta$-reduction
in the host language and make dealing with open terms less convenient.

Another option is to implement binders using de~Bruijn indices. A
convenient way is using well-scoped de~Bruijn
indices~\citep{Altenkirch:TLCA93} (and later extended to intrinsically
typed terms in ~\citep{Benton:JAR12}), where substitution is
implemented as an inductive data type and its structural properties
need to be proved manually. There are implementations of de~Bruijn
indices that provide and facilitate proving the required lemmas. One
example is Autosubst~\citep{Schaefer2015} that implements a decision
procedure to compare terms with explicit substitutions. The Nameless
Painless~\citep{Pouillard:2011} approach is based on de~Bruijn
indices, but it simplifies the arithmetic reasoning by hiding the
numbers in an abstract ``world'' representation. Finally,
GMeta~\citep{Lee2012} simplifies implementing first-order binding
representations (e.g.: de~Bruijn indices) using generic programming to
provide proofs for the lemmas about binders.

More recently, \citet{Allais:2017} propose a way of encoding syntax
(as opposed to syntax and judgments) in Agda using a generic type and
scope preserving semantics. This allows them to prove lemmas about the
semantics and then reuse them to implement renaming, substitution,
normalization by evaluation and CPS transformations.

Finally, a possible approach is to use nominal
logic~\citep{Pitts:2003} that uses an infinite set of atoms (abstract
names) to replace the position based approach. These ideas have been
implemented in Nominal~\citep{Urban:JAR08} package of
Isabelle~\citep{Nipkow-Paulson-Wenzel:2002}.

Finally, thanks to its specification framework, Orca is
well-positioned to implement meta-programming in dependently typed
systems where syntax can be reflected into LF objects where
computational functions can manipulate them. There is interest in
doing this in existing systems as a way of implementing tactics and
elaboration, for example the MTAC~\citep{ziliani:2015} system for Coq,
or elaboration reflection~\citep{Christiansen:2016} for Idris
implement meta-programming (The Agda proof assistant also implements
Idris-style reflection). What Orca offers to meta-programming is the
ability of processing structures with binders in a type safe and
convenient way, instead, for example, of the de~Bruijn representation
that the Idris reflection offers.

\section{Conclusion}

In this chapter we describe Orca, as an idea and a prototype that
implements a Martin-L\"{o}f style type theory with a specification
language based on the logical framework LF. This language can be
thought of as a dependently typed Beluga, and thus continues the theme
of cetacean inspired names (after Beluga and Delphin). The design of
Orca can also be seen as a next step extending the design of the
Babybel system from Chapter~\ref{chp:babybel} to total languages and
full LF and thus continues that line of work. However, this chapter
does not address the full spectrum of possibilities and questions
about Orca. This should not be seen as a limitation but as an
aspiration. Exploring the meta-theory of the system, its expressivity
and the fields where it will be applicable is all exciting future
work. Some of this work is already in progress, particularly the
meta-theory and improving and extending the prototype with new ideas.
However, the current state of affairs is promising, we show how Orca
can interleave computation and specifications, and how this saves some
lemmas that would be necessary in systems that need to represent
specification computation as relations. Moreover, the Orca
distribution contains several examples that hint to the power of
reasoning not only about specifications and inductive types like in
Beluga, but also about functions. Some of these examples are: a small
type preservation proof, the translation we used as an example here, a
proof that a particular computation preservers types (illustrated with
the simplest computation function, the copy function), a conversion
between two styles of operational semantics and several small programs
that we use as the beginnings of a test harness. In conclusion, Orca
represents a new bold first step in the reasoning about open objects
and offers a path to explore extensions to the ideas presented in this
chapter.


%% file: conclusion.tex
In this thesis, after introducing the problem of reasoning about open
terms using a HOAS representation, we show how to simplify writing
these programs using inference of omitted arguments, how to write
programs by integrating contextual types and HOAS with existing
programming languages, and finally we present the Orca system that
does the first step towards integrating a Martin-L\"of's style type
theory with contextual types and the logical framework LF, this allows
for the interleaving of specifications and proofs/computations.
Therefore, Orca, once the meta-theory is proven and the totality check
implemented, allows for proofs about meta-theoretic properties of
specifications together with proofs about computations which allows
the user to prove properties of their computations over
specifications. This provides for the seamless combination of programs
and proofs.

\section{Future Work}

\subsection{Implicit Parameter Reconstruction}

Implicit parameter reconstruction makes our lives easier, by allowing
us to concentrate in what is important. Reconstruction only fills in
parts that are forced by the surrounding program. As future work, one
would like to extend the kinds of information that can be inferred.
Furthermore, as the current algorithm works for indexed type systems,
it would need to be extended for full dependent types to support
reconstruction of Orca programs. Having a formally specified type
reconstruction for Orca will dramatically improve its ergonomics.
Another interesting approach would be to specify reconstruction as
form elaboration reflection~\citep{Christiansen:2016} that would allow
for reconstruction to be done in Orca code itself. Orca's
specification framework should make for a great target for this
reflection mechanism because the LF framework could be used to
represent the program being elaborated.

\subsection{Contextual Types and Programming Languages}

In the future, we plan to implement our approach also in other
languages. In particular, it would be natural to implement our
approach in Haskell. GHC Haskell offers the required syntax extension
mechanism and its more powerful type system offers interesting
possibilities. With regards to the implementation, the PPX syntax
extension mechanism is a bit limited (e.g.: the syntax of type
annotations cannot be extended), it would be interesting to either
extend PPX or to replace PPX with CamlP4 that would allow for a more
seamless syntax extension. And finally, while the current features of
Babybel allow for many interesting use cases, it would be interesting
to develop more substantial use cases to both validate the convenience
our approach and to possibly justify new features that a larger
program might require. For example, implementing the compiler for a
small language would neatly showcase the use of the syntactic
framework.

\subsection{Contextual Types and Type Theory}

This chapter represents an open path to future research, which is
another way to say that there are lots of open questions regarding
this subjects. So far, there is no meta-theory for the core calculus
we presented, as future work it would be important to show that
normalization is preserved when extending a dependently typed calculus
with the logical framework LF using contextual types in the way
described in this chapter\footnote{There is ongoing work on draft of
  meta-theory for this}. Also on the side of the theory, certain
needed features are missing to make it more expressive, like
substitution variables~\citep{Cave:LFMTP13}. Finally, inspired by the
Babybel theory where contexts are erased at run-time, in this calculus
contexts are run-time irrelevant. This means that properties about
contexts need to be established using inductive predicates (e.g.: the
\lstinline!rel! predicate in the example from
Section~\ref{sec:tranbool}). Adding Beluga style context schemas would
allow for more powerful use of contexts.

On the side of the Orca prototype, the immediate future work is
mechanically making sure that the functions are total, that is
implementing coverage and termination checking. Implementing coverage
would extend ideas from the Beluga system~\citep{Pientka:Cover10} and
for termination many choices are possible. The author would like to
explore the idea of sized types~\citep{Hughes:POPL96} and
\citep{Abel:phdthesis}. And from here the possibilities expand, as
Orca can be used as a vehicle for research in this area. Examples of
this are: adding generic judgments as in the work by
\citet{MillerTiu:TCL05} to have Abella~\citep{Gacek:IJCAR08} and
Delphin style introduction of fresh names, and exploring co-inductive
types with co-patterns.


%% file: reconstruction/extra.tex
\chapter{Proof of Soundness of Reconstruction}\label{sec:completesoundnessproof}

We begin with some lemmas that we will use to establish the main result:

\setcounter{lem}{0}
\begin{lem}[Implicit parameter instantiation] \label{lem:inst} Let's
  consider the judgment:
  $\Theta;\Delta;\Gamma\vdash E\oft T\recon E_1\oft T_1/\Theta_1$,
  where $\Theta_1$ is a weakening of $\Theta$.

  We want to prove that, if $\rho_g$ is a grounding instantiation such as $\cdot\vdash\rho_g\oft\Theta_1$
  where we split $\rho_g = \rho_g',\rho_g''$ and $\cdot\vdash\rho_g'\oft\Theta$ and
  $\cdot;\ahsub{\rho_g'}\Delta;\ahsub{\rho_g'}\Gamma\vdash\ahsub{\rho_g'}E\oft\ahsub{\rho_g'}T$ then
  $\cdot;\ahsub{\rho_g}\Delta;\ahsub{\rho_g}\Gamma\vdash\ahsub{\rho_g}E_1\oft\ahsub{\rho_g}T_1$.
\end{lem}
\begin{proof}
  The proof follows by induction on the rules of the judgment where
  the base case for $\elimpd$ is trivial and the inductive step for
  $\elimp$ has also a very direct proof.
\end{proof}

\begin{lem}[Pattern elaboration] \label{lem:patelab}~
  \begin{enumerate}
  \item If $\Theta;\Delta\vdash pat\recon\Pi\Delta_1;\Gamma_1 . Pat\oft T/\Theta_1;\rho_1$ and
     $\rho_r$ is a further refinement substitution, such as $\Theta_2\vdash\rho_r\oft\Theta_1$ and
    $\ep$ is a ground lifting substitution, such as $\Delta_i\vdash\ep\oft\Theta_1$ then\\
    $\Delta_i,\ahsub{\ep}\ahsub{\rho_r}\Delta_1;\ahsub{\ep}\ahsub{\rho_r}\Gamma_1\vdash
    \ahsub{\ep}\ahsub{\rho_r}Pat\checks\ahsub{\ep}\ahsub{\rho_r}T$.

  \item If $\Theta;\Delta\vdash pat\reconChk{T}\Pi\Delta_1;\Gamma_1 . Pat/\Theta_1;\rho_1$ and
     $\rho_r$ is a further refinement substitution, such as $\Theta_2\vdash\rho_r\oft\Theta_1$ and
    $\ep$ is a ground lifting substitution, such as $\Delta_i\vdash\ep\oft\Theta_1$ then\\
    $\Delta_i,\ahsub{\ep}\ahsub{\rho_r}\Delta_1;\ahsub{\ep}\ahsub{\rho_r}\Gamma_1\vdash
    \ahsub{\ep}\ahsub{\rho_r}Pat\checks\ahsub{\ep}\ahsub{\rho_r\circ\rho_1} T$.

  \item If $\Theta;\Delta\vdash\wvec{pat}\reconChk{T}\Pi\Delta_1;\Gamma_1 . \wvec{Pat}\spineret S/\Theta_1;\rho_1$ and
     $\rho_r$ is a further refinement substitution, such as $\Theta_2\vdash\rho_r\oft\Theta_1$ and
    $\ep$ is a ground lifting substitution, such as $\Delta_i\vdash\ep\oft\Theta_1$ then\\
    $\Delta_i,\ahsub{\ep}\ahsub{\rho_r}\Delta_1;\ahsub{\ep}\ahsub{\rho_r}\Gamma_1\vdash
    \ahsub{\ep}\ahsub{\rho_r}\wvec{Pat}\checks\ahsub{\ep}\ahsub{\rho_r\circ\rho_1} T\spineret\ahsub{\ep}\ahsub{\rho_r}S$.
  \end{enumerate}
\end{lem}

\begin{proof}
  \begin{flushleft}
By simultaneous induction on the first derivation.
\vs
For (1):
\paragraph{Case} $\mathcal{D} : \Theta;\Delta\vdash\yux{\const{c}}{\wvec{pat}}\recon
\Pi\Delta_1;\Gamma_1 . \yux{\const{c}}{\wvec{Pat}}\oft S/\Theta_1;\rho_1$\vs

$\Sigma(\const{c})=T$\\
$\Theta;\Delta\vdash\wvec{pat}\reconChk{T}
\Pi\Delta_1;\Gamma_1 . \wvec{Pat}\spineret S/\Theta_1;\rho_1$
\hfill by assumption\\

$\Delta_i,\ahsub{\ep}\ahsub{\rho_r}\Delta_1;\ahsub{\ep}\ahsub{\rho_r}\Gamma_1\vdash
    \ahsub{\ep}\ahsub{\rho_r}\wvec{Pat}\checks\ahsub{\ep}\ahsub{\rho_r\circ\rho_1} T\spineret\ahsub{\ep}\ahsub{\rho_r}S$
\hfill by i.h. (3)\\

Note that types in the signature (i.e. $\Sigma$) are ground so $\ahsub{\ep}\ahsub{\rho_r\circ\rho_1} T = T$\\

$
\Delta_i,\ahsub{\ep}\ahsub{\rho_r}\Delta_1;\ahsub{\ep}\ahsub{\rho_r}\Gamma_1\vdash
\yux{\const c}{(\ahsub{\ep}\ahsub{\rho_r}\wvec{Pat})}\checks\ahsub{\ep}\ahsub{\rho_r}S$
\hfill by \rl{t-pcon}.\\

$
\Delta_i,\ahsub{\ep}\ahsub{\rho_r}\Delta_1;\ahsub{\ep}\ahsub{\rho_r}\Gamma_1\vdash
\ahsub{\ep}\ahsub{\rho_r}(\yux{\const c}{\wvec{Pat})}\checks\ahsub{\ep}\ahsub{\rho_r}S$\\
\hfill by properties of substitution\\
which is what we wanted to show.\vs

For (2):
\paragraph{Case}$\mathcal{E}:
\Theta;\Delta\vdash x\reconChk{T}
      \Pi \Delta_1~;~\underbrace{x\oft T}_{\Gamma_1} . x ~/~\Theta ; \ids\Theta$

$\Gamma_1(x)= T$
\hfill by $x$ being the only variable in $\Gamma_1$\\

$\ahsub\ep\ahsub{\rho_r}\Gamma_1 = \ahsub\ep\ahsub{\rho_r}\Gamma_1 T$
\hfill by applying $\ep$ and $\rho_r$ to $\Delta_1$, $\Gamma_1$ and $T$\\

$\ahsub\ep\ahsub{\rho_r}\Delta_1;\ahsub\ep\ahsub{\rho_r}\Gamma_1\vdash
x \checks \ahsub\ep\ahsub{\rho_r}T $
\hfill by rule \rl{t-pvar}\\
which is what we wanted to prove \vs

For (3):
\paragraph{Case} $\mathcal{F} : \Theta;\Delta\vdash pat~~\wvec{pat}\reconChk{T_1\to T_2}
\Pi\Delta_2;\Gamma_1,\Gamma_2 . (\ahsub{\rho'}Pat)~~\wvec{Pat}\spineret S/\Theta_2;\rho_2\circ\rho_1$\vs

$\Theta;\Delta\vdash pat\reconChk{T_1}\Pi\Delta_1;\Gamma_1 . Pat/\Theta_1;\rho_1$\\
$\Theta_1;\Delta_1\vdash\wvec{pat}\reconChk{\ahsub\rho{T_2}}
\Pi\Delta_2;\Gamma_2 . \wvec{Pat}\spineret S/\Theta_2;\rho_2$
\hfill by assumption\\

$\Theta_2\vdash\rho_2\oft\Theta_1$\hfill by invariant of rule\\
$\Theta_3\vdash\rho_3\circ\rho_2\oft\Theta_1$\hfill \emph{(further refinement substitution)} by composition\\
$\Delta_i\vdash\ep\oft\Theta_3$\hfill lifting substitution\\

$\Delta_i,\ahsub\ep\ahsub{\rho_3\circ\rho_2}\Delta_1;\ahsub\ep\ahsub{\rho_3\circ\rho_2}\Gamma_1\vdash
\ahsub\ep\ahsub{\rho_3\circ\rho_2}Pat\checks\ahsub\ep\ahsub{\rho_3\circ\rho_2\circ\rho_1}T_1$\\
\hfill by i.h. on (1). \onestar\\

$\Delta_i,\ahsub\ep\ahsub{\rho_3}\Delta_2;\ahsub\ep\ahsub{\rho_3}\Gamma_2\vdash\ahsub\ep\ahsub{\rho_3}\wvec{Pat}\checks
\ahsub\ep\ahsub{\rho_3\circ\rho_2\circ\rho_1}T_2\spineret\ahsub\ep\ahsub{\rho_3} S$\\
\hfill by i.h. on (2)\\

we note that in pattern elaboration we have:\\
$\Delta_2=\ahsub{\rho_2}\Delta_1,\Delta'_2$\\
\hfill $\Delta_2$ is the context $\Delta_1$ with the hole instantiation applied and
some extra assumptions (i.e. $\Delta'_2$))).\\
and $\Gamma_2=\ahsub{\rho_2}\Gamma_1,\Gamma'_2$\\
\hfill $\Gamma_2$ is the context $\Gamma_1$ with the hole instantiation applied and
some extra assumptions (i.e. $\Gamma'_2$).\\

and we can weaken \onestar\ to:\\
$\Delta_i,\ahsub\ep\ahsub{\rho_3\circ\rho_2}\Delta_1,\ahsub\ep\ahsub{\rho_3}\Delta'_2;
\ahsub\ep\ahsub{\rho_3\circ\rho_2}\Gamma_1,\ahsub\ep\ahsub{\rho_3}\Gamma'_2\vdash
\ahsub\ep\ahsub{\rho_3\circ\rho_2}Pat\checks\ahsub\rho\ahsub{\rho_3\circ\rho_2\circ\rho_1}T_1$\\

$\Delta_i,\ahsub\ep\ahsub{\rho_3}\Delta_2;\ahsub\ep\ahsub{\rho_3}\Gamma_2\vdash
(\ahsub\ep\ahsub{\rho_3\circ\rho_2}Pat)(\ahsub\ep\ahsub{\rho_3}\wvec{Pat})
\checks\ahsub\ep\ahsub{\rho_3\circ\rho_2\circ\rho_1}T_1\to
\ahsub\ep\ahsub{\rho_3\circ\rho_2\circ\rho_1}T_2\spineret\ahsub\ep\ahsub{\rho_3} S$\\
\hfill by \rl{t-sarr}.\\

$\Delta_i,\ahsub\ep\ahsub{\rho_3}\Delta_2;\ahsub\ep\ahsub{\rho_3}\Gamma_2\vdash
\ahsub\ep\ahsub{\rho_3}(\ahsub{\rho_2}Pat\,\wvec{Pat})
\checks\ahsub\ep\ahsub{\rho_3\circ\rho_2\circ\rho_1}(T_1\to T_2)\spineret\ahsub\ep\ahsub{\rho_3} S$
\hfill by properties of substitution\\
which is what we wanted to show.\vs

\paragraph{Case} $\mathcal{F} :
\Theta;\Delta\vdash\yux{\ibox c}{\wvec{pat}}\reconChk{\pie{X\oft U}T}
\Pi\Delta_2;\Gamma_2 . \yux{(\ahsub{\rho_1}\ibox{C})}{\wvec{Pat}}\spineret S/\Theta_2;\rho_2\circ\rho_1$\vs

$\Theta;\Delta\vdash c\reconChk{U} C/\Theta_1;\Delta_1;\rho_1$\\
$\Theta_1;\Delta_1 \vdash\wvec{pat}\reconChk{\asub{C/X}{\ahsub{\rho_1} T}}
\Pi\Delta_2;\Gamma_2 . \wvec{Pat}\spineret S/\Theta_2;\rho_2$
\hfill by assumption\\

$\Theta_2\vdash\rho_2\oft\Theta_1$\hfill by invariant of rule\\
$\Theta_3\vdash\rho_3\circ\rho_2\oft\Theta_1$\hfill \emph{(further refinement substitution)} by composition\\
$\Delta_i\vdash\ep\oft\Theta_3$\hfill lifting substitution\\

$\Delta_i,\ahsub\ep\ahsub{\rho_3\circ\rho_2}\Delta_1\vdash
\ahsub\ep\ahsub{\rho_3\circ\rho_2} C\checks\ahsub\ep\ahsub{\rho_3\circ\rho_2\circ\rho_1} U$\\
\hfill by property of the index language\onestar\\

$\Delta_i,\ahsub\ep\ahsub{\rho_3}\Delta_2;\ahsub\ep\ahsub{\rho_3}\Gamma_2\vdash
\ahsub\ep\ahsub{\rho_3}\wvec{Pat}\checks
\ahsub\ep\ahsub{\rho_3\circ\rho_2}([C/X]\ahsub{\rho_1}T)\spineret\ahsub\ep\ahsub{\rho_3}S$\\
\hfill by i.h. (3)\\

as before, we note that:\\
$\Delta_2=\ahsub{\rho_2}\Delta_1,\Delta'_2$
\hfill $\Delta_2$ is the context $\Delta_1$ with the hole instantiation applied and
some extra assumptions (i.e. $\Delta'_2$).\\

and we can weaken \onestar to:\\
$\Delta_i,\ahsub\ep\ahsub{\rho_3\circ\rho_2}\Delta_1,\ahsub\ep\ahsub{\rho_3}\Delta'_2\vdash
\ahsub\ep\ahsub{\rho_3\circ\rho_2} C\checks\ahsub\ep\ahsub{\rho_3\circ\rho_2\circ\rho_1} U$\\

Note that $\ahsub\ep\ahsub{\rho_3\circ\rho_2}([C/X]\ahsub{\rho_1}T) =
[(\ahsub\ep\ahsub{\rho_3\circ\rho_2}C)/X](\ahsub\ep\ahsub{\rho_3\circ\rho_2\circ\rho_1}T)$\\
by properties of substitution\\

$\Delta_i,\ahsub\ep\ahsub{\rho_3}\Delta_2;\ahsub\ep\ahsub{\rho_3}\Gamma_2\vdash
\yux{\ibox{\ahsub\ep\ahsub{\rho_3\circ\rho_2} C}}{(\ahsub\ep\ahsub{\rho_3}\wvec{Pat})}\checks
\pie{X\oft (\ahsub\ep\ahsub{\rho_3\circ\rho_2\circ\rho_1} U)}(\ahsub\ep\ahsub{\rho_3\circ\rho_2\circ\rho_1}T)
\spineret\ahsub\ep\ahsub{\rho_3}S$
\hfill by \rl{t-spi}\\

$\Delta_i,\ahsub\ep\ahsub{\rho_3}\Delta_2;\ahsub\ep\ahsub{\rho_3}\Gamma_2\vdash
\ahsub\ep\ahsub{\rho_3}(\yux{\ibox{\ahsub{\rho_2} C}}{\wvec{Pat})})\checks
\ahsub\ep\ahsub{\rho_3\circ\rho_2\circ\rho_1}(\pie{X\oft U}T)
\spineret\ahsub\ep\ahsub{\rho_3}S$\\
\hfill by properties of substitution\\
which is what we wanted to show.\vs

\paragraph{Case} $\mathcal{F} : \Theta;\Delta\vdash\wvec{pat}\reconChk{\piei{X\oft U}T}
\Pi\Delta_1;\Gamma_1 . \yux{(\ahsub{\rho_1}C)}{\wvec{Pat}}\spineret S/\Theta_1;\rho_1$\vs

$\genHole (\hole Y: \Delta.U) = C$\\
$\Theta,\hole Y\oft{\Delta.U};\Delta\vdash\wvec{pat}\reconChk{\asub{C/X}T}
\Pi\Delta';\Gamma' . \wvec{Pat}/\Theta';\rho\spineret S$
\hfill by assumption\\

$\Theta,?Y\oft \Delta.U;\Delta\vdash C\checks U$ \hfill by genhole invariant\\
$\Delta_i,\ahsub\ep\ahsub{\rho_r\circ\rho_1}\Delta
\vdash\ahsub\ep\ahsub{\rho_r\circ\rho_1}C\checks \ahsub\ep\ahsub{\rho_r\circ\rho_1}U$\\
\hfill applying substitutions $\ep, \rho_r and \rho_1$\\

noting that $\Delta_1 = \ahsub{\rho_1}\Delta,\Delta'_1$\\

$\Delta_i,\ahsub\ep\ahsub{\rho_r}(\ahsub{\rho_1}\Delta,\Delta'_1)
\vdash\ahsub\ep\ahsub{\rho_r\circ\rho_1}C\checks \ahsub\ep\ahsub{\rho_r\circ\rho_1}U$
\hfill by weakening\\

$\Delta_i,\ahsub\ep\ahsub{\rho_r}\Delta_1;\ahsub\ep\ahsub{\rho_r}\Gamma'\vdash\ahsub\ep\ahsub{\rho_r}\wvec{Pat}
\checks\ahsub\ep\ahsub{\rho_r\circ\rho_1}[C/X]T
\spineret\ahsub\ep\ahsub{\rho_r}S$\\
\hfill by i.h. (3)\\

$\Delta_i,\ahsub\ep\ahsub{\rho_r}\Delta_1;\ahsub\ep\ahsub{\rho_r}\Gamma'\vdash\ahsub\ep\ahsub{\rho_r}\wvec{Pat}
\checks[\ahsub\ep\ahsub{\rho_r\circ\rho_1}C/X](\ahsub\ep\ahsub{\rho_r\circ\rho_1}T)
\spineret\ahsub\ep\ahsub{\rho_r}S$\\
\hfill by properties of substitution\\

$\Delta_i,\ahsub\ep\ahsub{\rho_r}\Delta_1;\ahsub\ep\ahsub{\rho_r}\Gamma'\vdash
\yux{\ibox{\ahsub\ep\ahsub{\rho_r\circ\rho_1}C}}{\ahsub\ep\ahsub{\rho_r}\wvec{Pat}}\checks
\piei{X\oft \ahsub\ep\ahsub{\rho_r\circ\rho_1}U}{(\ahsub\ep\ahsub{\rho_r\circ\rho_1}T)}
\spineret\ahsub\ep\ahsub{\rho_r}S$\\
\hfill by \rl{t-spi}\\

$\Delta_i,\ahsub\ep\ahsub{\rho_r}\Delta_1;\ahsub\ep\ahsub{\rho_r}\Gamma'\vdash
\ahsub\ep\ahsub{\rho_r}\yux{\ibox{\ahsub{\rho_1}C}}{\wvec{Pat}}\checks
\ahsub\ep\ahsub{\rho_r\circ\rho_1}(\piei{X\oft U}T)
\spineret\ahsub\ep\ahsub{\rho_r}S$\\
\hfill by properties of substitution\\
which is what we wanted to show\vs
  \end{flushleft}
\end{proof}

We can now prove the theorem:
\begin{thm}[Soundness]~
  \begin{enumerate}
  \item If $\Theta;\Delta;\Gamma\vdash \rclo e \theta \reconChk T E/\Theta_1;\rho_1$
    then for any grounding hole instantiation $\rho_g$ s.t.
    $\cdot \vdash \rho_g : \Theta_1$ and $\rho_0 = \rho_g \circ \rho_1$, we have\\
    $\ahsub{\rho_0}\Delta;\ahsub{\rho_0}\Gamma\vdash \ahsub{\rho_g}E \checks \ahsub{\rho_0}T$.

  \item If $\Theta;\Delta;\Gamma\vdash \rclo e \theta \recon E\oft T/\Theta_1;\rho_1$
    then for any grounding hole instantiation $\rho_g$ s.t. $\cdot \vdash \rho_g :
    \Theta_1$ and $\rho_0 = \rho_g\circ\rho_1$, we have\\
    $\ahsub{\rho_0}\Delta;\ahsub{\rho_0}\Gamma\vdash \ahsub{\rho_g}E \synths\ahsub{\rho_g}T$.

  \item If $\Delta;\Gamma\vdash \rclo {pat\mapsto e} \theta\reconChk {S\to T}
    \Pi\Delta';\Gamma' . Pat : \theta' \mapsto E$
    then $\Delta;\Gamma\vdash\Pi\Delta';\Gamma' . Pat : \theta' \mapsto E\checks S\to T$.
  \end{enumerate}
\end{thm}
\begin{proof}
  \begin{flushleft}
By simultaneous induction on the first derivation.

For (1):

\paragraph{Case} ${\mathcal{D}} : \Theta;\Delta;\Gamma\vdash\rclo{\casee e {\wvec b}}\theta\reconChk T
    \casee E {\wvec B} /\Theta';\rho$\vs

$\Theta;\Delta;\Gamma\vdash\rclo e\theta\recon E\oft S/\cdot;\rho$
\hfill by inversion on \rl{\elcase}\\

$\ahsub\rho\Delta;\ahsub\rho\Gamma\vdash\rclo{\wvec b}{\ahsub\rho\theta}\reconChk{S\to \ahsub{\rho}T}\wvec B$
\hfill by inversion on \rl{\elcase}\\

for any grounding hole inst. $\rho'$ we have
$\ahsub\rho\Delta;\ahsub\rho\Gamma\vdash E\synths S$
\hfill by I.H. noting $\rho'=\cdot$ and $\rho'\circ\rho=\rho$\\

$\ahsub\rho\Delta;\ahsub\rho\Gamma\vdash B\oft S\to\ahsub\rho T$
\hfill for every branch by (3)\\

$\ahsub\rho\Delta;\ahsub\rho\Gamma\vdash\casee E \wvec{B}\checks\ahsub\rho T$
\hfill by \rl{t-case}\vs
Note that because $E$ is ground then the only grounding hole to instantiate is the empty substitution.
\vs

\paragraph{Case} $\mathcal{D} :
\Theta;\Delta;\Gamma\vdash\rclo{\fne x e}\theta \reconChk{T_1\to T_2} \fne x E/\Theta_1;\rho_1$
\vs

$\Theta;\Delta;\Gamma,x\oft T_1\vdash \rclo e\theta \reconChk{T_2} E/\Theta_1;\rho_1$
\hfill by assumption\\

for any grounding hole inst. $\rho_g$ we have:
$\ahsub{\rho_0}\Delta;\ahsub{\rho_o}(\Gamma,x\oft T_1)\vdash\ahsub{\rho_g} E \checks \ahsub{\rho_0} T_2$
\hfill by i.h. (1) with $\rho_0=\rho_g\circ\rho_1$\\

$\ahsub{\rho_0}\Delta;(\ahsub{\rho_o}\Gamma),x\oft(\ahsub{\rho_0}T_1)\vdash\ahsub{\rho_g} E \checks \ahsub{\rho_0} T_2$\\
\hfill by properties of substitution\\

$\ahsub{\rho_0}\Delta;\ahsub{\rho_o}\Gamma\vdash\fne x {(\ahsub{\rho_g} E)} \checks(\ahsub{\rho_0}T_1)\to(\ahsub{\rho_0} T_2)$
\hfill by \rl{t-fn}\\

$\ahsub{\rho_0}\Delta;\ahsub{\rho_o}\Gamma\vdash\ahsub{\rho_g}(\fne x E) \checks\ahsub{\rho_0}(T_1)\to T_2)$\\
\hfill by properties of substitution\\

which is what we wanted to show
\vs

\paragraph{Case} $\mathcal{D} :
\Theta;\Delta;\Gamma\vdash\rclo{\mlame X e}\theta \reconChk{\pie{X\oft U}T} \mlame X E/\Theta_1;\rho_1$
\vs

$\Theta;\Delta,X\oft U;\Gamma\vdash \rclo e{\theta, X/X} \reconChk{T} E/\Theta_1;\rho_1$
\hfill by assumption\\

for any grounding hole inst. $\rho_g$ we have: $\ahsub{\rho_0}(\Delta,X\oft U);\ahsub{\rho_0}\Gamma\vdash
\ahsub{\rho_g}E \checks \ahsub{\rho_o}T$
\hfill by i.h.(1) with $\rho_0=\rho_g\circ\rho_1$\\

$\ahsub{\rho_o}\Delta,X\oft (\ahsub{\rho_0}U);\ahsub{\rho_0}\Gamma\vdash
\ahsub{\rho_g}E \checks \ahsub{\rho_o}T$
\hfill by properties of subst\\

$\ahsub{\rho_o}\Delta;\ahsub{\rho_0}\Gamma\vdash
\mlame X {\ahsub{\rho_g}E} \checks \pie{X\oft\ahsub{\rho_0} U}(\ahsub{\rho_o}T)$
\hfill by \rl{t-mlam}\\

$\ahsub{\rho_o}\Delta;\ahsub{\rho_0}\Gamma\vdash
\ahsub{\rho_g}\mlame X E \checks\ahsub{\rho_0}\pie{X\oft U}T$\\
\hfill by properties of substitution\\
which is what we wanted to show\vs

\paragraph{Case}$\mathcal{D} :
\Theta;\Delta;\Gamma\vdash\rclo{e}\theta \reconChk{\piei{X\oft U}T} \mlame X E/\Theta_1;\rho_1$\vs

this case follows the same structure as the previous\vs

\noindent
\paragraph{Case} $\mathcal{D} :
\Theta;\Delta;\Gamma\vdash\rclo{\ibox{c}}\theta\reconChk{\ibox U}\ibox C/\Theta_1;\rho_1$\vs

\noindent
$\Theta;\Delta\vdash\rclo c \theta\reconChk U C/\Theta_1;\rho_1$
\hfill by assumption\\

for any grounding inst. $\rho_g$ we have $\ahsub{\rho_0}\Delta;\ahsub{\rho_0}\Gamma\vdash \ahsub{\rho_g}C\checks \ahsub{\rho_0}U$\\
\hfill by properties of the index language and $\rho_0=\rho_g\circ\rho_1$\\

$\ahsub{\rho_0}\Delta;\ahsub{\rho_0}\Gamma\vdash \ahsub{\rho_g}\ibox C\checks \ahsub{\rho_0} \ibox U$
\hfill by \rl{t-box} and properties of subst.\\

which is what we wanted to show\vs

\paragraph{Case} $\mathcal{D} :
\Theta;\Delta;\Gamma\vdash \rclo e\theta \reconChk T
                \ahsub{\rho_2} E/\Theta_2;\rho_2\circ\rho_1$\vs

$\Theta;\Delta;\Gamma\vdash \rclo e\theta \recon E\oft T_1/\Theta_1;\rho_1$ \\
$\Theta_1 ;\ahsub{\rho_1}{\Delta} \vdash T_1 \doteq \ahsub{\rho_1} T/\Theta_2;\rho_2$
\hfill by assumption \\

for any grounding inst. $\rho_g$ we have
$\ahsub{\rho_o}\Delta;\ahsub{\rho_0}\Gamma\vdash\ahsub{\rho_g}E\synths\ahsub{\rho_g}T_1$
\hfill by i.h. (2) where $\rho_o=\rho_g\circ\rho_1$ \quad \onestar\\

for any grounding inst. $\rho_g'$ we have
$\ahsub{\rho_g'\circ\rho_2}T_1 = \ahsub{\rho_g'\circ\rho_2\circ\rho_1}T$
\hfill by prop of unification and applying a grounding subst \twostars\\

$\ahsub{\rho_g'\circ\rho_2\circ\rho_1}\Delta;\ahsub{\rho_g'\circ\rho_2\circ\rho_1}\Gamma\vdash
\ahsub{\rho_g'\circ\rho_2} E\synths\ahsub{\rho_g'\circ\rho_2}T_1$
\hfill from \onestar using $\rho_g=\rho_g'\circ\rho_2$\\

$\ahsub{\rho_g'\circ\rho_2\circ\rho_1}\Delta;\ahsub{\rho_g'\circ\rho_2\circ\rho_1}\Gamma\vdash
\ahsub{\rho_g'\circ\rho_2} E\synths\ahsub{\rho_g'\circ\rho_2\circ\rho_1}T$
\hfill by \twostars\\

$\ahsub{\rho_g'\circ\rho_2\circ\rho_1}\Delta;\ahsub{\rho_g'\circ\rho_2\circ\rho_1}\Gamma\vdash
\ahsub{\rho_g'\circ\rho_2} E\checks\ahsub{\rho_g'\circ\rho_2\circ\rho_1}T$
\hfill by \rl{t-syn}\\
which is what we wanted to show\vs

For(2):
\paragraph{Case}$\mathcal{E} :
\Theta;\Delta;\Gamma\vdash\rclo{\yux {e}{\ibox c}}\theta \recon
   \yux{E_1}{\ibox C}\oft \asub{C/X}(\ahsub{\rho_2}T)/\Theta_2;\rho_2\circ\rho_1$
\vs

$\Theta;\Delta;\Gamma\vdash \rclo {e} \theta\recon E_1\oft \pie{X\oft U}T/\Theta_1;\rho_1$\\
$\Theta_1;\ahsub{\rho_1}\Delta\vdash \rclo c {\ahsub{\rho_1}\theta}\reconChk U
C/\Theta_2;\rho_2$
\hfill by assumption\\

for any grounding instantiation $\rho_g$
s.t. $\cdot\vdash\rho_g\oft \Theta_1$ we have
$\ahsub{\rho_g\circ\rho_1}\Delta;\ahsub{\rho_g\circ\rho_1}\Gamma\vdash
\ahsub{\rho_g}E_1\synths\ahsub{\rho_g}\pie{X\oft U}T$
\hfill by i.h. (2)\onestar\\

for any grounding instantiation $\rho_g'$
s.t. $\cdot\vdash\rho_g'\oft \Theta_2$ we have
$\ahsub{\rho_g'\circ\rho_2\circ\rho_1}\Delta\vdash
\ahsub{\rho_g'}C\checks \ahsub{\rho_g'\circ\rho_2} U$
\hfill by soundness of index reconstruction\\

$\ahsub{\rho_g'\circ\rho_2\circ\rho_1}\Delta;\ahsub{\rho_g'\circ\rho_2\circ\rho_1}\Gamma\vdash
\ahsub{\rho_g'\circ\rho_2}E_1\synths\ahsub{\rho_g'\circ\rho_2}\pie{X\oft U}T$
\hfill Note that in \onestar $\cdot\vdash\rho_g\oft\Theta_1$ so we can instantiate $\rho_g=\rho_g'\circ\rho_2$\\

$\ahsub{\rho_g'\circ\rho_2\circ\rho_1}\Delta;\ahsub{\rho_g'\circ\rho_2\circ\rho_1}\Gamma\vdash
\ahsub{\rho_g'\circ\rho_2}E_1\synths\pie{X\oft (\ahsub{\rho_g'\circ\rho_2}U)}(\ahsub{\rho_g'\circ\rho_2}T)$\\
\hfill by properties of substitutions \\

$\ahsub{\rho_g'\circ\rho_2\circ\rho_1}\Delta;\ahsub{\rho_g'\circ\rho_2\circ\rho_1}\Gamma\vdash
\yux{(\ahsub{\rho_g'\circ\rho_2}E_1)}{\ahsub{\rho_g'}C}\synths
\asub{\ahsub{\rho_g'}C\}/X}(\ahsub{\rho_g'\circ\rho_2}T)$\\
\hfill by \rl{t-app-index} \\

$\ahsub{\rho_g'\circ\rho_2\circ\rho_1}\Delta;\ahsub{\rho_g'\circ\rho_2\circ\rho_1}\Gamma\vdash
\ahsub{\rho_g'}(\yux{(\ahsub{\rho_2}E_1)}{C})\synths
\ahsub{\rho_g'}(\asub{C/X}(\ahsub{\rho_2}T)$\\
\hfill by properties of substitutions \\
which is what we wanted to show.\vs

\paragraph{Case} $\mathcal{E} :
\Theta ; \Delta; \Gamma \vdash \rclo x \theta \recon E_1 \oft T_1
  ~/~\Theta_1; \ids{\Theta_1}$ \vs

$\Gamma(x) = T$ \\
$\Theta;\Delta;\Gamma\vdash x\oft T\recon E_1\oft T_1/\Theta_1$
\hfill by assumption\\

$\Delta;\Gamma\vdash x\synths T$
\hfill by rule \rl{t-var}\onestar\\

for any grounding inst. $\rho_g$ s.t. $\cdot\vdash\Theta_1$ we have:\\
$\ahsub{\rho_g\circ\rho_1}\Delta;\ahsub{\rho_g\circ\rho_1}\Gamma\vdash \ahsub{\rho_g}E_1 \oft \ahsub{\rho_g}T_1$
\hfill by \onestar, weakening and lemma~\ref{lem:inst} with $\rho_1 = \ids{\Theta_1}$\\

which is what we wanted to show\vs

For (3):
\paragraph{Case}$\mathcal{F} :
\Delta;\Gamma\vdash\rclo{pat\mapsto e}\theta\reconChk{S\to T}\Pi\Delta_r;\Gamma_r .Pat'\oft\theta\mapsto E$
\vs

$\Delta\vdash pat\reconChk S \Pi\Delta_r;\Gamma_r.Pat\oft\theta_r\mid\theta_e$
\hfill by assumption\\

$\cdot;\cdot\vdash pat\recon Pat : S'/\Theta_p;\Delta_p;\Gamma_p\mid\cdot$\\
$\Delta_p' \vdash \rho : \Theta_p$ and
$\Gamma_r = \asub{\theta_p}\ahsub\rho{\Gamma_p}$, $Pat' = \asub{\theta_p}\ahsub\rho{Pat}$
\hfill by inversion on \elsubst\\

$\Delta'_p,\ahsub\rho{\Delta_p};\ahsub\rho{\Gamma_p}\vdash\ahsub\rho{Pat}\checks\ahsub\rho{S'}$
\hfill by pattern elaboration lemma\\

$\Delta,\Delta'_p,\ahsub\rho{\Delta_p}\vdash\ahsub\rho{S'}\doteqdot S/\Delta_r, \theta$
\hfill by inversion on \elsubst\\

where we can split $\theta$ as $\theta=\theta_r,\theta_i,\theta_e$ so that:
$\left\{
  \begin{array}{l}
    \Delta_r \vdash\theta_r\oft\Delta \\
    \Delta_r \vdash\theta_i\oft\Delta'_p \\
    \Delta_r \vdash  \theta_i,\theta_e \oft \Delta'_p, \ahsub\rho{\Delta_p}
  \end{array}\right.$\\

let $\theta_p = \theta_i,\theta_e$\\
$\asub{\underbrace{\theta_i,\theta_e}_{\theta_p}}\ahsub\rho{S'} = \asub{\theta_r}S$
\hfill by soundness of unification and the fact that $\Delta$ and
$\Delta'_p,\ahsub\rho{\Delta_p}$ are distinct\\

$\Delta_r;\asub{\theta_p}\ahsub\rho{\Gamma_p}\vdash\asub{\theta_p}\ahsub\rho{Pat}\checks\asub{\theta_p}\ahsub\rho{S'}$
\hfill by substitution lemma\\

$\Delta_r;\underbrace{\asub{\theta_p}\ahsub\rho{\Gamma_p}}_{\Gamma_r}
\vdash\underbrace{\asub{\theta_p}\ahsub\rho{Pat}}_{Pat'}\checks\asub{\theta_r}S$
\hfill by $\asub\theta\ahsub\rho{S'}=\asub{\theta_r}S$\\

$\cdot ; \Delta_r;\asub{\theta_r}\Gamma,\Gamma_r\vdash\rclo e{\theta_r\circ\theta,~\theta_e}
\reconChk{\asub{\theta_r}{T}} E/\cdot;\cdot$
\hfill by assumption\\

$\Delta_r;\asub{\theta_r}\Gamma,\Gamma_r\vdash E \checks \asub{\theta_r}T$
\hfill by (1)\\

$\Delta;\Gamma\vdash\Pi\Delta_r;\Gamma_r . Pat'\oft\theta_r\mapsto E\checks S\to T$
\hfill by \rl{t-branch}\\
which is what we wanted to show.
\end{flushleft}
\end{proof}

%% file: babybel/proof.tex
\chapter{Babybel's Translation Meta-theory}\label{sec:bbproofs}

\begin{lem}[Ambient Context]\label{lem:ctxproof}
If $\Gamma(\code{u}) = [\Psi \vdash \const a]$ then\\
$ \eval{\Gamma} (u) =  \tp{sftm}{ \eval{\Psi}, \const a}$.
\end{lem}
\begin{proof}
Induction on the structure of $\Gamma$.
\end{proof}

\begin{lem}[Terms]\label{lem:tmproof}$\;$
  \begin{enumerate}
  \item   If $\Gamma;\Psi\vdash M \oft A$ then $\cdot;\eval\Gamma\vdash\eval{M}_{\Psi\vdash A} \oft \eval{\Psi\vdash A}$.
  \item If $\Gamma ; \Psi \vdash \sigma \oft \Phi$ then $\cdot ; \eval{\Gamma} \vdash \eval{\sigma}_ {\Psi \vdash \Phi} : \eval{\Psi\vdash\Phi}$
  \end{enumerate}
\end{lem}
\begin{proof}
  \begin{flushleft}
Induction on the typing derivation.

\paragraph{Case}
$\mathcal{D} = \infer[\rl{t-var}]
    {\Gamma ; \Psi \vdash x \oft \const a}
    {\Psi(x) = \const a}
$
\\[1em]
$\cdot ; \eval \Gamma \vdash \tp{Var}{\eval{\Psi},~\const a}~k : \tp{sftm}{\eval{\Psi},~\const a}$ \hfill
by the correctness of our translation function that computes  the position $k$ of $x$ in $\Psi$.
\\[0.5em]
$\cdot;\eval\Gamma\vdash \eval{x}_{\Psi\vdash\const a} : \eval{\Psi\vdash \const a}$ \hfill by definition\\

\paragraph{Case}
$\mathcal{D} = \infer[\rl{t-sub}]
{\Gamma ; \Phi \vdash M[\sigma]^{\Phi}_{\Psi} : A}
{\Gamma ; \Psi \vdash M \oft A &
                   \Gamma ; \Phi \vdash \sigma \oft \Psi }$\vs

$\cdot ; \eval{\Gamma} \vdash \eval{M}_{\Psi \vdash A} \oft \eval{\Psi
  \vdash A}$ \hfill by i.h.\\[0.5em]
$\cdot ; \eval{\Gamma} \vdash \eval{\sigma}_{\Phi \vdash \Psi} \oft
\eval{\Phi \vdash \Psi}$ \hfill by i.h.
\\[0.5em]
$e = \code{apply\_sub}~\eval{M}_{\Psi\vdash A}~ \eval{\sigma}_{\Phi\vdash\Psi}$ and\\
$\cdot ; \eval{\Gamma} \vdash e : \eval {\Phi \vdash A}$ \hfill by property of $\code{apply\_sub}$

\paragraph{Case}
$\mathcal{D} = \infer[\rl{t-qvar}]
{\Gamma ; \Psi \vdash \mvar u \oft \const a}
{\Gamma(\code{u}) = [\Psi \vdash \const a]}$
\\[1em]
$\eval\Gamma(\code{u}) = \tp{sftm}{ \eval{\Psi, \const a}}$ \hfill by
Lemma \ref{lem:ctxproof}  \\[0.5em]
$\cdot ; \eval \Gamma \vdash \eval{\mvar u} \oft \eval {\Psi \vdash \const a}$
\hfill by rule $\rl{g-var}$ and definition

\paragraph{Case}
$\mathcal{D} = \infer[\rl{t-lam}]
    {\Gamma ; \Psi \vdash \lam x M \oft \const a \to A}
    {\Gamma ; \Psi, x\oft \const a \vdash M \oft A}$\vs

$\cdot;\eval\Gamma\vdash \eval{M}_{\Psi,\const a\vdash A} \oft \eval
{\Psi, \const a \vdash A}$ \hfill by i.h.\\[0.5em]
$\eval{\Psi \vdash \const a \to A} = \tp{sftm}{\eval \Psi, \eval{\const a \to A}} = \tp{sftm}{\eval \Psi, \tp{arr}{\const a, \eval A}     }$\hfill by definition
\\[0.5em]
$\cdot;\eval\Gamma\vdash \tp{Lam~}{\tp{cons~}{\eval{\Psi},\const a, \eval A~}\,}~\eval{M}_{\Gamma,\const a\vdash A} : \eval{\Psi \vdash \const a \to A}$\hfill
by using \rl{g-con} \\[1em]

Similar for the other cases.
\end{flushleft}
\end{proof}

\begin{lem}[Pat.]\label{lem:patproof}
  If $\vdash\ptj {pat} \tau \Gamma$ then $\cdot\vdash\ptj{\eval{pat}_{\Psi\vdash A}^\Gamma}{\eval \tau}{\Gamma}$.
\end{lem}
\begin{proof}
   By induction on the type derivation for patterns.\vs
\end{proof}

\begin{lem}[Ctx. Pat.]\label{lem:ctxpatproof}
  If $\Psi\vdash\ptj R A \Gamma$ then $\cdot\vdash\ptj{\eval{R}_{\Psi\vdash A}^\Gamma}{\eval{\Psi \vdash A}}{\Gamma}$.
\end{lem}
\begin{proof}
  \begin{flushleft}
By induction on the typing derivation. The interesting case is the one where $R$ is a pattern variable.

\paragraph{Case:}
$\mathcal{D} = \infer[\rl{tp-mvar}]
    {\Psi\vdash\ptj{\pmetavar u} {\const a} {u\oft[\Psi\vdash \const a]}}{}$\vs

$\cdot\vdash \ptj {u} {\tp{sftm}{\eval\Psi, \const a}} {\cdot;u\oft \tp{sftm}{\eval\Psi, \const a}}$ \hfill by \rl{gp-var}\\[0.5em]
$\cdot\vdash \ptj{\eval{\pmetavar u}^{u\oft[\Psi\vdash \const a]}_{\const a}} {\eval{\Psi\vdash \const a}} {\cdot;\eval {u\oft[\Psi\vdash \const a]}}$
\hfill by definition
\\[1em]
The other cases are similar.
\end{flushleft}
\end{proof}

\begin{thm}[Main]$\;$
 \begin{enumerate}
 \item If $\Gamma\vdash e\checks \tau$ then $\cdot;\eval\Gamma\vdash\eval{e}_{\Gamma\vdash\tau} \oft \eval\tau$.
\item If $\Gamma\vdash i \synths \tau$ then $\cdot;\eval\Gamma\vdash\eval{i}_{\Gamma\vdash\tau} \oft \eval\tau$.
  \end{enumerate}
\end{thm}

\begin{proof}
  \begin{flushleft}
By mutual induction on the type derivations.\\

\paragraph{Case}
    $ \mathcal{D} = \infer[\rl{t-ctx-obj}]
    {\Gamma \vdash [\hat\Psi \vdash M] \checks [\Psi \vdash \const a]}
    {\Gamma;\Psi \vdash M \oft \const a}$\vs

$\cdot ; \eval{\Gamma} \vdash \eval{M}_{\Psi\vdash \const a} \oft \eval {\Psi \vdash \const a}$ \hfill from Lemma \ref{lem:tmproof}.\\[0.5em]
$\eval{\Psi\vdash\const a} = \tp{sftm}{ \eval{\Psi}, \eval{A}}$ and $\eval{\Psi\vdash\const M}_{\Psi\vdash \const a} = \eval{M}_{\Psi\vdash \const a}$ \hfill \\[0.5em]
$\eval{\Gamma} \vdash \eval{\hat\Psi\vdash M]}_{\Psi\vdash\const a}
\oft \eval{[\Psi \vdash \const a]}$ \hfill by definition

\paragraph{Case}
$\mathcal{D} = \infer[\rl{t-cm}]
{\Gamma \vdash \cmatche i \branches \checks \tau}
{\Gamma \vdash i \synths [\Psi \vdash \const a] &
  \forall b\in\many b\mathrel{.}\Gamma \vdash b \checks [\Psi \vdash \const a] \to \tau}$\vs

%
We note that each $b_i \in \many b$ is of the form $[\Psi\vdash R] \mapsto e$.\\[0.5em]
$\Psi \vdash \ptj R  A \Gamma$ \hfill \\[0.5em]
$\Gamma, \Gamma' \vdash e \checks \tau$  \hfill by typing inversion\\[0.5em]
$\cdot\vdash\ptj {\eval{R}^{\Gamma'}_{\Psi\vdash \const a}} {\const a} \Gamma'$ \hfill by Lemma \ref{lem:ctxpatproof}\\[0.5em]
$\cdot;\eval{\Gamma,\Gamma'}\vdash \eval{e}_{\Gamma,\Gamma'\vdash \tau} \oft \eval\tau$ \hfill by i.h. (1).\\[0.5em]
$\cdot;\Gamma,\Gamma' \vdash \eval{R}^{\Gamma'}_{\Psi\vdash\const a} \mapsto \eval{e}_{\Gamma,\Gamma'} \oft \const a \to \eval\tau$ \hfill by \rl{g-branch} \\[0.5em]
$\cdot;\eval{\Gamma} \vdash \eval{i}_{\Gamma\vdash\const a} \oft \const a$ \hfill by I.H (2). \\[0.5em]
$\cdot ; \eval{\Gamma} \vdash \eval{\cmatche i \branches}_{\Gamma  \vdash \tau} : \eval{\tau}$ \hfill
by $\rl{g-match}$

\paragraph{Case}
$\mathcal{D} =
\infer[\rl{t-emb}]
{\Gamma\vdash\embe i \checks \tau}
{\Gamma\vdash i \synths \tau' \qquad \tau=\tau'}$\vs

$\cdot ; \eval{\Gamma} \vdash \eval{i}_{\Gamma \vdash \tau} : \eval{\tau}$ \hfill by i.h.(2)\\[1em]

The other cases for part 1) are similar.

\paragraph{Case}
$\mathcal{D} = \infer[\rl{t-var}]
{\Gamma \vdash x \synths \tau}
{\Gamma(x) = \tau}$\vs

$\cdot ; \eval{\Gamma} \vdash \eval{x}_{\Gamma \vdash \tau} : \eval{\tau}$ \hfill trivial using \rl{g-var}.

\paragraph{Case}
$\mathcal{D} = \infer[\rl{t-app}]
{\Gamma\vdash \app{i}{e} \synths \tau}
{\Gamma\vdash i \synths \tau'\to\tau &
  \Gamma\vdash e \checks \tau'}$\vs

$\cdot; \eval{\Gamma} \vdash \eval{i}_{\Gamma\vdash\tau'\to\tau} \oft \eval{\tau' \to \tau}$ \hfill by i.h.\\[0.5em]
$\cdot; \eval{\Gamma} \vdash \eval{i}_{\Gamma\vdash\tau'\to\tau} \oft \eval{\tau'} \to \eval\tau$ \hfill by definition \\[0.5em]
$\cdot;\eval{\Gamma} \vdash \eval{e}_{\Gamma\vdash\tau'} \oft \eval{\tau'}$ \hfill by i.h.\\[0.5em]
$\cdot ; \eval{\Gamma} \vdash \eval{\app i e}_{\Gamma \vdash \tau} : \eval{\tau}$ \hfill by \rl{g-app}\\[1em]
The other cases for part 2) are similar.
\end{flushleft}
\end{proof}


%% file: orca/beltran.tex
\chapter[Translating Booleans in Beluga]{A Beluga Implementation of Orca's Translation Example}\label{chp:belugatran}

This chapter contains an implementation of the example from
Section~\ref{sec:tranbool} written using the Beluga language. We do
not explain this example in detail, but it is here as a reference to
show how similar the two languages are.

We start by defining the languages and translation of types using the
logical framework LF, this is similar to the Orca implementation,
except that where Orca uses a function to implement the translation of
types here we have a relation that is inhabited by the pairs of source
types and their translations. Notice the schema declaration that
allows for storing blocks of assumptions in the context thus avoiding
the need to declare a context relation.

\lstset{language=Beluga}
\begin{lstlisting}
LF s-tp : type =
| bool : s-tp
| arr : s-tp -> s-tp -> s-tp
;

LF s-tm : s-tp -> type =
| app : s-tm (arr S T) -> s-tm S -> s-tm T
| lam : (s-tm S -> s-tm T) -> s-tm (arr S T)
| tt : s-tm bool
| ff : s-tm bool
| ife : s-tm bool -> s-tm T -> s-tm T -> s-tm T
;

LF t-tp : type =
| tunit : t-tp
| tsum : t-tp -> t-tp -> t-tp
| tarr : t-tp -> t-tp -> t-tp
;

LF t-tm : t-tp -> type =
| tapp : t-tm (tarr S T) -> t-tm S -> t-tm T
| tlam : (t-tm S -> t-tm T) -> t-tm (tarr S T)
| tone : t-tm tunit
| tinl : t-tm s -> t-tm (tsum S T)
| tinr : t-tm t -> t-tm (tsum S T)
| tcase : t-tm (tsum S T) ->
          (t-tm S -> t-tm R) ->
          (t-tm T -> t-tm R) ->
          t-tm R
;

LF tran-tp : s-tp -> t-tp -> type =
| t-bool : tran-tp bool (tsum tunit tunit)
| t-arr : tran-tp S S' ->
          tran-tp T T' ->
          tran-tp (arr S T) (tarr S' T')
;

schema ctx = block (s: s-tm S, t: t-tm T, tr: tran-tp S T;
\end{lstlisting}

The type \lstinline!tran-tp! replaces the Orca function with the same
name. But then, we need to show that the relation is defined for all
terms, and that it is deterministic. We prove both facts in the
\lstinline!comp-tran-tp! and \lstinline!unique-tran-tp! lemmas,
respectively.

\begin{lstlisting}
LF ex-tran : s-tp -> type =
| ex : {T : t-tp} tran-tp S T -> ex-tran S
;

rec comp-tran-tp : {S : [|- s-tp]} [|- ex-tran S] =
mlam S => case [|- S] of
| [|- bool] => [|- ex (tsum tunit tunit) t-bool]
| [|- arr S S'] =>
  let [|- ex T TR] = comp-tran-tp [|- S] in
  let [|- ex T' TR'] = comp-tran-tp [|- S'] in
  [|- ex (tarr T T') (t-arr TR TR')]
;

LF eq : t-tp -> t-tp -> type =
| refl : eq T T
;

rec unique-tran-tp : (g:ctx) [g |- tran-tp S[] T[]] -> [|- tran-tp S T'] -> [|- eq T T'] =
fn tr1 => fn tr2 => case tr1 of
| [g |- t-bool] =>
  let [|- t-bool] = tr2 in
  [|- refl]
| [g |- t-arr T1 T2] =>
  let [|- t-arr T1' T2'] = tr2 in
  let [|- refl] = unique-tran-tp [g |- T1] [|- T1'] in
  let [|- refl] = unique-tran-tp [g |- T2] [|- T2'] in
  [|- refl]
;
\end{lstlisting}

\filbreak
With these lemmas in place, the translation can be implemented, notice
how the context relation was made superfluous by the information
carried in the context schema. Context schemas for Orca remain an
unexplored subject.

\begin{lstlisting}
rec tran : (g : ctx) [|- tran-tp S T] -> [g |- s-tm S[]] -> [g |- t-tm T[]] =
fn tr => fn e => case e of
| {M : [g|- s-tm (arr S[] T[])]} [g |- app M N] =>
  let [|- ex S' TR] = comp-tran-tp [|- S] in
  let [|- TR'] = tr in
  let [g |- M'] = tran [|- t-arr TR TR'] [g |- M] in
  let [g |- N'] = tran [|- TR] [g |- N] in
  [g |- tapp M' N']

| [g |- lam \x.M] =>
  let [|- t-arr TR TR'] :
        [|- tran-tp (arr S T) (tarr S' T')] = tr
  in
  let [g, x: block s:s-tm S[]
        , t:t-tm S'[]
        , tr:tran-tp S[] S'[] |- M'[..,x.t]] =
      tran [|- TR']
           [g, x: block s:s-tm S[]
             , t:t-tm S'[]
             , tr:tran-tp S[] S'[] |- M[..,x.s]]
  in
  [g |- tlam \x.M']

| [g |- tt] =>
  let [|- t-bool] = tr in
  [g |- tinr tone]

| [g |- ff] =>
  let [|- t-bool] = tr in
  [g |- tinl tone]

| [g |- ife C M N] =>
  let [g |- C'] = tran [|- t-bool] [g |- C] in
  let [g |- M'] = tran tr [g |- M] in
  let [g |- N'] = tran tr [g |- N] in
  [g |- tcase C' (\x. M'[..]) (\x. N'[..])]

| {#p:[g |- block s:s-tm S[]
                , t:t-tm T[]
                , tr:tran-tp S[] T[]]} [g |- #p.s] =>
  let [|- refl] = unique-tran-tp [g |- #p.tr] tr in
  [g |- #p.t]
;
\end{lstlisting}


%% file: thesis.bbl
\begin{thebibliography}{111}
\providecommand{\natexlab}[1]{#1}
\providecommand{\url}[1]{\texttt{#1}}
\expandafter\ifx\csname urlstyle\endcsname\relax
  \providecommand{\doi}[1]{doi: #1}\else
  \providecommand{\doi}{doi: \begingroup \urlstyle{rm}\Url}\fi

\bibitem[Abel(2006)]{Abel:phdthesis}
A.~Abel.
\newblock \emph{A Polymorphic Lambda-Calculus with Sized Higher-Order Types}.
\newblock PhD thesis, Ludwig-Maximilians University, 2006.

\bibitem[Abel and Pientka(2011)]{Abel:TLCA11}
A.~Abel and B.~Pientka.
\newblock Higher-order dynamic pattern unification for dependent types and
  records.
\newblock In L.~Ong, editor, \emph{10th {I}nternational {C}onference on {T}yped
  {L}ambda {C}alculi and {A}pplications {(TLCA'11)}}, Lecture Notes in Computer
  Science (LNCS 6690), pages 10--26. Springer, 2011.

\bibitem[Allais et~al.(2017)Allais, Chapman, McBride, and McKinna]{Allais:2017}
G.~Allais, J.~Chapman, C.~McBride, and J.~McKinna.
\newblock Type-and-scope safe programs and their proofs.
\newblock In \emph{Proceedings of the 6th ACM SIGPLAN Conference on Certified
  Programs and Proofs}, CPP 2017, pages 195--207. ACM, 2017.

\bibitem[Altenkirch(1993)]{Altenkirch:TLCA93}
T.~Altenkirch.
\newblock A formalization of the strong normalization proof for {System F} in
  {LEGO}.
\newblock In M.~Bezem and J.~F. Groote, editors, \emph{International Conference
  on Typed Lambda Calculi and Applications (TLCA '93)}, volume 664 of
  \emph{Lecture Notes in Computer Science}, pages 13--28. Springer, 1993.
\newblock ISBN 3-540-56517-5.

\bibitem[Asperti et~al.(2012)Asperti, Ricciotti, Coen, and Tassi]{Asperti:2012}
A.~Asperti, W.~Ricciotti, C.~S. Coen, and E.~Tassi.
\newblock A bi-directional refinement algorithm for the calculus of
  (co)inductive constructions.
\newblock \emph{Logical Methods in Computer Science}, 8:\penalty0 1--49, 2012.

\bibitem[Augustsson(1985)]{Augustsson:FPCA85}
L.~Augustsson.
\newblock Compiling pattern matching.
\newblock In J.-P. Jouannaud, editor, \emph{Functional Programming Languages
  and Computer Architecture (FPCA'85)}, volume 201 of \emph{Lecture Notes in
  Computer Science (LNCS)}, pages 368--381. Springer, 1985.

\bibitem[Baelde et~al.(2014)Baelde, Chaudhuri, Gacek, Miller, Nadathur, Tiu,
  and Wang]{Baelde:AbellaTutorial}
D.~Baelde, K.~Chaudhuri, A.~Gacek, D.~Miller, G.~Nadathur, A.~Tiu, and Y.~Wang.
\newblock Abella: {A} system for reasoning about relational specifications.
\newblock \emph{Journal of Formalized Reasoning}, 7\penalty0 (2):\penalty0
  1--89, 2014.

\bibitem[Barendregt(1992)]{Barendregt:1992}
H.~P. Barendregt.
\newblock Lambda calculi with types.
\newblock In S.~Abramsky, D.~M. Gabbay, and S.~E. Maibaum, editors,
  \emph{Handbook of Logic in Computer Science (Vol. 2)}, pages 117--309. Oxford
  University Press, Inc., New York, NY, USA, 1992.
\newblock ISBN 0-19-853761-1.

\bibitem[Belanger et~al.(2013)Belanger, Monnier, and Pientka]{Belanger:CPP13}
O.~S. Belanger, S.~Monnier, and B.~Pientka.
\newblock Programming type-safe transformations using higher-order abstract
  syntax.
\newblock In G.~Gonthier and M.~Norrish, editors, \emph{Third International
  Conference on Certified Programs and Proofs (CPP'13)}, Lecture Notes in
  Computer Science (LNCS 8307), pages 243--258. Springer, 2013.

\bibitem[Benton et~al.(2012{\natexlab{a}})Benton, Hur, Kennedy, and
  McBride]{Benton:JAR12}
N.~Benton, C.~Hur, A.~Kennedy, and C.~McBride.
\newblock Strongly typed term representations in coq.
\newblock \emph{J. Autom. Reasoning}, 49\penalty0 (2):\penalty0 141--159,
  2012{\natexlab{a}}.

\bibitem[Benton et~al.(2012{\natexlab{b}})Benton, Hur, Kennedy, and
  McBride]{Benton:12}
N.~Benton, C.-K. Hur, A.~J.~. Kennedy, and C.~McBride.
\newblock Strongly typed term representations in {C}oq.
\newblock \emph{Journal of Automated Reasoning}, 49\penalty0 (2):\penalty0
  141--159, 2012{\natexlab{b}}.

\bibitem[Boespflug and Pientka(2011)]{Boespflug:LFMTP11}
M.~Boespflug and B.~Pientka.
\newblock Multi-level contextual modal type theory.
\newblock In G.~Nadathur and H.~Geuvers, editors, \emph{6th International
  Workshop on Logical Frameworks and Meta-languages: Theory and Practice
  (LFMTP'11)}, Electronic Proceedings in Theoretical Computer Science (EPTCS),
  2011.

\bibitem[Brady(2013)]{Brady:JFP13}
E.~Brady.
\newblock Idris, a general-purpose dependently typed programming language:
  Design and implementation.
\newblock \emph{Journal of Functional Programming}, 23:\penalty0 552--593, 9
  2013.

\bibitem[Cave and Pientka(2012)]{Cave:POPL12}
A.~Cave and B.~Pientka.
\newblock Programming with binders and indexed data-types.
\newblock In \emph{{39th Annual {ACM} SIGPLAN-SIGACT Symposium on Principles of
  Programming Languages (POPL'12)}}, pages 413--424. ACM Press, 2012.

\bibitem[Cave and Pientka(2013)]{Cave:LFMTP13}
A.~Cave and B.~Pientka.
\newblock First-class substitutions in contextual type theory.
\newblock In \emph{Proceedings of the Eighth ACM SIGPLAN International Workshop
  on Logical Frameworks and Meta-Languages: Theory and Practice (LFMTP'13)},
  pages 15--24. ACM Press, 2013.
\newblock ISBN 978-1-4503-2382-6.

\bibitem[Cave et~al.(2014)Cave, Ferreira, Panangaden, and Pientka]{Cave:POPL14}
A.~Cave, F.~Ferreira, P.~Panangaden, and B.~Pientka.
\newblock Fair reactive programming.
\newblock In S.~Jagannathan and P.~Sewell, editors, \emph{The 41st Annual {ACM}
  {SIGPLAN-SIGACT} Symposium on Principles of Programming Languages, {POPL}
  '14}, pages 361--372. {ACM}, 2014.

\bibitem[Cervesato et~al.(2002)Cervesato, Pfenning, Walker, and
  Watkins]{Cervesato02tr}
I.~Cervesato, F.~Pfenning, D.~Walker, and K.~Watkins.
\newblock A concurrent logical framework {II}: Examples and applications.
\newblock Technical Report CMU-CS-02-102, Department of Computer Science,
  Carnegie Mellon University, 2002.

\bibitem[Chen and Xi(2005)]{Chen05:CombiningProgramsWithProofs}
C.~Chen and H.~Xi.
\newblock Combining programming with theorem proving.
\newblock In O.~Danvy and B.~C. Pierce, editors, \emph{10th International
  Conference on Functional Programming}, pages 66--77, 2005.

\bibitem[Cheney and Hinze(2003{\natexlab{a}})]{ChH03Pha}
J.~Cheney and R.~Hinze.
\newblock First-class phantom types.
\newblock Technical Report CUCIS TR2003-1901, Cornell University,
  2003{\natexlab{a}}.

\bibitem[Cheney and Hinze(2003{\natexlab{b}})]{Cheney:2003}
J.~Cheney and R.~Hinze.
\newblock First-class phantom types.
\newblock Technical report, Cornell University, 2003{\natexlab{b}}.

\bibitem[Chlipala(2008)]{Chlipala:ICFP08}
A.~J. Chlipala.
\newblock Parametric higher-order abstract syntax for mechanized semantics.
\newblock In J.~Hook and P.~Thiemann, editors, \emph{13th ACM SIGPLAN
  {I}nternational {C}onference on {F}unctional {P}rogramming (ICFP'08)}, pages
  143--156. ACM, 2008.

\bibitem[Christiansen and Brady(2016)]{Christiansen:2016}
D.~Christiansen and E.~Brady.
\newblock Elaborator reflection: Extending idris in idris.
\newblock In \emph{Proceedings of the 21st ACM SIGPLAN International Conference
  on Functional Programming}, ICFP 2016, pages 284--297, New York, NY, USA,
  2016. ACM.

\bibitem[Church(1940)]{Church:1940}
A.~Church.
\newblock A formulation of the simple theory of types.
\newblock \emph{The Journal of Symbolic Logic}, 5:\penalty0 56--68, 6 1940.

\bibitem[Constable(1986)]{Constable:1986}
R.~L. Constable.
\newblock \emph{Implementing Mathematics with the Nuprl Proof Development
  System}.
\newblock Prentice-Hall, Inc., Upper Saddle River, NJ, USA, 1986.
\newblock ISBN 0-134-51832-2.

\bibitem[Coquand and Huet(1988)]{Coquand:1988}
T.~Coquand and G.~Huet.
\newblock The calculus of constructions.
\newblock \emph{Information and Computation}, 76\penalty0 (2):\penalty0 95 --
  120, 1988.

\bibitem[Danielsson and Norell(2011)]{Danielsson:2011}
N.~A. Danielsson and U.~Norell.
\newblock \emph{Parsing Mixfix Operators}, pages 80--99.
\newblock Springer Berlin Heidelberg, 2011.

\bibitem[Davies and Pfenning(2001)]{Davies:ACM01}
R.~Davies and F.~Pfenning.
\newblock A modal analysis of staged computation.
\newblock \emph{Journal of the ACM}, 48\penalty0 (3):\penalty0 555--604, 2001.
\newblock \doi{10.1145/382780.382785}.

\bibitem[de~Bruijn(1972)]{de1972lambda}
N.~de~Bruijn.
\newblock {Lambda calculus notation with nameless dummies, a tool for automatic
  formula manipulation, with application to the Church-Rosser theorem}.
\newblock \emph{Indag. Math}, 34\penalty0 (5):\penalty0 381--392, 1972.

\bibitem[de~Bruijn(1983)]{debruijn:1983}
N.~G. de~Bruijn.
\newblock \emph{AUTOMATH, a Language for Mathematics}, pages 159--200.
\newblock Springer Berlin Heidelberg, Berlin, Heidelberg, 1983.

\bibitem[de~Bruijn(1991)]{deBruijNAMS:1991}
N.~G. de~Bruijn.
\newblock {C}hecking {M}athematics with {C}omputer {A}ssistance.
\newblock \emph{Notices of the American Mathematical Society}, 38\penalty0
  (1):\penalty0 8--15, 1991.

\bibitem[Despeyroux and Hirschowitz(1994)]{Despeyroux:1994:HAS:645708.664171}
J.~Despeyroux and A.~Hirschowitz.
\newblock Higher-order abstract syntax with induction in coq.
\newblock In \emph{Proceedings of the 5th International Conference on Logic
  Programming and Automated Reasoning}, LPAR '94, pages 159--173, London, UK,
  UK, 1994. Springer-Verlag.
\newblock ISBN 3-540-58216-9.

\bibitem[Despeyroux et~al.(1997)Despeyroux, Pfenning, and
  Sch{\"u}rmann]{Despeyroux97}
J.~Despeyroux, F.~Pfenning, and C.~Sch{\"u}rmann.
\newblock Primitive recursion for higher-order abstract syntax.
\newblock In \emph{Proceedings of the Third International Conference on Typed
  Lambda Calculus and Applications (TLCA'97)}, pages 147--163. Springer, 1997.
\newblock Extended version available as Technical Report CMU-CS-96-172,
  Carnegie Mellon University.

\bibitem[Dowek et~al.(1996)Dowek, Hardin, Kirchner, and
  Pfenning]{Dowek96jicslp}
G.~Dowek, T.~Hardin, C.~Kirchner, and F.~Pfenning.
\newblock Unification via explicit substitutions: The case of higher-order
  patterns.
\newblock In M.~Maher, editor, \emph{{Joint International Conference and
  Symposium on Logic Programming}}, pages 259--273. MIT Press, Sept. 1996.

\bibitem[Dunfield and Pientka(2009)]{Dunfield:coverage08}
J.~Dunfield and B.~Pientka.
\newblock Case analysis of higher-order data.
\newblock In \emph{International Workshop on Logical Frameworks and
  Meta-Languages: Theory and Practice (LFMTP'08)}, volume 228 of
  \emph{Electronic Notes in Theoretical Computer Science (ENTCS)}, pages
  69--84. Elsevier, June 2009.

\bibitem[Dybjer(1994)]{Dybjer:1994}
P.~Dybjer.
\newblock {I}nductive {F}amilies.
\newblock \emph{Formal Aspects of Computing}, 6\penalty0 (4):\penalty0
  440--465, Jul 1994.
\newblock ISSN 1433-299X.

\bibitem[Felty and Momigliano(2012)]{Felty12}
A.~Felty and A.~Momigliano.
\newblock Hybrid: A definitional two-level approach to reasoning with
  higher-order abstract syntax.
\newblock \emph{Journal of Automated Reasoning}, 48\penalty0 (1):\penalty0
  43--105, 2012.

\bibitem[Ferreira and Pientka(2014)]{Ferreira:2014}
F.~Ferreira and B.~Pientka.
\newblock Bidirectional elaboration of dependently typed programs.
\newblock In \emph{Proceedings of the 16th International Symposium on
  Principles and Practice of Declarative Programming}, PPDP '14, pages
  161--174. ACM, 2014.

\bibitem[Ferreira and Pientka(2017)]{Ferreira:2017}
F.~Ferreira and B.~Pientka.
\newblock Programs using syntax with first-class binders.
\newblock In H.~Yang, editor, \emph{Programming Languages and Systems: 26th
  European Symposium on Programming, ESOP 2017, Uppsala, Sweden}, pages
  504--529. Springer Berlin Heidelberg, 2017.

\bibitem[Ferreira et~al.(2017)Ferreira, Thibodeau, and Pientka]{Ferreira:2017b}
F.~Ferreira, D.~Thibodeau, and B.~Pientka.
\newblock Dependent type theory with contextual types.
\newblock In \emph{23rd International Conference on Types for Proofs and
  Programs {TYPES} 2017, Budapest, Hungary}, pages 61--62, 2017.

\bibitem[Fiore and Hur(2008)]{Fiore:2008}
M.~Fiore and C.-K. Hur.
\newblock Term equational systems and logics: (extended abstract).
\newblock \emph{Electronic Notes in Theoretical Computer Science},
  218:\penalty0 171 -- 192, 2008.
\newblock ISSN 1571-0661.
\newblock Proceedings of the 24th Conference on the Mathematical Foundations of
  Programming Semantics (MFPS XXIV).

\bibitem[Gabbay and Pitts(1999)]{Gabbay:LICS99}
M.~Gabbay and A.~Pitts.
\newblock A new approach to abstract syntax involving binders.
\newblock In G.~Longo, editor, \emph{Proceedings of the 14th Annual Symposium
  on Logic in Computer Science ({LICS}'99)}, pages 214--224. IEEE Computer
  Society Press, 1999.

\bibitem[Gacek(2008)]{Gacek:IJCAR08}
A.~Gacek.
\newblock The {A}bella interactive theorem prover (system description).
\newblock In \emph{4th International Joint Conference on Automated Reasoning},
  volume 5195 of \emph{Lecture Notes in Artificial Intelligence}, pages
  154--161. Springer, 2008.

\bibitem[Gacek et~al.(2012)Gacek, Miller, and Nadathur]{Gacek:JAR12}
A.~Gacek, D.~Miller, and G.~Nadathur.
\newblock A two-level logic approach to reasoning about computations.
\newblock \emph{Journal of Automated Reasoning}, 49\penalty0 (2):\penalty0
  241--273, 2012.

\bibitem[Garrigue and R{\'e}my(2013)]{Garrigue:2013}
J.~Garrigue and D.~R{\'e}my.
\newblock Ambivalent types for principal type inference with gadts.
\newblock In \emph{Proceedings of the 11th Asian Symposium on Programming
  Languages and Systems (APLAS'13)}, Lecture Notes in Computer Science (LNCS
  8301), pages 257--272. Springer, 2013.
\newblock ISBN 978-3-319-03541-3.

\bibitem[Gonthier et~al.(2013)Gonthier, Asperti, Avigad, Bertot, Cohen,
  Garillot, Roux, Mahboubi, O'Connor, Biha, Pasca, Rideau, Solovyev, Tassi, and
  Th{\'e}ry]{Gonthier:2013}
G.~Gonthier, A.~Asperti, J.~Avigad, Y.~Bertot, C.~Cohen, F.~Garillot, S.~L.
  Roux, A.~Mahboubi, R.~O'Connor, S.~O. Biha, I.~Pasca, L.~Rideau, A.~Solovyev,
  E.~Tassi, and L.~Th{\'e}ry.
\newblock \emph{A Machine-Checked Proof of the Odd Order Theorem}, pages
  163--179.
\newblock Springer Berlin Heidelberg, Berlin, Heidelberg, 2013.

\bibitem[Harper and Licata(2007)]{HarperLicata:JFP07}
R.~Harper and D.~R. Licata.
\newblock {Mechanizing Metatheory in a Logical Framework}.
\newblock \emph{{Journal of Functional Programming}}, 17\penalty0
  (4-5):\penalty0 613--673, 2007.

\bibitem[Harper and Pfenning(2005)]{Harper03tocl}
R.~Harper and F.~Pfenning.
\newblock On equivalence and canonical forms in the {LF} type theory.
\newblock \emph{ACM Transactions on Computational Logic}, 6\penalty0
  (1):\penalty0 61--101, 2005.

\bibitem[Harper et~al.(1993)Harper, Honsell, and Plotkin]{Harper93jacm}
R.~Harper, F.~Honsell, and G.~Plotkin.
\newblock A framework for defining logics.
\newblock \emph{Journal of the ACM}, 40\penalty0 (1):\penalty0 143--184,
  January 1993.

\bibitem[Harrison(2009)]{Harrison09a}
J.~Harrison.
\newblock Hol light: An overview.
\newblock In S.~Berghofer, T.~Nipkow, C.~Urban, and M.~Wenzel, editors,
  \emph{TPHOLs}, volume 5674 of \emph{Lecture Notes in Computer Science}, pages
  60--66. Springer, 2009.
\newblock ISBN 978-3-642-03358-2.

\bibitem[Harrison et~al.(2014)Harrison, Urban, and Wiedijk]{Harrison:2014}
J.~Harrison, J.~Urban, and F.~Wiedijk.
\newblock History of interactive theorem proving.
\newblock In J.~H. Siekmann, editor, \emph{Computational Logic}, volume 9
  Suplement C of \emph{Handbook of the History of Logic}, pages 135 -- 214.
  North-Holland, 2014.

\bibitem[Howard(1980)]{Howard80}
W.~A. Howard.
\newblock The formulae-as-types notion of construction.
\newblock In J.~P. Seldin and J.~R. Hindley, editors, \emph{{To H.~B.~Curry;
  Essays on Combinatory Logic, Lambda Calculus and Formalism}}, pages 479 --
  490. Academic Press, 1980.

\bibitem[Hughes et~al.(1996)Hughes, Pareto, and Sabry]{Hughes:POPL96}
J.~Hughes, L.~Pareto, and A.~Sabry.
\newblock Proving the correctness of reactive systems using sized types.
\newblock In \emph{23rd ACM SIGPLAN-SIGACT {S}ymposium on {P}rinciples of
  {P}rogramming {L}anguages (POPL'96)}, pages 410--423, New York, NY, USA,
  1996. ACM.
\newblock ISBN 0-89791-769-3.

\bibitem[Jacob-Rao(2017)]{JacobRao:2017}
R.~Jacob-Rao.
\newblock Well-founded recursion in terms and types.
\newblock Master's thesis, McGill University, 2017.

\bibitem[Jia et~al.(2008)Jia, Vaughan, Mazurak, Zhao, Zarko, Schorr, and
  Zdancewic]{Jia:2008}
L.~Jia, J.~A. Vaughan, K.~Mazurak, J.~Zhao, L.~Zarko, J.~Schorr, and
  S.~Zdancewic.
\newblock Aura: A programming language for authorization and audit.
\newblock In \emph{13th ACM SIGPLAN {I}nternational {C}onference on
  {F}unctional {P}rogramming}, pages 27--38. ACM, 2008.

\bibitem[Lee et~al.(2012)Lee, Oliveira, Cho, and Yi]{Lee2012}
G.~Lee, B.~C. D.~S. Oliveira, S.~Cho, and K.~Yi.
\newblock \emph{GMeta: A Generic Formal Metatheory Framework for First-Order
  Representations}, pages 436--455.
\newblock Springer Berlin Heidelberg, 2012.

\bibitem[Leroy et~al.(2016{\natexlab{a}})Leroy, Doligez, Frisch, Garrigue,
  R\'emy, and Vouillon]{Ocaml:2016}
X.~Leroy, D.~Doligez, A.~Frisch, J.~Garrigue, D.~R\'emy, and J.~Vouillon.
\newblock \emph{The OCaml System Release 4.03 -- Documentation and user's
  manual}.
\newblock Institut National de Recherche en Informatique et en Automatique,
  2016{\natexlab{a}}.

\bibitem[Leroy et~al.(2016{\natexlab{b}})Leroy, Doligez, Frisch, Garrigue,
  R\'emy, and Vouillon]{ocamlManual}
X.~Leroy, D.~Doligez, A.~Frisch, J.~Garrigue, D.~R\'emy, and J.~Vouillon.
\newblock \emph{The OCaml system (release 4.04): Documentation and user's
  manual}.
\newblock Institut National de Recherche en Informatique et en Automatique,
  2016{\natexlab{b}}.

\bibitem[Luo(1994)]{Luo:1994}
Z.~Luo.
\newblock \emph{Computation and reasoning: a type theory for computer science.}
\newblock Oxford University Press, Inc., New York, NY, USA, 1994.

\bibitem[Luther(2001)]{Luther:IJCAR01}
M.~Luther.
\newblock More on implicit syntax.
\newblock In R.~Gore, A.~Leitsch, and T.~Nipkow, editors, \emph{{First
  International Joint Conference on Automated Reasoning (IJCAR'01)}}, Lecture
  Notes in Artificial Intelligence (LNAI) 2083, pages 386--400. Springer, 2001.

\bibitem[Martin-L{\"o}f(1982)]{Martin-Loef79a}
P.~Martin-L{\"o}f.
\newblock Constructive mathematics and computer programming.
\newblock In \emph{6-th International Congress for Logic, Methodology and
  Philosophy of Science, 1979}, pages 153--175. North-Holland, 1982.

\bibitem[Martin-L{\"o}f(1984)]{Martin-Loef84a}
P.~Martin-L{\"o}f.
\newblock \emph{Intuitionistic {Type Theory}}, volume~1 of \emph{Studies in
  Proof Theory Lecture Notes}.
\newblock Bibliopolis, Napoli, 1984.

\bibitem[Miller(1991)]{Miller91iclp}
D.~Miller.
\newblock Unification of simply typed lambda-terms as logic programming.
\newblock In \emph{{8th International Logic Programming Conference}}, pages
  255--269. MIT Press, 1991.

\bibitem[Miller and Palamidessi(1999)]{Miller:1999}
D.~Miller and C.~Palamidessi.
\newblock Foundational aspects of syntax.
\newblock \emph{ACM Comput. Surv.}, 31\penalty0 (3es), Sept. 1999.

\bibitem[Miller and Tiu(2005)]{MillerTiu:TCL05}
D.~Miller and A.~Tiu.
\newblock A proof theory for generic judgments.
\newblock \emph{ACM Transactions on Computational Logic}, 6\penalty0
  (4):\penalty0 749--783, 2005.
\newblock ISSN 1529-3785.

\bibitem[Nanevski et~al.(2008)Nanevski, Pfenning, and Pientka]{Nanevski:ICML05}
A.~Nanevski, F.~Pfenning, and B.~Pientka.
\newblock {C}ontextual {M}odal {T}ype {T}heory.
\newblock \emph{{ACM Transactions on Computational Logic}}, 9\penalty0
  (3):\penalty0 1--49, 2008.

\bibitem[Nipkow et~al.(2002)Nipkow, Paulson, and
  Wenzel]{Nipkow-Paulson-Wenzel:2002}
T.~Nipkow, L.~C. Paulson, and M.~Wenzel.
\newblock \emph{Isabelle/{HOL}: A Proof Assistant for Higher-Order Logic},
  volume 2283 of \emph{Lecture Notes in Computer Science}.
\newblock Springer, 2002.

\bibitem[Nordstr{\"o}m et~al.(1990)Nordstr{\"o}m, Petersson, and
  Smith]{NordstroemPetersonSmith90a}
B.~Nordstr{\"o}m, K.~Petersson, and J.~M. Smith.
\newblock \emph{Programming in Martin-L{\"o}fs Type Theory. An introduction}.
\newblock Clarendon Press, Oxford, 1990.

\bibitem[Norell(2007)]{Norell:phd07}
U.~Norell.
\newblock \emph{Towards a practical programming language based on dependent
  type theory}.
\newblock PhD thesis, Department of Computer Science and Engineering, Chalmers
  University of Technology, Sept. 2007.
\newblock {T}echnical {R}eport 33D.

\bibitem[Paulin-Mohring(1993)]{Paulin-Mohring:TLCA93}
C.~Paulin-Mohring.
\newblock Inductive definitions in the system coq - rules and properties.
\newblock In M.~Bezem and J.~F. Groote, editors, \emph{International Conference
  on Typed Lambda Calculi and Applications(TLCA '93)}, volume 664 of
  \emph{Lecture Notes in Computer Science}, pages 328--345. Springer, 1993.

\bibitem[Paulson(1988)]{Paulson88cade}
L.~C. Paulson.
\newblock Isabelle: The next seven hundred theorem provers.
\newblock In E.~Lusk and R.~Overbeek, editors, \emph{9th International
  Conference on Automated Deduction (CADE-9)}, pages 772--773, Argonne,
  Illinois, 1988. Springer Verlag Lecture Notes in Computer Science (LNCS) 310.
\newblock System abstract.

\bibitem[Pfenning(1989)]{Pfenning89lics}
F.~Pfenning.
\newblock {Elf}: A language for logic definition and verified meta-programming.
\newblock In \emph{Fourth Annual Symposium on Logic in Computer Science}, pages
  313--322, Pacific Grove, California, June 1989. IEEE Computer Society Press.

\bibitem[Pfenning(2001)]{Pfenning:LICS01}
F.~Pfenning.
\newblock Intensionality, extensionality, and proof irrelevance in modal type
  theory.
\newblock In J.~Halpern, editor, \emph{Proceedings of the 16th Annual Symposium
  on Logic in Computer Science (LICS'01)}, pages 221--230, Boston,
  Massachusetts, June 2001. IEEE Computer Society Press.

\bibitem[Pfenning and Davies(2001)]{Pfenning01mscs}
F.~Pfenning and R.~Davies.
\newblock A judgmental reconstruction of modal logic.
\newblock \emph{Mathematical Structures in Computer Science}, 11\penalty0
  (4):\penalty0 511--540, 2001.

\bibitem[Pfenning and Elliott(1988)]{Pfenning88pldi}
F.~Pfenning and C.~Elliott.
\newblock Higher-order abstract syntax.
\newblock In \emph{{ACM SIGPLAN } Symposium on Language Design and
  Implementation (PLDI'88)}, pages 199--208, June 1988.

\bibitem[Pfenning and Paulin-Mohring(1990)]{Pfenning:1990}
F.~Pfenning and C.~Paulin-Mohring.
\newblock \emph{Inductively defined types in the Calculus of Constructions},
  pages 209--228.
\newblock Springer-Verlag, New York, NY, 1990.

\bibitem[Pfenning and Sch{\"u}rmann(1999)]{Pfenning99cade}
F.~Pfenning and C.~Sch{\"u}rmann.
\newblock System description: {Twelf} --- a meta-logical framework for
  deductive systems.
\newblock In H.~Ganzinger, editor, \emph{{16th International Conference on
  Automated Deduction (CADE-16)}}, Lecture Notes in Artificial Intelligence
  (LNAI 1632), pages 202--206. Springer, 1999.

\bibitem[Pientka(2008)]{Pientka:POPL08}
B.~Pientka.
\newblock A type-theoretic foundation for programming with higher-order
  abstract syntax and first-class substitutions.
\newblock In \emph{{35th Annual {ACM} SIGPLAN-SIGACT Symposium on Principles of
  Programming Languages (POPL'08)}}, pages 371--382. ACM Press, 2008.

\bibitem[Pientka(2009)]{Pientka:TOCL09}
B.~Pientka.
\newblock Higher-order term indexing using substitution trees.
\newblock \emph{ACM Transactions on Computational Logic}, 11\penalty0
  (1):\penalty0 1--40, 2009.

\bibitem[Pientka(2013)]{Pientka:JFP13}
B.~Pientka.
\newblock An insider's look at {LF} type reconstruction: {E}verything you
  (n)ever wanted to know.
\newblock \emph{Journal of Functional Programming}, 1\penalty0 (1--37), 2013.

\bibitem[Pientka and Abel(2015)]{Pientka:TLCA15}
B.~Pientka and A.~Abel.
\newblock Structural recursion over contextual objects.
\newblock In T.~Altenkirch, editor, \emph{13th International Conference on
  Typed Lambda Calculi and Applications (TLCA'15)}. Leibniz International
  Proceedings in Informatics (LIPIcs) of Schloss Dagstuhl, 2015.

\bibitem[Pientka and Cave(2015)]{Pientka:CADE15}
B.~Pientka and A.~Cave.
\newblock Inductive {B}eluga:{P}rogramming {P}roofs ({S}ystem {De}scription).
\newblock In A.~P. Felty and A.~Middeldorp, editors, \emph{25th International
  Conference on Automated Deduction (CADE-25)}, Lecture Notes in Computer
  Science (LNCS 9195), pages 272--281. Springer, 2015.

\bibitem[Pientka and Dunfield(2010{\natexlab{a}})]{Pientka:Cover10}
B.~Pientka and J.~Dunfield.
\newblock {C}overing all bases: design and implementation of case analysis for
  contextual objects.
\newblock Technical report, McGill University, 2010{\natexlab{a}}.

\bibitem[Pientka and Dunfield(2010{\natexlab{b}})]{Pientka:IJCAR10}
B.~Pientka and J.~Dunfield.
\newblock Beluga: a framework for programming and reasoning with deductive
  systems ({System Description}).
\newblock In J.~Giesl and R.~Haehnle, editors, \emph{{5th International Joint
  Conference on Automated Reasoning (IJCAR'10)}}, Lecture Notes in Artificial
  Intelligence (LNAI 6173), pages 15--21. Springer-Verlag, 2010{\natexlab{b}}.

\bibitem[Pientka and Pfenning(2003)]{PientkaPfenning:CADE03}
B.~Pientka and F.~Pfenning.
\newblock Optimizing higher-order pattern unification.
\newblock In F.~Baader, editor, \emph{19th International Conference on
  Automated Deduction (CADE-19)}, Lecture Notes in Artificial Intelligence
  (LNAI) 2741, pages 473--487. Springer-Verlag, 2003.

\bibitem[Pierce and Turner(2000)]{Pierce:2000}
B.~C. Pierce and D.~N. Turner.
\newblock Local type inference.
\newblock \emph{ACM Transaction on Programming Languages and Systems},
  22\penalty0 (1):\penalty0 1--44, jan 2000.

\bibitem[Pitts(2003)]{Pitts:2003}
A.~Pitts.
\newblock Nominal logic, a first order theory of names and binding.
\newblock \emph{Information and Computation}, 186\penalty0 (2):\penalty0
  165--193, Nov. 2003.
\newblock ISSN 0890-5401.

\bibitem[Pollack(1990)]{Pollack90}
R.~Pollack.
\newblock Implicit syntax.
\newblock In G.~Huet and G.~Plotkin, editors, \emph{Proceedings of First
  Workshop on Logical Frameworks}, pages 421--434, 1990.

\bibitem[Poswolsky and Sch{\"u}rmann(2009)]{Poswolsky:DelphinDesc08}
A.~Poswolsky and C.~Sch{\"u}rmann.
\newblock System description: {Delphin}---a functional programming language for
  deductive systems.
\newblock In \emph{International Workshop on Logical Frameworks and
  Meta-Languages: Theory and Practice (LFMTP'08)}, volume 228 of
  \emph{Electronic Notes in Theoretical Computer Science (ENTCS)}, pages
  135--141. Elsevier, 2009.

\bibitem[Poswolsky and Sch{\"u}rmann(2008)]{Schuermann:ESOP08}
A.~B. Poswolsky and C.~Sch{\"u}rmann.
\newblock Practical programming with higher-order encodings and dependent
  types.
\newblock In \emph{{17th European Symposium on Programming (ESOP '08)}}, volume
  4960, pages 93--107. Springer, 2008.

\bibitem[Pottier(2006)]{Pottier:2006}
F.~Pottier.
\newblock An overview of {C}$\alpha$ml.
\newblock \emph{Electronic Notes in Theoretical Computer Science}, 148\penalty0
  (2):\penalty0 27 -- 52, 2006.
\newblock Proceedings of the ACM-SIGPLAN Workshop on {ML} (ML'05).

\bibitem[Pottier(2007)]{Pottier:LICS07}
F.~Pottier.
\newblock Static name control for {FreshML}.
\newblock In \emph{22nd {IEEE} Symposium on Logic in Computer Science
  (LICS'07)}, pages 356--365. IEEE Computer Society, July 2007.

\bibitem[Pouillard(2011)]{Pouillard:2011}
N.~Pouillard.
\newblock Nameless, painless.
\newblock In \emph{Proceedings of the 16th ACM SIGPLAN International Conference
  on Functional Programming}, ICFP '11, pages 320--332. ACM, 2011.
\newblock ISBN 978-1-4503-0865-6.

\bibitem[Sch{\"a}fer et~al.(2015)Sch{\"a}fer, Tebbi, and Smolka]{Schaefer2015}
S.~Sch{\"a}fer, T.~Tebbi, and G.~Smolka.
\newblock Autosubst: Reasoning with de bruijn terms and parallel substitutions.
\newblock In C.~Urban and X.~Zhang, editors, \emph{6th International Conference
  of Interactive Theorem Proving ({ITP})}, Lecture Notes in Computer Science
  (9236), pages 359--374. Springer, Aug 2015.

\bibitem[Schrijvers et~al.(2009)Schrijvers, {Peyton Jones}, Sulzmann, and
  Vytiniotis]{Schrijvers:2009}
T.~Schrijvers, S.~{Peyton Jones}, M.~Sulzmann, and D.~Vytiniotis.
\newblock Complete and decidable type inference for gadts.
\newblock In \emph{14th ACM SIGPLAN {I}nternational {C}onference on
  {F}unctional {P}rogramming}, ICFP '09, pages 341--352. ACM, 2009.

\bibitem[Sch{\"u}rmann et~al.(2001)Sch{\"u}rmann, Despeyroux, and
  Pfenning]{Schurmann:TCS01}
C.~Sch{\"u}rmann, J.~Despeyroux, and F.~Pfenning.
\newblock Primitive recursion for higher-order abstract syntax.
\newblock \emph{Theoretical Computer Science}, 266\penalty0 (1-2):\penalty0
  1--57, 2001.

\bibitem[Sheard and Pasalic(2008)]{Sheard:2008}
T.~Sheard and E.~Pasalic.
\newblock Meta-programming with built-in type equality.
\newblock \emph{Electronic Notes in Theoretical Computer Science},
  199:\penalty0 49 -- 65, 2008.
\newblock Proceedings of the Fourth International Workshop on Logical
  Frameworks and Meta-Languages (LFM 2004).

\bibitem[Shinwell et~al.(2003)Shinwell, Pitts, and Gabbay]{Shinwell:ICFP03}
M.~R. Shinwell, A.~M. Pitts, and M.~J. Gabbay.
\newblock {FreshML}: programming with binders made simple.
\newblock In \emph{8th International Conference on Functional Programming
  (ICFP'03)}, pages 263--274. ACM Press, 2003.

\bibitem[Slind and Norrish(2008)]{Slind:2008}
K.~Slind and M.~Norrish.
\newblock \emph{A Brief Overview of HOL4}, pages 28--32.
\newblock Springer Berlin Heidelberg, Berlin, Heidelberg, 2008.

\bibitem[Stansifer and Wand(2014)]{Stansifer:2014}
P.~Stansifer and M.~Wand.
\newblock Romeo: A system for more flexible binding-safe programming.
\newblock In \emph{Proceedings of the 19th {ACM} {SIGPLAN} International
  Conference on Functional Programming}, ICFP '14, pages 53--65, 2014.

\bibitem[{The Coq Developement Team}(2016)]{CoqManual}
{The Coq Developement Team}.
\newblock \emph{The Coq Proof Assistant Reference Manual v. 8.6.1}.
\newblock Institut National de Recherche en Informatique et en Automatique,
  2016.

\bibitem[Urban(2008)]{Urban:JAR08}
C.~Urban.
\newblock Nominal techniques in {I}sabelle/{HOL}.
\newblock \emph{Journal of Automated Reasoning}, 40\penalty0 (4):\penalty0
  327--356, 2008.

\bibitem[Washburn and Weirich(2008)]{WashburnWeirich:JFP06}
G.~Washburn and S.~Weirich.
\newblock Boxes go bananas: Encoding higher-order abstract syntax with
  parametric polymorphism.
\newblock \emph{Journal of Functional Programming}, 18\penalty0 (01):\penalty0
  87--140, 2008.

\bibitem[Watkins et~al.(2002)Watkins, Cervesato, Pfenning, and
  Walker]{Watkins02tr}
K.~Watkins, I.~Cervesato, F.~Pfenning, and D.~Walker.
\newblock A concurrent logical framework {I}: Judgments and properties.
\newblock Technical Report CMU-CS-02-101, Department of Computer Science,
  Carnegie Mellon University, 2002.

\bibitem[Weirich et~al.(2017)Weirich, Voizard, de~Amorim, and
  Eisenberg]{Weirich:2017}
S.~Weirich, A.~Voizard, P.~H.~A. de~Amorim, and R.~A. Eisenberg.
\newblock A specification for dependent types in {H}askell.
\newblock \emph{Proc. ACM Program. Lang.}, 1\penalty0 (ICFP):\penalty0
  31:1--31:29, Aug. 2017.
\newblock ISSN 2475-1421.

\bibitem[Westbrook et~al.(2011)Westbrook, Frisby, and Brauner]{Westbrook:2011}
E.~Westbrook, N.~Frisby, and P.~Brauner.
\newblock Hobbits for {H}askell: a library for higher-order encodings in
  functional programming languages.
\newblock In \emph{4th ACM Symposium on Haskell (Haskell'11)}, pages 35--46.
  ACM, 2011.

\bibitem[Xi(2004)]{Xi03:ATS}
H.~Xi.
\newblock Applied type system.
\newblock In \emph{TYPES 2003}, volume 3085 of \emph{Lecture Notes in Computer
  Science}, pages 394--408. Springer, 2004.

\bibitem[Xi(2007)]{Xi:JFP}
H.~Xi.
\newblock Dependent ml an approach to practical programming with dependent
  types.
\newblock \emph{Journal of Functional Programming}, 17:\penalty0 215--286, 3
  2007.

\bibitem[Xi and Pfenning(1999)]{XI99popl}
H.~Xi and F.~Pfenning.
\newblock Dependent types in practical programming.
\newblock In \emph{26th ACM SIGPLAN-SIGACT Symposium on Principles of
  Programming Languages (POPL'99)}, pages 214--227. ACM Press, 1999.

\bibitem[Xi et~al.(2003)Xi, Chen, and Chen]{Xi03:guarded}
H.~Xi, C.~Chen, and G.~Chen.
\newblock Guarded recursive datatype constructors.
\newblock In \emph{30th ACM SIGPLAN-SIGACT Symposium on Principles of
  Programming Languages (POPL '03)}, pages 224--235. ACM Press, 2003.

\bibitem[Zenger(1997)]{Zenger:TCS97}
C.~Zenger.
\newblock Indexed types.
\newblock \emph{Theoretical Computer Science}, 187\penalty0 (1-2):\penalty0
  147--165, 1997.

\bibitem[Ziliani et~al.(2015)Ziliani, Dreyer, Krishnaswami, Nanevski, and
  Vafeiadis]{ziliani:2015}
B.~Ziliani, D.~Dreyer, N.~R. Krishnaswami, A.~Nanevski, and V.~Vafeiadis.
\newblock Mtac: A monad for typed tactic programming in coq.
\newblock \emph{Journal of Functional Programming}, 25, 2015.

\end{thebibliography}
